\pdfoutput=1

\RequirePackage[l2tabu,orthodox]{nag}

\documentclass[11pt,letterpaper,reqno,oneside]{article}

\DeclareTextFontCommand{\emph}{\slshape}

\makeatletter
\renewcommand{\paragraph}{%
	\@startsection{paragraph}{4}%
	{\z@}{1.75ex \@plus 1ex \@minus .2ex}{-0.7em}%
	{\normalfont\normalsize\bfseries}%
}
\makeatother

\usepackage[T1]{fontenc}
\usepackage{lmodern}
\usepackage[utf8]{inputenc}

\usepackage[american,german,british]{babel}

\usepackage[activate={true,nocompatibility}]{microtype}

\usepackage{csquotes}

\usepackage[backend=biber,autolang=other,style=apa,maxcitenames=5,sortcites=false]{biblatex}
\DeclareLanguageMapping{british}{british-apa}
\DeclareLanguageMapping{american}{american-apa}

\usepackage{amsmath}
\usepackage{amssymb}
\usepackage{amsfonts}
\usepackage{amsthm}
\usepackage{mathtools}
\usepackage{stmaryrd}

\let\originalleft\left
\let\originalright\right
\renewcommand{\left}{\mathopen{}\mathclose\bgroup\originalleft}
\renewcommand{\right}{\aftergroup\egroup\originalright}

\usepackage{graphicx}
\usepackage{tikz,pgfplots}
\pgfplotsset{compat=1.10}
\usetikzlibrary{shapes.multipart}
\usetikzlibrary{decorations.markings}

\usepackage[title]{appendix}

\usepackage{datetime}
\newdateformat{datestyle}{ {\THEDAY} \monthname[\THEMONTH] \THEYEAR }

\usepackage{enumerate}
\usepackage{enumitem}
\setlist[enumerate,1]{label=(\arabic*)}
\setlist[itemize,1]{label=--}
\setlist[itemize,2]{label=--}
\setlist[itemize,3]{label=--}
\setlist[itemize,4]{label=--}

\usepackage{xcolor}
\usepackage{verbatim}
\usepackage{gensymb}
\usepackage{rotating}
\usepackage{subcaption}
\usepackage{floatpag}
\floatpagestyle{plain}
\usepackage{marvosym}

\usepackage{hyphenat}
\theoremstyle{definition}

\newtheorem{theorem}{Theorem}
\newtheorem{proposition}{Proposition}
\newtheorem{lemma}{Lemma}
\newtheorem{corollary}{Corollary}
\newtheorem{remark}{Remark}
\newtheorem{observation}{Observation}
\newtheorem{example}{Example}
\newtheorem{definition}{Definition}
\newtheorem*{claim}{Claim}

\newtheoremstyle{named}
	{\topsep}					
	{\topsep}					
	{}							
	{0pt}						
	{\bfseries}					
	{}							
	{5pt plus 1pt minus 1pt}	
	{\thmnote{#3}}				
\theoremstyle{named}
\newtheorem{namedthm}{}

\renewcommand{\qedsymbol}{$\blacksquare$}

\usepackage{xpatch}
\xpatchcmd{\proof}{\itshape}{\proofheadfont}{}{}
\newcommand{\proofheadfont}{\slshape}

\usepackage{hyperref}
\hypersetup{pdfborder=0 0 0}

\usepackage[nameinlink]{cleveref}
\crefname{page}{p.}{pp.}
\crefname{equation}{equation}{equations}
\crefname{section}{section}{sections}
\crefname{subsection}{section}{sections}
\crefname{subsubsection}{section}{sections}
\crefname{appsec}{appendix}{appendices}
\crefname{supplsec}{supplemental appendix}{supplemental appendices}
\crefname{footnote}{footnote}{footnotes}
\crefname{figure}{figure}{figures}
\crefname{table}{table}{tables}
\crefname{theorem}{theorem}{theorems}
\crefname{proposition}{proposition}{propositions}
\crefname{lemma}{lemma}{lemmata}
\crefname{corollary}{corollary}{corollaries}
\crefname{remark}{remark}{remarks}
\crefname{observation}{observation}{observations}
\crefname{example}{example}{examples}
\crefname{fact}{fact}{facts}
\crefname{definition}{definition}{definitions}
\crefname{assumption}{assumption}{assumptions}
\crefname{exercise}{exercise}{exercises}
\crefname{notation}{notation}{notation}
\crefname{claim}{claim}{claims}
\crefname{conjecture}{conjecture}{conjectures}

\newcommand{\eps}{\varepsilon}

\newcommand{\dd}{\mathrm{d}}

\DeclareMathOperator*{\argmin}{arg\,min}

\DeclareMathOperator*{\argmax}{arg\,max}

\DeclareMathOperator*{\interior}{int}

\DeclareMathOperator*{\cl}{cl}

\DeclareMathOperator*{\co}{co}

\DeclareMathOperator*{\supp}{supp}

\DeclareMathOperator*{\cav}{cav}

\newcommand{\R}{\mathbf{R}}

\newcommand{\Q}{\mathbf{Q}}

\newcommand{\C}{\mathbf{C}}

\newcommand{\N}{\mathbf{N}}

\newcommand{\1}{\boldsymbol{1}}

\newcommand{\union}{\cup}

\newcommand{\intersect}{\cap}

\newcommand{\Union}{\bigcup}

\DeclarePairedDelimiter\abs{\lvert}{\rvert}

\DeclarePairedDelimiter\norm{\lVert}{\rVert}

\newcommand*{\xslant}[2][76]{%
	\begingroup
	\sbox0{#2}%
	\pgfmathsetlengthmacro\wdslant{\the\wd0 + cos(#1)*\the\wd0}%
	\leavevmode
	\hbox to \wdslant{\hss
		\tikz[
			baseline=(X.base),
			inner sep=0pt,
			transform canvas={xslant=cos(#1)},
		] \node (X) {\usebox0};%
		\hss
		\vrule width 0pt height\ht0 depth\dp0 %
	}%
	\endgroup
}
\makeatletter
\newcommand*{\xslantmath}{}
\def\xslantmath#1#{%
	\@xslantmath{#1}%
}
\newcommand*{\@xslantmath}[2]{%
	\ensuremath{%
		\mathpalette{\@@xslantmath{#1}}{#2}%
	}%
}
\newcommand*{\@@xslantmath}[3]{%
	\xslant#1{$#2#3\m@th$}%
}
\makeatother

\makeatletter
\def\namedlabel#1#2{\begingroup
	#2%
	\def\@currentlabel{#2}%
	\phantomsection\label{#1}\endgroup
}
\makeatother

\makeatletter
\let\save@mathaccent\mathaccent
\newcommand*\if@single[3]{%
	\setbox0\hbox{${\mathaccent"0362{#1}}^H$}%
	\setbox2\hbox{${\mathaccent"0362{\kern0pt#1}}^H$}%
	\ifdim\ht0=\ht2 #3\else #2\fi
	}
\newcommand*\rel@kern[1]{\kern#1\dimexpr\macc@kerna}
\newcommand*\widebar[1]{\@ifnextchar^{{\wide@bar{#1}{0}}}{\wide@bar{#1}{1}}}
\newcommand*\wide@bar[2]{\if@single{#1}{\wide@bar@{#1}{#2}{1}}{\wide@bar@{#1}{#2}{2}}}
\newcommand*\wide@bar@[3]{%
	\begingroup
	\def\mathaccent##1##2{%
	  \let\mathaccent\save@mathaccent
	  \if#32 \let\macc@nucleus\first@char \fi
	  \setbox\z@\hbox{$\macc@style{\macc@nucleus}_{}$}%
	  \setbox\tw@\hbox{$\macc@style{\macc@nucleus}{}_{}$}%
	  \dimen@\wd\tw@
	  \advance\dimen@-\wd\z@
	  \divide\dimen@ 3
	  \@tempdima\wd\tw@
	  \advance\@tempdima-\scriptspace
	  \divide\@tempdima 10
	  \advance\dimen@-\@tempdima
	  \ifdim\dimen@>\z@ \dimen@0pt\fi
	  \rel@kern{0.6}\kern-\dimen@
	  \if#31
	    \overline{\rel@kern{-0.6}\kern\dimen@\macc@nucleus\rel@kern{0.4}\kern\dimen@}%
	    \advance\dimen@0.4\dimexpr\macc@kerna
	    \let\final@kern#2%
	    \ifdim\dimen@<\z@ \let\final@kern1\fi
	    \if\final@kern1 \kern-\dimen@\fi
	  \else
	    \overline{\rel@kern{-0.6}\kern\dimen@#1}%
	  \fi
	}%
	\macc@depth\@ne
	\let\math@bgroup\@empty \let\math@egroup\macc@set@skewchar
	\mathsurround\z@ \frozen@everymath{\mathgroup\macc@group\relax}%
	\macc@set@skewchar\relax
	\let\mathaccentV\macc@nested@a
	\if#31
	  \macc@nested@a\relax111{#1}%
	\else
	  \def\gobble@till@marker##1\endmarker{}%
	  \futurelet\first@char\gobble@till@marker#1\endmarker
	  \ifcat\noexpand\first@char A\else
	    \def\first@char{}%
	  \fi
	  \macc@nested@a\relax111{\first@char}%
	\fi
	\endgroup
}
\makeatother

\usepackage{sidecap}
\sidecaptionvpos{figure}{c}

\usepackage{calc}

\makeatletter
	\newcommand{\hyperdest}[1]{\Hy@raisedlink{\hypertarget{#1}{}}}
\makeatother

\newcommand{\mpc}[1]{\text{MPC}(#1)}
\newcommand{\mps}[1]{\text{MPS}(#1)}
\newcommand{\mpsI}[2]{\text{MPS}_{#1}(#2)}

\newcommand{\lessvexu}[1]{u:#1}
\newcommand{\lessvexv}[1]{v:#1}

\newcommand{\lowerthan}{\quad\text{is lower than}\quad}
\newcommand{\higherthan}{\quad\text{is higher than}\quad}
\newcommand{\notlowerthan}{\quad\text{is not lower than}\quad}

\newcommand{\customspace}{\quad\;}

\newcommand{\strictlyhigherthan}{\customspace\text{\parbox{\widthof{higher than}}{\centering is strictly higher than}}\customspace}

\newcommand{\notstrictlyhigherthan}{\customspace\text{\parbox{\widthof{is not strictly}}{\centering is not strictly higher than}}\customspace}

\newcommand{\Blowerthan}{\quad\text{is B-lower than}\quad}
\newcommand{\Bhigherthan}{\quad\text{is B-higher than}\quad}

\newcommand{\Bstrictlyhigherthan}{\customspace\text{\parbox{\widthof{B-higher than}}{\centering is strictly B-higher than}}\customspace}

\newcommand{\Bnotstrictlyhigherthan}{\customspace\text{\parbox{\widthof{B-higher than}}{\centering is not strictly B-higher than}}\customspace}

\newif\ifviolation
\newif\ifbig
\newif\ifu
\newif\ifp
\newif\ifX
\newif\ifv
\newif\ifq
\newif\ifsshape
\newif\ifconcav
\newif\ifaffine
\newif\ifinterpol
\newif\ifdrawu
\newif\ifdrawv
\newif\ifdrawb
\newif\ifdrawl

\title{\scshape The comparative statics of persuasion%
\thanks{We are grateful for comments from four referees, Da\u{g}han Carlos Akkar, Itai Arieli, Yunus Aybas, Gabe Carroll, Tommaso Denti, Piotr Dworczak, Michael Eldar, Pia Ennuschat, Matteo Escudé, Alkis Georgiadis-Harris, Alexis Ghersengorin, Ben Golub, Duarte Gonçalves, Olivier Gossner, Ian Jewitt, Andreas Kleiner, Paul Klemperer, Peter Klibanoff, Anton Kolotilin, Annie Liang, Elliot Lipnowski, Thomas Mariotti, Laurent Mathevet, Meg Meyer, Wojciech Olszewski, Paula Onuchic, Marco Ottaviani, Alessandro Pavan, Antonio Penta, Jacopo Perego, Philipp Strack, Bruno Strulovici, Alex Wolitzky, Kun Zhang and audiences at Arizona State, Berlin, Bonn, Caltech, Cambridge, Cergy, City, CREST, EUI, Glasgow, LSE, Mannheim, Michigan, Nottingham, Oxford, Pompeu Fabra, Surrey, Tilburg, UCL, VSET and several conferences. Curello acknowledges support from the German Research Foundation (DFG) through CRC TR 224 (Project B02). Sinander acknowledges support from the Leverhulme Trust in the form of a Philip Leverhulme Prize.}}

\author{%
Gregorio Curello \\
University of Mannheim
\and
Ludvig Sinander \\
University of Oxford}

\date{27 November 2026}

\makeatletter
	\AtBeginDocument{ \hypersetup{
		pdftitle = {The comparative statics of persuasion},
		pdfauthor = {Gregorio Curello and Ludvig Sinander}
		} }
\makeatother

\begin{document}

\maketitle

\begin{abstract}
	In the persuasion model, apart from a few special cases, comparative statics has been an open question.
	We answer it,
	delineating which shifts of the sender's interim payoff
	lead her optimally to choose a more informative signal.
	Our first theorem identifies a coarse notion of `increased convexity' that we show characterises those shifts of the sender's interim payoff that lead her optimally to choose \emph{no less} informative signals.
	To strengthen this conclusion to `\emph{more} informative'
	requires further assumptions:
	our second theorem identifies the necessary and sufficient condition on the sender's interim payoff, which strictly generalises the convex--concave (`S') shape commonly imposed in the literature.
	We apply our results to delegation, costly information acquisition and quality certification.
\end{abstract}

\section{Introduction}
\label{sec:intro}

Many economic design problems boil down to choosing a distribution $F \in \mpc{F_0}$ to maximise $\int u \dd F$, where $u : [0,1] \to \R$ is a given upper semi-continuous function, $F_0$ is a given distribution, and `$\mpc{F_0}$' denotes the set of all mean-preserving contractions of $F_0$. Examples include persuasion \parencite{KamenicaGentzkow2011} and delegation \parencite{Holmstrom1977,Holmstrom1984}.

A large literature has sought to characterise and solve for the optimal distribution $F$.
In this paper, we ask a different question: comparative statics. In particular, which shifts of model primitives, particularly of the function $u$, lead to a more dispersed optimal choice $F$? By `more dispersed', we mean `mean-preserving spread', which corresponds in persuasion to choosing a more informative signal, and in delegation to granting the agent greater discretion.

For concreteness, we focus on persuasion, and treat other interpretations (such as delegation) separately as applications.
Recall that the persuasion model features an uncertain state of the world, whose distribution $F_0$ is called the \emph{prior,} and a character called the \emph{sender.}
The sender flexibly designs what will and won't be revealed about the state, by choosing a \emph{signal.}

Motivated by applications, we primarily focus on the case in which the sender's interim payoff depends on only one moment of the posterior belief---without loss, the mean.
This `single-moment' assumption is satisfied by several important economic models, and is therefore maintained in much of the recent literature on persuasion in economic environments.%
	\footnote{\label{footnote:mean-meas_ex}For example, \textcite{RoeslerSzentes2017,RavidRoeslerSzentes2022,DoganHu2022,ArmstrongZhou2022,HwangKimBoleslavsky2023,BergemannHeumannMorris2022,MenschRavidThereze2026,Thereze2023adv,Kreutzkamp2023}. See also \textcite{Bergemannetal2022}.}
We write $u(m)$ for the sender's interim payoff when the posterior mean is $m$. (The function $u$ is a reduced-form description of a downstream interaction.)

It turns out that the random posterior mean induced (via Bayes's rule) by a signal must have a distribution $F$ belonging to $\mpc{F_0}$, and conversely that any element of $\mpc{F_0}$ is induced by some signal. Hence the sender's problem is to choose $F \in \mpc{F_0}$ to maximise $\int u \dd F$. The model's primitives are $u$ and $F_0$, and our question is which shifts of $u$ lead the sender to choose a more informative distribution $F$, whatever the prior $F_0$. By `a more informative distribution', we mean a mean-preserving spread.

Following the comparative-statics literature, we divide our comparative-statics question into two parts: we first ask which shifts of the sender's interim payoff $u$ lead her optimally to choose \emph{not strictly less} informative distributions, and then ask what is the maximal domain of interim payoffs $u$ on which these shifts actually lead \emph{(weakly) more} informative distributions to be chosen.%
	\footnote{`Not strictly less' does not imply `more' since `less informative than' is a \emph{partial} order.}
This division of comparative-statics questions into a `non-decreasing' question (characterising payoff shifts) and an `increasing' question (identifying a payoff domain) is fundamental to the theory of comparative statics, e.g. \textcite{Topkis1978,MilgromShannon1994,QuahStrulovici2009,QuahStrulovici2007extensions}.%
	\footnote{\label{footnote:mcs_properties}In these classic papers, the first (`non-decreasing') question is answered by payoff-shift notions called (respectively) increasing differences, single-crossing differences, and interval dominance, while the second question is answered by domain restrictions called (respectively) supermodularity, quasi-supermodularity, and I-quasi-supermodularity.}
It tends to yield economically interpretable and easily applicable conditions.%
	\footnote{There are two papers \parencite{CheKimKojima2021,AmirRietzke2025} which eschew the canonical bipartite division, instead identifying `omnibus' conditions for comparative statics which are formally weaker, but harder to interpret or verify in applications.}

\Cref{theorem:nondecr} shows that a coarse notion of `less convex than'
characterises `non-decreasing' comparative statics in the persuasion model:
coarsely more convex interim payoffs $u$ are exactly those
that lead \emph{not strictly less} informative distributions to be chosen by the sender, whatever the prior $F_0$.

Our main result, \Cref{theorem:incr}, characterises what more is needed
to obtain `increasing' comparative statics:
it identifies a property of interim payoffs $u$
that is necessary and sufficient
for coarse-convexity shifts to cause \emph{(weakly) more} informative distributions to be chosen, whatever the prior.
This property, the \emph{crater property,} is a geometric condition
that strictly generalises the convex--concave (`S') shape
commonly assumed in the literature.

The crater property is demanding. Nevertheless, we show that it is satisfied in some applications, permitting new comparative-statics conclusions to be drawn about the problems of designing quality ratings \parencite{Zapechelnyuk2020} and delegation \parencite[e.g.][]{KolotilinZapechelnyuk2025}. On the other hand, our theorems suggest that comparative-statics conclusions for costly information acquisition \parencite[e.g.][]{RavidRoeslerSzentes2022} will prove elusive.

What if the `single-moment' assumption is dropped? And what if it is the prior $F_0$ (rather than the interim payoff $u$) which shifts? We answer both of these questions with `impossibility' results, which assert that comparative-statics conclusions cannot be drawn except in trivial cases.

We conclude by showing how our results also yield comparative statics for the `dual' problem of choosing $F \in \mps{F_0}$ to maximise $\int u \dd F$, where `$\mps{F_0}$' denotes the set of all mean-preserving \emph{spreads} of $F_0$. This problem also has many applications \parencite[e.g.][]{KleinerMoldovanuStrack2021}.

\subsection{Relation to the persuasion literature}
\label{sec:intro:lit_pers}

The persuasion model was introduced by \textcite{KamenicaGentzkow2011}, with precedents in \textcite{AumannMaschler1995,BrocasCarillo2007,RayoSegal2010}. A rich literature has sought to characterise and solve for optimal signals \parencite[e.g.][]{Kolotilin2014,Kolotilin2018,DworczakMartini2019,KleinerMoldovanuStrack2021,ArieliBabichenkoSmorodinskyYamashita2023,DworczakKolotilin2023,KolotilinCorraoWolitzky2023}.

Comparative statics has been an open question, except in three special cases. Each of these concerns particular shifts of interim payoffs, and additionally restricts attention either to `S'-shaped interim payoffs \parencite{KolotilinMylovanovZapechelnyuk2022}, to binary priors \parencite{Yoder2022}, or both \parencite{GitmezMolavi2023,GitmezMolaviSonin2025}. We discuss these special cases in §\ref{sec:mcs_incr:s} below.

\textcite{CerreiavioglioDillenbergerOrtoleva2020} give a comparative-statics result for a variant of the persuasion model in which the sender has non-expected-utility preferences. We discuss this result in \cref{app:robust}.

There are classic comparative-statics results for the problem of choosing $F \in \mathcal{C}$ to maximise $\int u \dd F$ when the constraint set $\mathcal{C}$ is totally ordered by a suitable stochastic order.%
	\footnote{A \emph{stochastic order} is a partial order $\lesssim$ on the set of all distributions. $\mathcal{C}$ is \emph{totally ordered by $\lesssim$} iff for any $F,G \in \mathcal{C}$, either $F \lesssim G$ or $G \lesssim F$.}
In particular, if $\mathcal{C}$ is totally ordered by mean-preserving single-crossing,%
	\footnote{$F$ is \emph{mean-preserving single-crossed} by $G$ iff $\int x F(\dd x) = \int x G(\dd x)$ and for any $x<y$ in $[0,1]$, $F(x) \geq G(x)$ implies $F(y) \geq G(y)$.}
then the optimal choice $F$ becomes more informative when $u$ becomes more convex in the conventional `Arrow--Pratt' sense \parencite[][p.~345]{DiamondStiglitz1974},%
	\footnote{We are grateful to a referee for drawing our attention to this paper. The underlying mathematical result appears also in e.g. \textcite{KarlinNovikoff1963,Jewitt1987}.}
and if $\mathcal{C}$ is totally ordered by `less informative than', then the optimum becomes more informative when a convex function is added to $u$ \parencite[][§3]{Jewitt1986}. These beautiful results neither imply nor are implied by our theorems, since our constraint set is $\mathcal{C}=\mpc{F_0}$, which is not totally ordered except if $F_0$ is degenerate. Conventional increased convexity of $u$ and adding a convex function to $u$ are both special cases of our notion of a `coarse-convexity shift' of $u$.

As mentioned, some delegation problems are equivalent to persuasion (see §\ref{sec:appl} below for details). The delegation literature began with \textcite{Holmstrom1977,Holmstrom1984}. To our knowledge, all existing comparative-statics results in this literature either rely on parametric assumptions,%
	\footnote{For example, \textcite[§4]{Holmstrom1984}, \textcite[§4.4]{Kleiner2026}, and \textcite[§3.2]{Zeng2026} place parametric assumptions on preferences and/or distributions, and \textcite[Proposition~5]{AmadorWerningAngeletos2006} assume that the agent and principal have identical preferences modulo a scalar `mis-alignment' parameter.}
or else concern the simplified problem in which the principal is constrained to only use `interval delegation' mechanisms.%
	\footnote{E.g. Holmström's (\citeyear{Holmstrom1984}, Theorem~3) `ally principle', \textcite{Armstrong1995}, and \textcite[Proposition~4]{KartikKleinerVanweelden2021}. The ally principle fails when the principal can use arbitrary mechanisms \parencite[§6.4]{AlonsoMatouschek2008}.}

\subsection{Relation to the comparative-statics literature}
\label{sec:intro:lit_mcs}

The comparative-statics literature \parencite[e.g.][]{Topkis1978,MilgromShannon1994,QuahStrulovici2009} asks, for any problem in which an agent chooses an action $a$ from a partially ordered set $\mathcal{A}$, which shifts of the agent's objective function $U : \mathcal{A} \to \R$ lead her 	optimally to choose a higher action.%
	\footnote{A detail: the literature actually restricts attention to action sets $\mathcal{A}$ whose partial order has a \emph{lattice} structure.
	This proviso is satisfied in the persuasion model (see \cref{app:proofs:product}).}
In the persuasion model, the sender's action is a distribution $F$ drawn from the set of all inducible (by a signal) posterior-mean distributions, ordered by `less informative than', and her objective function is $U(F) \coloneqq \int u \dd F$, where $u(m)$ is the sender's interim payoff in case the posterior mean is $m$.

As mentioned above, the literature features two types of properties: notions of \emph{shift} of the payoff $U$, and \emph{domain restrictions} on the payoff $U$ (classic properties of each type are listed in \cref{footnote:mcs_properties} above). The role of shifts is to encourage the agent to take higher actions; formally, a shift ensures that the agent chooses a \emph{not strictly lower} action.%
	\footnote{See e.g. \textcite[Proposition~5]{QuahStrulovici2007extensions} and \textcite{AndersonSmith2024}.}
The second type of property identifies a domain of payoffs $U$ on which shifts in fact lead the agent to choose a \emph{(weakly) higher} action.

Our first theorem identifies the correct notion of `shift' in the persuasion model: coarse-convexity shifts of the interim payoff $u$ are exactly those which lead the sender optimally to choose a \emph{not strictly less} informative signal (whatever the prior). The proof uses comparative-statics theory: in particular, we show that coarse-convexity shifts of $u$ produce interval-dominance shifts of $U$, then invoke Proposition~5 in \textcite{QuahStrulovici2007extensions}.%
	\footnote{Definitions and results from \textcite{QuahStrulovici2007extensions} are reviewed in \cref{app:mcs_lit}.}

Our main theorem identifies the maximal domain of interim payoffs $u$ on which coarse-convexity shifts lead the sender to choose a \emph{more} informative signal. The answer (the crater property) is restrictive but non-trivial, with several applications. Our proof exploits the specific structure of the persuasion model. Existing comparative-statics theory is of no help, since even the weakest domain-restriction property in the literature, Quah and Strulovici's (\citeyear{QuahStrulovici2007extensions}) `I-quasi-supermodularity', is satisfied by the objective function $U(F) = \int u \dd F$ only if $u$ is either concave or convex (i.e. the trivial cases).%
	\footnote{We prove this in \cref{app:mcs_lit}. These cases are trivial because if $u$ is concave (convex) then `no information' (`full information') is optimal whatever the prior.}

The literature contains one broadly analogous result: \textcite{AndersonSmith2024} identify conditions for comparative statics in the \textcite{Becker1973} sorting model beyond the supermodular (positive assortative matching) case. Standard comparative-statics theory is inapplicable, so the authors argue from first principles, using the particular structure of the sorting model.

There is no contradiction with the literature's results on supermodularity-type properties being necessary (in a sense) for comparative statics. Rather, those necessity results are weaker and subtler than they may seem, as explained by \textcite{AmirRietzke2025} and exemplified by our main theorem.

Our main theorem is a proof of concept, showing that non-trivial comp\-arative-statics conclusions can sometimes be drawn beyond the standard supermodularity-type payoff domains. This matters because in our experience, these supermodularity-type properties often fail in economic applications, except if actions are totally ordered (e.g. scalars). Our main theorem (and Anderson and Smith's) may thus be viewed as a step toward a widely applicable comparative-statics theory allowing for \emph{partially} ordered actions.

\subsection{Roadmap}
\label{sec:intro:roadmap}

We describe the persuasion model in the next section.
In §\ref{sec:mcs}, we characterise `non-decreasing' comparative statics in terms of a coarse notion of `less convex than' (\Cref{theorem:nondecr}).
We then (§\ref{sec:mcs_incr}) characterise `increasing' comparative statics (\Cref{theorem:incr}, our main result).
In §\ref{sec:limits}, we consider shifts of the prior and dropping the `single-moment' assumption.
We study applications in §\ref{sec:appl}. We conclude (§\ref{sec:mps}) by addressing the dual mean-preserving-\emph{spread} problem.

\section{The persuasion model}
\label{sec:model}

There is an uncertain state of the world,
formally a random variable taking values in a bounded interval $\left[ \underline{x}, \bar{x} \right]$. We assume without loss of generality that $\underline{x}=0$ and $\bar{x}=1$.
We shall use the term \emph{distribution} to refer to CDFs $[0,1] \to [0,1]$.
We write $F_0$ for the distribution of the state, and refer to it as `the prior (distribution)'.
For two distributions $F$ and $G$,
recall that $F$ is a mean-preserving contraction of $G$ if and only if
\begin{equation*}
	\int_0^x F
	\leq \int_0^x G
	\quad \text{for every $x \in [0,1]$, with equality at $x=1$,}
\end{equation*}
or equivalently iff $\int \psi \dd F \leq \int \psi \dd G$
for every convex $\psi : [0,1] \to \R$.%
	\footnote{See e.g. \textcite[§3.A]{ShakedShanthikumar2007}. Another equivalent condition is $\int \psi \dd F \leq \int \psi \dd G$ for every convex \emph{and continuous} $\psi : [0,1] \to \R$.}
We write `$\mpc{G}$' for the set of all mean-preserving contractions of $G$.

A sender chooses a \emph{signal,} i.e. a random variable jointly distributed with the state.%
	\footnote{Formally, a signal is $(M,\pi)$, where $M$ is a non-empty compact metric space and $\pi$ is a Borel measurable map $[0,1] \to \Delta(M)$, where $\Delta(M)$ is the set of all Borel probabilities on $M$, with the topology of weak convergence.
	The interpretation is that $M$ is a set of messages, and that $\pi(x) \in \Delta(M)$ is the distribution of messages sent if the state is $x \in [0,1]$.}
Given a signal, each signal realisation induces a posterior belief via Bayes's rule, whose expectation we call the \emph{posterior mean.}
Each signal thus induces a random posterior mean, with some distribution.
Call a distribution \emph{feasible (given $F_0$)} iff it is the posterior-mean distribution induced by some signal.
As is well-known, a distribution is feasible given $F_0$ if and only if it belongs to $\mpc{F_0}$.%
	\footnote{See e.g. \textcite{SinanderStrassen} for a proof. Versions of this result were stated (or at least alluded to) by \textcite{HardyLittlewoodPolya1929,Blackwell1951,Kolotilin2014}\nocite{Kolotilin2018}.}

The sender's (interim) payoff at a given realised posterior belief
is assumed to depend only on its mean:
her payoff at posterior mean $m \in [0,1]$ is $u(m)$,
where $u : [0,1] \to \R$ is upper semi-continuous.
Her problem is to choose among the feasible distributions $F$ to maximise her expected payoff $\int u \dd F$.

\begin{remark}
	\label{remark:mean_meas}
	Our assumption that only the mean matters is motivated by applications, where it is common for payoffs to depend on a single moment of the posterior distribution---without loss, the mean.%
		\footnote{This is without loss because if payoffs depend on the interim expectation of $f(X)$, where $X$ is the state of the world and $f : [0,1] \to \R$ is continuous, then we may re-define the state of the world to be $Y \coloneqq f(X)$.}
	This `single-moment' assumption holds in much of the recent literature on persuasion in classic economic environments (see \cref{footnote:mean-meas_ex} above).
	We relax it in §\ref{sec:limits:multi} below.
\end{remark}

\subsection{Interpretation}
\label{sec:model:interpretation}

The interim payoff $u : [0,1] \to \R$ is a reduced-form object, capturing the (expected) payoff consequences for the sender of whatever downstream interaction takes place after her chosen signal realises.

In the simplest case, the downstream interaction involves a (single) \emph{receiver} taking an action.
Formally, there is a non-empty set $\mathcal{A}$ of actions, and the sender's and receiver's interim payoffs $U_S(a,m)$ and $U_R(a,m)$ depend on the chosen action $a \in \mathcal{A}$ and on the mean $m \in [0,1]$ of their (posterior) belief about the state. The assumption that only the mean matters means that both parties have a risk-neutral attitude to uncertainty about the state (for each $a \in \mathcal{A}$, their ex-post payoffs $x \mapsto \widebar{u}_S(a,x)$ and $x \mapsto \widebar{u}_R(a,x)$ are affine). When the posterior mean is $m \in [0,1]$, the receiver chooses action $A(m) \in \argmax_{a \in \mathcal{A}} U_R(a,m)$, so the sender's interim payoff is $u(m) \coloneqq U_S( A(m), m )$.
The assumption that $u$ is upper semi-continuous can be micro-founded by assuming that ($A$ is such that) the receiver breaks ties in the sender's favour.%
	\footnote{The micro-foundation assumes that there is a topology on $\mathcal{A}$ which makes $\mathcal{A}$ compact, $U_S$ (jointly) upper semi-continuous and $U_R$ (jointly) continuous. Finiteness of $\mathcal{A}$ is a sufficient condition.}
This simple sender--receiver model of a downstream interaction nests some but not all of our applications in §\ref{sec:appl} below.

Our analysis will be robust to the details of the downstream interaction, giving conditions directly on the interim payoff $u$ that are necessary and sufficient for comparative statics. These conditions may then be checked in particular applications, as we illustrate in §\ref{sec:appl} below.

\subsection{Informativeness}
\label{sec:model:informativeness}

Following \textcite{GanuzaPenalva2010}, we compare informativeness according to the variability of the random posterior mean.

\begin{definition}
	\label{definition:informativeness}
	For distributions $F$ and $G$,
	we call $F$ \emph{less informative} than $G$
	if and only if $F$ is a mean-preserving contraction of $G$.
\end{definition}

This captures informativeness in the spirit of Blackwell:
a more informative distribution is precisely one that is preferred ex-ante
by every expected-utility decision-maker who cares about the state only through its mean, i.e. who has a risk-neutral attitude to uncertainty about the state.%
	\footnote{Explicitly,
	$F$ is less informative than $G$ iff
	for any non-empty (action) set $\mathcal{A}$ and any bounded-above (payoff) $U : \mathcal{A} \times [0,1] \to \R$ such that $U(a,\cdot)$ is affine (i.e. risk-neutral) for each $a \in \mathcal{A}$,
	we have $\int \sup_{a \in \mathcal{A}} U(a,m) F( \dd m )
	\leq \int \sup_{a \in \mathcal{A}} U(a,m) G( \dd m )$.
	This is because a function $\psi : [0,1] \to \R$ is convex and continuous iff it equals $m \mapsto \sup_{a \in \mathcal{A}} U(a,m)$ for some such $\mathcal{A}$ and $U$.}
It is not equivalent to the standard Blackwell order, which demands preference by \emph{all} expected-utility decision-makers. Our notion seems appropriate since the sender herself is assumed to care only about the mean.%
	\footnote{In \cref{app:robust}, we explore how our results change if the Blackwell order is used instead.}

Since there need not be a unique optimal posterior-mean distribution,
comparative statics requires comparing \emph{sets} of distributions.
We handle this in standard fashion by using the \emph{weak set order:}
for two sets $S,S' \subseteq \mpc{F_0}$,
we call $S$ \emph{lower} than $S'$
if and only if for any $F \in S$ and $G \in S'$,
there is a distribution in $S'$ that is more informative than $F$,
and there is a distribution in $S$ that is less informative than $G$.%
	\footnote{The literature often instead uses the \emph{strong} set order; we discuss this in \cref{app:mcs_lit}.}
We say that $S$ is \emph{strictly lower} than $S'$
if and only if $S$ is lower than $S'$
and $S'$ is not lower than $S$.
Clearly for singletons $S = \{F\}$ and $S' = \{G\}$, $S$ is (strictly) lower than $S'$ if and only if ($F \neq G$ and) $F$ is less informative than $G$.
Finally, we call $S'$ \emph{(strictly) higher} than $S$
if and only if $S$ is (strictly) lower than $S'$.

\section{\texorpdfstring{`}{‘}Non-decreasing\texorpdfstring{'}{’} comparative statics}
\label{sec:mcs}

In this section,
we ask a preliminary `non-decreasing' comparative-statics question:
which shifts of the sender's interim payoff $u$
ensure that she does \emph{not} choose a \emph{strictly less} informative distribution?
Intuition suggests that convexity should be decisive,
since a `more convex' $u$ embodies a greater liking for informative distributions.
We validate this intuition, by defining a new \emph{coarse} notion of comparative convexity and proving that it is the necessary and sufficient condition for `non-decreasing' comparative statics.

\begin{definition}
	\label{definition:lessconvex_coarse}
	For functions $u,v : [0,1] \to \R$,
	we say that $u$ is \emph{coarsely less convex} than $v$
	if and only if for any $x<y$ in $[0,1]$ such that
	\begin{equation}
		u( \alpha x + (1-\alpha) y )
		\leq \alpha u(x) + (1-\alpha) u(y) 
		\label{eq:C_u}
		\tag{$\lessvexu{\alpha}$}
	\end{equation}
	holds for every $\alpha \in (0,1)$, we also have
	\begin{equation}
		v( \alpha x + (1-\alpha) y )
		\leq \alpha v(x) + (1-\alpha) v(y) 
		\label{eq:C_v}
		\tag{$\lessvexv{\alpha}$}
	\end{equation}
	for every $\alpha \in (0,1)$,
	and furthermore any $\alpha \in (0,1)$ at which the inequality \eqref{eq:C_u} is strict is also one at which \eqref{eq:C_v} is strict.
\end{definition}

We call $v$ \emph{coarsely more convex} than $u$
if and only if
$u$ is coarsely less convex than $v$.
By inspection, the relation `coarsely less convex than' is transitive and reflexive, but not anti-symmetric.

There is a simple sufficient condition (let `$\co$' denote the convex hull):

\begin{lemma}
	\label{lemma:CLC_suff}
	For functions $u,v : [0,1] \to \R$,
	if $v(x) = \Phi(u(x),x)$ for every $x \in [0,1]$, where $\Phi : \co(u([0,1])) \times [0,1] \to \R$ is convex with $\Phi(\cdot,x)$ strictly increasing for every $x \in (0,1)$,
	then $u$ is coarsely less convex than $v$.
\end{lemma}

\begin{proof}
	Fix $x<y$ in $[0,1]$ and $\alpha \in (0,1)$ such that $u( \alpha x + (1-\alpha) y ) \leq \mathrel{(<)} \alpha u(x) + (1-\alpha) u(y)$.
	Since $\alpha x + (1-\alpha) y \in (0,1)$, $\Phi(\cdot,\alpha x + (1-\alpha) y)$ is strictly increasing, so
	\begin{multline*}
		v( \alpha x + (1-\alpha) y )
		\leq \mathrel{(<)} 
		\Phi\Bigl( \alpha u(x) + (1-\alpha) u(y), \alpha x + (1-\alpha) y \Bigr)
		\\
		\leq \alpha v(x) + (1-\alpha) v(y) ,
	\end{multline*}
	where the latter inequality follows from the convexity of $\Phi$.
\end{proof}

Thus $u$ is coarsely less convex than $v$ whenever $u$ is less convex than $v$ in the conventional sense: $v = \phi \circ u$ for some convex and strictly increasing function $\phi : \co(u([0,1])) \to \R$ (to see this, take $\Phi(k,x) \coloneqq \phi(k)$ in \Cref{lemma:CLC_suff}). Such a shift from $u$ to $v$ occurs whenever the sender becomes less risk-averse \parencite{Pratt1964} or gains access to an outside option \parencite{CurelloSinanderWhitmeyer2025}.
A different sufficient condition for $u$ to be coarsely less convex than $v$ is that $v = u + \psi$ for some convex $\psi : [0,1] \to \R$ (take $\Phi(k,x) \coloneqq k + \psi(x)$ in \Cref{lemma:CLC_suff}). Such shifts from $u$ to $v$, which feature in the literature on costly information acquisition, occur when the stakes are raised or the sender gains flexibility (\cite{Whitmeyer2024,Delara2025}; see §\ref{sec:appl:costly} below).
In case $u$ and $v$ are twice continuously differentiable, the latter sufficient condition is equivalent to $u'' \leq v''$, and if $u'>0<v'$ or $u'<0>v'$ then the former sufficient condition is equivalent to $u''/\abs*{u'} \leq v''/\abs*{v'}$.
For later reference, we summarise these findings in a corollary:

\begin{corollary}
	\label{corollary:CLC_suffsuff}
	For $u,v : [0,1] \to \R$,
	$u$ is coarsely less convex than $v$ whenever either
	(i)~$v = \phi \circ u$ for some convex and strictly increasing $\phi : \co(u([0,1])) \to \R$ or
	(ii)~$v = u + \psi$ for some convex $\psi : [0,1] \to \R$.
\end{corollary}

We show in \cref{app:more_clc} that \Cref{lemma:CLC_suff} is nearly tight, by giving a partial converse, as well as an exact (but more complicated) characterisation of coarse-convexity-increasing transformations $\Phi$. In the same appendix, we extend \Cref{lemma:CLC_suff} to directly address the sender--receiver interpretation described in §\ref{sec:model:interpretation}, by identifying when increased alignment of interests between the sender and the receiver leads the interim payoff $u$ to become coarsely more convex.

The following result characterises `non-decreasing' comparative statics.

\begin{theorem}
	\label{theorem:nondecr}
	Let $u,v : [0,1] \to \R$ be upper semi-continuous.
	If $u$ is coarsely less convex than $v$,
	then for any distribution $F_0$,
	\begin{equation}
		\argmax_{F \in \mpc{F_0}} \int u \dd F
		\notstrictlyhigherthan
		\argmax_{F \in \mpc{F_0}} \int v \dd F .
		\label{eq:mcs}
		\tag{$\star$}
	\end{equation}
	Conversely, if \eqref{eq:mcs} holds for every distribution $F_0$,
	then $u$ must be coarsely less convex than $v$.
\end{theorem}

To interpret \eqref{eq:mcs}, consider the (generic) case in which the maximisers are unique:
\begin{equation*}
	\argmax_{F \in \mpc{F_0}} \int u \dd F
	= \{G\}
	\quad \text{and} \quad
	\argmax_{F \in \mpc{F_0}} \int v \dd F
	= \{H\} .
\end{equation*}
In this case, property \eqref{eq:mcs} reads `$G$ is not strictly more informative than $H$'. In other words, \emph{either} $G$ is less informative than $H$, \emph{or} $G$ is neither more nor less informative than $H$.

The proof is in \cref{app:proofs:pf_thm_nondecr}.
The second half (necessity) is straightforward.
For sufficiency, we show that if $u$ is coarsely less convex than $v$,
then $F \mapsto \int u \dd F$ is \emph{interval-dominated}
by $F \mapsto \int v \dd F$;
a standard comparative-statics theorem due to \textcite{QuahStrulovici2007extensions} then implies that \eqref{eq:mcs} must hold for every distribution $F_0$.

Our argument for interval dominance runs as follows.
What interval dominance demands is, roughly speaking, that if a distribution $G$ is optimal for $u$ given some prior $F_0$, then any less informative distribution $F$ is dis-preferred by $v$: $\int v \dd F \leq \int v \dd G$. We prove this in three cases of increasing generality.%
	\footnote{We thank Ian Jewitt for suggesting this tripartite argument.}
\emph{Case~1: $F$ is a point mass and $G$ is binary.} In this case, $\int v \dd F \leq \int v \dd G$ follows directly from $u$ being coarsely less convex than $v$.
\emph{Case~2: $F$ is a point mass and $G$ is arbitrary.} Write $m$ for the (common) mean of $F$ and $G$. $G$ is a weighted average of binary distributions with mean $m$, by Choquet's theorem and the fact that all extreme points of the space of mean-$m$ distributions are binary \parencite{Karr1983}. Hence $\int v \dd F \leq \int v \dd G$ by Case~1.
\emph{Case~3: both $F$ and $G$ are arbitrary.} By Blackwell's theorem, $G$ can be obtained from $F$ by first drawing a point $x$ from $F$, then applying an ($x$-contingent) mean-preserving spread. By Case~2, each of these mean-preserving spreads has non-negative expected value under $v$. Hence, integrating out $x$ under its distribution $F$, we have $\int v \dd F \leq \int v \dd G$.

The proof in \cref{app:proofs:pf_thm_nondecr} formalises the above argument, taking care of the measure-theoretic niceties and adding arguments to handle the gap between the rough definition of interval dominance given above and the (more involved) exact definition.

\begin{remark}
	\label{remark:thm1_id}
	A looser sense of why the first (sufficiency) part of \Cref{theorem:nondecr} is true can be obtained by assuming that $v-u$ is convex (a special case of $u$ being coarsely less convex than $v$, by \Cref{corollary:CLC_suffsuff}). Then, writing $U(F) \coloneqq \int u \dd F$ and $V(F) \coloneqq \int v \dd F$ for each distribution $F$, the difference $V-U$ is increasing (i.e. $(V-U)(F) \leq (V-U)(G)$ whenever $F$ is less informative than $G$, a special case of $F \mapsto \int u \dd F$ being interval-dominated by $F \mapsto \int v \dd F$). Hence for any distribution $F_0$, if $F \in \argmax_{\mpc{F_0}} U$ is \emph{more} informative than $G \in \argmax_{\mpc{F_0}} V$, then $V(F) \leq V(G) = U(G) + (V-U)(G) \leq U(F) + (V-U)(F) = V(F)$, whence $G \in \argmax_{\mpc{F_0}} U$ and $F \in \argmax_{\mpc{F_0}} V$. So \eqref{eq:mcs} holds for every distribution $F_0$.
\end{remark}


\section{\texorpdfstring{`}{‘}Increasing\texorpdfstring{'}{’} comparative statics}
\label{sec:mcs_incr}

In this section,
we ask what is required
for a shift of the sender's interim payoff
to lead her optimally to choose a \emph{more} informative distribution.
By \Cref{theorem:nondecr}, it is necessary that the payoff become coarsely more convex.

This condition is not sufficient
if all upper semi-continuous interim payoffs $u,v : [0,1] \to \R$
and all prior distributions $F_0$
are considered.
(We will see this explicitly in §\ref{sec:mcs_incr:pf_thm2} below, in a proof.)
Our question is thus:
on what restricted domain of interim payoffs $u,v$ and/or priors $F_0$
are coarse-convexity shifts \emph{sufficient}
for `increasing' comparative statics?

Our main result (\Cref{theorem:incr}) describes the maximal domain of interim payoffs on which `increasing' comparative statics holds.
Concretely, it identifies the condition on the interim payoff $u$ that is necessary and sufficient
for `increasing' comparative statics to hold between $u$ and any coarsely more convex $v$, given any prior $F_0$.
This condition is called the \emph{crater property.}
We also give the corresponding `decreasing' result (\Cref{proposition:W_down}): for an interim payoff $v$, comparative statics holds between $u$ and $v$ for any coarsely \emph{less} convex $u$ and any prior $F_0$ if and only if $v$ is either concave or convex.

We furthermore exhibit a suitable domain of priors.
A \emph{binary} prior is one with a two-point support;
under such a prior, the state is effectively binary.
We show (\Cref{proposition:ido_binary}) that for `increasing' comparative statics between payoffs $u$ and $v$ to hold
across all \emph{binary} priors $F_0$, it is both necessary and sufficient that $u$ be coarsely less convex than $v$.

The crater property is demanding. A key message of this section is therefore that comparative statics are often highly prior-sensitive. On the other hand, the crater property does often hold in applications, allowing comparative-statics conclusions to be drawn, as we show in §\ref{sec:appl} below.

Finally (§\ref{sec:mcs_incr:s}), we specialise our results to the three cases in which comparative statics have previously been obtained in the literature \parencite{KolotilinMylovanovZapechelnyuk2022,Yoder2022,GitmezMolavi2023}.

\subsection{Regularity}
\label{sec:mcs_incr:regularity}

We shall mostly restrict attention to well-behaved payoffs:

\begin{definition}
	\label{definition:regular}
	Call a function $u : [0,1] \to \R$ \emph{regular}
	iff
	(i)~$u$ is continuous and possesses a continuous and bounded derivative $u' : (0,1) \to \R$, and
	(ii)~$[0,1]$ may be partitioned into finitely many intervals,
	on each of which $u$ is either strictly convex or strictly concave.
\end{definition}

Part~(ii) of regularity rules out affine segments. This is merely for simplicity: we show in \cref{app:incr_affine} that our results below remain true (though with much longer proofs) when regularity is weakened by replacing part~(ii) with the requirement that on each of the finitely many intervals, $u$ is either strictly convex, strictly concave, \emph{or affine.} This weakening of property~(ii) is one of the maintained assumptions of \textcite{DworczakMartini2019}.

For a regular function $u : [0,1] \to \R$,
we extend the derivative $u' : (0,1) \to \R$
to a continuous map $[0,1] \to \R$
by letting $u'(0)$ be the right-hand derivative of $u$ at $0$
and $u'(1)$ the left-hand derivative at $1$.

\subsection{Maximal domain of interim payoffs}
\label{sec:mcs_incr:w}

The following property will be the key to comparative statics.

\begin{definition}
	\label{definition:crater}
	A regular function $u : [0,1] \to \R$ satisfies the
	\emph{crater property}
	if and only if
	for any $x < y < z < w$ in $[0,1]$
	such that $u$ is concave on $[x,y]$ and $[z,w]$ and convex on $[y,z]$,
	the tangents to $u$ at $x$ and at $w$
	cross at coordinates $(X,Y) \in \R^2$
	satisfying $Y \leq u(X)$.
\end{definition}

The property is illustrated in \Cref{fig:W}.
Loosely, it requires that any `valley' of $u$
be sufficiently steep-walled, wide, and shallow---like a crater.

\begin{figure}
	\centering
	\begin{subfigure}{0.48\textwidth}
		\centering
		\violationtrue
		\bigfalse
		\utrue
		\ptrue
		\Xtrue
		\vfalse
		\qfalse
		\input{tikz/W}
		\caption{A violation of the crater property: the tangents to $u$ at $x$ and at $w$ cross above the graph of $u$.}
		\label{fig:W:violationtrue}
	\end{subfigure}
	\hfill
	\begin{subfigure}{0.48\textwidth}
		\centering
		\violationfalse
		\bigfalse
		\utrue
		\ptrue
		\Xtrue
		\vfalse
		\qfalse
		\input{tikz/W}
		\caption{Not a violation of the crater property: the tangents to $u$ at $x$ and at $w$ cross below the graph of $u$.}
		\label{fig:W:violationfalse}
	\end{subfigure}
	\caption{Illustration of the crater property.}
	\label{fig:W}
\end{figure}

The crater property is demanding. It rules out `double tangencies', for example: there cannot exist $x<y<z<w$ in $[0,1]$ and an affine $p : [0,1] \to \R$ such that $p \geq u$ on $[x,w]$ and $p=u$ on $\{y,z\}$. It follows that the crater property implies \emph{affine-closedness,} a property which \textcite[Theorem~3]{DworczakMartini2019} showed is necessary and sufficient for a monotone partitional distribution to be optimal given any atomless convex-support prior. (Monotone partitional distributions are those which are induced by a particular type of deterministic signal; the formal definition is in §\ref{sec:mcs_incr:pf_thm2} below.)

Nevertheless, there are important classes of interim payoffs which satisfy the crater property.
Call a function $u : [0,1] \to \R$ \emph{S-shaped} iff it is continuous and, for some $x \in [0,1]$, convex on $[0,x]$ and concave on $[x,1]$.
Examples include the logistic function and all unimodal CDFs.
All regular S-shaped functions satisfy the crater property (vacuously).
The same is true of \emph{reverse-S-shaped} functions, meaning those $u : [0,1] \to \R$ such that $x \mapsto u(1-x)$ is S-shaped; for example, the logit and probit functions.
S-shaped interim payoffs $u$ are important in the literature, both because they arise naturally in applications (see §\ref{sec:appl} below) and because they permit a sharp characterisation of optimal posterior-mean distributions (see §\ref{sec:mcs_incr:pf_thm2} below).

More generally, the crater property is satisfied by all \emph{W-shaped} functions, meaning continuous functions that are convex on $[0,x]$ and on $[y,1]$ and concave on $[x,y]$, for some $x \leq y$ in $[0,1]$.

\begin{theorem}
	\label{theorem:incr}
	Let $u : [0,1] \to \R$ be regular.
	If $u$ satisfies the crater property,
	then for every regular $v : [0,1] \to \R$ that is coarsely more convex than $u$
	and every distribution $F_0$,
	\begin{equation}
		\argmax_{F \in \mpc{F_0}} \int u \dd F
		\lowerthan
		\argmax_{F \in \mpc{F_0}} \int v \dd F .
		\label{eq:mcs_incr}
		\tag{$\star\star$}
	\end{equation}
	Conversely, if \eqref{eq:mcs_incr} holds
	for every regular $v$ that is coarsely more convex than $u$
	and every distribution $F_0$,
	then $u$ satisfies the crater property.
\end{theorem}

In short, the crater property is necessary and sufficient for coarse-convexity shifts to yield `increasing' comparative statics. Since the crater property is demanding, this may be viewed as a negative result: comparative statics is prior-sensitive, so that conclusions often cannot be drawn robustly across all priors $F_0$.
On the other hand, the crater property does hold in several applications, and in such cases \Cref{theorem:incr} delivers comparative statics. We treat several such applications in §\ref{sec:appl} below.

We show in \cref{app:incr_affine} that regularity may be weakened to allow affine segments, at the cost of a longer proof. In \cref{app:robust}, we show that \Cref{theorem:incr} remains true as stated if instead of contemplating every regular $v : [0,1] \to \R$ that is \emph{coarsely} more convex than $u$, we consider only those which are more convex than $u$ in the conventional sense ($v = \phi \circ u$ for some convex and strictly increasing function $\phi : \co(u([0,1])) \to \R$), or only those such that $v-u$ is convex. Finally, \Cref{theorem:incr} also remains true if only `nice' priors $F_0$ are considered; in particular, the proof in §\ref{sec:mcs_incr:pf_thm2} below of the converse (necessity) part uses only atomless convex-support priors $F_0$, and can easily be tweaked to focus on e.g. full-support or finite-support priors.

The question answered by \Cref{theorem:incr} has a symmetric counterpart: what is the necessary and sufficient condition on an interim payoff $v$ for every coarse-convexity \emph{decrease} (to some $u$) to yield a \emph{decrease} of informativeness, regardless of the prior distribution $F_0$? The answer is as follows.

\begin{proposition}
	\label{proposition:W_down}
	Let $v : [0,1] \to \R$ be regular. Then \eqref{eq:mcs_incr} holds for every regular $u : [0,1] \to \R$ that is coarsely \emph{less} convex than $v$ and every distribution $F_0$ if, and only if, $v$ is either concave or convex.
\end{proposition}

In other words, `decreasing' comparative statics are \emph{highly} prior-sensitive: a coarse-convexity decrease from $v$ yields decreased informativeness whatever the prior $F_0$ only in the trivial cases of a concave $v$ (when full pooling is optimal) or a convex $v$ (when full revelation is optimal). Like \Cref{theorem:incr}, \Cref{proposition:W_down} remains true if only `nice' priors $F_0$ are considered.

The asymmetry between \Cref{theorem:incr} and \Cref{proposition:W_down} is explained in §\ref{sec:mcs_incr:pf_thm2} below, where we show that the same logic which implies the necessity of the crater property of $u$ in \Cref{theorem:incr} also implies the necessity of concavity or convexity of $v$ in \Cref{proposition:W_down}.

\begin{remark}
	\label{remark:crater_local}
	The crater property is local in character: it can be checked by separately inspecting each maximal interval $[x,w]$ on which $u$ is concave--convex--concave.
	This is noteworthy since it contrasts with the global character of the persuasion problem, in which a change of $u$ on an interval $I \subseteq [0,1]$ can impact optimal information-provision about states far from $I$.
\end{remark}

\subsection{Proof of \texorpdfstring{\Cref{theorem:incr} and \Cref{proposition:W_down}}{Theorem \ref{theorem:incr} and Proposition \ref{proposition:W_down}}}
\label{sec:mcs_incr:pf_thm2}

Before proving \Cref{theorem:incr}, we review a few standard concepts and results. Fix a (prior) distribution $F_0$ that is atomless with convex support. A distribution $F \in \mpc{F_0}$ is called \emph{monotone partitional (with respect to $F_0$)} iff $[0,1]$ may be partitioned into proper intervals, where each of these intervals $I$ satisfies $\int_0^{\inf I} (F-F_0) = 0 = \int_0^{\sup I} (F-F_0)$ and is either an \emph{$F$-pooling} interval ($F$ is a point mass on $I$, i.e. constant but for a single jump) or an \emph{$F$-revealing} interval ($F=F_0$ on $I$). If (and only if) a distribution is monotone partitional, there is a simple deterministic signal which induces it: when the state $x \in [0,1]$ falls in a pooling interval $I$, the signal conveys `the state belongs to $I$', and when $x$ falls in a revealing interval, the signal conveys `the state is $x$'. If $F,G \in \mpc{F_0}$ are both monotone partitional, then $F$ is less informative than $G$ iff every $G$-pooling interval is contained in an $F$-pooling interval.

If $u : [0,1] \to \R$ is regular and has the crater property, and the prior $F_0$ is atomless with convex support, then there is exactly one distribution which is optimal for $u$ given $F_0$, and this distribution is monotone partitional \parencite[][Theorem~7]{DworczakKolotilin2023}. Furthermore, a distribution $F$ that is monotone partitional with respect to $F_0$ is optimal for $u$ given $F_0$ if and only if there exists a convex function $p : [0,1] \to \R$ such that $p \geq u$, $p=u$ $F$-a.e., and $p$ is affine on each $F$-pooling interval \parencite[Theorems~1--2]{DworczakMartini2019}.%
	\footnote{Proof of the sufficiency (`if') part: $F$ is optimal because for any $G \in \mpc{F_0}$,
	$\int u \dd F
	= \int p \dd F
	= \int p \dd F_0
	\geq \int p \dd G
	\geq \int u \dd G$,
	where the steps hold, respectively,
	since $u=p$ $F$-a.e.,
	since $p$ is affine on each $F$-pooling interval,
	since $p$ is convex and $G$ belongs to $\mpc{F_0}$,
	and since $p \geq u$.}
(The function $p$ solves the \emph{dual} of the persuasion problem.) One implication is that if $I \subseteq \supp(F_0)$ is an $F$-revealing interval, then $u$ is strictly convex on $I$.%
	\footnote{We have $p=u$ on $\interior(I)$ as $p=u$ $F$-a.e., which since $p$ is convex and $u$ is regular means that $u$ is strictly convex on $I$.}

\begin{proof}[Proof of the converse (necessity) part]
	Suppose that $u$ is regular and violates the crater property;
	we shall find a regular and coarsely more convex $v : [0,1] \to \R$
	and a distribution $F_0$ such that \eqref{eq:mcs_incr} fails.

	\begin{figure}
		\centering
		\violationtrue
		\bigtrue
		\utrue
		\ptrue
		\Xtrue
		\vtrue
		\qfalse
		\input{tikz/W}
		\caption{Proof of the converse part of \Cref{theorem:incr}. When $u$ violates the crater property, there are tangents to $u$ at some $x$ and $w$ which cross above the graph of $u$. These $x$ and $w$ may be chosen so that two tangents' upper envelope $p$ is convex. And we may choose an S-shaped function $v$ which exceeds (equals) $u$ to the left (right) of the tangents' crossing.}
		\label{fig:Wproof}
	\end{figure}

	Since $u$ violates the crater property
	(refer to \Cref{fig:Wproof}),
	there are $x' < x < y < z < w < w'$ in $[0,1]$
	such that $u$ is strictly concave on $[x',y]$ and $[z,w']$ and strictly convex on $[y,z]$,
	and there is a convex function $p : [0,1] \to \R$
	and an $X \in (y,z)$
	such that
	$p$ is affine on $[x',X]$ and on $[X,w']$,
	weakly exceeds $u$ on $[x',w']$,
	strictly exceeds $u$ at $X$,
	and is tangent to $u$ at $x$ and at $w$.
	Let $F_0$
	be a distribution that is atomless
	with support $[x',w']$, and
	\begin{equation*}
		\frac{1}{F_0(X)} \int_0^X \xi F_0(\dd \xi) = x
		\quad \text{and} \quad
		\frac{1}{1-F_0(X)} \int_X^1 \xi F_0(\dd \xi) = w .
	\end{equation*}
	
	Let $F \in \mpc{F_0}$ be monotone partitional with pooling intervals $[0,X]$ and $[X,1]$. (This distribution is induced by a signal which conveys whether or not the state exceeds $X$, and conveys nothing else.) The distribution $F$ is optimal for $u$, because $p$ is convex with $p \geq u$ on $[x',w']$ and $p=u$ $F$-a.e., and is affine on $[0,X]$ and on $[X,1]$. ($F$ is in fact \emph{uniquely} optimal for $u$.)

	Since $u'$ is bounded,
	we may choose a regular $v : [0,1] \to \R$
	that coincides with $u$ on $[X,1]$
	and that weakly exceeds $u$ and is strictly convex on $[0,X]$
	(refer to \Cref{fig:Wproof}).
	It is easily seen that $v$ is coarsely more convex than $u$.
	On $[x',w'] = \supp(F_0)$, the function $v$ is S-shaped, so satisfies the crater property. Denote by $G$ the distribution that is (uniquely) optimal for $v$ given $F_0$. We must show that $F$ is \emph{not} less informative than $G$, so that \eqref{eq:mcs_incr} fails.

	As $v$ is S-shaped on $\supp(F_0)$, `upper censorship' is optimal by Kolotilin's (\citeyear{Kolotilin2014}, p.~14) well-known result: $G$ is monotone partitional with revealing interval $[0,a]$ and pooling interval $[a,1]$, where $a \in (0,1)$ satisfies
	\begin{equation*}
		\frac{ v(b) - v(a) }{ b-a } = v'(b) ,
		\quad \text{where} \quad
		b = \frac{1}{1-F_0(a)} \int_a^1 \xi F_0( \dd \xi ) .\footnote{\label{footnote:upper_censor_opt}To see directly why $G$ is optimal, let $q : [0,1] \to \R$ match $v$ on $[0,a]$ and match $x \mapsto v(a) + (x-a)v'(b)$ on $[a,1]$, and note that $q$ is convex with $q \geq v$ on $\supp(F_0)$ and $q = v$ $G$-a.e., and that $q$ is affine on $[a,1]$.}
	\end{equation*}
	A simple graphical argument shows that
	$a$ must be strictly smaller than $X$.%
		\footnote{\label{footnote:a_less_X}We have $b < w$,
		since $b \geq w$ would imply both $a < X$ (for tangency, as $p>u=v$ at $X$)
		and $a \geq X$ (as $b$ is the mean conditional on the event $[a,1]$).
		Then since 		
		$b$ ($w$) equals the mean conditional on the event $[a,1]$ ($[X,1]$),
		we must have $a<X$.}
	Thus the $G$-pooling interval $[a,1]$ is \emph{not} contained in either of the $F$-pooling intervals, $[0,X]$ and $[X,1]$, so $F$ fails to be less informative than $G$.
\end{proof}

The logic of the above necessity argument also explains why \Cref{proposition:W_down} is true. Call a function $u : [0,1] \rightarrow \R$ \emph{M-shaped} iff it is continuous and, for some $x \leq y$ in $[0,1]$, concave on $[0,x]$ and on $[y,1]$ and convex on $[x,y]$.

\begin{proof}[Sketch proof of \Cref{proposition:W_down}]
	The `if' part is close to obvious. For (the contra-positive of) the `only if' part, the key observation is that if $v$ is neither concave nor convex, then it must be S- or reverse-S-shaped on some interval, in which case we may find a regular, M-shaped and coarsely less convex $u : [0,1] \to \R$ as in \Cref{fig:Wproof}, so that \eqref{eq:mcs_incr} fails by the logic of the proof of the converse (necessity) part of \Cref{theorem:incr}. The full proof is in \cref{app:proofs:pf_thm_incr_down}.
\end{proof}

To prove the first (sufficiency) part of \Cref{theorem:incr}, we must show that given any regular $u,v : [0,1] \to \R$ and any distribution $F_0$, if $u$ satisfies the crater property and is coarsely less convex than $v$, then \eqref{eq:mcs_incr} holds. We prove this below under two auxiliary assumptions: that $v$ (also) satisfies the crater property, and that $F_0$ is atomless with convex support. The argument below forms the basis of the complete proof in \cref{app:proofs:pf_thm_incr}, which dispenses with the auxiliary assumptions (at the cost of significant additional work).

The argument makes no use of general-purpose `increasing' comparative-statics results \parencite[e.g.][]{QuahStrulovici2009,QuahStrulovici2007extensions}, because such results are inapplicable to the persuasion model except in trivial cases, as discussed in §\ref{sec:intro:lit_mcs} above. Instead, we proceed from first principles, exploiting the particular structure of the persuasion model, via the dual (the convex function $p$).

\begin{proof}[Proof of the first (sufficiency) part under auxiliary assumptions]
	Fix regular $u,v : [0,1] \to \R$ and a distribution $F_0$. Assume that $u$ and $v$ both satisfy the crater property, that $u$ is coarsely less convex than $v$, and that $F_0$ is atomless with convex support. Then there are monotone partitional distributions $F$ and $G$ which are uniquely optimal given $F_0$ for $u$ and for $v$, respectively. We must show that $F$ is less informative than $G$, so that \eqref{eq:mcs_incr} holds.

	To that end, fix an arbitrary $G$-pooling interval $[x',z'] \subseteq \supp(F_0)$; we must show that $[x',z']$ is contained in an $F$-pooling interval. Denote by $y'$ the (by definition of `pooling', unique) atom of $G$ in $[x',z']$.

	It cannot be that $y'$ belongs to an open $F$-revealing interval $I$, since then $u$ would be strictly convex on $I$, in which case $v$ (being coarsely more convex) would also be strictly convex on $I$, contradicting the optimality of $G$ for $v$ given $F_0$ (as $G$ could be strictly improved by moving probability mass `outward' from $y'$). Hence there is an $F$-pooling interval $[x,z]$ which contains $y'$. We will show that $[x,z] \supseteq [x',z']$.

	Write $y$ for the (unique) atom of $F$ in $[x,z]$. It suffices to rule out each of the following four scenarios, since at least one of these must hold if $[x,z] \nsupseteq [x',z']$: (i)~$x'<x$ and $y' \leq y$, (ii)~$x'<x$ and $y < y'$, (iii)~$z<z'$ and $y \leq y'$, and (iv)~$z<z'$ and $y' < y$. If scenario~(ii) holds, then $z<z'$,%
		\footnote{It cannot be that $z' \leq z$ because this (together with $x'<x$) would imply that $y' = [F_0(z')-F_0(x')]^{-1} \int_{x'}^{z'} \xi F_0(\dd \xi) \leq [F_0(z)-F_0(x)]^{-1} \int_{x}^{z} \xi F_0(\dd \xi) = y$.}
	so scenario~(iii) also holds. Similarly, (iv) implies (i). It thus suffices to rule out scenarios~(i) and (iii). We focus on (i), omitting the symmetric argument for (iii).

	So, suppose toward a contradiction that $x'<x$ and $y' \leq y$. Since $y'$ belongs to $[x,z]$, we have $x'<x \leq y' \leq y$.

	\begin{namedthm}[Claim.]
		\label{claim:incr_suff_s}
		There exists an $\widebar{x} \in (x',x)$ such that $u$ is S-shaped on $\left( \widebar{x}, y \right)$. Furthermore, $u$ is strictly exceeded on $(x,y)$ by the (unique) affine function $[0,1] \to \R$ which intersects $u$ at $x$ and at $y$.
	\end{namedthm}

	\begin{figure}
		\centering
		\begin{subfigure}{\textwidth}
			\begin{minipage}[c]{0.73\textwidth}
				\centering
				\drawutrue
				\drawvfalse
				\drawbfalse
				\drawlfalse
				\ifdrawu\drawvfalse\drawbfalse\drawlfalse\fi
\ifdrawv\drawufalse\drawbfalse\drawlfalse\fi
\ifdrawb\drawufalse\drawvfalse\drawlfalse\fi
\ifdrawl\drawufalse\drawvfalse\drawbfalse\fi

\begin{tikzpicture}[scale=1, line cap=round,
	declare function={
		functionu(\x,\min,\infl,\A,\B,\C)
			= \A*min(\x-\min,0)^4
			+ \B*max(\x-\min,0)^2
			- \C*max(\x-\infl,0)^2;
		functionq(\x,\xmin)
			= 2.0 - 0.1*(\x-\xmin);
		functionv(\x,\xmin,\tang)
			= functionq(\x,\xmin)
			- 0.285*(\x-\tang)^2
			* (1 + 0.45*sin(deg(5*(\x-\xmin))));
		functionp(\x,\kink)
			= 1.25 + 0.4*max((\x-\kink),-0.5*(\x-\kink));
		functionull(\x,\kink,\tang)
			= functionp(\x,\kink)
			- 1*min(\x-\tang,0)^2
			- 0.8*max(\x-\tang,0)^2;
		functionulr(\x,\kink,\tang,\infl)
			= functionp(\x,\kink)
			- 1*min(\x-\tang,0)^2
			- 1*max(\x-\tang,0)^2
			+ 3.2*min(\x-\infl,0)^2;
		functionubr(\x,\kink,\tang,\infl)
			= functionp(\x,\kink)
			- 1.62*min(\x-\tang,0)^2
			- 1.5*max(\x-\tang,0)^2
			+ 4.05*min(\x-\infl,0)^2;
		functionubr_deriv(\x,\tang)
			= 0.4
			- 2*1.62*(\x-\tang);
		functionubl(\x,\kink,\tang,\infl,\start)
			= functionubr(\x,\kink,\tang,\infl)
			- 3.5*min(\x-\start,0)^2;
		functionubl_deriv(\x,\kink,\tang,\infl,\start)
			= 0.4
			- 2*1.62*(\x-\tang)
			+ 2*4.05*(\x-\infl)
			- 2*3.5*(\x-\start);
		functionub_tangs(\x,\tangl,\tangr,\kink,\tang,\infl,\start)
			= max( functionubl(\tangl,\cutx,\cuty,\infl,\start) + functionubl_deriv(\tangl,\cutx,\cuty,\infl,\start) * (\x-\tangl) , functionubr(\tangr,\cutx,\cuty,\infl) + functionubr_deriv(\tangr,\cuty) * (\x-\tangr) );
		}]

	\pgfmathsetmacro{\ticklength}{1/14};	
	\pgfmathsetmacro{\samples}{100};		
	\pgfmathsetmacro{\radius}{1.9};			

	\pgfmathsetmacro{\xmax}{5};
	\pgfmathsetmacro{\ymax}{5};
	\pgfmathsetmacro{\axisbuff}{0.35};
	\pgfmathsetmacro{\xmult}{1.65};
	\pgfmathsetmacro{\ymult}{1.3};
	\pgfmathsetmacro{\xplus}{0};
	\pgfmathsetmacro{\yplus}{0};

	\pgfmathsetmacro{\cutw}{0};
	\pgfmathsetmacro{\cutystar}{1};
	\pgfmathsetmacro{\cutxprime}{1};
	\pgfmathsetmacro{\cutxbar}{1.85};
	\pgfmathsetmacro{\cutx}{2};
	\pgfmathsetmacro{\cutyprime}{3};
	\pgfmathsetmacro{\cutybar}{3.7};
	\pgfmathsetmacro{\cuty}{4};
	\pgfmathsetmacro{\cutz}{5};
	\pgfmathsetmacro{\cutzprime}{5};

	\draw ( { \xplus + \xmult * \cutz },
		{ \yplus + \ymult * functionp(\cutz,\cutx) } )
		node[anchor=west] {$\strut \boldsymbol{\color{white}p}$};

	\pgfmathsetmacro{\min}{\cutyprime};
	\pgfmathsetmacro{\infl}{0.5*\cutyprime+0.5*\cutybar};
	\pgfmathsetmacro{\A}{1.425};
	\pgfmathsetmacro{\B}{1.0};
	\pgfmathsetmacro{\C}{1.8};
	\ifdrawu
		\draw[thick, dashed]
			( { \xplus + \xmult * \cutx }, { \yplus + \ymult * functionu(\cutx,\min,\infl,\A,\B,\C) } )
			-- ( { \xplus + \xmult * \cuty }, { \yplus + \ymult * functionu(\cuty,\min,\infl,\A,\B,\C) } );
		\draw[thick, dashed]
			( { \xplus + \xmult * \cutxbar }, { \yplus + \ymult * functionu(\cutxbar,\min,\infl,\A,\B,\C) } )
			-- ( { \xplus + \xmult * \cutybar }, { \yplus + \ymult * functionu(\cutybar,\min,\infl,\A,\B,\C) } );
		\draw
			[domain={\cutxbar}:{\cuty}, variable=\x,
			samples=\samples, very thick]
			plot ( { \xplus + \xmult * \x },
			{ \yplus + \ymult * functionu(\x,\min,\infl,\A,\B,\C) } );
		\pgfmathsetmacro{\xlabel}{0.8*\cutx+0.2*\cutyprime}
		\draw ( { \xplus + \xmult * \xlabel },
			{ \yplus + \ymult * functionu(\xlabel,\min,\infl,\A,\B,\C) } )
			node[anchor=east] {$\strut \boldsymbol{u}$};
	\fi

	\ifdrawv
		\draw
			[domain={\cutxprime}:{\cutzprime}, variable=\x,
			samples=2, very thick]
			plot ( { \xplus + \xmult * \x },
			{ \yplus + \ymult * functionq(\x,\cutxprime) } );
		\draw[thick, dashed]
			( { \xplus + \xmult * \cutxbar }, { \yplus + \ymult * functionv(\cutxbar,\cutxprime,\cutyprime) } )
			-- ( { \xplus + \xmult * \cutybar }, { \yplus + \ymult * functionv(\cutybar,\cutxprime,\cutyprime) } );
		\draw
			[domain={\cutxprime}:{\cutzprime}, variable=\x,
			samples=\samples, very thick]
			plot ( { \xplus + \xmult * \x },
			{ \yplus + \ymult * functionv(\x,\cutxprime,\cutyprime) } );
		\pgfmathsetmacro{\xlabel}{0.2*\cuty+0.8*\cutzprime}
		\draw ( { \xplus + \xmult * \xlabel },
			{ \yplus + \ymult * functionv(\xlabel,\cutxprime,\cutyprime) } )
			node[anchor=west] {$\strut \boldsymbol{v}$};
		\draw ( { \xplus + \xmult * \cutz },
			{ \yplus + \ymult * functionq(\cutzprime,\cutxprime) } )
			node[anchor=west] {$\strut \boldsymbol{q}$};
	\fi

	\ifdrawu\else\ifdrawv\else
		\ifdrawb		
			\pgfmathsetmacro{\infl}{0.75*\cutyprime+0.25*\cuty}
			\pgfmathsetmacro{\start}{0.5*\cutx+0.5*\cutyprime}
			\pgfmathsetmacro{\tangl}{\cutx}
			\pgfmathsetmacro{\tangr}{0.3*\cutyprime+0.7*\cuty}
			\draw
				[domain={\cutx}:{0.8*\cuty+0.2*\cutz}, variable=\x,
				samples=\samples, thick, dashed]
				plot ( { \xplus + \xmult * \x },
				{ \yplus + \ymult * functionub_tangs(\x,\tangl,\tangr,\cutx,\cuty,\infl,\start) } );
			\draw
				[domain={\cutx}:{\cutz}, variable=\x,
				samples=\samples, very thick]
				plot ( { \xplus + \xmult * \x },
				{ \yplus + \ymult * functionp(\x,\cutx) } );
			\draw
				[domain={\cutx}:{\start}, variable=\x,
				samples=\samples, very thick]
				plot ( { \xplus + \xmult * \x },
				{ \yplus + \ymult * functionubl(\x,\cutx,\cuty,\infl,\start) } );
			\draw
				[domain={\start}:{((\cuty+\cutz)/2)}, variable=\x,
				samples=\samples, very thick]
				plot ( { \xplus + \xmult * \x },
				{ \yplus + \ymult * functionubr(\x,\cutx,\cuty,\infl) } );
			\draw ( { \xplus + \xmult * ((\cuty+\cutz)/2) },
				{ \yplus + \ymult * functionubr(((\cuty+\cutz)/2),\cutx,\cuty,\infl) } )
				node[anchor=west] {$\strut \boldsymbol{u}$};
		\fi
		\ifdrawl
			\draw
				[domain={\cutw}:{\cutx}, variable=\x,
				samples=2, very thick]
				plot ( { \xplus + \xmult * \x },
				{ \yplus + \ymult * functionp(\x,\cutx) } );
			\draw
				[domain={\cutx}:{\cutz}, variable=\x,
				samples=2, very thick]
				plot ( { \xplus + \xmult * \x },
				{ \yplus + \ymult * functionp(\x,\cutx) } );
			\draw
				[domain={((\cutw+\cutystar)/2)}:{\cutx}, variable=\x,
				samples=\samples, very thick]
				plot ( { \xplus + \xmult * \x },
				{ \yplus + \ymult * functionull(\x,\cutx,\cutystar) } );
			\draw
				[domain={\cutx}:{((\cuty+\cutz)/2)}, variable=\x,
				samples=\samples, very thick]
				plot ( { \xplus + \xmult * \x },
				{ \yplus + \ymult * functionulr(\x,\cutx,\cuty,\cutyprime) } );
			\draw ( { \xplus + \xmult * ((\cuty+\cutz)/2) },
				{ \yplus + \ymult * functionulr(((\cuty+\cutz)/2),\cutx,\cuty,\cutyprime) } )
				node[anchor=west] {$\strut \boldsymbol{u}$};
		\fi
		\draw ( { \xplus + \xmult * \cutz },
			{ \yplus + \ymult * functionp(\cutz,\cutx) } )
			node[anchor=west] {$\strut \boldsymbol{p}$};
	\fi\fi


	\ifdrawl
		\draw[-] ( { \xplus + \xmult * \cutw }, - \ticklength )
			-- ( { \xplus + \xmult * \cutw }, \ticklength );
		\draw ( { \xplus + \xmult * \cutw }, 0 )
			node[anchor=north] {$\strut w$};
		\draw[-] ( { \xplus + \xmult * \cutystar }, - \ticklength )
			-- ( { \xplus + \xmult * \cutystar }, \ticklength );
		\draw ( { \xplus + \xmult * \cutystar }, 0 )
			node[anchor=north] {$\strut y^*$};
	\fi

	\ifdrawl\else
		\ifdrawv\else\color{gray}\fi
		\draw[-] ( { \xplus + \xmult * \cutxprime }, - \ticklength )
			-- ( { \xplus + \xmult * \cutxprime }, \ticklength );
		\draw ( { \xplus + \xmult * \cutxprime }, 0 )
			node[anchor=north] {$\strut x'$};
		\color{black}
	\fi

	\ifdrawb\else\ifdrawl\else
		\draw[-] ( { \xplus + \xmult * \cutxbar }, - \ticklength )
			-- ( { \xplus + \xmult * \cutxbar }, \ticklength );
		\draw ( { \xplus + \xmult * \cutxbar }, 0 )
			node[anchor=north] {$\strut \widebar{x}$};
		\draw[-] ( { \xplus + \xmult * \cutybar }, - \ticklength )
			-- ( { \xplus + \xmult * \cutybar }, \ticklength );
		\draw ( { \xplus + \xmult * \cutybar }, 0 )
			node[anchor=north] {$\strut \widebar{y}$};
	\fi\fi

	\ifdrawv\color{gray}\fi
	\draw[-] ( { \xplus + \xmult * \cutx }, - \ticklength )
		-- ( { \xplus + \xmult * \cutx }, \ticklength );
	\draw ( { \xplus + \xmult * \cutx }, 0 )
		node[anchor=north] {$\strut x$};
	\draw[-] ( { \xplus + \xmult * \cuty }, - \ticklength )
		-- ( { \xplus + \xmult * \cuty }, \ticklength );
	\draw ( { \xplus + \xmult * \cuty }, 0 )
		node[anchor=north] {$\strut y$};
	\color{black}

	\ifdrawv\else\color{gray}\fi
	\draw[-] ( { \xplus + \xmult * \cutyprime }, - \ticklength )
		-- ( { \xplus + \xmult * \cutyprime }, \ticklength );
	\draw ( { \xplus + \xmult * \cutyprime }, 0 )
		node[anchor=north] {$\strut y'$};
	\color{black}

	\ifdrawb\else\ifdrawl\else
		\ifdrawu\color{gray}\fi
		\draw[-] ( { \xplus + \xmult * \cutzprime }, - \ticklength )
			-- ( { \xplus + \xmult * \cutzprime }, \ticklength );
		\draw ( { \xplus + \xmult * \cutzprime }, 0 )
			node[anchor=north] {$\strut z'$};
		\color{black}
	\fi\fi

	\ifdrawu\else\ifdrawv\else
		\draw[-] ( { \xplus + \xmult * \cutz }, - \ticklength )
			-- ( { \xplus + \xmult * \cutz }, \ticklength );
		\draw ( { \xplus + \xmult * \cutz }, 0 )
			node[anchor=north] {$\strut z$};
	\fi\fi

	\draw[<->] ( { \xplus + \xmult * ( \cutw ) - \axisbuff }, 0 )
		-- ( { \xplus + \xmult * ( \cutz ) + \axisbuff }, 0 );

\end{tikzpicture}

			\end{minipage}
			\hfill
			\begin{minipage}[c]{0.24\textwidth}
				\caption{\raggedright The \hyperref[claim:incr_suff_s]{claim} asserts that if $x' < x \leq y' \leq y$, then $u$ is S-shaped on $(\widebar{x},y)$, and is exceeded by the chord from $(x,u(x))$ to $(y,u(y))$.}
				\label{fig:suff:u}
			\end{minipage}
		\end{subfigure}
		\\
		\begin{subfigure}{\textwidth}
			\begin{minipage}[c]{0.73\textwidth}
				\centering
				\drawufalse
				\drawvtrue
				\ifdrawu\drawvfalse\drawbfalse\drawlfalse\fi
\ifdrawv\drawufalse\drawbfalse\drawlfalse\fi
\ifdrawb\drawufalse\drawvfalse\drawlfalse\fi
\ifdrawl\drawufalse\drawvfalse\drawbfalse\fi

\begin{tikzpicture}[scale=1, line cap=round,
	declare function={
		functionu(\x,\min,\infl,\A,\B,\C)
			= \A*min(\x-\min,0)^4
			+ \B*max(\x-\min,0)^2
			- \C*max(\x-\infl,0)^2;
		functionq(\x,\xmin)
			= 2.0 - 0.1*(\x-\xmin);
		functionv(\x,\xmin,\tang)
			= functionq(\x,\xmin)
			- 0.285*(\x-\tang)^2
			* (1 + 0.45*sin(deg(5*(\x-\xmin))));
		functionp(\x,\kink)
			= 1.25 + 0.4*max((\x-\kink),-0.5*(\x-\kink));
		functionull(\x,\kink,\tang)
			= functionp(\x,\kink)
			- 1*min(\x-\tang,0)^2
			- 0.8*max(\x-\tang,0)^2;
		functionulr(\x,\kink,\tang,\infl)
			= functionp(\x,\kink)
			- 1*min(\x-\tang,0)^2
			- 1*max(\x-\tang,0)^2
			+ 3.2*min(\x-\infl,0)^2;
		functionubr(\x,\kink,\tang,\infl)
			= functionp(\x,\kink)
			- 1.62*min(\x-\tang,0)^2
			- 1.5*max(\x-\tang,0)^2
			+ 4.05*min(\x-\infl,0)^2;
		functionubr_deriv(\x,\tang)
			= 0.4
			- 2*1.62*(\x-\tang);
		functionubl(\x,\kink,\tang,\infl,\start)
			= functionubr(\x,\kink,\tang,\infl)
			- 3.5*min(\x-\start,0)^2;
		functionubl_deriv(\x,\kink,\tang,\infl,\start)
			= 0.4
			- 2*1.62*(\x-\tang)
			+ 2*4.05*(\x-\infl)
			- 2*3.5*(\x-\start);
		functionub_tangs(\x,\tangl,\tangr,\kink,\tang,\infl,\start)
			= max( functionubl(\tangl,\cutx,\cuty,\infl,\start) + functionubl_deriv(\tangl,\cutx,\cuty,\infl,\start) * (\x-\tangl) , functionubr(\tangr,\cutx,\cuty,\infl) + functionubr_deriv(\tangr,\cuty) * (\x-\tangr) );
		}]

	\pgfmathsetmacro{\ticklength}{1/14};	
	\pgfmathsetmacro{\samples}{100};		
	\pgfmathsetmacro{\radius}{1.9};			

	\pgfmathsetmacro{\xmax}{5};
	\pgfmathsetmacro{\ymax}{5};
	\pgfmathsetmacro{\axisbuff}{0.35};
	\pgfmathsetmacro{\xmult}{1.65};
	\pgfmathsetmacro{\ymult}{1.3};
	\pgfmathsetmacro{\xplus}{0};
	\pgfmathsetmacro{\yplus}{0};

	\pgfmathsetmacro{\cutw}{0};
	\pgfmathsetmacro{\cutystar}{1};
	\pgfmathsetmacro{\cutxprime}{1};
	\pgfmathsetmacro{\cutxbar}{1.85};
	\pgfmathsetmacro{\cutx}{2};
	\pgfmathsetmacro{\cutyprime}{3};
	\pgfmathsetmacro{\cutybar}{3.7};
	\pgfmathsetmacro{\cuty}{4};
	\pgfmathsetmacro{\cutz}{5};
	\pgfmathsetmacro{\cutzprime}{5};

	\draw ( { \xplus + \xmult * \cutz },
		{ \yplus + \ymult * functionp(\cutz,\cutx) } )
		node[anchor=west] {$\strut \boldsymbol{\color{white}p}$};

	\pgfmathsetmacro{\min}{\cutyprime};
	\pgfmathsetmacro{\infl}{0.5*\cutyprime+0.5*\cutybar};
	\pgfmathsetmacro{\A}{1.425};
	\pgfmathsetmacro{\B}{1.0};
	\pgfmathsetmacro{\C}{1.8};
	\ifdrawu
		\draw[thick, dashed]
			( { \xplus + \xmult * \cutx }, { \yplus + \ymult * functionu(\cutx,\min,\infl,\A,\B,\C) } )
			-- ( { \xplus + \xmult * \cuty }, { \yplus + \ymult * functionu(\cuty,\min,\infl,\A,\B,\C) } );
		\draw[thick, dashed]
			( { \xplus + \xmult * \cutxbar }, { \yplus + \ymult * functionu(\cutxbar,\min,\infl,\A,\B,\C) } )
			-- ( { \xplus + \xmult * \cutybar }, { \yplus + \ymult * functionu(\cutybar,\min,\infl,\A,\B,\C) } );
		\draw
			[domain={\cutxbar}:{\cuty}, variable=\x,
			samples=\samples, very thick]
			plot ( { \xplus + \xmult * \x },
			{ \yplus + \ymult * functionu(\x,\min,\infl,\A,\B,\C) } );
		\pgfmathsetmacro{\xlabel}{0.8*\cutx+0.2*\cutyprime}
		\draw ( { \xplus + \xmult * \xlabel },
			{ \yplus + \ymult * functionu(\xlabel,\min,\infl,\A,\B,\C) } )
			node[anchor=east] {$\strut \boldsymbol{u}$};
	\fi

	\ifdrawv
		\draw
			[domain={\cutxprime}:{\cutzprime}, variable=\x,
			samples=2, very thick]
			plot ( { \xplus + \xmult * \x },
			{ \yplus + \ymult * functionq(\x,\cutxprime) } );
		\draw[thick, dashed]
			( { \xplus + \xmult * \cutxbar }, { \yplus + \ymult * functionv(\cutxbar,\cutxprime,\cutyprime) } )
			-- ( { \xplus + \xmult * \cutybar }, { \yplus + \ymult * functionv(\cutybar,\cutxprime,\cutyprime) } );
		\draw
			[domain={\cutxprime}:{\cutzprime}, variable=\x,
			samples=\samples, very thick]
			plot ( { \xplus + \xmult * \x },
			{ \yplus + \ymult * functionv(\x,\cutxprime,\cutyprime) } );
		\pgfmathsetmacro{\xlabel}{0.2*\cuty+0.8*\cutzprime}
		\draw ( { \xplus + \xmult * \xlabel },
			{ \yplus + \ymult * functionv(\xlabel,\cutxprime,\cutyprime) } )
			node[anchor=west] {$\strut \boldsymbol{v}$};
		\draw ( { \xplus + \xmult * \cutz },
			{ \yplus + \ymult * functionq(\cutzprime,\cutxprime) } )
			node[anchor=west] {$\strut \boldsymbol{q}$};
	\fi

	\ifdrawu\else\ifdrawv\else
		\ifdrawb		
			\pgfmathsetmacro{\infl}{0.75*\cutyprime+0.25*\cuty}
			\pgfmathsetmacro{\start}{0.5*\cutx+0.5*\cutyprime}
			\pgfmathsetmacro{\tangl}{\cutx}
			\pgfmathsetmacro{\tangr}{0.3*\cutyprime+0.7*\cuty}
			\draw
				[domain={\cutx}:{0.8*\cuty+0.2*\cutz}, variable=\x,
				samples=\samples, thick, dashed]
				plot ( { \xplus + \xmult * \x },
				{ \yplus + \ymult * functionub_tangs(\x,\tangl,\tangr,\cutx,\cuty,\infl,\start) } );
			\draw
				[domain={\cutx}:{\cutz}, variable=\x,
				samples=\samples, very thick]
				plot ( { \xplus + \xmult * \x },
				{ \yplus + \ymult * functionp(\x,\cutx) } );
			\draw
				[domain={\cutx}:{\start}, variable=\x,
				samples=\samples, very thick]
				plot ( { \xplus + \xmult * \x },
				{ \yplus + \ymult * functionubl(\x,\cutx,\cuty,\infl,\start) } );
			\draw
				[domain={\start}:{((\cuty+\cutz)/2)}, variable=\x,
				samples=\samples, very thick]
				plot ( { \xplus + \xmult * \x },
				{ \yplus + \ymult * functionubr(\x,\cutx,\cuty,\infl) } );
			\draw ( { \xplus + \xmult * ((\cuty+\cutz)/2) },
				{ \yplus + \ymult * functionubr(((\cuty+\cutz)/2),\cutx,\cuty,\infl) } )
				node[anchor=west] {$\strut \boldsymbol{u}$};
		\fi
		\ifdrawl
			\draw
				[domain={\cutw}:{\cutx}, variable=\x,
				samples=2, very thick]
				plot ( { \xplus + \xmult * \x },
				{ \yplus + \ymult * functionp(\x,\cutx) } );
			\draw
				[domain={\cutx}:{\cutz}, variable=\x,
				samples=2, very thick]
				plot ( { \xplus + \xmult * \x },
				{ \yplus + \ymult * functionp(\x,\cutx) } );
			\draw
				[domain={((\cutw+\cutystar)/2)}:{\cutx}, variable=\x,
				samples=\samples, very thick]
				plot ( { \xplus + \xmult * \x },
				{ \yplus + \ymult * functionull(\x,\cutx,\cutystar) } );
			\draw
				[domain={\cutx}:{((\cuty+\cutz)/2)}, variable=\x,
				samples=\samples, very thick]
				plot ( { \xplus + \xmult * \x },
				{ \yplus + \ymult * functionulr(\x,\cutx,\cuty,\cutyprime) } );
			\draw ( { \xplus + \xmult * ((\cuty+\cutz)/2) },
				{ \yplus + \ymult * functionulr(((\cuty+\cutz)/2),\cutx,\cuty,\cutyprime) } )
				node[anchor=west] {$\strut \boldsymbol{u}$};
		\fi
		\draw ( { \xplus + \xmult * \cutz },
			{ \yplus + \ymult * functionp(\cutz,\cutx) } )
			node[anchor=west] {$\strut \boldsymbol{p}$};
	\fi\fi


	\ifdrawl
		\draw[-] ( { \xplus + \xmult * \cutw }, - \ticklength )
			-- ( { \xplus + \xmult * \cutw }, \ticklength );
		\draw ( { \xplus + \xmult * \cutw }, 0 )
			node[anchor=north] {$\strut w$};
		\draw[-] ( { \xplus + \xmult * \cutystar }, - \ticklength )
			-- ( { \xplus + \xmult * \cutystar }, \ticklength );
		\draw ( { \xplus + \xmult * \cutystar }, 0 )
			node[anchor=north] {$\strut y^*$};
	\fi

	\ifdrawl\else
		\ifdrawv\else\color{gray}\fi
		\draw[-] ( { \xplus + \xmult * \cutxprime }, - \ticklength )
			-- ( { \xplus + \xmult * \cutxprime }, \ticklength );
		\draw ( { \xplus + \xmult * \cutxprime }, 0 )
			node[anchor=north] {$\strut x'$};
		\color{black}
	\fi

	\ifdrawb\else\ifdrawl\else
		\draw[-] ( { \xplus + \xmult * \cutxbar }, - \ticklength )
			-- ( { \xplus + \xmult * \cutxbar }, \ticklength );
		\draw ( { \xplus + \xmult * \cutxbar }, 0 )
			node[anchor=north] {$\strut \widebar{x}$};
		\draw[-] ( { \xplus + \xmult * \cutybar }, - \ticklength )
			-- ( { \xplus + \xmult * \cutybar }, \ticklength );
		\draw ( { \xplus + \xmult * \cutybar }, 0 )
			node[anchor=north] {$\strut \widebar{y}$};
	\fi\fi

	\ifdrawv\color{gray}\fi
	\draw[-] ( { \xplus + \xmult * \cutx }, - \ticklength )
		-- ( { \xplus + \xmult * \cutx }, \ticklength );
	\draw ( { \xplus + \xmult * \cutx }, 0 )
		node[anchor=north] {$\strut x$};
	\draw[-] ( { \xplus + \xmult * \cuty }, - \ticklength )
		-- ( { \xplus + \xmult * \cuty }, \ticklength );
	\draw ( { \xplus + \xmult * \cuty }, 0 )
		node[anchor=north] {$\strut y$};
	\color{black}

	\ifdrawv\else\color{gray}\fi
	\draw[-] ( { \xplus + \xmult * \cutyprime }, - \ticklength )
		-- ( { \xplus + \xmult * \cutyprime }, \ticklength );
	\draw ( { \xplus + \xmult * \cutyprime }, 0 )
		node[anchor=north] {$\strut y'$};
	\color{black}

	\ifdrawb\else\ifdrawl\else
		\ifdrawu\color{gray}\fi
		\draw[-] ( { \xplus + \xmult * \cutzprime }, - \ticklength )
			-- ( { \xplus + \xmult * \cutzprime }, \ticklength );
		\draw ( { \xplus + \xmult * \cutzprime }, 0 )
			node[anchor=north] {$\strut z'$};
		\color{black}
	\fi\fi

	\ifdrawu\else\ifdrawv\else
		\draw[-] ( { \xplus + \xmult * \cutz }, - \ticklength )
			-- ( { \xplus + \xmult * \cutz }, \ticklength );
		\draw ( { \xplus + \xmult * \cutz }, 0 )
			node[anchor=north] {$\strut z$};
	\fi\fi

	\draw[<->] ( { \xplus + \xmult * ( \cutw ) - \axisbuff }, 0 )
		-- ( { \xplus + \xmult * ( \cutz ) + \axisbuff }, 0 );

\end{tikzpicture}

			\end{minipage}
			\hfill
			\begin{minipage}[c]{0.24\textwidth}
				\caption{\raggedright Such a $q$ exists since $G$ is optimal for $v$ and $[x',z']$ is a $G$-pooling interval.}
				\label{fig:suff:v}
			\end{minipage}
		\end{subfigure}
		\\
		\begin{subfigure}{\textwidth}
			\begin{minipage}[c]{0.73\textwidth}
				\centering
				\drawvfalse
				\drawbtrue
				\ifdrawu\drawvfalse\drawbfalse\drawlfalse\fi
\ifdrawv\drawufalse\drawbfalse\drawlfalse\fi
\ifdrawb\drawufalse\drawvfalse\drawlfalse\fi
\ifdrawl\drawufalse\drawvfalse\drawbfalse\fi

\begin{tikzpicture}[scale=1, line cap=round,
	declare function={
		functionu(\x,\min,\infl,\A,\B,\C)
			= \A*min(\x-\min,0)^4
			+ \B*max(\x-\min,0)^2
			- \C*max(\x-\infl,0)^2;
		functionq(\x,\xmin)
			= 2.0 - 0.1*(\x-\xmin);
		functionv(\x,\xmin,\tang)
			= functionq(\x,\xmin)
			- 0.285*(\x-\tang)^2
			* (1 + 0.45*sin(deg(5*(\x-\xmin))));
		functionp(\x,\kink)
			= 1.25 + 0.4*max((\x-\kink),-0.5*(\x-\kink));
		functionull(\x,\kink,\tang)
			= functionp(\x,\kink)
			- 1*min(\x-\tang,0)^2
			- 0.8*max(\x-\tang,0)^2;
		functionulr(\x,\kink,\tang,\infl)
			= functionp(\x,\kink)
			- 1*min(\x-\tang,0)^2
			- 1*max(\x-\tang,0)^2
			+ 3.2*min(\x-\infl,0)^2;
		functionubr(\x,\kink,\tang,\infl)
			= functionp(\x,\kink)
			- 1.62*min(\x-\tang,0)^2
			- 1.5*max(\x-\tang,0)^2
			+ 4.05*min(\x-\infl,0)^2;
		functionubr_deriv(\x,\tang)
			= 0.4
			- 2*1.62*(\x-\tang);
		functionubl(\x,\kink,\tang,\infl,\start)
			= functionubr(\x,\kink,\tang,\infl)
			- 3.5*min(\x-\start,0)^2;
		functionubl_deriv(\x,\kink,\tang,\infl,\start)
			= 0.4
			- 2*1.62*(\x-\tang)
			+ 2*4.05*(\x-\infl)
			- 2*3.5*(\x-\start);
		functionub_tangs(\x,\tangl,\tangr,\kink,\tang,\infl,\start)
			= max( functionubl(\tangl,\cutx,\cuty,\infl,\start) + functionubl_deriv(\tangl,\cutx,\cuty,\infl,\start) * (\x-\tangl) , functionubr(\tangr,\cutx,\cuty,\infl) + functionubr_deriv(\tangr,\cuty) * (\x-\tangr) );
		}]

	\pgfmathsetmacro{\ticklength}{1/14};	
	\pgfmathsetmacro{\samples}{100};		
	\pgfmathsetmacro{\radius}{1.9};			

	\pgfmathsetmacro{\xmax}{5};
	\pgfmathsetmacro{\ymax}{5};
	\pgfmathsetmacro{\axisbuff}{0.35};
	\pgfmathsetmacro{\xmult}{1.65};
	\pgfmathsetmacro{\ymult}{1.3};
	\pgfmathsetmacro{\xplus}{0};
	\pgfmathsetmacro{\yplus}{0};

	\pgfmathsetmacro{\cutw}{0};
	\pgfmathsetmacro{\cutystar}{1};
	\pgfmathsetmacro{\cutxprime}{1};
	\pgfmathsetmacro{\cutxbar}{1.85};
	\pgfmathsetmacro{\cutx}{2};
	\pgfmathsetmacro{\cutyprime}{3};
	\pgfmathsetmacro{\cutybar}{3.7};
	\pgfmathsetmacro{\cuty}{4};
	\pgfmathsetmacro{\cutz}{5};
	\pgfmathsetmacro{\cutzprime}{5};

	\draw ( { \xplus + \xmult * \cutz },
		{ \yplus + \ymult * functionp(\cutz,\cutx) } )
		node[anchor=west] {$\strut \boldsymbol{\color{white}p}$};

	\pgfmathsetmacro{\min}{\cutyprime};
	\pgfmathsetmacro{\infl}{0.5*\cutyprime+0.5*\cutybar};
	\pgfmathsetmacro{\A}{1.425};
	\pgfmathsetmacro{\B}{1.0};
	\pgfmathsetmacro{\C}{1.8};
	\ifdrawu
		\draw[thick, dashed]
			( { \xplus + \xmult * \cutx }, { \yplus + \ymult * functionu(\cutx,\min,\infl,\A,\B,\C) } )
			-- ( { \xplus + \xmult * \cuty }, { \yplus + \ymult * functionu(\cuty,\min,\infl,\A,\B,\C) } );
		\draw[thick, dashed]
			( { \xplus + \xmult * \cutxbar }, { \yplus + \ymult * functionu(\cutxbar,\min,\infl,\A,\B,\C) } )
			-- ( { \xplus + \xmult * \cutybar }, { \yplus + \ymult * functionu(\cutybar,\min,\infl,\A,\B,\C) } );
		\draw
			[domain={\cutxbar}:{\cuty}, variable=\x,
			samples=\samples, very thick]
			plot ( { \xplus + \xmult * \x },
			{ \yplus + \ymult * functionu(\x,\min,\infl,\A,\B,\C) } );
		\pgfmathsetmacro{\xlabel}{0.8*\cutx+0.2*\cutyprime}
		\draw ( { \xplus + \xmult * \xlabel },
			{ \yplus + \ymult * functionu(\xlabel,\min,\infl,\A,\B,\C) } )
			node[anchor=east] {$\strut \boldsymbol{u}$};
	\fi

	\ifdrawv
		\draw
			[domain={\cutxprime}:{\cutzprime}, variable=\x,
			samples=2, very thick]
			plot ( { \xplus + \xmult * \x },
			{ \yplus + \ymult * functionq(\x,\cutxprime) } );
		\draw[thick, dashed]
			( { \xplus + \xmult * \cutxbar }, { \yplus + \ymult * functionv(\cutxbar,\cutxprime,\cutyprime) } )
			-- ( { \xplus + \xmult * \cutybar }, { \yplus + \ymult * functionv(\cutybar,\cutxprime,\cutyprime) } );
		\draw
			[domain={\cutxprime}:{\cutzprime}, variable=\x,
			samples=\samples, very thick]
			plot ( { \xplus + \xmult * \x },
			{ \yplus + \ymult * functionv(\x,\cutxprime,\cutyprime) } );
		\pgfmathsetmacro{\xlabel}{0.2*\cuty+0.8*\cutzprime}
		\draw ( { \xplus + \xmult * \xlabel },
			{ \yplus + \ymult * functionv(\xlabel,\cutxprime,\cutyprime) } )
			node[anchor=west] {$\strut \boldsymbol{v}$};
		\draw ( { \xplus + \xmult * \cutz },
			{ \yplus + \ymult * functionq(\cutzprime,\cutxprime) } )
			node[anchor=west] {$\strut \boldsymbol{q}$};
	\fi

	\ifdrawu\else\ifdrawv\else
		\ifdrawb		
			\pgfmathsetmacro{\infl}{0.75*\cutyprime+0.25*\cuty}
			\pgfmathsetmacro{\start}{0.5*\cutx+0.5*\cutyprime}
			\pgfmathsetmacro{\tangl}{\cutx}
			\pgfmathsetmacro{\tangr}{0.3*\cutyprime+0.7*\cuty}
			\draw
				[domain={\cutx}:{0.8*\cuty+0.2*\cutz}, variable=\x,
				samples=\samples, thick, dashed]
				plot ( { \xplus + \xmult * \x },
				{ \yplus + \ymult * functionub_tangs(\x,\tangl,\tangr,\cutx,\cuty,\infl,\start) } );
			\draw
				[domain={\cutx}:{\cutz}, variable=\x,
				samples=\samples, very thick]
				plot ( { \xplus + \xmult * \x },
				{ \yplus + \ymult * functionp(\x,\cutx) } );
			\draw
				[domain={\cutx}:{\start}, variable=\x,
				samples=\samples, very thick]
				plot ( { \xplus + \xmult * \x },
				{ \yplus + \ymult * functionubl(\x,\cutx,\cuty,\infl,\start) } );
			\draw
				[domain={\start}:{((\cuty+\cutz)/2)}, variable=\x,
				samples=\samples, very thick]
				plot ( { \xplus + \xmult * \x },
				{ \yplus + \ymult * functionubr(\x,\cutx,\cuty,\infl) } );
			\draw ( { \xplus + \xmult * ((\cuty+\cutz)/2) },
				{ \yplus + \ymult * functionubr(((\cuty+\cutz)/2),\cutx,\cuty,\infl) } )
				node[anchor=west] {$\strut \boldsymbol{u}$};
		\fi
		\ifdrawl
			\draw
				[domain={\cutw}:{\cutx}, variable=\x,
				samples=2, very thick]
				plot ( { \xplus + \xmult * \x },
				{ \yplus + \ymult * functionp(\x,\cutx) } );
			\draw
				[domain={\cutx}:{\cutz}, variable=\x,
				samples=2, very thick]
				plot ( { \xplus + \xmult * \x },
				{ \yplus + \ymult * functionp(\x,\cutx) } );
			\draw
				[domain={((\cutw+\cutystar)/2)}:{\cutx}, variable=\x,
				samples=\samples, very thick]
				plot ( { \xplus + \xmult * \x },
				{ \yplus + \ymult * functionull(\x,\cutx,\cutystar) } );
			\draw
				[domain={\cutx}:{((\cuty+\cutz)/2)}, variable=\x,
				samples=\samples, very thick]
				plot ( { \xplus + \xmult * \x },
				{ \yplus + \ymult * functionulr(\x,\cutx,\cuty,\cutyprime) } );
			\draw ( { \xplus + \xmult * ((\cuty+\cutz)/2) },
				{ \yplus + \ymult * functionulr(((\cuty+\cutz)/2),\cutx,\cuty,\cutyprime) } )
				node[anchor=west] {$\strut \boldsymbol{u}$};
		\fi
		\draw ( { \xplus + \xmult * \cutz },
			{ \yplus + \ymult * functionp(\cutz,\cutx) } )
			node[anchor=west] {$\strut \boldsymbol{p}$};
	\fi\fi


	\ifdrawl
		\draw[-] ( { \xplus + \xmult * \cutw }, - \ticklength )
			-- ( { \xplus + \xmult * \cutw }, \ticklength );
		\draw ( { \xplus + \xmult * \cutw }, 0 )
			node[anchor=north] {$\strut w$};
		\draw[-] ( { \xplus + \xmult * \cutystar }, - \ticklength )
			-- ( { \xplus + \xmult * \cutystar }, \ticklength );
		\draw ( { \xplus + \xmult * \cutystar }, 0 )
			node[anchor=north] {$\strut y^*$};
	\fi

	\ifdrawl\else
		\ifdrawv\else\color{gray}\fi
		\draw[-] ( { \xplus + \xmult * \cutxprime }, - \ticklength )
			-- ( { \xplus + \xmult * \cutxprime }, \ticklength );
		\draw ( { \xplus + \xmult * \cutxprime }, 0 )
			node[anchor=north] {$\strut x'$};
		\color{black}
	\fi

	\ifdrawb\else\ifdrawl\else
		\draw[-] ( { \xplus + \xmult * \cutxbar }, - \ticklength )
			-- ( { \xplus + \xmult * \cutxbar }, \ticklength );
		\draw ( { \xplus + \xmult * \cutxbar }, 0 )
			node[anchor=north] {$\strut \widebar{x}$};
		\draw[-] ( { \xplus + \xmult * \cutybar }, - \ticklength )
			-- ( { \xplus + \xmult * \cutybar }, \ticklength );
		\draw ( { \xplus + \xmult * \cutybar }, 0 )
			node[anchor=north] {$\strut \widebar{y}$};
	\fi\fi

	\ifdrawv\color{gray}\fi
	\draw[-] ( { \xplus + \xmult * \cutx }, - \ticklength )
		-- ( { \xplus + \xmult * \cutx }, \ticklength );
	\draw ( { \xplus + \xmult * \cutx }, 0 )
		node[anchor=north] {$\strut x$};
	\draw[-] ( { \xplus + \xmult * \cuty }, - \ticklength )
		-- ( { \xplus + \xmult * \cuty }, \ticklength );
	\draw ( { \xplus + \xmult * \cuty }, 0 )
		node[anchor=north] {$\strut y$};
	\color{black}

	\ifdrawv\else\color{gray}\fi
	\draw[-] ( { \xplus + \xmult * \cutyprime }, - \ticklength )
		-- ( { \xplus + \xmult * \cutyprime }, \ticklength );
	\draw ( { \xplus + \xmult * \cutyprime }, 0 )
		node[anchor=north] {$\strut y'$};
	\color{black}

	\ifdrawb\else\ifdrawl\else
		\ifdrawu\color{gray}\fi
		\draw[-] ( { \xplus + \xmult * \cutzprime }, - \ticklength )
			-- ( { \xplus + \xmult * \cutzprime }, \ticklength );
		\draw ( { \xplus + \xmult * \cutzprime }, 0 )
			node[anchor=north] {$\strut z'$};
		\color{black}
	\fi\fi

	\ifdrawu\else\ifdrawv\else
		\draw[-] ( { \xplus + \xmult * \cutz }, - \ticklength )
			-- ( { \xplus + \xmult * \cutz }, \ticklength );
		\draw ( { \xplus + \xmult * \cutz }, 0 )
			node[anchor=north] {$\strut z$};
	\fi\fi

	\draw[<->] ( { \xplus + \xmult * ( \cutw ) - \axisbuff }, 0 )
		-- ( { \xplus + \xmult * ( \cutz ) + \axisbuff }, 0 );

\end{tikzpicture}

			\end{minipage}
			\hfill
			\begin{minipage}[c]{0.24\textwidth}
				\caption{\raggedright Proof of the \hyperref[claim:incr_suff_s]{claim}: $u$ not being S-shaped on $[x,y]$ implies a crater property violation.}
				\label{fig:suff:b}
			\end{minipage}
		\end{subfigure}
		\\
		\begin{subfigure}{\textwidth}
			\begin{minipage}[c]{0.73\textwidth}
				\centering
				\drawbfalse
				\drawltrue
				\ifdrawu\drawvfalse\drawbfalse\drawlfalse\fi
\ifdrawv\drawufalse\drawbfalse\drawlfalse\fi
\ifdrawb\drawufalse\drawvfalse\drawlfalse\fi
\ifdrawl\drawufalse\drawvfalse\drawbfalse\fi

\begin{tikzpicture}[scale=1, line cap=round,
	declare function={
		functionu(\x,\min,\infl,\A,\B,\C)
			= \A*min(\x-\min,0)^4
			+ \B*max(\x-\min,0)^2
			- \C*max(\x-\infl,0)^2;
		functionq(\x,\xmin)
			= 2.0 - 0.1*(\x-\xmin);
		functionv(\x,\xmin,\tang)
			= functionq(\x,\xmin)
			- 0.285*(\x-\tang)^2
			* (1 + 0.45*sin(deg(5*(\x-\xmin))));
		functionp(\x,\kink)
			= 1.25 + 0.4*max((\x-\kink),-0.5*(\x-\kink));
		functionull(\x,\kink,\tang)
			= functionp(\x,\kink)
			- 1*min(\x-\tang,0)^2
			- 0.8*max(\x-\tang,0)^2;
		functionulr(\x,\kink,\tang,\infl)
			= functionp(\x,\kink)
			- 1*min(\x-\tang,0)^2
			- 1*max(\x-\tang,0)^2
			+ 3.2*min(\x-\infl,0)^2;
		functionubr(\x,\kink,\tang,\infl)
			= functionp(\x,\kink)
			- 1.62*min(\x-\tang,0)^2
			- 1.5*max(\x-\tang,0)^2
			+ 4.05*min(\x-\infl,0)^2;
		functionubr_deriv(\x,\tang)
			= 0.4
			- 2*1.62*(\x-\tang);
		functionubl(\x,\kink,\tang,\infl,\start)
			= functionubr(\x,\kink,\tang,\infl)
			- 3.5*min(\x-\start,0)^2;
		functionubl_deriv(\x,\kink,\tang,\infl,\start)
			= 0.4
			- 2*1.62*(\x-\tang)
			+ 2*4.05*(\x-\infl)
			- 2*3.5*(\x-\start);
		functionub_tangs(\x,\tangl,\tangr,\kink,\tang,\infl,\start)
			= max( functionubl(\tangl,\cutx,\cuty,\infl,\start) + functionubl_deriv(\tangl,\cutx,\cuty,\infl,\start) * (\x-\tangl) , functionubr(\tangr,\cutx,\cuty,\infl) + functionubr_deriv(\tangr,\cuty) * (\x-\tangr) );
		}]

	\pgfmathsetmacro{\ticklength}{1/14};	
	\pgfmathsetmacro{\samples}{100};		
	\pgfmathsetmacro{\radius}{1.9};			

	\pgfmathsetmacro{\xmax}{5};
	\pgfmathsetmacro{\ymax}{5};
	\pgfmathsetmacro{\axisbuff}{0.35};
	\pgfmathsetmacro{\xmult}{1.65};
	\pgfmathsetmacro{\ymult}{1.3};
	\pgfmathsetmacro{\xplus}{0};
	\pgfmathsetmacro{\yplus}{0};

	\pgfmathsetmacro{\cutw}{0};
	\pgfmathsetmacro{\cutystar}{1};
	\pgfmathsetmacro{\cutxprime}{1};
	\pgfmathsetmacro{\cutxbar}{1.85};
	\pgfmathsetmacro{\cutx}{2};
	\pgfmathsetmacro{\cutyprime}{3};
	\pgfmathsetmacro{\cutybar}{3.7};
	\pgfmathsetmacro{\cuty}{4};
	\pgfmathsetmacro{\cutz}{5};
	\pgfmathsetmacro{\cutzprime}{5};

	\draw ( { \xplus + \xmult * \cutz },
		{ \yplus + \ymult * functionp(\cutz,\cutx) } )
		node[anchor=west] {$\strut \boldsymbol{\color{white}p}$};

	\pgfmathsetmacro{\min}{\cutyprime};
	\pgfmathsetmacro{\infl}{0.5*\cutyprime+0.5*\cutybar};
	\pgfmathsetmacro{\A}{1.425};
	\pgfmathsetmacro{\B}{1.0};
	\pgfmathsetmacro{\C}{1.8};
	\ifdrawu
		\draw[thick, dashed]
			( { \xplus + \xmult * \cutx }, { \yplus + \ymult * functionu(\cutx,\min,\infl,\A,\B,\C) } )
			-- ( { \xplus + \xmult * \cuty }, { \yplus + \ymult * functionu(\cuty,\min,\infl,\A,\B,\C) } );
		\draw[thick, dashed]
			( { \xplus + \xmult * \cutxbar }, { \yplus + \ymult * functionu(\cutxbar,\min,\infl,\A,\B,\C) } )
			-- ( { \xplus + \xmult * \cutybar }, { \yplus + \ymult * functionu(\cutybar,\min,\infl,\A,\B,\C) } );
		\draw
			[domain={\cutxbar}:{\cuty}, variable=\x,
			samples=\samples, very thick]
			plot ( { \xplus + \xmult * \x },
			{ \yplus + \ymult * functionu(\x,\min,\infl,\A,\B,\C) } );
		\pgfmathsetmacro{\xlabel}{0.8*\cutx+0.2*\cutyprime}
		\draw ( { \xplus + \xmult * \xlabel },
			{ \yplus + \ymult * functionu(\xlabel,\min,\infl,\A,\B,\C) } )
			node[anchor=east] {$\strut \boldsymbol{u}$};
	\fi

	\ifdrawv
		\draw
			[domain={\cutxprime}:{\cutzprime}, variable=\x,
			samples=2, very thick]
			plot ( { \xplus + \xmult * \x },
			{ \yplus + \ymult * functionq(\x,\cutxprime) } );
		\draw[thick, dashed]
			( { \xplus + \xmult * \cutxbar }, { \yplus + \ymult * functionv(\cutxbar,\cutxprime,\cutyprime) } )
			-- ( { \xplus + \xmult * \cutybar }, { \yplus + \ymult * functionv(\cutybar,\cutxprime,\cutyprime) } );
		\draw
			[domain={\cutxprime}:{\cutzprime}, variable=\x,
			samples=\samples, very thick]
			plot ( { \xplus + \xmult * \x },
			{ \yplus + \ymult * functionv(\x,\cutxprime,\cutyprime) } );
		\pgfmathsetmacro{\xlabel}{0.2*\cuty+0.8*\cutzprime}
		\draw ( { \xplus + \xmult * \xlabel },
			{ \yplus + \ymult * functionv(\xlabel,\cutxprime,\cutyprime) } )
			node[anchor=west] {$\strut \boldsymbol{v}$};
		\draw ( { \xplus + \xmult * \cutz },
			{ \yplus + \ymult * functionq(\cutzprime,\cutxprime) } )
			node[anchor=west] {$\strut \boldsymbol{q}$};
	\fi

	\ifdrawu\else\ifdrawv\else
		\ifdrawb		
			\pgfmathsetmacro{\infl}{0.75*\cutyprime+0.25*\cuty}
			\pgfmathsetmacro{\start}{0.5*\cutx+0.5*\cutyprime}
			\pgfmathsetmacro{\tangl}{\cutx}
			\pgfmathsetmacro{\tangr}{0.3*\cutyprime+0.7*\cuty}
			\draw
				[domain={\cutx}:{0.8*\cuty+0.2*\cutz}, variable=\x,
				samples=\samples, thick, dashed]
				plot ( { \xplus + \xmult * \x },
				{ \yplus + \ymult * functionub_tangs(\x,\tangl,\tangr,\cutx,\cuty,\infl,\start) } );
			\draw
				[domain={\cutx}:{\cutz}, variable=\x,
				samples=\samples, very thick]
				plot ( { \xplus + \xmult * \x },
				{ \yplus + \ymult * functionp(\x,\cutx) } );
			\draw
				[domain={\cutx}:{\start}, variable=\x,
				samples=\samples, very thick]
				plot ( { \xplus + \xmult * \x },
				{ \yplus + \ymult * functionubl(\x,\cutx,\cuty,\infl,\start) } );
			\draw
				[domain={\start}:{((\cuty+\cutz)/2)}, variable=\x,
				samples=\samples, very thick]
				plot ( { \xplus + \xmult * \x },
				{ \yplus + \ymult * functionubr(\x,\cutx,\cuty,\infl) } );
			\draw ( { \xplus + \xmult * ((\cuty+\cutz)/2) },
				{ \yplus + \ymult * functionubr(((\cuty+\cutz)/2),\cutx,\cuty,\infl) } )
				node[anchor=west] {$\strut \boldsymbol{u}$};
		\fi
		\ifdrawl
			\draw
				[domain={\cutw}:{\cutx}, variable=\x,
				samples=2, very thick]
				plot ( { \xplus + \xmult * \x },
				{ \yplus + \ymult * functionp(\x,\cutx) } );
			\draw
				[domain={\cutx}:{\cutz}, variable=\x,
				samples=2, very thick]
				plot ( { \xplus + \xmult * \x },
				{ \yplus + \ymult * functionp(\x,\cutx) } );
			\draw
				[domain={((\cutw+\cutystar)/2)}:{\cutx}, variable=\x,
				samples=\samples, very thick]
				plot ( { \xplus + \xmult * \x },
				{ \yplus + \ymult * functionull(\x,\cutx,\cutystar) } );
			\draw
				[domain={\cutx}:{((\cuty+\cutz)/2)}, variable=\x,
				samples=\samples, very thick]
				plot ( { \xplus + \xmult * \x },
				{ \yplus + \ymult * functionulr(\x,\cutx,\cuty,\cutyprime) } );
			\draw ( { \xplus + \xmult * ((\cuty+\cutz)/2) },
				{ \yplus + \ymult * functionulr(((\cuty+\cutz)/2),\cutx,\cuty,\cutyprime) } )
				node[anchor=west] {$\strut \boldsymbol{u}$};
		\fi
		\draw ( { \xplus + \xmult * \cutz },
			{ \yplus + \ymult * functionp(\cutz,\cutx) } )
			node[anchor=west] {$\strut \boldsymbol{p}$};
	\fi\fi


	\ifdrawl
		\draw[-] ( { \xplus + \xmult * \cutw }, - \ticklength )
			-- ( { \xplus + \xmult * \cutw }, \ticklength );
		\draw ( { \xplus + \xmult * \cutw }, 0 )
			node[anchor=north] {$\strut w$};
		\draw[-] ( { \xplus + \xmult * \cutystar }, - \ticklength )
			-- ( { \xplus + \xmult * \cutystar }, \ticklength );
		\draw ( { \xplus + \xmult * \cutystar }, 0 )
			node[anchor=north] {$\strut y^*$};
	\fi

	\ifdrawl\else
		\ifdrawv\else\color{gray}\fi
		\draw[-] ( { \xplus + \xmult * \cutxprime }, - \ticklength )
			-- ( { \xplus + \xmult * \cutxprime }, \ticklength );
		\draw ( { \xplus + \xmult * \cutxprime }, 0 )
			node[anchor=north] {$\strut x'$};
		\color{black}
	\fi

	\ifdrawb\else\ifdrawl\else
		\draw[-] ( { \xplus + \xmult * \cutxbar }, - \ticklength )
			-- ( { \xplus + \xmult * \cutxbar }, \ticklength );
		\draw ( { \xplus + \xmult * \cutxbar }, 0 )
			node[anchor=north] {$\strut \widebar{x}$};
		\draw[-] ( { \xplus + \xmult * \cutybar }, - \ticklength )
			-- ( { \xplus + \xmult * \cutybar }, \ticklength );
		\draw ( { \xplus + \xmult * \cutybar }, 0 )
			node[anchor=north] {$\strut \widebar{y}$};
	\fi\fi

	\ifdrawv\color{gray}\fi
	\draw[-] ( { \xplus + \xmult * \cutx }, - \ticklength )
		-- ( { \xplus + \xmult * \cutx }, \ticklength );
	\draw ( { \xplus + \xmult * \cutx }, 0 )
		node[anchor=north] {$\strut x$};
	\draw[-] ( { \xplus + \xmult * \cuty }, - \ticklength )
		-- ( { \xplus + \xmult * \cuty }, \ticklength );
	\draw ( { \xplus + \xmult * \cuty }, 0 )
		node[anchor=north] {$\strut y$};
	\color{black}

	\ifdrawv\else\color{gray}\fi
	\draw[-] ( { \xplus + \xmult * \cutyprime }, - \ticklength )
		-- ( { \xplus + \xmult * \cutyprime }, \ticklength );
	\draw ( { \xplus + \xmult * \cutyprime }, 0 )
		node[anchor=north] {$\strut y'$};
	\color{black}

	\ifdrawb\else\ifdrawl\else
		\ifdrawu\color{gray}\fi
		\draw[-] ( { \xplus + \xmult * \cutzprime }, - \ticklength )
			-- ( { \xplus + \xmult * \cutzprime }, \ticklength );
		\draw ( { \xplus + \xmult * \cutzprime }, 0 )
			node[anchor=north] {$\strut z'$};
		\color{black}
	\fi\fi

	\ifdrawu\else\ifdrawv\else
		\draw[-] ( { \xplus + \xmult * \cutz }, - \ticklength )
			-- ( { \xplus + \xmult * \cutz }, \ticklength );
		\draw ( { \xplus + \xmult * \cutz }, 0 )
			node[anchor=north] {$\strut z$};
	\fi\fi

	\draw[<->] ( { \xplus + \xmult * ( \cutw ) - \axisbuff }, 0 )
		-- ( { \xplus + \xmult * ( \cutz ) + \axisbuff }, 0 );

\end{tikzpicture}

			\end{minipage}
			\hfill
			\begin{minipage}[c]{0.24\textwidth}
				\caption{\raggedright Proof of the \hyperref[claim:incr_suff_s]{claim}: $p(x)>u(x)$ implies a crater property violation.}
				\label{fig:suff:l}
			\end{minipage}
		\end{subfigure}
		\caption{Proof of the first (sufficiency) part of \Cref{theorem:incr}.}
		\label{fig:suff}
	\end{figure}

	Graphically, the \hyperref[claim:incr_suff_s]{claim} asserts that $u$ must look as in \Cref{fig:suff:u}.

	By the \hyperref[claim:incr_suff_s]{claim} (and regularity of $u$), we may choose $\widebar{y} \in (y',z')$ such that $u$ is strictly exceeded on $\left( \widebar{x}, \widebar{y} \right)$ by the (unique) affine function $[0,1] \to \R$ which intersects $u$ at $\widebar{x}$ and at $\widebar{y}$, as depicted in \Cref{fig:suff:u}. Since $G$ is optimal for $v$ and $[x',z']$ is a $G$-pooling interval, with $y'$ being the unique atom of $G$ in $[x',z']$, there exists a function $q : [0,1] \to \R$ which weakly exceeds $v$, is affine on $[x',z']$, and is tangent to $v$ at $y'$, as depicted in \Cref{fig:suff:v}.
	It follows that the (unique) affine function $[0,1] \to \R$ which intersects $v$ at $\widebar{x}$ and at $\widebar{y}$ does \emph{not} strictly exceed $v$ at $y'$ (see \Cref{fig:suff:v} again). But then $u$ fails to be coarsely less convex than $v$---a contradiction.

	It remains only to prove the \hyperref[claim:incr_suff_s]{claim}. Since $F$ is optimal for $u$, there is a convex function $p : [0,1] \to \R$ such that $p \geq u$, $p=u$ $F$-a.e., and $p$ is affine on each $F$-pooling interval. Thus $p$ is affine on $[x,z]$ and tangent to $u$ at $y$, and (thus, using regularity) $u$ is strictly concave on an open interval containing $y$. It follows (again using regularity) that $u$ is S-shaped on $[x,y]$, since if $u$ were strictly concave on $[w_1,w_2]$ and on $[w_3,y]$ and strictly convex on $[w_2,w_3]$ where $x \leq w_1 < w_2 < w_3 < y$, then as shown in \Cref{fig:suff:b}, tangents to $u$ could be chosen in $[w_1,w_2]$ and near $y$ such that the crater property would be violated. Hence $p>u$ on $(x,y)$, and if $p(x)=u(x)$ then $u$ is strictly convex on $[x,x+\delta)$ for some $\delta > 0$.

	It thus suffices to show that $p(x)=u(x)$ and that $u$ is strictly convex on $(x-\eps,x]$ for some $\eps>0$. For the former, suppose toward a contradiction that $p(x)>u(x)$. Then $[w,x]$ is an $F$-pooling interval for some $w \in [0,x)$ (because $p=u$ on any $F$-revealing interval, as $p=u$ $F$-a.e.). Thus, letting $y^*$ be the (unique) atom of $F$ in $[w,x]$ and applying the reasoning in the previous paragraph, we have that $u$ is reverse-S-shaped on $[y^*,x]$ and strictly concave on an open interval containing $y^*$, that $p$ is affine on $[w,x]$ and tangent to $u$ at $y^*$, and that $u$ is S-shaped on $[x,y]$, as depicted in \Cref{fig:suff:l}---a violation of the crater property.

	It remains only to show that $u$ is strictly convex on $(x-\eps,x]$ for some $\eps>0$. If $[w,x]$ is an $F$-pooling interval for some $w \in [0,x)$, then $p$ is affine on $[w,x]$, which together with $p(x)=u(x)$, the tangency of $p$ to $u$ at $y^*$ and the reverse-S-shapedness of $u$ on $[y^*,x]$ implies that $u$ is strictly convex on $(x-\eps,x]$ for some $\eps>0$. If instead $[w,x]$ is an $F$-revealing interval for some $w \in [0,x)$, then $u$ is strictly convex on $[w,x]$.
\end{proof}

\subsection{The domain of binary priors}
\label{sec:mcs_incr:binary}

Call a distribution $F$ \emph{binary} iff its support comprises at most two values:
$F = p \1_{[x,1]} + (1-p) \1_{[y,1]}$ for some $p,x,y \in [0,1]$.
When the prior distribution $F_0$ is binary,
the persuasion model is equivalent to
a simpler model in which there are just \emph{two} states,
and the sender's interim payoff at posterior belief $(q,1-q)$ is $u(q)$, for some upper semi-continuous function $u : [0,1] \to \R$.

\begin{proposition}
	\label{proposition:ido_binary}
	For upper semi-continuous $u,v : [0,1] \to \R$, \eqref{eq:mcs_incr} holds for every \emph{binary} distribution $F_0$ if, and only if, $u$ is coarsely less convex than $v$.
\end{proposition}

Thus restricting attention to binary priors obviates the need for the crater property, or indeed for any condition at all on $u$.
The proof (\cref{app:proofs:pf_thm_nondecr_binary}) is straightforward: the `if' part follows from inspection of the concave envelopes of $u$ and $v$ \parencite[à la][]{KamenicaGentzkow2011}, while the `only if' part follows from a simple construction.

We show in \cref{app:robust} that \Cref{proposition:ido_binary} is tight (so \Cref{theorem:incr} is robust): on any restricted domain of priors that contains at least one non-binary prior, `increasing' comparative statics conclusions cannot generally be drawn without a crater-like shape restriction on the interim payoff $u$.

\subsection{Special cases}
\label{sec:mcs_incr:s}

We now relate our comparative-statics results to those of \textcite{KolotilinMylovanovZapechelnyuk2022,Yoder2022,GitmezMolavi2023}. The first paper's Proposition~1 assumes that $u$ and $v$ are S-shaped and that $u$ is less convex than $v$ in the conventional sense ($v = \phi \circ u$ for some convex and strictly increasing $\phi : \co(u([0,1])) \to \R$).%
	\footnote{The authors' \emph{proof} (p.~580) identifies and then verifies a sufficient condition for comparative statics in the special case of S-shaped $u$ and $v$.
	This condition can be shown to be equivalent, in (only) that S-shaped special case, to $u$ being coarsely less convex than $v$.}
\Cref{theorem:incr} shows that S-shapedness of $v$ is superfluous, that S-shapedness of $u$ can be weakened to W-shapedness (or, more generally, the crater property), and that $u$ need only be \emph{coarsely} less convex than $v$, which admits e.g. convexity of $v-u$ as an alternative sufficient condition.
Similarly for these authors' Proposition~2.


Suppose that $u$ and $v$ are regular, that $u$ is S-shaped, and that $u'$ is more convex than $v'$ in the conventional sense: $u' = \phi \circ v'$ for some convex and strictly increasing $\phi : \co(v'([0,1])) \to \R$.
A slight extension of Gitmez and Molavi's (\citeyear{GitmezMolavi2023}, appendix~A) core argument shows that under these hypotheses, $u$ is coarsely less convex than $v$ (see \cref{app:more_clc} for the proof). Thus by \Cref{theorem:incr}, less information is provided under $u$ than under $v$, whatever the prior.
By symmetry, the same is true if $u$ is \emph{reverse-}S-shaped and $u'$ is \emph{less} convex than $v'$.
These findings generalise the main result of \textcite{GitmezMolavi2023}, which draws the same conclusion under the additional assumption that the prior is \emph{binary.}%
	\footnote{\textcite{GitmezMolavi2023} was recently superseded by \textcite{GitmezMolaviSonin2025}, which has a broader scope. Gitmez and Molavi's (\citeyear{GitmezMolavi2023}, appendix~A) core argument is no longer present as such, but a version of their main result remains (`Proposition~3').}

Finally, \textcite{Yoder2022} likewise restricts attention to binary priors, and assumes that $v-u$ is convex. This is a special case of \Cref{proposition:ido_binary}.

\section{The limits of comparative statics}
\label{sec:limits}

Our main result, \Cref{theorem:incr}, shows that the crater property is necessary and sufficient for every coarse-convexity shift to lead the sender optimally to provide more information, whatever the prior. In this section, we explore the limits of this result.

We firstly (§\ref{sec:limits:prior}) ask whether comparative-statics conclusions can be drawn when what shifts is not the sender's interim payoff $u$, but rather the prior $F_0$. We prove a negative result (\Cref{proposition:priorshift}): no shift of the prior leads the sender optimally to provide more information whatever the interim payoff $u$.

Secondly (§\ref{sec:limits:multi}), we drop the `single-moment' assumption, allowing the interim payoff $u$ to depend on any (finite) number of moments of the posterior belief. \Cref{theorem:nondecr} extends directly: coarse-convexity shifts continue to characterise `non-decreasing' comparative statics. For `increasing' comparative statics, we prove an impossibility result (\hyperref[theorem:incr_general]{\Cref*{theorem:incr}$'$}): when $u$ depends on more than a single moment, coarse-convexity shifts lead more information to be provided whatever the prior only in the trivial cases in which $u$ is concave or convex.

Some further results are given in \cref{app:robust}. We establish there that our binary-priors result \Cref{proposition:ido_binary} is tight (it fails for non-binary priors), and that the crater property remains necessary even if only specific shifts are considered (e.g. increased convexity in the conventional sense). We also examine the role of regularity, in particular differentiability.

\subsection{Shifts of the prior distribution}
\label{sec:limits:prior}

Our main results concerned comparative statics with respect to shifts of the sender's interim payoff $u$.
In this section, we consider shifts of the other primitive of the persuasion model: the distribution $F_0$ of the state.

Shifts of $F_0$ may be interpreted as changes in the information available to the sender.
In particular, suppose that the state of the world is drawn from some fixed distribution $H_0$, and that the sender herself observes only a signal, inducing posterior-mean distribution $F_0 \in \mpc{H_0}$.%
	\footnote{The sender then chooses a public signal conveying information about what she knows (i.e. the public signal must be independent of the state conditional on the sender's private signal), inducing a posterior-mean distribution $F \in \mpc{F_0}$.}
Then if the sender's signal becomes more informative, this manifests precisely as increased informativeness of $F_0$.

\begin{proposition}
	\label{proposition:priorshift}
	There exist no atomless full-support distributions $F_0 \neq G_0$ such that for every regular $u : [0,1] \to \R$,
	\begin{equation}
		\argmax_{F \in \mpc{F_0}} \int u \dd F
		\lowerthan
		\argmax_{F \in \mpc{G_0}} \int u \dd F .
		\label{eq:mcs_incr_prior}
		\tag{$\dag$}
	\end{equation}
\end{proposition}

In other words, the effect on optimal information-provision of a shift of the prior distribution $F_0$ depends finely on the interim payoff $u$: there are no shifts which deliver `increasing' comparative statics robustly across all possible interim payoffs.

The proof of \Cref{proposition:priorshift} is in \cref{app:proofs:pf_priorshift}.

\subsection{Beyond the \texorpdfstring{`}{‘}single-moment\texorpdfstring{'}{’} case}
\label{sec:limits:multi}

Our analysis has focussed on the salient case in which interim payoffs depend on only a single moment of the posterior belief---without loss, the mean. In this section, we extend our theorems to the `multi-moment' case. We find that whereas \Cref{theorem:nondecr} extends directly, yielding `non-decreasing' comparative statics, the analogue of \Cref{theorem:incr} is a negative result stating that there is no hope of `increasing' comparative statics beyond the `single-moment' case.

\subsubsection{The `multi-moment' persuasion model}
\label{sec:limits:multi:model}

In the general `multi-moment' persuasion model \parencite[e.g.][§4]{DworczakKolotilin2023}, the uncertain state of the world is a random \emph{vector} drawn from a non-empty, compact and convex set $E \subseteq \R^n$, where $n \in \N$. By \emph{distribution,} we shall mean a CDF $\R^n \to [0,1]$ concentrated on $E$. The distribution of the state (`the prior') is denoted by $F_0$. 
For distributions $F$ and $G$,
we call $F$ \emph{less informative} than $G$
iff $\int \psi \dd F \leq \int \psi \dd G$
for every continuous convex $\psi : E \to \R$.
Write `$\mpc{G}$' for the set of all distributions which are less informative than $G$.

A sender chooses a signal.
Given a signal, each signal realisation induces a posterior belief via Bayes's rule, whose expectation (a vector) we call the \emph{posterior mean.}
Each signal thus induces a random posterior mean, with some distribution.
Call a distribution \emph{feasible (given $F_0$)} iff it is the posterior-mean distribution induced by some signal.
The feasible distributions are exactly those that are less informative than the prior $F_0$ \parencite[e.g.][p.~94]{Phelps2001}.

The sender's (interim) payoff at a given realised posterior belief
is assumed to depend only on its mean:
her payoff at posterior mean $m \in E$ is $u(m)$,
where $u : E \to \R$ is upper semi-continuous.
Her problem is to choose among the feasible distributions $F$ to maximise her expected payoff $\int u \dd F$.

\begin{remark}
	\label{remark:nest}
	The special case $E \subseteq \R$ is the one studied in the rest of this paper. The persuasion model of \textcite{KamenicaGentzkow2011} is the special case in which $E$ is a simplex, i.e. the convex hull of an affinely independent set $\Omega \subseteq \R^n$,%
		\footnote{A set $S \subseteq \R^n$ is called \emph{affinely independent} iff it is finite and for any $\alpha : S \to \R$ such that $\sum_{x \in S} \alpha(x) = \sum_{x \in S} \alpha(x)x = 0$, we have $\alpha(x) = 0$ for each $x \in S$.}
	and the prior $F_0$ is concentrated on the vertices $\Omega$. The interpretation is that $\Omega = \supp(F_0)$ is the set of states of the world, the simplex $\Delta(\Omega) = E$ is the set of all possible beliefs about the state, and the interim payoff $u$ depends in an arbitrary way on the posterior belief.
\end{remark}

\subsubsection{Comparative statics}
\label{sec:limits:multi:mcs}

For any non-empty and finite set $S \subseteq E$, let $\Delta(S)$ denote the set of all maps $\alpha : S \to [0,1]$ such that $\sum_{x \in S} \alpha(x) = 1$.

\begin{definition}
	\label{definition:lessconvex_coarse_general}
	For functions $u,v : E \to \R$,
	we say that $u$ is \emph{coarsely less convex} than $v$
	if and only if for any affinely independent $S\subseteq E$ such that
	$u\left( \sum_{x \in S} \alpha(x) x \right)
	\leq \sum_{x \in S} \alpha(x) u(x)$
	holds for every $\alpha \in \Delta(S)$, we also have
	$v\left( \sum_{x \in S} \alpha(x) x \right)
	\leq \sum_{x \in S} \alpha(x) v(x)$
	for every $\alpha \in \Delta(S)$,
	and furthermore any $\alpha \in \Delta(S)$ at which the former inequality is strict is also one at which the latter inequality is strict.
\end{definition}

`Coarsely less convex than' admits the same sufficient conditions as in the `single-moment' case: \Cref{lemma:CLC_suff} and \Cref{corollary:CLC_suffsuff} remain true as stated, except with `$[0,1]$' and `$(0,1)$' replaced by `$E$'.

Our `non-decreasing' comparative-statics result, \Cref{theorem:nondecr}, remains true exactly as stated, except with `$[0,1]$' replaced by `$E$'. The same proof (\cref{app:proofs:pf_thm_nondecr}) applies, except with `$[0,1]$' replaced by `$E$' and binary distributions replaced by distributions with affinely independent support, plus a few smaller changes (e.g. replacing `$\R$' by `$\R^n$').%
	\footnote{One further small tweak to the proof is required: see \Cref{remark:app:proofs:pf_thm_nondecr:cs_priors_implies_lessvex_multi} in \cref{app:proofs:pf_thm_nondecr}.}

Recall (\cpageref{definition:regular}) our definition of regularity for functions $u : [0,1] \to \R$. We call a function $u : E \to \R$ \emph{strongly regular} iff it is twice continuously differentiable with bounded derivatives and, for all distinct $x,y \in E$, the map $[0,1] \to \R$ given by $\alpha \mapsto u(\alpha x + (1-\alpha)y)$ is regular. (We insist on second derivatives merely in order to rule out uninteresting complications.)

\begin{namedthm}[\Cref*{theorem:incr}$\boldsymbol{'}$.]
	\label{theorem:incr_general}
	Suppose that $E$ is not one-dimensional,%
		\footnote{The \emph{dimension} of a convex set $E \subseteq \R^n$ is $\max\{ \abs*{S} : \text{$S \subseteq E$ affinely independent} \} - 1$.}
	and let $u : E \to \R$ be strongly regular.
	If $u$ is either concave or convex,
	then for every strongly regular $v : E \to \R$ that is coarsely more convex than $u$
	and every distribution $F_0$,
	\begin{equation}
		\argmax_{F \in \mpc{F_0}} \int u \dd F
		\lowerthan
		\argmax_{F \in \mpc{F_0}} \int v \dd F .
		\label{eq:mcs_incr_general}
		\tag{$\star\star$}
	\end{equation}
	Conversely, if \eqref{eq:mcs_incr_general} holds
	for every strongly regular $v$ that is coarsely more convex than $u$
	and every distribution $F_0$,
	then $u$ is either concave or convex.
\end{namedthm}

In other words, comparative statics are highly prior-sensitive outside of the `single-moment' case: a coarse-convexity increase from $u$ yields increased informativeness whatever the prior $F_0$ only in the trivial cases of a concave $u$ (when full pooling is optimal) or a convex $u$ (when full revelation is optimal). The proof (\cref{app:proofs:pf_theorem_incr_general}) shows that this remains true even if only atomless convex-support priors $F_0$ are considered. The argument uses Dworczak and Kolotilin's (\citeyear{DworczakKolotilin2023}) duality techniques.

\Cref{proposition:ido_binary} similarly fails outside of the `single-moment' case: there exist priors $F_0$ with affinely independent support (the \textcite{KamenicaGentzkow2011} special case) such that \eqref{eq:mcs_incr_general} fails for some strongly regular $u,v : E \to \R$ with $u$ coarsely less convex than $v$.

\section{Applications}
\label{sec:appl}

In this section, we apply our theorems to the problems of designing quality ratings \parencite{Zapechelnyuk2020}, costly information acquisition \parencite[e.g.][]{RavidRoeslerSzentes2022}, and delegation \parencite[e.g.][]{KolotilinZapechelnyuk2025}.

\subsection{Certification design}
\label{sec:appl:certification}

In the model of \textcite{Zapechelnyuk2020}, a regulator designs a `certification rule' $\pi : \R_+ \to \R_+$. Then, a producer chooses a price $p \in \R_+$ and quality $x \in \R_+$, at a cost of $c(x)/t$, where $t \in (0,\infty)$ is a parameter. A representative consumer then observes the price $p$ and the `(quality) rating' $\pi(x)$ (but not the actual quality $x$), and chooses whether to buy ($a=1$) or not ($a=0$). If she does not buy, she consumes an outside option worth $\theta$.

The types $\theta$ and $t$ are statistically independent, with respective CDFs $F$ and $G$, where $\supp(F)=[0,1]$, and where $\supp(G) \subseteq (0,\infty)$ is bounded. Only the consumer (the producer) knows the realisation of $\theta$ (of $t$). The producer's payoff is $ap-c(x)/t$, the consumer's payoff is $a(x-p-\theta)$, and the regulator seeks to maximise consumer surplus.

Under the additional assumption that, for some $\alpha \in (0,\infty)$, $F(\theta)=\theta^\alpha$ for each $\theta \in [0,1]$ and $x \mapsto c'(x)/x^\alpha$ is strictly increasing and not pointwise lower than $\min \supp(G)$, \textcite{Zapechelnyuk2020} shows that the regulator's problem is equivalent to persuasion with interim payoff $u \coloneqq -G$ and a prior $F_0$ that depends on $\alpha$ but not on $G$. Certification rules correspond to monotone partitional signals, and a more informative rule is exactly one under which every producer type is weakly better off.

This permits comparative-statics conclusions to be drawn. In particular, the interim payoff $u$ becomes coarsely more convex whenever the CDF $G$ worsens, e.g. in the monotone-likelihood-ratio sense. If $G$ is unimodal (as \textcite{Zapechelnyuk2020} mostly assumes), or more generally M-shaped, then $u$ satisfies the crater property. \Cref{theorem:incr} then implies that when quality becomes costlier to produce (i.e. $G$ worsens), quality ratings should optimally become more informative. This has implications for e.g. the design of health insurance enrolment platforms \parencite[see][]{Vatter2022}.

\subsection{Costly information acquisition (\texorpdfstring{`}{‘}rational inattention\texorpdfstring{'}{’})}
\label{sec:appl:costly}

In the literature on costly information acquisition with mean-measurable costs \parencite[e.g.][]{RavidRoeslerSzentes2022,MenschRavidThereze2026,Thereze2023adv,Kreutzkamp2023,MenschMalik2023}, a decision-maker chooses flexibly how to learn before taking an action. Each posterior-mean distribution $F$ has a cost $C(F)$ and a benefit $W(F)$. These are assumed to be posterior-mean-separable: $C(F) = \int c \dd F - c(\mu_0)$ and $W(F) = \int w \dd F$ for each $F \in \mpc{F_0}$, where $c,w : [0,1] \to \R$ are convex and continuous, and $\mu_0 \coloneqq \int x F_0(\dd x)$ denotes the prior mean.
The interim benefit $w$ is interpreted as arising from a decision problem: $w(x) = \sup_{a \in \mathcal{A}} U(a,x)$ for each $x \in [0,1]$, where $U(a,m)$ denotes the interim payoff of action $a \in \mathcal{A}$ given posterior mean $m \in [0,1]$.
The decision-maker's flexible-learning problem is to choose $F \in \mpc{F_0}$ to maximise $W(F)-C(F)$.
This is nested by the persuasion model, with $u \coloneqq w-c$.

Following the literature,%
	\footnote{E.g. \textcite{ChambersLiuRehbeck2020,Denti2022,Whitmeyer2024}.}
we say that \emph{information becomes more valuable} when the interim benefit $w$ shifts to $\widetilde{w} = w + \psi$, where $\psi : [0,1] \to \R$ is convex. Changes of the underlying decision problem $(\mathcal{A},U)$ which cause information to become more valuable include raising the stakes \parencite{Whitmeyer2024}, adding actions (in some cases---see \textcite{Whitmeyer2024}), and adding decisions \parencite{Delara2025}.%
	\footnote{`Raising the stakes' means replacing $U$ by $kU$ for some $k \in [1,\infty)$, and `adding actions' means replacing $\mathcal{A}$ by $\mathcal{A} \cup \mathcal{B}$ for some set $\mathcal{B}$. `Adding decisions' means replacing $\mathcal{A}$ by $\mathcal{A} \times \mathcal{B}$ and $U$ by $(a,b,m) \mapsto U(a,m)+V(b,m)$ for some non-empty set $\mathcal{B}$ and function $V : \mathcal{B} \times [0,1] \to \R$ such that $V(b,\cdot)$ is affine for each $b \in \mathcal{B}$.}
When information becomes more valuable, the interim payoff $u=w-c$ becomes coarsely more convex by \Cref{corollary:CLC_suffsuff}. Hence by \Cref{theorem:nondecr}, the agent optimally learns no less.%
	\footnote{This recovers part of Theorem~1 in \textcite{Whitmeyer2024}.}
The same occurs when information becomes cheaper in the sense that the interim cost $c$ shifts to $\widetilde{c} = c - \psi$, where $\psi : [0,1] \to \R$ is convex.

In case the prior $F_0$ is binary, \Cref{proposition:ido_binary} provides that when information becomes cheaper or more valuable, the decision-maker optimally learns \emph{more.} This result directly applies to Denti's (\citeyear{Denti2022}, §IV) experimental test of the costly-information-acquisition model's comparative-statics predictions.

Beyond the binary-prior case, \Cref{theorem:incr} suggests that results about the decision-maker optimally learning \emph{more} will prove elusive. The mere fact that $u$ is the difference of two convex functions implies almost nothing.%
	\footnote{For any continuous $v : [0,1] \to \R$ and any $\eps>0$, there are convex $c,w : [0,1] \to \R$ such that $\sup_{x \in [0,1]} \abs*{ v(x) - [w(x)-c(x)] } < \eps$ \parencite[see e.g.][Lemma~S.3]{Sinander2022}.}
Rather, satisfaction by $u=w-c$ of the crater property (or sufficient conditions like W-shapedness) depends on the relative curvatures of the interim cost $c$ and interim benefit $w$, requiring either strong assumptions or hard-to-interpret joint restrictions. Some simple examples exist: for instance, $u=w-c$ is W-shaped if $c(x) \coloneqq \kappa \abs*{x-\mu_0}$ for each $x \in [0,1]$, where $\kappa>0$.

This negative finding seems robust. For example, replacing `less informative than' with the standard Blackwell order does not help (since `Blackwell less informative than' implies `less informative than'---see \cref{app:robust}). Dropping the posterior-mean-separability assumption yields an even more negative result, in light of \hyperref[theorem:incr_general]{\Cref*{theorem:incr}$'$} (§\ref{sec:limits:multi} above).

\subsection{Delegation}
\label{sec:appl:delegation}

In this section, we apply our results to draw comparative-statics conclusions about delegation, exploiting the well-known equivalence between delegation and persuasion. This complements the existing delegation literature, which contains few general comparative-statics results, as discussed in §\ref{sec:intro:lit_pers} above.

There is a principal with authority over which action $a \in \R$ is taken, and an agent with private information in the form of a type $x \in X \subseteq \R$. The agent's type is drawn from a continuous CDF $G$, and $\supp(G)=X$ is compact and convex. The agent's and principal's payoffs are, respectively,
\begin{equation*}
	U(a,x) = ax - c(a)
	\quad \text{and} \quad
	V(a,x) = ab(x)- c(a) ,
\end{equation*} 
where $b : X \to \R$ is bounded and Borel measurable, and $c : \R \to \R$ is differentiable and strictly convex with $\abs*{c'(a)} \to \infty$ as $\abs*{a} \to \infty$.

A \emph{direct mechanism} is a family $A = (A_x)_{x \in X}$ of probability distributions on $\R$ such that $U(\cdot,x)$ and $V(\cdot,x)$ are $A_x$-integrable for each $x \in X$ and the map $x \mapsto A_x$ is Borel measurable with respect to the topology of weak convergence. Given a direct mechanism $A$, define $U_A,V_A : X \to \R$ by
\begin{equation*}
	U_A(x) \coloneqq \int U(a,x) A_x(\dd a)
	\quad \text{and} \quad
	V_A(x) \coloneqq \int V(a,x) A_x(\dd a)
\end{equation*} 
for each $x \in X$. Say that $A$ is \emph{incentive compatible (IC)} iff $U_A(x) \ge \int U_A(a,x)A_y(\dd a)$ for all $x,y \in X$. By the revelation principle, it is without loss to restrict attention to IC direct mechanisms, henceforth `IC mechanisms'.

A direct mechanism $A$ is called \emph{undominated} iff there is no direct mechanism $A^\dag$ that is constant ($A^\dag_x = A^\dag_y$ for all $x,y \in X$) and satisfies $U_{A^\dag}(x) > U_A(x)$ and $V_{A^\dag}(x) > V_A(x)$ for every $x \in X$. Without loss of optimality, the principal chooses among undominated IC mechanisms.

\textcite{KolotilinZapechelnyuk2025} and \textcite[§4.3]{KleinerMoldovanuStrack2021} show that the problem of choosing an undominated IC mechanism $A$ to maximise the principal's expected payoff $\int V_A \dd G$ is equivalent to a persuasion problem, and that choosing a more informative distribution in the latter problem is equivalent to granting the agent greater \emph{discretion,} in the sense that every agent type earns a higher expected utility.
The argument runs as follows.%
	\footnote{There are some differences between \textcite{KolotilinZapechelnyuk2025} and \textcite[§4.3]{KleinerMoldovanuStrack2021}; our argument here is closer to that in the former paper.}
Given any $B > 0$, there are $\underline a \in (-\infty,(c')^{-1}(\min X))$ and $\bar a \in ((c')^{-1}(\max X),\infty)$ such that, for any Borel-measurable $b$ with $\abs*{b} \le B$, an IC mechanism $A$ is undominated only if $U_A(x) \ge \max\{U(\underline a,x),U(\bar a,x)\}$ for every $x \in X$.%
	\footnote{Proposition~1 in \textcite{KolotilinZapechelnyuk2025} actually delivers such $\underline a$ and $\bar a$ only for each \emph{fixed} $b$. We provide the missing details in \Cref{lemma:delegation_oo} in \cref{app:delegation}.
}
Assume without loss that $\underline a = 0 = c'(0)$ and that $\bar a = 1 = c'(1)$.%
	\footnote{\Cref{lemma:delegation_wlog} in \cref{app:delegation} establishes that this involves no loss of generality.}
Let $F_0$ be the distribution given by $F_0(x) \coloneqq (c')^{-1}(x)$ for each $x \in [0,1]$.
Extend $b$ from its domain $X$ ($\subseteq (0,1)$) to a bounded and Borel-measurable map $[0,1] \to \R$, and define $u : [0,1] \to \R$ by
\begin{equation*}
	u(x) \coloneqq \int_0^x [x - b(y)] G(\dd y)
	\quad \text{for each $x \in [0,1]$.}
\end{equation*}
Set $\beta \coloneqq \int b \dd G - \int x F_0(\dd x)-c(0)$. 
For any undominated IC mechanism $A$, there exists an $F \in \mpc{F_0}$ such that
\begin{align}
	\int V_A \dd G
	&= \int u \dd F + \beta
	\tag{D.1}
	\label{eq:delegation_1}
	\\
	\text{and} \quad
	U_A(x)
	&= \int_0^x F-c(0) \qquad \text{for each $x \in X$,}
	\tag{D.2}
	\label{eq:delegation_2}
\end{align}
and conversely, for any $F \in \mpc{F_0}$, there exists an IC mechanism $A$ such that \eqref{eq:delegation_1} and \eqref{eq:delegation_2} hold.%
	\footnote{This follows from Lemma~3, Claim~3 and the proof of Claim~4 in \textcite{KolotilinZapechelnyuk2025}.}
Hence the principal's problem of choosing an undominated IC mechanism $A$ to maximise $\int V_A \dd G$ is equivalent to the persuasion problem with prior $F_0$ and utility function $u$, and greater informativeness in persuasion corresponds in delegation to `greater agent discretion', i.e. a higher expected payoff for each agent type.

Observe that $F_0$ does not vary with either $b$ or $G$. Hence by \Cref{theorem:nondecr,theorem:incr}, any shift of $b$ or $G$ which makes $u$ coarsely more convex cannot lead the principal to grant the agent strictly less discretion, and leads her to grant the agent \emph{more} discretion in case $u$ satisfies the crater property.

It remains only to identify sufficient conditions on shifts of the primitive parameters $b$ and $G$ which produce coarse-convexity shifts of $u$, and to identify primitive sufficient conditions for $u$ to satisfy the crater property. We do this in \Cref{proposition:delegation_clc_b,proposition:delegation_clc_G,proposition:delegation_crater} below. It is routine to verify the sufficiency of these conditions, but for completeness, the steps are recounted in \cref{app:delegation}.

We first identify shifts of $b$ which lead $u$ to become coarsely more convex.

\begin{observation}
	\label{proposition:delegation_clc_b}
	Let $d : [0,1] \to \R$ be bounded and Borel measurable, and let $v : [0,1] \to \R$ be given by $v(x) = \int_0^x [x - d(y)] G(\dd y)$ for each $x \in [0,1]$. Define $X' \coloneqq \{x \in X : b(x) \ne x \}$, and let $\lambda : X' \to \R$ be given by
	\begin{equation*}
		\lambda(x) \coloneqq \frac{x-d(x)}{x-b(x)} \quad \text{for each $x \in X'$.}
	\end{equation*} 
	Then $u$ is coarsely less convex than $v$ if either of the following holds: 
	
	\begin{enumerate}[label=(\roman*)]
		
		\item \label{item:delegation_b_pos} 
		$x \le d(x) \le b(x)$ for each $x \in X$, and $\lambda$ is decreasing.
		
		\item \label{item:delegation_b_neg} 
		$b(x) \le d(x) \le x$ for each $x \in X$, and $\lambda$ is increasing.
	
	\end{enumerate}
\end{observation}

The difference $b(x)-x$ is the principal's \emph{bias.} \Cref{proposition:delegation_clc_b} provides conditions under which a decrease in (the magnitude of) the bias leads to a coarse-convexity increase of $u$. It suffices that the sign of the bias be independent of the agent's type (positive in case \ref{item:delegation_b_pos} and negative in \ref{item:delegation_b_neg}) and that the (relative) reduction in magnitude be suitably monotone in the agent's type. The conditions of \Cref{proposition:delegation_clc_b} are satisfied if the bias $x \mapsto b(x)-x$ is constant and its magnitude falls, or if the principal's utility shifts to $(a,x) \mapsto (1-\alpha) V(a,x)+\alpha U(a,x)$ where $\alpha \in (0,1)$.

We now consider which shifts of $G$ lead $u$ to become coarsely more convex.

\begin{observation}
	\label{proposition:delegation_clc_G}
	Let $H$ be a continuous CDF with support $X$, and let $v : [0,1] \to \R$ be given by $v(x) \coloneqq \int_0^x [x - b(y)] H(\dd y)$ for each $x \in [0,1]$.
	Then $u$ is coarsely less convex than $v$ if $b$ is increasing and one of the following holds:
	
	\begin{enumerate}[label=(\roman*)]
		
		\item \label{item:delegation_G_down} $G$ is higher than $H$ in the monotone-likelihood-ratio sense, and $b(x) \ge x$ for each $x \in X$.

		\item \label{item:delegation_G_up} $G$ is lower than $H$ in the monotone-likelihood-ratio sense, and $b(x) \le x$ for each $x \in X$.

	\end{enumerate}
\end{observation}

Finally, we identify primitive sufficient conditions for $u$ to satisfy the crater property and regularity, allowing \Cref{theorem:incr} to be applied. As mentioned in §\ref{sec:mcs_incr:w} and proved in \cref{app:incr_affine}, \Cref{theorem:incr} remains true when regularity is replaced by \emph{weak regularity,} which differs from regularity (defined on \cpageref{definition:regular}) only in that it allows affine segments.

\begin{observation}
	\label{proposition:delegation_crater}
	Suppose that $G$ is twice differentiable with $G'$ log-concave%
		\footnote{That is, $\{x \in \R: G'(x) > 0\}$ is an interval, and $\log G'$ is concave on this interval.}
	and $G'(\min X) = 0 = G'(\max X)$, that $b$ is continuously differentiable on $\interior(X)$ with $b(x) \ne x$ for each $x \in \interior(X)$, that $x \mapsto (2-b'(x))/(x-b(x))$ is decreasing on $\interior(X)$, and that $\int b \dd G \in X$.
	Then $u$ is weakly regular and satisfies the crater property.
\end{observation}

A constant bias $x \mapsto b(x)-x$ satisfies the monotonicity requirement. The hypothesis $\int b \dd G \in X$ means that the principal's bias is not too strong on average (a demanding sufficient condition is $b(X) \subseteq X$). The hypotheses of \Cref{proposition:delegation_crater} guarantee that $u$ is either S-shaped or reverse-S-shaped (depending on the sign of $x \mapsto x-b(x)$). If instead $G'$ is log-convex and $x \mapsto (2-b'(x))/(x-b(x))$ is increasing, then the conclusion remains true provided the hypothesis $G'(\min X) = 0 = G'(\max X)$, which is incompatible with log-convexity of $G'$, is replaced by a (demanding) restriction on $b$.

\section{Mean-preserving spreads}
\label{sec:mps}

For two distributions $F$ and $G$, recall that $G$ is called a mean-preserving spread of $F$ if and only if $F \in \mpc{G}$. We write `$\mps{F}$' for the set of all mean-preserving spreads of $F$. Persuasion (under the `single-moment' assumption) amounts to choosing a distribution $F \in \mpc{F_0}$ to maximise $\int u \dd F$, where the function $u : [0,1] \to \R$ and the distribution $F_0$ are given.

The analogous problem of choosing $F \in \mps{F_0}$ to maximise $\int u \dd F$ also has economic applications, including persuasion with heterogeneous priors \parencite{OnuchicRay2023}, persuasion with `Yaari dual' preferences \parencite{JewittQuigley2023}, information disclosure in auctions \parencite{Bergemannetal2022}, monopoly pricing with rationing and quality conflation \parencite{LoertscherMuir2022}, redistributive allocation of goods of heterogeneous quality \parencite{AkbarpourDworczakKominers2024}, and certain delegation problems \parencite{KolotilinZapechelnyuk2025}.%
	\footnote{See also \textcite{AugiasUhe2025}, who consider constraint sets of the form $\mpc{F_0} \cap \mps{G_0}$, and identify economic applications in which these arise.}

In this section, we show that when the constraint set is $\mps{F_0}$ rather than $\mpc{F_0}$, comparative statics is straightforward: a more informative distribution is chosen whenever $u$ becomes coarsely more convex (with no need for any crater-type shape restriction---cf. \Cref{theorem:incr}), and also whenever $F_0$ becomes more informative (cf. \Cref{proposition:priorshift}).

Given an interval $I \subseteq [0,1]$, we write `$\mpsI{I}{F}$' for the set of all $G \in \mps{F}$ with $\supp(G) \subseteq I$. The following shows that when the constraint set is $\mpsI{I}{F_0}$, any coarse-convexity shift produces comparative statics.

\begin{proposition}
	\label{proposition:mpsxy}
	Let $u,v : [0,1] \to \R$ be upper semi-continuous.
	The following are equivalent: 
	
	\begin{enumerate}[label=(\roman*)]
		
		\item \label{item:mpsxy_clc} 
		$u$ is coarsely less convex than $v$.
		
		\item \label{item:mpsxy_nondecr}
		For any closed interval $I \subseteq [0,1]$ and distribution $F_0$ with $\supp(F_0) \subseteq I$,
		\begin{equation*}
			\argmax_{F \in \mpsI{I}{F_0}} \int u \dd F
			\notstrictlyhigherthan
			\argmax_{F \in \mpsI{I}{F_0}} \int v \dd F .
		\end{equation*}
		
		\item \label{item:mpsxy_incr} 
		For any closed interval $I \subseteq [0,1]$ and distribution $F_0$ with $\supp(F_0) \subseteq I$,
		\begin{equation*}
			\argmax_{F \in \mpsI{I}{F_0}} \int u \dd F
			\lowerthan
			\argmax_{F \in \mpsI{I}{F_0}} \int v \dd F .
		\end{equation*}

	\end{enumerate}
\end{proposition}

\Cref{proposition:mpsxy} is one analogue of \Cref{theorem:nondecr,theorem:incr} for mean-preserving spreads. The interpretation of the constraint set $\mpsI{I}{F_0}$ is that we consider all mean-preserving spreads of $F_0$ in a problem in which the state is known to be drawn from $I$ (rather than $[0,1]$). An alternative analogue of \Cref{theorem:nondecr,theorem:incr}, in which only constraint sets of the form $\mps{F_0}$ ($=\mpsI{[0,1]}{F_0}$) are contemplated, is given in \cref{app:mps} (obviously `$u$ coarsely less convex than $v$' remains sufficient, but it is not necessary). The upshot from either analogue is that comparative statics with a mean-preserving-spread constraint set is similar to the binary-prior special case of the persuasion model (\Cref{proposition:ido_binary}): any coarse-convexity shift will suffice.

One implication of \Cref{proposition:mpsxy} (together with \Cref{corollary:CLC_suffsuff}) is that if $u$ is less convex than $v$ in the conventional sense ($v = \phi \circ u$ for some convex and strictly increasing $\phi : \co(u([0,1])) \to \R$), then \ref{item:mpsxy_incr} holds. This recovers Proposition~6 in \textcite{OnuchicRay2023} and Proposition~1 in \textcite{JewittQuigley2023}, which (to our knowledge) are the only extant comparative-statics results for $\mps{F_0}$ constraint sets.

\Cref{proposition:mpsxy} is a consequence of the following fundamental result, which establishes an equivalence between comparative statics for all constraint sets $\mps{F_0}$ with arbitrary $F_0$ and comparative statics for all constraint sets $\mpc{F_0}$ with special $F_0$, namely those satisfying $\supp(F_0) \subseteq \{0,1\}$.

\begin{lemma}
	\label{lemma:mps_mpc}
	Let $u,v : [0,1] \to \R$ be upper semi-continuous. Then
	\begin{equation}
		\argmax_{F \in \mps{F_0}} \int u \dd F
		\notstrictlyhigherthan			
		\argmax_{F \in \mps{F_0}} \int v \dd F 
		\label{eq:mps_nohigher}
		\tag{$\star_\text{S}$}
	\end{equation}
	holds for every distribution $F_0$ if and only if
	\begin{equation}
		\argmax_{F \in \mpc{F_0}} \int u \dd F
		\notstrictlyhigherthan
		\argmax_{F \in \mpc{F_0}} \int v \dd F 
		\label{eq:mpc_nohigher}
		\tag{$\star$}
	\end{equation}
	holds for every distribution $F_0$ with $\supp(F_0) \subseteq \{0,1\}$. Similarly,
	\begin{equation}
		\argmax_{F \in \mps{F_0}} \int u \dd F
		\lowerthan
		\argmax_{F \in \mps{F_0}} \int v \dd F 
		\label{eq:mps_lower}
		\tag{$\star\star_\text{S}$}
	\end{equation}
	holds for every distribution $F_0$ if and only if
	\begin{equation}
		\argmax_{F \in \mpc{F_0}} \int u \dd F
		\lowerthan
		\argmax_{F \in \mpc{F_0}} \int v \dd F 
		\label{eq:mpc_lower}
		\tag{$\star\star$}
	\end{equation}
	holds for every distribution $F_0$ with $\supp(F_0) \subseteq \{0,1\}$.	
\end{lemma}

To understand the special role of `extreme binary' priors $F_0$, i.e. those with $\supp(F_0) \subseteq \{0,1\}$, observe that among all distributions with a given mean $x \in [0,1]$, the least informative is the degenerate one, $\1_{[x,1]}$, and the most informative is the extreme binary one, call it $F^x$, so that $\mpc{F^x} = \mps{\1_{[x,1]}}$. It follows that considering $\mpc{F_0}$ for every extreme binary distribution $F_0$ is equivalent to considering $\mps{F_0}$ for every degenerate distribution $F_0$. This immediately implies the `only if' claims in \Cref{lemma:mps_mpc}.

In light of this observation, what the `if' claims in \Cref{lemma:mps_mpc} assert is simply that comparative statics on $\mps{F_0}$ for every \emph{degenerate} $F_0$ implies comparative statics on $\mps{F_0}$ for \emph{every} $F_0$. The (straightforward but not trivial) proof of this assertion is given in \cref{app:mps}.

\begin{proof}[Proof of \Cref{proposition:mpsxy}]
	By \Cref{lemma:mps_mpc}, \ref{item:mpsxy_nondecr} is equivalent to 
	
	\begin{enumerate}[label=(\roman*$'$)]
		\setcounter{enumi}{1}
		\item \label{item:mpsxy_nondecr_binary} \eqref{eq:mpc_nohigher} holds for all binary distributions $F_0$,
	\end{enumerate}
	and \ref{item:mpsxy_incr} is equivalent to 
	
	\begin{enumerate}[label=(\roman*$'$),resume]
		\item \label{item:mpsxy_incr_binary} \eqref{eq:mpc_lower} holds for all binary distributions $F_0$.
	\end{enumerate}
	And \ref{item:mpsxy_clc} implies \ref{item:mpsxy_incr_binary} implies \ref{item:mpsxy_nondecr_binary} implies \ref{item:mpsxy_clc}: the first implication follows from \Cref{proposition:ido_binary}, the second is immediate, and the third follows from the proof in \cref{app:proofs:pf_thm_nondecr:cs_priors_implies_lessvex} of the second (converse) part of {\hyperref[theorem:nondecr_extended]{\Cref*{theorem:nondecr}$^*$}}.%
		\footnote{The argument there shows that
		if $u$ is not coarsely less convex than $v$,
		then we can construct a prior $F_0$
		such that \eqref{eq:mpc_nohigher} fails.
		And the constructed prior is, in fact, binary.}
\end{proof}

Finally, we examine comparative statics for shifts of the prior $F_0$.

\begin{proposition}
	\label{proposition:mps_prior}
	For distributions $F_0$ and $G_0$, $F_0$ is less informative than $G_0$ if and only if for any upper semi-continuous $u : [0,1] \to \R$,
	\begin{equation*}
		\argmax_{F \in \mps{F_0}} \int u \dd F
		\lowerthan
		\argmax_{F \in \mps{G_0}} \int u \dd F .
	\end{equation*}
\end{proposition}

The proof of \Cref{proposition:mps_prior} is in \cref{app:mps}.



\begin{appendices}

\crefalias{section}{appsec}
\crefalias{subsection}{appsec}
\crefalias{subsubsection}{appsec}

\section{Proofs of main results}
\label{app:proofs}

\subsection{Product structure of \texorpdfstring{`}{‘}less informative than\texorpdfstring{'}{’}}
\label{app:proofs:product}

In this appendix, we characterise
the `less informative than' order on distributions
in terms of the product order on convex functions $[0,1] \to \R$.
This result will be used in \cref{app:proofs:pf_thm_nondecr,app:proofs:pf_thm_incr} below.

For each $F \in \mpc{F_0}$,
let $C_F$ denote the function $[0,1] \to \R$
given by $C_F(x) \coloneqq \int_0^x F$ for each $x \in [0,1]$.
Let $\mathcal{C}$ be the space of all continuous convex functions $C : [0,1] \to \R$
whose right-hand derivative $C^+ : [0,1) \to \R$ satisfies $0 \leq C^+ \leq 1$
and which obey $C(x) \leq \int_0^x F_0$ for every $x \in [0,1]$, with equality at $x=0$ and $x=1$.
Given any $C \in \mathcal{C}$,
define $C^+(1) \coloneqq 1$ by convention.
The \emph{product order} (or `pointwise order') on $\mathcal{C}$
is the partial order in which $C$ smaller than $C'$
if and only if $C(x) \leq C'(x)$ for every $x \in [0,1]$.

The following extends Gentzkow and Kamenica's (\citeyear{GentzkowKamenica2016}) observation:
not only do distributions $F$ correspond one-to-one with convex functions $C_F$,
but greater informativeness of $F$
is equivalent to $C_F$ being pointwise higher.

\begin{lemma}
	\label{lemma:isomorphism}
	Fix a prior $F_0$.
	The map $F \mapsto C_F$ is a bijection $\mpc{F_0} \to \mathcal{C}$
	(with inverse $C \mapsto C^+$),
	and is increasing when $\mpc{F_0}$ is ordered by `less informative than' and $\mathcal{C}$ has the product order.
	Thus $\mpc{F_0}$ and $\mathcal{C}$ are order-isomorphic.
\end{lemma}

\begin{proof}
	Clearly the map $F \mapsto C_F$ carries $\mpc{F_0}$ into $\mathcal{C}$, and is increasing.
	The map $C \mapsto C^+$ similarly carries $\mathcal{C}$ into $\mpc{F_0}$, and by inspection $F = C_F^+$ for every $F \in \mpc{F_0}$;
	so we've found an inverse of $F \mapsto C_F$ defined on all of $\mathcal{C}$, meaning that $F \mapsto C_F$ is bijective.
\end{proof}

\begin{corollary}
	\label{corollary:lattice}
	For any given prior $F_0$,
	the set $\mpc{F_0}$ ordered by `less informative than'
	is a complete lattice.%
		\footnote{This result is well-known: see e.g. \textcite[Theorem~3.3]{MullerScarsini2006}.}
\end{corollary}

\begin{proof}
	By \Cref{lemma:isomorphism},
	we need only show that when $\mathcal{C}$ has the product order,
	it holds for any non-empty $\mathcal{C}' \subseteq \mathcal{C}$
	that $C^\star \coloneqq \sup_{C \in \mathcal{C}'} C$ is its least upper bound,
	and that the convex envelope of $\inf_{C \in \mathcal{C}'} C$, which we'll call $C_\star$, is its greatest lower bound.
	For the former, $C^\star$ clearly belongs to $\mathcal{C}$, is clearly an upper bound of $\mathcal{C}'$,
	and is clearly pointwise smaller than any other upper bound.
	For the latter, $C_\star$ is an element of $\mathcal{C}$, is clearly a lower bound of $\mathcal{C}'$,
	and exceeds every other lower bound by definition of the convex envelope.
\end{proof}

\subsection{Proof of \texorpdfstring{\Cref{theorem:nondecr} (\cpageref{theorem:nondecr})}{Theorem \ref{theorem:nondecr} (p. \pageref{theorem:nondecr})}}
\label{app:proofs:pf_thm_nondecr}

We shall prove the following generalisation of \Cref{theorem:nondecr}. For two distributions $F$ and $G$, write $[G,F]$ for the set of all distributions that are more informative than $G$ and less informative than $F$. ($[G,F]$ is called an \emph{order interval.})

\begin{namedthm}[\Cref*{theorem:nondecr}$\boldsymbol{^*}$.]
	\label{theorem:nondecr_extended}
	For upper semi-continuous $u,v : [0,1] \to \R$, the following are equivalent:
	
	\begin{enumerate}[label=(\roman*)]

		\item \label{item:lessvex}
		$u$ is coarsely less convex than $v$.

		\item \label{item:cs_priors}
		For every distribution $F_0$,
		\eqref{eq:mcs} holds.

		\item \label{item:cs_intervals}
		For all distributions $G_0,F_0$
		such that $G_0$ is less informative than $F_0$
		and $\int u \dd F, \int v \dd G > -\infty$
		for some $F,G \in [G_0,F_0]$,
		\begin{equation*}
			\argmax_{F \in [G_0,F_0]}
			\int u \dd F
			\notstrictlyhigherthan
			\argmax_{F \in [G_0,F_0]}
			\int v \dd F .
		\end{equation*}

	\end{enumerate}
	
\end{namedthm}

In proving \hyperref[theorem:nondecr_extended]{\Cref*{theorem:nondecr}$^*$}, we shall write $\mu_F$ for the mean of a distribution $F$,
and shall sometimes abbreviate `$F$ is less informative than $G$' to `$F \preceq G$'.
For $x,y \in \R$ and $\alpha \in [0,1]$, we shall write $x_\alpha y \coloneqq \alpha x + (1-\alpha) y$.

In \hyperref[theorem:nondecr_extended]{\Cref*{theorem:nondecr}$^*$},
property \ref{item:cs_intervals} implies property \ref{item:cs_priors}
because $\mpc{F_0} = \left[ \nu, F_0 \right]$ where $\nu$ is the point mass concentrated on $\mu_{F_0}$, and obviously $\int u \dd \nu, \int v \dd \nu > -\infty$.
We shall prove that \ref{item:cs_priors} implies \ref{item:lessvex} and that \ref{item:lessvex} implies \ref{item:cs_intervals}.

\subsubsection{Proof that \ref{item:cs_priors} implies \ref{item:lessvex}}
\label{app:proofs:pf_thm_nondecr:cs_priors_implies_lessvex}

Observe that given $u,v : [0,1] \rightarrow \R$, $u$ is coarsely less convex than $v$ iff for any $x<z$ in $[0,1]$ satisfying
\begin{equation}
	u(x_\alpha z)
	\leq u(x)_\alpha u(z)
	\quad \text{for all $\alpha \in (0,1)$,}
	\tag{$\triangle$}
	\label{eq:chord_convex:if}
\end{equation}
it holds for each $\alpha \in (0,1)$ that
\begin{equation}
	u(x_\alpha z)
	\leq \mathrel{(<)} u(x)_\alpha u(z)
	\quad \text{implies} \quad
	v(x_\alpha z) \leq \mathrel{(<)} v(x)_\alpha v(z) .
	\tag{$\mathord{\Rightarrow} : \alpha$}
	\label{eq:chord_convex:then}
\end{equation}

We prove the contra-positive.
Assume that \ref{item:lessvex} fails,
meaning there are $x<z$ in $[0,1]$ and an $\alpha \in (0,1)$
such that \eqref{eq:chord_convex:if} holds and \eqref{eq:chord_convex:then} fails;
we seek a distribution $F_0$ such that
\begin{equation*}
	M_{F_0}(u)
	\coloneqq \argmax_{F \in \left[\nu,F_0\right]} \int u \dd F
	\strictlyhigherthan
	M_{F_0}(v) \coloneqq \argmax_{F \in \left[\nu,F_0\right]} \int v \dd F ,
\end{equation*}
where $\nu$ denotes the point mass concentrated on $\mu_{F_0}$.
For this, it suffices that
$F_0 \in M_{F_0}(u)$ and $\mu_{F_0} \in M_{F_0}(v)$
(so that $M_{F_0}(u)$ is higher than $M_{F_0}(v)$)
and that either $F_0 \notin M_{F_0}(v)$ or $\mu_{F_0} \notin M_{F_0}(u)$
(so that $M_{F_0}(v)$ is not higher than $M_{F_0}(u)$).
We shall use the standard `concavification' reasoning \parencite[see][]{KamenicaGentzkow2011}. Consider two cases.

\emph{Case~1: $v(x_\beta z) \leq v(x)_\beta v(z)$ for all $\beta \in (0,1)$.}
Let $F_0$ be the distribution assigning weight $\alpha$ to $x$ and $1-\alpha$ to $z$, so that $\mu_{F_0} = x_\alpha z$.
$F_0$ belongs to $M_{F_0}(v)$ by the hypothesis of this case, and to $M_{F_0}(u)$ by \eqref{eq:chord_convex:if}.
Since $v(x_\alpha z) \leq v(x)_\alpha v(z)$ and \eqref{eq:chord_convex:then} fails by hypothesis,
it must be that
$u(x_\alpha z) < u(x)_\alpha u(z)$ and $v(x_\alpha z) = v(x)_\alpha v(z)$, or equivalently
$u(\mu_{F_0}) < \int u \dd F_0$
and $v(\mu_{F_0}) = \int v \dd F_0$.
Then $\mu_{F_0}$ belongs to $M_{F_0}(v)$ but not to $M_{F_0}(u)$.

\emph{Case~2: $v(x_\beta z) > v(x)_\beta v(z)$ for some $\beta \in (0,1)$.}
Assume without loss that $v(x_\alpha z) > v(x)_\alpha v(z)$.
Let $\widehat{v}$ be the concave envelope (i.e. pointwise least majorant) of the restriction of $v$ to $[x,z]$, and note that $\widehat{v}(x_\alpha z) \geq v(x_\alpha z) > v(x)_\alpha v(z)$
and (since $v$ is upper semi-continuous)
that $\widehat{v}(x) = v(x)$ and $\widehat{v}(z) = v(z)$.
Then there is a $\gamma \in (0,1)$ such that $\widehat{v}$ is not affine on any neighbourhood of $x_{\gamma} z$,
and $\widehat{v}(x_{\gamma} z) = v(x_{\gamma}z)$ since $v$ is upper semi-continuous.
Let $F_0$ be the distribution assigning weight $\gamma$ to $x$ and $1-\gamma$ to $z$, so that $\mu_{F_0} = x_{\gamma} z$.
Then $F_0 \notin M_{F_0}(v)$ since $\widehat{v}$ is not affine, and $\mu_{F_0} \in M_{F_0}(v)$ since $\widehat{v}(\mu_{F_0}) = v(\mu_{F_0})$.
And $F_0$ belongs to $M_{F_0}(u)$ by \eqref{eq:chord_convex:if}.
\qed

\begin{remark}
	\label{remark:app:proofs:pf_thm_nondecr:cs_priors_implies_lessvex_multi}
	We claimed in §\ref{sec:limits:multi} that in the `multi-moment' persuasion model, \Cref{theorem:nondecr} remains true exactly as stated, with almost the same proof. To adapt the above argument, instead of starting with an interval $[x,z] \subseteq [0,1]$, start with an affinely independent set $S \subseteq E$ such that $u(\sum_{x \in S} \beta(x)x) \le \sum_{x \in S} \beta(x)u(x)$ for all $\beta \in \Delta(S)$. Case~1 is when $v(\sum_{x \in S} \beta(x)x) \le \sum_{x \in S} \beta(x)v(x)$ for every $\beta \in \Delta(S)$, and case~2 is when case~1 does not obtain. The argument in case~1 is unchanged, except with `$[x,z]$' replaced by `$S$', `$[0,1]$' replaced by `$E$' and binary distributions replaced by distributions with affinely independent support, plus a few smaller changes (e.g. replacing `$\R$' by `$\R^n$'). In case~2, replace the sentence `Then there is a $\gamma \in (0,1)$ such that $\widehat{v}$ is not affine on any neighbourhood of $x_\gamma z$, and $\widehat{v}(x_\gamma z) = v(x_\gamma z)$ since $v$ is upper semi-continuous' with `Then there is a $\gamma \in \Delta(S)$ which is not a point mass such that the restriction of $\widehat{v}$ to $\co(\supp(\gamma))$ is not affine on any neighbourhood of $\sum_{x \in S} \gamma(x)x$, and $\widehat{v}(\sum_{x \in S} \gamma(x)x) = v(\sum_{x \in S} \gamma(x)x)$ since $v$ is upper semi-continuous'.
\end{remark}

\subsubsection{Proof that \ref{item:lessvex} implies \ref{item:cs_intervals}, using lemmata}
\label{app:proofs:pf_thm_nondecr:lessvex_implies_cs_intervals}

\begin{definition}
	\label{definition:ido}
	Let $u,v : [0,1] \to \R$ be upper semi-continuous.
	Given distributions $F,H$ such that $F \preceq H$, we say that \emph{$u$ is dominated by $v$ on $[F,H]$} iff
	\begin{equation*}
		\int u \dd H > -\infty , \quad
		\int v \dd F > -\infty , \quad
		\text{and} \quad
		H \in \argmax_{G \in [F,H]} \int u \dd G
	\end{equation*}
	implies that $\int v \dd H \geq \int v \dd F$, with the inequality strict if $\int u \dd H > \int u \dd F$.
	We say that \emph{$u$ is interval-dominated by $v$} iff for all distributions $F \preceq H$, $u$ is dominated by $v$ on $[F,H]$.
\end{definition}

Interval-dominance is a standard comparative-statics concept due to \textcite{QuahStrulovici2009,QuahStrulovici2007extensions}; see \cref{app:mcs_lit} for discussion and context.

Our proof will use some measure-theoretic concepts and lemmata. Recall that a \emph{distribution} is a CDF $[0,1] \to [0,1]$. A \emph{distribution family} is a collection $\lambda = (\lambda_x)_{x \in [0,1]}$, where $\lambda_x$ is a distribution for each $x \in [0,1]$, and $x \mapsto \int w \dd\lambda_x$ is Borel measurable for any continuous $w : [0,1] \to \R$.
For any distribution family $\lambda$ and any distribution $F$, define $F^\lambda : [0,1] \to [0,1]$ by
\begin{equation*}
	F^\lambda(x) \coloneqq \int \lambda_y(x)F(\dd y)
	\quad \text{for each $x \in [0,1]$.}
\end{equation*}
It follows from the next result that $F^\lambda$ is well-defined. (Specifically, part~\ref{item:fam_usc} yields that $y \mapsto \int \1_{[0,x]} \dd \lambda_y = \lambda_y(x)$ is Borel measurable, hence $F$-integrable.)

\begin{lemma}
	\label{lemma:dist_family}
	Let $\lambda$ be a distribution family, let $F$ be a distribution, and let $u : [0,1] \to \R$ be upper semi-continuous.
	Then 
	\begin{enumerate}[label=(\alph*)]
		
		\item \label{item:fam_usc} 
		$x \mapsto \int u \dd\lambda_x$ is Borel measurable, and
		
		\item \label{item:fam_product} 
		$F^\lambda$ is a distribution, and $\int u \dd F^\lambda = \int \int u \dd \lambda_x F(\dd x)$.
		
	\end{enumerate}
	Moreover, for any distribution family $\nu$ such that $\nu_x \preceq \lambda_x$ for $F$-a.e. $x \in [0,1]$,
	\begin{enumerate}[label=(\alph*), resume]
			
		\item \label{item:fam_info}
		$F^\nu \preceq F^\lambda$, and 
		
		\item \label{item:fam_opt}
		there exists a distribution family $(\rho_x)_{x \in [0,1]}$ such that 
		\begin{equation}
			\label{eq:fam_opt}
			\rho_x \in \argmax_{G \in [\nu_x,\lambda_x]} \int u \dd G \quad \text{for $F$-a.e. $x \in [0,1]$.}
		\end{equation}
		
	\end{enumerate}
\end{lemma}

\begin{lemma}
	\label{lemma:ido}
	Fix a distribution $F$ and distribution families $\lambda, \nu$ such that $\nu_x \preceq \lambda_x$ for all $x \in [0,1]$. 
	Let $u,v : [0,1] \to \R$ be upper semi-continuous, and suppose that $u$ is dominated by $v$ on $[\nu_x,\lambda_x]$ for all $x \in [0,1]$.
	Then $u$ is dominated by $v$ on $\bigl[F^{\nu},F^\lambda\bigr]$.
\end{lemma}

We relegate the proofs of \Cref{lemma:dist_family,lemma:ido} to \cref{app:proofs:pf_thm_nondecr:lemmata_pfs} below.

\begin{proof}[Proof that \ref{item:lessvex} implies \ref{item:cs_intervals}]
	Let $u,v : [0,1] \to \R$ be upper semi-continuous, with $u$ coarsely less convex than $v$.
	We shall show that $u$ is interval-dominated by $v$.
	This suffices by Proposition~5 in \textcite{QuahStrulovici2007extensions}.%
		\footnote{For completeness, this proposition is stated and proved in \cref{app:mcs_lit}.}

	So fix any distributions $F \preceq H$; we must show that $u$ is dominated by $v$ on $[F,H]$.
	We consider three cases of increasing generality.%
		\footnote{We thank Ian Jewitt for suggesting this tripartite argument.}
	Recall that we call a distribution \emph{binary} iff its support comprises at most two values.

	\medskip

	\emph{Case~1: $F$ is a point mass and $H$ is binary.}
	If $H$ is a point mass, then there is nothing to prove. Assume for the remainder that $\supp(H) = \{x,z\}$ where $x < z$.
	By the standard `concavification' reasoning \parencite[see][]{KamenicaGentzkow2011}, $H \in \argmax_{G \in [F,H]} \int u \dd G$ iff \eqref{eq:chord_convex:if} holds.
	Moreover, choosing $\alpha \in (0,1)$ so that $F$ is the point mass at $x_\alpha z$, $\int u \dd H \geq \mathrel{(>)} \int u \dd F$ holds iff $u(x_\alpha z) \leq \mathrel{(<)} u(x)_\alpha u(z)$, and $\int v \dd H \geq \mathrel{(>)} \int v \dd F$ holds iff $v(x_\alpha z) \leq \mathrel{(<)} v(x)_\alpha v(z)$. Since $u$ is coarsely less convex than $v$, it follows that $u$ is dominated by $v$ on $[F,H]$.

	\medskip

	\emph{Case~2: $F$ is a point mass.}
	By \Cref{lemma:ido} and the previous case, it suffices to exhibit distribution families $\nu,\lambda$ and a distribution $G$ such that $F = G^\nu$, $H = G^\lambda$, and for all $x \in [0,1]$, $\nu_x \preceq \lambda_x$, $\nu_x$ is a point mass, and $\lambda_x$ is binary.
	Moreover, we may assume without loss that $\mu_H \in (0,1)$.
	For each $x \in [0,1]$, let $\nu_x$ be the point mass at $\mu_H$; clearly $\nu = (\nu_x)_{x \in [0,1]}$ is a distribution family.
	Toward constructing $\lambda$ and $G$, let $\mathcal{F}$ be the set of all distributions with mean $\mu_H$, and let $\mathcal{B}$ be the set of all elements of $\mathcal{F}$ that are binary.
	By Theorem~2.1 in \textcite{Karr1983}, $\mathcal{B}$ is precisely the set of extreme points of $\mathcal{F}$.
	Moreover, the topology of weak convergence makes $\mathcal{F}$ compact and metrisable by Prokhorov's theorem \parencite[e.g.][Theorems~5.1 and 6.8]{Billingsley1999}.
	Hence $\mathcal{B}$ is a Borel subset of $\mathcal{F}$ \parencite[e.g.][Proposition~1.3]{Phelps2001} and, by Choquet's theorem \parencite[e.g.][p.~14]{Phelps2001}, there is a Borel probability measure $\pi$ on $\mathcal{F}$ such that $\pi(\mathcal{B})=1$ and
	\begin{equation*}
		\int w \dd H = \int \int w \dd L \pi(\dd L)
		\quad \text{for any continuous $w : [0,1] \to \R$.}
	\end{equation*}

	Since $\mathcal{F}$ is compact, metrisable, and uncountable it is a standard Borel space. Hence by the Borel isomorphism theorem \parencite[e.g.][Theorem~3.3.13]{Srivastava1998}, there exists a Borel measurable bijection $\phi : [0,1] \to \mathcal{F}$ with Borel measurable inverse $\phi^{-1}$.
	Let $G$ be the CDF of the pushforward of $\pi$ by $\phi^{-1}$.
	Since $\pi(\mathcal{B})=1$, there exists a Borel measurable $\lambda : [0,1] \to \mathcal{B}$ such that $\lambda = \phi$ $G$-a.e.%
		\footnote{For example, $\lambda \coloneqq \1_{\phi^{-1}(\mathcal{B})} \times \phi + \1_{[0,1] \setminus \phi^{-1}(\mathcal{B})} \times L$, where $L \in \mathcal{B}$.}
	Since $\lambda$ is Borel measurable, it is a distribution family \parencite[e.g.][Theorem~IV.1.6]{Warga1972}.
	We have $G^\nu = F$ since $F$ is the point mass at $\mu_H$ (as $F \preceq H$), and for all $x \in [0,1]$, $\nu_x \preceq \lambda_x$ and $\lambda_x$ is binary.
	Finally, to show that $G^\lambda = H$, observe that $\pi$ is the pushforward by $\phi$ of the Borel measure $A \mapsto \int_A \dd G$, and thus $\pi$ equals the pushforward by $\lambda$ of $A \mapsto \int_A \dd G$. Hence
	\begin{equation*}
		\int w \dd H
		= \int \int w \dd L \pi(\dd L)
		= \int \int w \dd \lambda_x G(\dd x)
		= \int w \dd G^\lambda
	\end{equation*}
	for all continuous $w : [0,1] \to \R$, where the last equality follows from \Cref{lemma:dist_family}\ref{item:fam_product}. It follows that $H = G^\lambda$.

	\medskip

	\emph{Case~3: $F$ and $H$ are arbitrary.}
	By \Cref{lemma:ido} and the previous case, it suffices to exhibit distribution families $\nu,\lambda$ such that $F = F^\nu$, $H = F^\lambda$, and for all $x \in [0,1]$, $\nu_x \preceq \lambda_x$ and $\nu_x$ is a point mass. To that end, for each $x \in [0,1]$, let $\nu_x$ be the point mass at $x$; clearly $\nu = (\nu_x)_{x \in [0,1]}$ is a distribution family, and $F = F^\nu$.
	By Blackwell's theorem \parencite[e.g.][p.~94]{Phelps2001}, there exists a distribution family $\lambda = (\lambda_x)_{x \in [0,1]}$ such that $\mu_{\lambda_x} = x$ for all $x \in [0,1]$ and $H = F^\lambda$; clearly $\nu_x \preceq \lambda_x$ for each $x \in [0,1]$.
\end{proof}

\subsubsection{Proofs of the measure-theoretic lemmata}
\label{app:proofs:pf_thm_nondecr:lemmata_pfs}

\begin{proof}[Proof of \Cref{lemma:dist_family}]
	For \ref{item:fam_usc}, recall that since $u$ is upper semi-continuous, it is the pointwise limit of a pointwise decreasing sequence $(u_n)_{n \in \N}$ of continuous functions. By the monotone convergence theorem, $x \mapsto \int u \dd\lambda_x$ is the pointwise limit of the (pointwise decreasing) sequence $x \mapsto \int u_n \dd\lambda_x$ of Borel measurable functions. Hence $x \mapsto \int u \dd\lambda_x$ is Borel measurable.

	For \ref{item:fam_product}, note that $w \mapsto \int\int w \dd \lambda_x F(\dd x)$ defines a continuous linear functional on the space of continuous functions $w : [0,1] \to \R$ endowed with the supremum norm, mapping positive functions to positive values and constant functions to their images.
	Hence by the Riesz--Markov representation theorem \parencite[e.g.][Theorem~14.12]{AliprantisBorder2006}, there exists a unique distribution $G$ such that $\int w \dd G = \int\int w \dd \lambda_x F(\dd x) $ for every continuous $w : [0,1] \to \R$. Moreover,
	\begin{multline*}
		\int u \dd G 
		= \lim_{n \to \infty} \int u_n \dd G 
		= \lim_{n \to \infty} \int \int u_n \dd \lambda_x F(\dd x)
		\\
		= \int \lim_{n \to \infty} \int u_n \dd \lambda_x F(\dd x)
		= \int \int u \dd \lambda_x F(\dd x) ,
	\end{multline*}
	where the first, third and fourth equalities follow from the monotone convergence theorem.
	For any $x \in [0,1]$, the above argument with $u$ replaced by $\1_{[0,x]}$ yields $G(x) = F^\lambda(x)$, showing that $G = F^\lambda$; hence \ref{item:fam_product} holds.

	For \ref{item:fam_info}, given any convex function $\phi : [0,1] \to \R$, we have $\int \phi \dd\nu_x \leq \int \phi \dd \lambda_x$ for $F$-a.e. $x \in [0,1]$ since $\nu_x \preceq \lambda_x$ for $F$-a.e. $x \in [0,1]$, so that 
	\begin{equation*}
		\int \phi \dd F^{\nu} = \int \int \phi \dd \nu_x F(\dd x) \leq \int \int \phi \dd\lambda_x F(\dd x) = \int \phi \dd F^\lambda
	\end{equation*}
	where the equalities follow from part~\ref{item:fam_product} since $\phi$ is convex and thus upper semi-continuous.

	For \ref{item:fam_opt}, let $\mathcal{G}$ be the space of all distributions endowed with the topology of weak convergence, and let $\mathcal{D}$ be the set of all pairs $(G,H) \in \mathcal{G}^2$ that satisfy $G \preceq H$, equipped with the product topology. $\mathcal{G}$ is separable and metrisable by Prokhorov's theorem \parencite[e.g.][Theorem~6.8]{Billingsley1999}.
	Hence by the measurable maximum theorem \parencite[e.g.][Theorem~18.19]{AliprantisBorder2006}, for each $n \in \N$, the correspondence $\mathcal{D} \Rightarrow \mathcal{G}$ given by 
	\begin{equation*}
		(G,H) \mapsto \argmax_{L \in [G,H]}\int u_n \dd L
	\end{equation*}
	admits a Borel measurable selection $R^n : \mathcal{D} \rightarrow \mathcal{G}$, since the correspondence $(G,H) \mapsto [G,H]$ is continuous with non-empty and compact values, and the map $L \mapsto \int u_n \dd L$ is continuous. 

	A collection $(\pi_x)_{x \in [0,1]} \subseteq \mathcal{G}$ is a distribution family if and only if $x \mapsto \pi_x$ is a Borel measurable map $[0,1] \to \mathcal{G}$ \parencite[e.g.][Theorem~IV.1.6]{Warga1972}.
	Hence $x \mapsto \lambda_x$ and $x \mapsto \nu_x$ are Borel measurable maps $[0,1] \to \mathcal{G}$.
	Since $\nu_x \preceq \lambda_x$ for $F$-a.e. $x \in [0,1]$, it follows (possibly after modifying $x \mapsto \lambda_x$ and $x \mapsto \nu_x$ on an $F$-null set) that $x \mapsto (\nu_x,\lambda_x)$ is a Borel measurable map $[0,1] \to \mathcal{D}$.
	Then for each $n \in \N$, $x \mapsto R^n(\nu_x,\lambda_x) \eqqcolon \rho^n_x$ is a Borel measurable map $[0,1] \to \mathcal{G}$, so $(\rho^n_x)_{x \in [0,1]}$ is a distribution family.

	By Theorem~IV.2.1 in \textcite{Warga1972}, we may assume (passing to a subsequence is necessary) that there is a distribution family $(\rho_x)_{x \in [0,1]}$ such that
	\begin{equation}
		\label{eq:fam_lim}
		\lim_{n \to \infty} \int\int w(x,y) \rho^n_x(\dd y)F(\dd x)
		= \int\int w(x,y) \rho_x(\dd y)F(\dd x)
	\end{equation}
	for any $w : [0,1]^2 \to \R$ with $w(x,\cdot)$ continuous for each $x \in [0,1]$, $w(\cdot,y)$ Borel measurable for each $y \in [0,1]$, and $x \mapsto \max_{y \in [0,1]} \abs{w(x,y)}$ $F$-integrable.
	
	It remains to establish \eqref{eq:fam_opt}.
	To this end, note that $W : \mathcal{D} \to \R \cup \{-\infty\}$ defined by
	\begin{equation*}
		W(G,H) \coloneqq \max_{L \in [G,H]}\int u \dd L
		\quad \text{for each $(G,H) \in \mathcal{D}$}
	\end{equation*}
	is upper semi-continuous \parencite[e.g.][Lemma~17.30]{AliprantisBorder2006}, so that the map $U : [0,1] \to \R \cup \{-\infty\}$ given by
	\begin{equation*}
		U(x) \coloneqq \max_{G \in [\nu_x,\lambda_x]} \int u \dd G \quad \text{for each $x \in [0,1]$}
	\end{equation*}
	is Borel measurable, being the composition of the Borel measurable map $x \mapsto (\nu_x,\lambda_x)$ with $W$.
	For each $n \in \N$, the map $U_n : [0,1] \to \R$ defined by
	\begin{equation*}
		U_n(x)
		\coloneqq \int u_n \dd \rho^n_x
		= \max_{G \in [\nu_x,\lambda_x]} \int u_n \dd G
		\quad \text{for each $x \in [0,1]$}
	\end{equation*}
	is Borel measurable since $(\rho^n_x)_{x \in [0,1]}$ is a distribution family, and satisfies $U_n \geq U$ since $u_n \geq u$.
	Hence for any Borel $A \subseteq [0,1]$, 
	\begin{multline*}
		\int_A\int u \dd \rho_x F(\dd x)
		= \int_A \lim_{m \to \infty} \int u_m \dd \rho_x F(\dd x) 
		= \lim_{m \to \infty} \int_A \int u_m \dd \rho_x F(\dd x) 
		\\
		= \lim_{m \to \infty}\lim_{n \to \infty}\int_A\int u_m \dd\rho^n_xF(\dd x) 
		\geq \lim_{n \to \infty} \int_A U_n \dd F 
		\geq \int_A U \dd F
	\end{multline*}
	where the first two equalities follow from the monotone convergence theorem,
	the third equality follows from \eqref{eq:fam_lim} above since $u_m$ is continuous for each $m \in \N$, 
	the first inequality holds since $(u_m)_{m \in \N}$ is pointwise decreasing,%
		\footnote{For any $m \leq n$, we have $u_m \geq u_n$, hence $\int u_m \dd \rho^n_x \geq \int u_n \dd \rho^n_x = U_n(x)$ for every $x \in [0,1]$, hence $\int_A \int u_m \dd \rho^n_x F(\dd x) \geq \int_A U_n \dd F$. Now let $n \to \infty$, then $m \to \infty$.}
	and the final inequality holds since $U_n \geq U$ for all $n \in \N$.
	Thus $\int u \dd \rho_x \geq U(x)$ for $F$-a.e. $x \in [0,1]$.

	Hence to establish \eqref{eq:fam_opt}, it suffices to show that $\nu_x \preceq \rho_x \preceq \lambda_x$ for $F$-a.e. $x \in [0,1]$.
	We shall prove that $\nu_x \preceq \rho_x$ for $F$-a.e. $x \in [0,1]$, omitting the analogous argument for the other half.
	To this end, note that there exists a countable set $\Phi$ of continuous convex functions $[0,1] \to \R$ such that any convex function $[0,1] \to \R$ is the pointwise limit of a pointwise decreasing sequence of functions in $\Phi$.%
		\footnote{For example, the set of all maps of the form $x \mapsto \max_{k \in \{1,\dots,K\}} [ \alpha(k) x + \beta(k) ]$ where $K \in \N$ and $\alpha,\beta : \{1,\dots,K\} \to \Q$. (Here $\Q \subseteq \R$ denotes the rational numbers.)}
	Moreover, it holds for any $\phi \in \Phi$ that $\int \phi \dd \nu_x \leq \int \phi \dd \rho_x$ for $F$-a.e. $x \in [0,1]$, since if this inequality were to fail for all $x \in A$ where $A \subseteq [0,1]$ is $F$-non-null, then
	\begin{equation*}
		\int_A\int \phi \dd\rho_x F(\dd x) < \int_A\int \phi \dd\nu_x F(\dd x) \leq \int_A\int \phi \dd\rho^n_x F(\dd x)
		\quad \text{for all $n \in \N$,}
	\end{equation*}
	(where the second inequality holds since $\nu_x \preceq \rho^n_x$ for all $x \in [0,1]$,) which would contradict \eqref{eq:fam_lim} with $w(x,y) \coloneqq \1_{A}(x) \phi(y)$.
	Since $\Phi$ is countable, it follows that there is an $F$-null set $A \subseteq [0,1]$ such that $\int \phi \dd \nu_x \leq \int \phi \dd \rho_x$ for every $x \in [0,1] \setminus A$ and every $\phi \in \Phi$.
	Hence by the monotone convergence theorem, $\int \phi \dd \nu_x \leq \int \phi \dd \rho_x$ holds for every $x \in [0,1] \setminus A$ and every convex $\phi : [0,1] \to \R$. Equivalently (since $A$ is $F$-null), $\nu_x \preceq \rho_x$ for $F$-a.e. $x \in [0,1]$.
\end{proof}

\begin{proof}[Proof of \Cref{lemma:ido}]
	Suppose that 
	\begin{equation}
		\label{eq:ido_if}
		\int u \dd F^\lambda > -\infty , \quad
		\int v \dd F^\nu > -\infty , \quad
		\text{and} \quad
		F^\lambda \in \argmax_{G \in \left[F^\nu,F^\lambda\right]} \int u \dd G .
	\end{equation}
	We must show that $\int v \dd F^\lambda \geq \int v \dd F^{\nu}$, and that the inequality is strict if $\int u \dd F^\lambda > \int u \dd F^{\nu}$.

	By \Cref{lemma:dist_family}\ref{item:fam_opt}, we may choose a distribution family $(\rho_x)_{x \in [0,1]}$ such that 
	\begin{equation*}
		\rho_x \in \argmax_{G \in [\nu_x,\lambda_x]} \int u \dd G 
		\quad \text{for $F$-a.e. $x \in [0,1]$.}
	\end{equation*}
	Define
	\begin{align*}
		A &\coloneqq \left\{x \in [0,1] : \text{$\int u \dd \lambda_x > - \infty$ and $\int v \dd \nu_x > -\infty$} \right\}
		\\
		\text{and}\quad
		B &\coloneqq \left\{x \in A : \lambda_x \in \argmax_{G \in [\nu_x,\lambda_x]} \int u \dd G\right\}.
	\end{align*}
	We have $\int_A \dd F = 1$ by \eqref{eq:ido_if} and \Cref{lemma:dist_family}\ref{item:fam_product}.
	We further claim that $\int_B \dd F = 1$. Suppose toward a contradiction that $\int_B \dd F < 1$; then
	\begin{equation*}
		\int u \dd F^\lambda = \int \int u \dd \lambda_x F(\dd x) < \int \int u \dd \rho_x F(\dd x) = \int u \dd F^{\rho} ,
	\end{equation*}
	where the equalities follow from \Cref{lemma:dist_family}\ref{item:fam_product}, and the inequality is strict since $\int u \dd F^\lambda > -\infty$. But $F^{\rho}$ belongs to $\bigl[F^{\nu},F^\lambda\bigr]$ by \Cref{lemma:dist_family}\ref{item:fam_info} since $\rho_x \in \left[\nu_x,\lambda_x\right]$ for $F$-a.e. $x \in [0,1]$, so $\int u \dd F^\lambda \geq \int u \dd F^{\rho}$ by \eqref{eq:ido_if}; a contradiction.

	Since $u$ is dominated by $v$ on $[\nu_x,\lambda_x]$ for all $x \in [0,1]$, we have $\int v \dd\lambda_x \geq \int v \dd \nu_x$ for every $x \in B$.
	This together with $\int_B \dd F = 1$ implies that
	\begin{equation*}
		\int v \dd F^\lambda
		= \int \int v \dd \lambda_x F(\dd x) 
		\geq \int \int v \dd \nu_x F(\dd x) 
		= \int v \dd F^{\nu} ,
	\end{equation*}
	where the equalities follow from \Cref{lemma:dist_family}\ref{item:fam_product}.

	It remains only to show that if $\int u \dd F^\lambda > \int u \dd F^{\nu}$, then $\int v \dd F^\lambda > \int v \dd F^{\nu}$. So assume that $\int u \dd F^\lambda > \int u \dd F^{\nu}$, and let
	\begin{equation*}
		C^w \coloneqq \left\{x \in B : \int w \dd \lambda_x > \int w \dd \nu_x \right\}
		\quad \text{for $w \in \{u,v\}$.}
	\end{equation*}
	Note that $C^u$ is $F$-non-null, since otherwise 
	\begin{equation*}
		\int u \dd F^\lambda = \int \int u \dd \lambda_x F(\dd x) \leq \int \int u \dd \nu_x F(\dd x) = \int u \dd F^{\nu} ,
	\end{equation*}
	where the inequality holds since $\int_B \dd F = 1$, and the equalities follow from \Cref{lemma:dist_family}\ref{item:fam_product}.
	We have $C^u \subseteq C^v$ since $u$ is dominated by $v$ on $[\nu_x,\lambda_x]$ for all $x \in [0,1]$.
	Thus $C^v$ is $F$-non-null, so
	\begin{equation*}
		\int v \dd F^\lambda = \int \int v \dd \lambda_x F(\dd x) > \int \int v \dd \nu_x F(\dd x) = \int v \dd F^{\nu} ,
	\end{equation*}
	where the inequality holds since $\int v \dd \lambda_x \geq \int v \dd \nu_x$ for $F$-a.e. $x \in [0,1]$ and $\int v \dd F^{\nu} > -\infty$, and the equalities follow from \Cref{lemma:dist_family}\ref{item:fam_product}.
\end{proof}

\subsection{Proof of \texorpdfstring{\Cref{theorem:incr} (\cpageref{theorem:incr})}{Theorem \ref{theorem:incr} (p. \pageref{theorem:incr})}}
\label{app:proofs:pf_thm_incr}

The necessity of the crater property for comparative statics was proved in the text (§\ref{sec:mcs_incr:pf_thm2}). In this appendix, we prove sufficiency.

Given any distribution $F$, let $C_F : [0,1] \rightarrow \R$ be given by $C_F(x) \coloneqq \int_0^x F$ for each $x \in [0,1]$.
We shall make free use of the order isomorphism described in \cref{app:proofs:product}
between distributions $F$ ordered by informativeness and functions $C_F$ ordered by pointwise inequality.

We shall use two lemmata.
The first is a version of Dworczak and Martini's (\citeyear{DworczakMartini2019}) duality theorem.
Given any regular $u : [0,1] \rightarrow \R$, let $\mathcal{M}(u)$ denote the space of all convex and Lipschitz continuous functions $p : [0,1] \rightarrow \R$ satisfying $p \geq u$.

\begin{lemma}
    \label{lemma:duality}
    Let $u : [0,1] \rightarrow \R$ be regular, and let $F_0$ be a distribution. Then
    \begin{equation*}
        \min_{p \in \mathcal{M}(u)} \int p \dd F_0
        = \max_{F \in \mpc{F_0}} \int u \dd F ,
    \end{equation*}
    where both sides are well-defined.
    Moreover, if
	\begin{equation*}
        p \in \argmin_{r \in \mathcal{M}(u)} \int r \dd F_0
        \quad \text{and} \quad
        F \in \argmax_{G \in \mpc{F_0}} \int u \dd G ,
    \end{equation*}
    then both
    
    \begin{enumerate}[label=(\alph*)]
        \item \label{item:duality:convex}
        $C_F = C_{F_0}$ on $\{x,y\}$ for any $x \leq y$ such that $[x,y]$ is a maximal interval of affineness of $p$, and

        \item \label{item:duality:support}
        $p = u$ on $\supp(F)$.
    \end{enumerate}
    
\end{lemma}

\Cref{lemma:duality} follows directly from Theorem~2 in \textcite{DizdarKovac2020}.

\begin{lemma}
	\label{lemma:crater}
	Let $u : [0,1] \to \R$ be regular and satisfy the crater property, and let $F_0$ be any distribution.
	Then there exists a
	\begin{equation*}
		p \in \argmin_{r \in \mathcal{M}(u)} \int r \dd F_0
	\end{equation*}
	such that for any $x \leq z$ such that $[x,z]$ is a maximal interval of affineness of $p$, both of the following hold:

	\begin{enumerate}[label=(\alph*)]
		
		\item \label{item:crater:bound} 

		\begin{enumerate}[label=(\roman*)]
		
			\item Either $x = 0$, or else $p(x) = u(x)$ and $u$ is strictly convex on $[x-\eps,x]$ for all sufficiently small $\eps > 0$, and

			\item either $z = 1$, or else $p(z) = u(z)$ and $u$ is strictly convex on $[z,z+\eps]$ for all sufficiently small $\eps > 0$.
		
		\end{enumerate}
		
		\item \label{item:crater:tangent} 
		If $x < z$, then either 

		\begin{enumerate}[label=(\roman*)]
		
			\item $p$ is tangent to $u$ at some $y \in (x,z)$ and $u$ is W-shaped on $[x,z]$, or 
		
			\item 
			$p$ is tangent to $u$ at $z$ and $u$ is S-shaped on $[x,z]$, or
		
			\item 
			$p$ is tangent to $u$ at $x$ and $u$ is reverse-S-shaped on $[x,z]$.
		
		\end{enumerate}
		
	\end{enumerate}
\end{lemma}

\begin{proof}
	We first consider the case in which $u$ is W-shaped, and then the case in which it is not. Recall that W-shapedness encompasses as special cases all convex, concave, S-shaped and reverse-S-shaped functions.

	\medskip

	\emph{Case~1: $u$ is W-shaped.} If $u$ is strictly convex, choose $p = u$. Suppose for the remainder that $u$ is not convex, and write $[a,b]$ for the maximal (proper) interval of concavity of $u$. For each $y \in [a,b]$, let $p_y : [0,1] \to \R$ be the pointwise maximum of $u$ and the tangent to $u$ at $y$, and let $[x_y,z_y]$ be the maximal interval of affineness of $p_y$ containing $y$. Note that for any $y \in [a,b]$, if $x \leq z$ are such that $[x,z]$ is a maximal interval of affineness of $p_y$, then $p_y$ satisfies \ref{item:crater:bound} and \ref{item:crater:tangent}. It therefore suffices to find a $y \in [a,b]$ such that $p_y \in \argmin_{r \in \mathcal{M}(u)} \int r \dd F_0$. 
	We shall first do so for the special case in which $F_0$ is atomless with support $[0,1]$, then in generality.

	\smallskip 

	\emph{Sub-case~1(a): $F_0$ is atomless with support $[0,1]$.}
	Observe that the map $\psi : [a,b] \to \R$ given by 
	\begin{equation*}
		\psi(y) \coloneqq \frac{\int_{x_y}^{z_y} \xi F_0(\dd \xi)}
		{F_0(z_y)-F_0(x_y)} \quad \text{for all $y \in [a,b]$}
	\end{equation*}
	is continuous with $\psi(a) > a$ and $\psi(b) < b$, so that $\psi(y) = y$ for some $y \in (a,b)$ by the intermediate value theorem.
	Then $\int p_y \dd F_0 = \int u \dd G$, where $G \coloneqq \1_{[x_y,y)}F_0(x_y) + \1_{[y,z_y)}F_0(z_y) + \1_{[0,1] \setminus [x_y,z_y)} F_0$.
	Hence $p_y \in \argmin_{r \in \mathcal{M}(u)} \int r \dd F_0$ by \Cref{lemma:duality}, so we may choose $p \coloneqq p_y$.

	\smallskip 

	\emph{Sub-case~1(b): $F_0$ is arbitrary.}
	Choose a sequence $(F^n_0)_{n \in \N}$ of atomless distributions with support $[0,1]$ converging weakly to $F_0$. 
	By the previous subcase, for each $n \in \N$ there is $y^n \in (a,b)$ such that $p_{y^n} \in \argmin_{r \in \mathcal{M}(u)} \int r \dd F^n_0$.
	After passing to subsequence if necessary, we may assume without loss that $y^n$ converges to some $y \in [a,b]$.
	Given any $r \in M(u)$, $\int p_{y^n} \dd F^n_0 \le \int r \dd F^n_0$ for all $n \in \N$, and letting $n \to \infty$ yields $\int p_y \dd F_0 \le \int r \dd F_0$, since $p_{y^n} \to p_y$ uniformly. Hence $p_y \in \argmin_{r \in \mathcal{M}(u)} \int r \dd F_0$, as desired.

	\medskip

	\emph{Case~2: $u$ is not W-shaped.}
	Enumerate the maximal intervals on which $u$ is W-shaped as $(I_i)_{i=1}^n$, noting that $n \geq 2$ (by the Case-2 hypothesis) and that these intervals overlap.
	Let $I_0$ ($I_{n+1}$) be the maximal interval of convexity of $u$ containing $0$ ($1$).
	We may assume without loss (relabelling the $(I_i)_{i=1}^n$ if required) 
	that $\min I_i \leq \min I_{i+1}$ for each $i \in \{1,\dots,n\}$, 
	that $I_i \cap I_j = \varnothing$ for any $i<j$ in $\{1,\dots,n\}$ such that $i+1 \neq j$, 
	and that $u$ is concave on $K_i \coloneqq \cl(I_i \setminus (I_{i-1} \cup I_{i+1}))$ for each $i \in \{1,\dots,n\}$,
	and convex on $J_i \coloneqq I_i \cap I_{i+1}$ for each $i \in \{0,\dots,n\}$.

	Since $u$ satisfies the crater property, for each $i \in \{1,\dots,n-1\}$, there are $\bar z_i \le \bar x_{i+1}$ in $\interior(J_i)$ such that the tangents to $u$ at $\min K_i$ and $\max K_{i+1}$ cross $u$ at $\bar z_i$ and $\bar x_{i+1}$, respectively. Let $w_i \coloneqq (\bar z_i+\bar x_{i+1})/2$ for each $i \in \{1,\dots,n-1\}$, $L_1 \coloneqq [0,w_1]$, $L_i \coloneqq (w_{i-1},w_{i}]$ for all $i \in \{2,\dots,n-1\}$, and $L_n \coloneqq (w_{n-1},1]$.
	For each $i \in \{1,\dots,n\}$, let $\mathcal{M}_i(u)$ denote the space of all convex and Lipschitz continuous $p : I_i \rightarrow \R$ satisfying $p \geq u$ on $I_i$. Applying Case~1 above yields%
		\footnote{If $L_i$ is $F_0$-null, apply Case 1 to an arbitrary prior with support within $I_i$. Otherwise, apply Case $1$ to the distribution derived by conditioning $F_0$ on $L_i$.}
	a
	\begin{equation*}
		p_i \in \argmin_{r \in \mathcal{M}_i(u)} \int_{L_i} r \dd F_0
	\end{equation*}
	that is affine on $[x_i,z_i]$ for some $x_i < z_i$ in $\cl(L_i)$, strictly convex on $[\min I_i,x_i]$ and on $[z_i,\max I_i]$, and satisfies property~\ref{item:crater:bound} (with $0$ and $1$ replaced by $\min I_i$ and $\max I_i$, respectively) and property~\ref{item:crater:tangent}. 
	It follows that $p \coloneqq \sum_{i = 1}^n \1_{L_i} p_i$ belongs to $\mathcal{M}(u)$ and, for any $x \leq z$ such that $[x,z]$ is a maximal interval of affineness of $p$, there exists an $i \in \{1,\dots,n\}$ such that $[x,z] \subseteq I_i$ and $p = p_i$ on $[x,z]$, and furthermore $\min I_i < x$ ($z < \max I_i$) unless $x = z$ or $i = 1$ ($i = n$). Hence \ref{item:crater:bound} and \ref{item:crater:tangent} are satisfied.
	Finally, $p \in \argmin_{r \in \mathcal{M}(u)} \int r \dd F_0$ since $p_i \in \argmin_{r \in \mathcal{M}_i(u)} \int_{L_i} r \dd F_0$ for each $i \in \{1,\dots,n\}$ and $\{L_i\}_{i = 1}^n$ is a partition of $[0,1]$.
\end{proof}

\begin{proof}[Proof of the first (sufficiency) part of \Cref{theorem:incr}]
	Fix regular $u,v : [0,1] \to \R$
	such that $u$ satisfies the crater property
	and is coarsely less convex than $v$.
	Let $F_0$ be any distribution,
	and fix
	\begin{equation*}
	    G \in \argmax_{F \in \mpc{F_0}} \int u \dd F
	    \quad \text{and} \quad
	    H \in \argmax_{F \in \mpc{F_0}} \int v \dd F .
	\end{equation*}
	We must construct
	\begin{equation*}
	    G' \in \argmax_{F \in \mpc{F_0}} \int u \dd F
	    \quad \text{and} \quad
	    H' \in \argmax_{F \in \mpc{F_0}} \int v \dd F 
	\end{equation*}
	such that $C_G \leq C_{H'}$ and $C_{G'} \leq C_H$.
		
	By \Cref{lemma:crater}, there is a $p \in \mathcal{M}(u)$ satisfying \Cref{lemma:crater}'s properties~\ref{item:crater:bound} and \ref{item:crater:tangent}. Fix a
	\begin{equation*}
		q \in \argmin_{r \in \mathcal{M}(v)} \int r \dd F_0 .
	\end{equation*}
	Note that (because $u$ is regular,) there are finitely many triplets $x < y < z$ such that $[x,z]$ is a maximal interval of affineness of $p$ and $p = u$ on $\{x,y,z\}$. Enumerate these as $(x_k,y_k,z_k)_{k=1}^K$ where $x_1 \leq \dots \leq x_K$, and note that $z_k < x_{k+1}$ for all $k \in \{1,\dots,K-1\}$ since $u$ is W-shaped on $[x_k,z_k]$ by \Cref{lemma:crater}\ref{item:crater:tangent} and satisfies the crater property. 
	
	Let $G'$ be obtained from $G$ by, for each $k \in \{1,\dots,K\}$, shifting as much probability mass as possible (subject to feasibility given $F_0$) from $\{x_k,z_k\}$ to $y_k$.
	Symmetrically, let $H'$ be obtained from $H$ by, for each $k \in \{1,\dots,K\}$ such that $q = v$ on $\{x_k,y_k,z_k\}$, shifting as much probability mass as possible (subject to feasibility given $F_0$) from $y_k$ to $\{x_k,z_k\}$.
	Note that
	\begin{equation*}
		\int u \dd (G'-G) = 0
		\quad \text{and} \quad
		\int v \dd (H'-H) \geq 0
	\end{equation*} 
	where the inequality holds since $q$ is convex, so that $G'$ ($H'$) is optimal for $u$ (for $v$) given $F_0$.

	It remains to prove that $C_{G'} \leq C_H$ and $C_G \leq C_{H'}$.
	By \Cref{lemma:duality}\ref{item:duality:convex}, it suffices to show that $C_{G'} \leq C_H$ and $C_G \leq C_{H'}$ on $[x,z]$ for any $x < z$ such that $[x,z]$ is a maximal interval of affineness of $q$, since $C_{G'} \leq C_{F_0} \geq C_{H'}$.
	So fix such $x<z$.
	We shall use the following claim.

	\begin{namedthm}[Claim.]
		\label{claim:concave}
		$u$ is concave on an open interval containing $\supp(H) \cap (x,z)$.
	\end{namedthm}

	\begin{proof}[Proof of the {\hyperref[claim:concave]{claim}}]
		\renewcommand{\qedsymbol}{$\square$}
		It suffices to show that, given any $y \leq y'$ in $\supp(H) \cap (x,z)$, $u$ is strictly concave on $(y-\eps,y'+\eps)$ for $\eps >0$ sufficiently small.
		So fix $y \leq y'$ in $\supp(H) \cap (x,z)$, and note that $q = v$ on $\{y,y'\}$ by \Cref{lemma:duality}\ref{item:duality:support}, so that $v$ is strictly concave on $(y-\eps,y+\eps)$ and on $(y'-\eps,y'+\eps)$ for small enough $\eps > 0$.
		Hence so is $u$, since it is coarsely less convex than $v$.
		
		Let $\widehat{u}$ be the concave envelope of the restriction of $u$ to $[y-\eps,y'+\eps]$.
	    We shall show that $\widehat{u}$ is not affine.
	    To see why this suffices, note that in this case $\widehat{u}(b) = u(b)$ for some $y-\eps < b < y'+\eps$, so that the tangent to $u$ at $b$ lies above $u$ on $[y-\eps,y'+\eps]$.
	    Then $u$ is strictly concave on $[y-\eps,y'+\eps]$, since it satisfies the crater property and is strictly concave on $[y-\eps,y+\eps]$ and on $[y'-\eps,y'+\eps]$.

		So suppose toward a contradiction that $\widehat{u}$ is affine.
	    Then 
	    \begin{equation*}
	        \alpha u(y-\eps) + (1-\alpha)u(y'+\eps)
	        \geq u\bigl(\alpha (y-\eps) + (1-\alpha)(y'+\eps) \bigr)
	        \quad \text{for all $\alpha \in (0,1)$,}
	    \end{equation*}
	    and the inequality is strict when $\alpha = \beta$ where $\beta(y-\eps) + (1-\beta)(y'+\eps) = y$, since $u$ is strictly concave on $[y-\eps,y+\eps]$.
	    However,
	    \begin{equation*}
	        \beta v(y-\eps) + (1-\beta)v(y'+\eps)
	        \leq \beta q(y-\eps) + (1-\beta)q(y'+\eps)
	        = q(y)
	        = v(y) ,
	    \end{equation*}
	    where the inequality holds since $v \leq q$, and the first equality holds since $q$ is affine on $[x,z] \supseteq [y-\eps,y'+\eps]$ for small enough $\eps > 0$. This is a contradiction with the fact that $u$ is coarsely less convex than $v$.
	\end{proof}%
	\renewcommand{\qedsymbol}{$\blacksquare$}

	By the \hyperref[claim:concave]{claim} and \Cref{lemma:crater}\ref{item:crater:bound}, there are $x' \leq z'$ such that $[x',z']$ is a maximal interval 
	of affineness of $p$ and $\supp(H) \cap (x,z) \subseteq (x',z')$, so that $C_G = C_{F_0}$ on $\{x',z'\}$ by \Cref{lemma:duality}\ref{item:duality:convex}.
	We consider two cases: first the (generic) case in which $p(y') = u(y')$ and $q(y) = v(y)$ for some $y \in (x,z) \cap (x',z')$ and $y' \in (x',z')$, and then the (non-generic) complementary case.

	\medskip

	\emph{Case~1: $p(y') = u(y')$ and $q(y) = v(y)$ for some $y \in (x,z) \cap (x',z')$ and $y' \in (x',z')$.}
	Note that $x' < z'$ in this case, so that $p > u$ on $(x',y') \cup (y',z')$ by \Cref{lemma:crater}\ref{item:crater:tangent}.
	Then $C_G$ is affine on $[x',y']$ and on $[y',z']$ by \Cref{lemma:duality}\ref{item:duality:support}; that is, $G$ pools $(x',z')$ into $y'$.
	Note also that $C_{G'}(x') = C_{F_0}(x')$ by \Cref{lemma:duality}\ref{item:duality:convex}.
	
	We first show that if $x' > 0$, then $G'(x'-) = F_0(x'-)$;
	that is, $G'$ reveals whether or not the state weakly exceeds $x'$.
	If there is no $a \in [0,x')$ such that $p$ is affine on $[a,x']$, then this follows from \Cref{lemma:duality}\ref{item:duality:convex}.
	Suppose for the remainder that $p$ is affine on $[a,x']$, where $a \in [0,x')$.
	Note that $p$ is tangent to $u$ at some $w \in \{x',y',z'\}$ by \Cref{lemma:crater}\ref{item:crater:tangent}.
	It cannot be that $p$ is tangent to $u$ at some $b \in (a,x')$, since then $u$ would be convex on an open interval containing $x'$ and satisfy $p(x') = u(x')$ by \Cref{lemma:crater}\ref{item:crater:bound}, in which case $u$ would violate the crater property on any open interval containing $[b,w]$---a contradiction.
	Thus $p > u$ on $(a,x')$, so that $C_{G'}$ is affine on $[a,x']$ by \Cref{lemma:duality}\ref{item:duality:support}.
	Moreover, choosing $a$ such that $[a,x']$ is a maximal interval of affineness of $p$ yields $C_{G'} = C_{F_0}$ on $\{a,x'\}$ by \Cref{lemma:duality}\ref{item:duality:convex}, so that $C_{G'} = C_{F_0}$ on $[a,x']$ since $C_{G'} \leq C_{F_0}$, and thus $G'(x'-) = F_0(x'-)$.
	
	A symmetric argument shows that $C_{G'}(z') = C_{F_0}(z')$ and $G'(z') = F_0(z')$: that is, $G'$ reveals whether or not the state is weakly below $z'$.

	Since $G'$ is less informative than $G$, we deduce that $G'$ pools states $(x',z')$ into $y'$ (that is, $C_{G'}$ is affine on $[x',y']$ and on $[y',z']$), and that $G$ reveals whether or not the state weakly exceeds $x'$ and whether or not the state is weakly below $z'$ (i.e. $G(x'-) = F_0(x'-)$ if $x' > 0$ and $G(z') = F_0(z')$).

	By \Cref{lemma:duality}\ref{item:duality:support}, if $x' \leq x$ and $z \leq z'$ and the first (second) inequality is strict unless $p(x') > u(x')$ ($p(z') > u(z')$), then $G$ has no atom at $x'$ (at $z'$) if $x' = x$ ($z = z'$).
	Since $G$ reveals whether or not the state belongs to $[x',1]$, also reveals whether or not the state belongs to $[0,z']$, and pools states $(x',z')$, it follows that $G$ pools $[x,z]$, so that $C_G \leq C_H$ on $[x,z]$ since $C_H = C_{F_0}$ on $\{x,z\}$.

	Assume for the remainder that either (i)~$x \leq x'$ with equality only if $p(x') = u(x')$ or (ii)~$z' \leq z$ with equality only if $p(z') = u(z')$.
	In particular, consider the former case; the argument for the latter is symmetric.
	Then $p(x') = u(x')$, since this holds if $x=x'$ by hypothesis and holds if $x<x'$ by \Cref{lemma:crater}\ref{item:crater:bound}.
	We consider two sub-cases: first the (non-generic) case in which $z' \leq z$ with equality only if $p(z') = u(z')$, and then the (generic) case in which $z \leq z'$ with equality only if $p(z') > u(z')$.

	\smallskip

	\emph{Sub-case~1(a): $z' \leq z$ with equality only if $p(z') = u(z')$.}
	In this sub-case, $p=u$ on $\{x',z'\}$ by \Cref{lemma:crater}\ref{item:crater:bound}. 
	Then, given $a \in [x,x']$ and $b \in [z',z]$ such that $u$ is strictly convex on $[a,x']$ and $[z',b]$, we have
	\begin{equation*}
		\alpha u(a) + (1-\alpha) u(b)
		\geq u(\alpha a + (1-\alpha)b)
		\quad \text{for all $\alpha \in (0,1)$,}
	\end{equation*}
	where the inequality is strict unless $a = x'$, $b = z'$ and $\alpha x' + (1-\alpha)z' = y'$, since $p$ is affine and exceeds $u$ on $[x',z']$, strictly so on $(x',y') \cup (y',z')$.
	However, choosing $\alpha=\beta$ where $\beta a + (1-\beta) b = y$ yields
	\begin{equation*}
		\beta v(a) + (1-\beta)v(b)
		\leq \beta q(a) + (1-\beta)q(b)
		= q(y)
		= v(y) ,
	\end{equation*}
	where the inequality holds since $v \leq q$ and the first equality holds since $q$ is affine on $[x,z] \supseteq [a,b]$.
	Since $u$ is coarsely less convex than $v$, it must therefore be that $q = v$ on $\{a,b\}$, $a = x'$, $b = z'$ and $y = y'$, and thus $x = x'$ and $z' = z$ by \Cref{lemma:crater}\ref{item:crater:bound}.
	Moreover, $q > v$ on $(x,y') \cup (y',z)$, since $y$ was chosen (in the Case-1 hypothesis) as an arbitrary $\xi \in (x',z')$ such that $q(\xi)=v(\xi)$, and we just showed that this necessitates $\xi=y'$. Hence $C_{H'}$ and $C_H$ are affine on $[x,y']$ and on $[y',z]$ by \Cref{lemma:duality}\ref{item:duality:support}.
	Then $C_G \leq C_{H'}$ ($C_{G'} \leq C_H$) on $[x,z]$ as desired by construction of $H'$ (of $G'$), since $C_G$ ($C_{G'}$) is likewise affine on $[x,y']$ and on $[y',z]$, and $C_{H'} = C_{F_0}$ ($C_H = C_{F_0}$) on $\{x,z\}$ by \Cref{lemma:duality}\ref{item:duality:convex} (and $G'(x'-) = F_0(x'-)$ if $x' > 0$ and $G'(z) = F_0(z)$).

	\smallskip

	\emph{Sub-case~1(b): $z \leq z'$ with equality only if $p(z') > u(z')$.}
	Since $p = u$ on $\{x',y'\}$ and $x' < y' < z'$, $u$ is strictly convex on $[x',w']$ and strictly concave on $[w',y']$ for some $x' < w' < y'$, by \Cref{lemma:crater}\ref{item:crater:tangent}.
	Hence, setting $w \coloneqq \min\{y',z\}$, we have
	\begin{equation*}
	    \alpha u(x') + (1-\alpha)u(w)
	    > u(\alpha x' + (1-\alpha)w)
	    \quad \text{for all $\alpha \in (0,1)$.}
	\end{equation*}
	It must be that $y' \leq y$, since if $y < y'$ then for $\alpha=\beta$ where $\beta x' + (1-\beta)w = y$, we would have
	\begin{equation*}
	    \beta v(x') + (1-\beta)v(w)
	    \leq \beta q(x') + (1-\beta)q(w)
	    = q(y)
	    = v(y) ,
	\end{equation*}
	where the inequality holds since $v \leq q$ and the first equality holds since $q$ is affine on $[x,z] \supseteq [x',w]$---a contradiction with the fact that $u$ is coarsely less convex than $v$.

	Furthermore, by \Cref{lemma:duality}\ref{item:duality:support}, $C_G$ is affine on $[y',z']$ since $p > u$ on $(y',z')$, and differentiable at $z$ if $z=z'$ by the sub-case-1(b) hypothesis.
	Since $y' \leq y < z \leq z'$ and $C_G(z') = C_{F_0}(z')$, the slope of $C_G$ on $(y',z')$ must be at least $F_0(z)$.
	Since $C_H \leq C_{F_0}$ with equality at $z$, $C_H$ also has slope (formally, right-hand derivative) at most $F_0(z)$ on $[y',z]$.
	Since $C_G(z) \leq C_{F_0}(z) = C_H(z)$, it follows that $C_G \leq C_H$ on $[y',z]$.

	Since $y$ was chosen (in the Case-1 hypothesis) as an arbitrary $\xi \in (x',z')$ such that $q(\xi)=v(\xi)$, and we showed above that this necessitates $y' \leq \xi$, it must be that $q > v$ on $(x',y')$, so that $C_H$ is affine on $[x,y']$ by \Cref{lemma:duality}\ref{item:duality:support}, since $\supp(H) \cap (x,z) \subseteq (x',z')$.
	This, together with the fact that $C_G(x) \leq C_{F_0}(x) = C_H(x)$ and (from above) $C_G(y') \leq C_H(y')$, implies that $C_G \leq C_H$ on $[x,y']$.
	Since $G'$ is less informative than $G$ and $H$ is less informative than $H'$, it follows that $C_{G'} \leq C_H$ and $C_G \leq C_{H'}$ on $[x,z]$, as desired.

	\medskip

	\emph{Case~2: either $q > v$ on $(x,z) \cap (x',z')$ or $p > u$ on $(x',z')$.}
	Suppose first that $q > v$ on $(x,z) \cap (x',z')$, or $C_H$ is affine on $[x',z']$. In this case, since $\supp(H) \cap (x,z) \subseteq (x',z')$, $C_H$ is affine on $[x,z]$ by \Cref{lemma:duality}\ref{item:duality:support}. Thus $C_H = C_{F_0} \geq C_G$ on $[x,z]$, where the equality holds since $C_H \leq C_{F_0}$ with equality on $\{x,z\}$.

	Suppose for the remainder that $(x,z) \cap (x',z') \neq \varnothing$ and $p > u$ on $(x',z')$.
	If $p(x') > u(x')$, then $x' = 0$ by \Cref{lemma:crater}\ref{item:crater:bound}, so that $0 = C_G(z') = C_{F_0}(z')$, where the first equality follows from \Cref{lemma:duality}\ref{item:duality:support}. Then $C_H = 0$ on $[0,z'] \supseteq [x',z']$ since $0 \le C_H \le C_{F_0}$ and $C_{F_0}$ is increasing, so that $C_H \geq C_G$ on $[x,z]$ by the argument in the previous paragraph. Similarly, $C_H \geq C_G$ on $[x,z]$ if $p(z') > u(z')$, so that we may assume that $p = u$ on $\{x',z'\}$.
	Then 
	\begin{equation*}
		\alpha u(x') + (1-\alpha) u(z') = p(\alpha x' + (1-\alpha)z') > u(\alpha x' + (1-\alpha)z')
	\end{equation*}
	for all $\alpha \in (0,1)$, where the equality holds since $p$ is affine on $[x',z']$. 
	
	If $x \leq x'$ and $z' \leq z$, then
	\begin{align*}
		q(\alpha x' + (1-\alpha)z') &= \alpha q(x') + (1-\alpha) q(z') \\&\geq \alpha v(x') + (1-\alpha)v(z') 
		> v(\alpha x' + (1-\alpha)z')
	\end{align*}
	for all $\alpha \in (0,1)$, where the equality holds since $q$ is affine on $[x,z] \supseteq [x',z']$, the weak inequality holds since $q \geq v$, and the strict inequality holds since $v$ is coarsely more convex than $u$; hence $q > v$ on $(x',z') \supseteq (x,z) \cap (x',z')$, so that $C_H \geq C_G$ on $[x,z]$ by the argument at the beginning of Case~2.

	Assume for the remainder that either $x' < x < z'$ or $x' < z < z'$.
	$C_G$ is affine on $[x',z']$ by \Cref{lemma:duality}\ref{item:duality:support} since $p > u$ on $(x',z')$, so $C_{F_0}$ is also affine on $[x',z']$ (and equal to $C_G$), as $C_G \leq C_{F_0}$ with equality on $\{x',z'\}$.
	Since $C_H \leq C_{F_0}$ with equality on $\{x,z\}$, and either $x' < x < z'$ or $x' < z < z'$, we must then have $C_H = C_{F_0}$ on $[x',z']$.
	Since $C_{F_0}$ is affine on $[x',z']$, it follows that $\varnothing = \supp(H) \cap (x',z') \supseteq \supp(H) \cap (x,z)$, so that $C_H$ is affine on $[x,z]$ and, hence, $C_G \le C_{F_0} = C_H$ on this interval.
	Therefore (since $G'$ is less informative than $G$ and $H$ is less informative than $H'$) $C_{G'} \leq C_H$ and $C_G \leq C_{H'}$ on $[x,z]$, as desired.
\end{proof}

\subsection{Proof of \texorpdfstring{\Cref{proposition:W_down} (\cpageref{proposition:W_down})}{Proposition \ref{proposition:W_down} (p. \pageref{proposition:W_down})}}
\label{app:proofs:pf_thm_incr_down}

For the `if' part (sufficiency), fix a distribution $F_0$, and let $u,v : [0,1] \to \R$ be regular with $u$ coarsely less convex than $v$.
If $v$ is concave, then it is \emph{strictly} concave since regular, so $u$ is strictly concave since coarsely less convex; hence the point mass concentrated on the prior mean $\int x F_0(\dd x)$ is uniquely optimal for $u$ given $F_0$, so \eqref{eq:mcs_incr} holds. If instead $v$ is convex, then it is \emph{strictly} convex since regular, so $F_0$ is uniquely optimal for $v$ given $F_0$; hence \eqref{eq:mcs_incr} holds.

For the `only if' part (necessity), fix a regular $v : [0,1] \to \R$ that is neither concave nor convex; we shall exhibit a regular $u : [0,1] \to \R$ that is coarsely less convex than $v$, an (atomless convex-support) prior distribution $F_0$, and a distribution $F$ that is optimal for $v$ given $F_0$
such that no distribution optimal for $u$ given $F_0$ is less informative than $F$. The argument will be similar to the proof in §\ref{sec:mcs_incr:pf_thm2} of the converse (necessity) part of \Cref{theorem:incr}.

By hypothesis (and using regularity), there are $x' < z < w'$ in $[0,1]$ such that either $v$ is strictly convex on $[x',z]$ and strictly concave on $[z,w']$, or $v$ is strictly concave on $[x',z]$ and strictly convex on $[z,w']$. We consider the former case (the latter is analogous).

Choose a $w \in (z,w')$ such that the tangent to $v$ at $w$ crosses $v$ on $[x',w)$ exactly once, at some $a' \in (x',z)$.
Since $v'$ is bounded, we may choose a regular $u : [0,1] \rightarrow \R$ such that $u-v$ is concave (so $u$ is coarsely less convex than $v$), $u$ is strictly concave on $[0,a']$ and on $[w,1]$, and $u \leq v$ on $[x',w']$, with equality on $[a',w]$.
Then since $u$ is strictly concave on $[x',a']$ and strictly convex on $[a',z]$, we may choose an $x \in (x',a')$ such that the tangent to $u$ at $x$ lies strictly above (below) $u$ at $a'$ (at $z$).
It follows that there is a convex $p : [0,1] \to \R$ and an $X \in (a',z)$ such that $p$ is affine on $[x',X]$ and on $[X,w']$, and $p \geq u$ on $[x',w']$, with equality on $\{x,w\}$ and with strict inequality at $X$.

Let $F_0$
be a distribution that is atomless
with support $[x',w']$,
\begin{equation*}
	\frac{1}{F_0(X)} \int_0^X \xi F_0(\dd \xi) = x
	\quad \text{and} \quad
	\frac{1}{1-F_0(X)} \int_X^1 \xi F_0(\dd \xi) = w .
\end{equation*}
As $v$ is S-shaped on $[x',w']$,
an `upper censorship' distribution $G$ is optimal
by Kolotilin's (\citeyear{Kolotilin2014}, p.~14) well-known result:
for $a \in (0,1)$ satisfying
\begin{equation*}
	\frac{ v(b) - v(a) }{ b-a } = v'(b) ,
	\quad \text{where} \quad
	b \coloneqq \frac{1}{1-F_0(a)} \int_a^1 \xi F_0( \dd \xi ) ,
\end{equation*}
this distribution $G$ is monotone partitional, where $[0,a)$ is a $G$-revealing interval and $[a,1]$ is $G$-pooling. It is easy to see graphically (in \Cref{fig:Wproof} on \cpageref{fig:Wproof}, paying attention to $p$)
that $a$ must be strictly smaller than $X$.
Thus the optimal distribution $G$ pools some states to the left of $X$
with states to its right.
For the payoff $u$, however,
it is strictly sub-optimal to pool states on either side of $X$ together.
Thus \eqref{eq:mcs_incr} fails:
no distribution optimal for $u$ given $F_0$
is less informative than $G$,
since the latter pools across $X$
while the former do not.
\qed

\subsection{Proof of \texorpdfstring{\Cref{proposition:ido_binary} (\cpageref{proposition:ido_binary})}{Proposition \ref{proposition:ido_binary} (p. \pageref{proposition:ido_binary})}}
\label{app:proofs:pf_thm_nondecr_binary}
	
The `only if' part of \Cref{proposition:ido_binary} follows from the proof in \cref{app:proofs:pf_thm_nondecr:cs_priors_implies_lessvex} of the second (converse) part of {\hyperref[theorem:nondecr_extended]{\Cref*{theorem:nondecr}$^*$}}.%
	\footnote{The argument there shows that
	if $u$ is not coarsely less convex than $v$,
	then we can construct a prior $F_0$
	such that $\argmax_{F \in \mpc{F_0}} \int u \dd F$ is strictly higher than (a fortiori not lower than) $\argmax_{F \in \mpc{F_0}} \int v \dd F$.
	And the constructed prior is, in fact, binary.}

To prove the `if' part, we first review a few standard concepts and results.
Given $0 \le x_0 \le x_1 \le 1$ and an upper semi-continuous $u : [x_0,x_1] \to \R$, write $\cav u$ for the concave envelope of $u$, and let $\mathcal{I}^u$ be the set of maximal non-empty open intervals $\subseteq [x_0,x_1]$ on which $\cav u$ is affine. For any $x \in [x_0,x_1]$, if $x \in I$ for some $I \in \mathcal{I}^u$ then let $H^u_x$ be the binary distribution with support $\{\inf I,\sup I\}$ and mean $x$, and if $x \notin I$ for every $I \in \mathcal{I}^u$ then let $H^u_x \coloneqq \1_{[x,1]}$ (the point mass at $x$). For each $x \in [x_0,x_1]$, let $G^u_x$ be the binary distribution with support
\begin{equation*}
	\Bigl\{ \sup \bigl( \{u = \cav u\} \cap [x_0,x] \bigr),\,
	\inf \bigl( \{u = \cav u\} \cap [x,x_1] \bigr) \Bigr\}
\end{equation*}
and mean $x$.
\textcite{KamenicaGentzkow2011} showed that, given a distribution $F_0$ with $\supp(F_0) = \{x_0,x_1\}$ and mean $\mu$,
\begin{equation*}
	M^u_\mu \coloneqq \argmax_{F \in \mpc{F_0}} \int u \dd F
\end{equation*}
equals the set of all mean-$\mu$ distributions $F$ which satisfy $\int u \dd F = (\cav u)(\mu)$. Hence if $\mu \in I$ for some $I \in \mathcal{I}^u$, then $M^u_\mu$ is the set of all mean-$\mu$ distributions supported on $\cl{I} \cap \{u = \cav u\}$, while if $\mu \notin I$ for every $I \in \mathcal{I}^u$, then $M^u_\mu= \{\1_{[\mu,1]}\}$. It follows that $G^u_\mu$ ($H^u_\mu$) is the least (most) informative element of $M^u_\mu$.

We now prove the `if' part of \Cref{proposition:ido_binary}. 
Let $u,v : [0,1] \to \R$ be upper semi-continuous,
assume that $u$ is coarsely less convex than $v$,
and let $F_0$ be a binary distribution with support $\{x_0,x_1\}$ for some $0 \le x_0 \le x_1 \le 1$, and mean denoted $\mu$. We must show that $M^{\bar u}_\mu$ is lower than $M^{\bar v}_\mu$, where $\bar u$ and $\bar v$ are the restrictions of $u$ and $v$ to $[x_0,x_1]$, respectively. This holds trivially if $\mu \notin I$ for every $I \in \mathcal{I}^{\bar u}$ or if $F_0 = \1_{[\mu,1]}$, since in either of these cases $M^{\bar u}_\mu = \{\1_{[\mu,1]}\}$. Assume for the remainder that $\mu \in I$ for some $I \in \mathcal{I}^{\bar u}$ and that $x_0 < \mu < x_1$. Assume without loss of generality that $x_0=0$ and $x_1=1$.%
	\footnote{All distributions in $\mpc{F_0}$ have support in $[x_0,x_1]$,
	and $\bar u$ is coarsely less convex than $\bar v$;
	so the interval $[x_0,x_1]$ may as well be $[0,1]$.}

Write $x \coloneqq \inf I$ and $w \coloneqq \sup I$, and define
\begin{equation*}
	\mathcal{U} \coloneqq \left\{u = \cav u \right\} \intersect [x,w] .
\end{equation*}
Note that $\supp(H^u_\mu) = \{x,w\}$, and that $\supp(G^u_\mu) = \{y,z\}$ for some $y \leq z$ in $[0,1]$. Clearly $x \leq y \leq \mu \leq z \leq w$, and we have $\{x,w,y,z\} \subseteq \mathcal{U}$ since $u$ is upper semi-continuous.

For the function $v$,
analogously define $x',y',z',w'$ and $\mathcal{V} \subseteq [0,1]$.
We must show that $H^u_\mu$ is less informative than $H^v_\mu$
and that $G^u_\mu$ is less informative than $G^v_\mu$.
The former requires precisely that $x' \leq x$ and $w \leq w'$,
while the latter requires that $y' \leq y$ and $z \leq z'$.
Given $\widehat{x},\widehat{y} \in \R$ and $\alpha \in [0,1]$, let us write $\widehat{x}_\alpha \widehat{y} \coloneqq \alpha \widehat{x} + (1-\alpha) \widehat{y}$.

We first show that $x' \leq x$ and $w \leq w'$.
Since $x,w \in \mathcal{U}$,
we have $u(x_\alpha w) \leq u(x)_\alpha u(w)$ for every $\alpha \in (0,1)$.
As $u$ is coarsely less convex than $v$,
it follows that $v(x_\alpha w) \leq v(x)_\alpha v(w)$ for each $\alpha \in (0,1)$, implying that $[x,w] \subseteq [x',w']$.

\begin{namedthm}[Claim.]
	\label{claim:binary_inclusion}
	$\mathcal{V} \intersect [x,w] \subseteq \mathcal{U}$.
\end{namedthm}

\begin{proof}[Proof]%
	\renewcommand{\qedsymbol}{$\square$}
	Take any $\widehat y \in \mathcal{V} \intersect [x,w]$.
	The result is trivial if $\widehat y=x$ or $\widehat y=w$, so suppose not:
	$\widehat y = x_\alpha w$ for some $\alpha \in (0,1)$.
	Then
	\begin{equation*}
		v(x_\alpha w)
		= (\cav v)(x_\alpha w)
		\geq (\cav v)(x)_\alpha (\cav v)(w)
		\geq v(x)_\alpha v(w)
	\end{equation*}
	since $x_\alpha w \in \mathcal{V}$ (the equality),
	$\cav v$ is concave (first inequality),
	and $\cav v \geq v$ (second inequality),
	whence $u(x_\alpha w) \geq u(x)_\alpha u(w)$
	because $u$ is coarsely less convex than $v$.
	So $u(x_\alpha w) = u(x)_\alpha u(w)$,
	and thus $\widehat y = x_\alpha w \in \mathcal{U}$.
\end{proof}%
\renewcommand{\qedsymbol}{$\blacksquare$}

We now show that $y' \leq y$; the argument for $z \leq z'$ is analogous.
If $y' < x$, then $y' < x \leq y$ since $y \in \mathcal{U} \subseteq [x,w]$.
Suppose instead that $x \leq y'$.
Then since $y' \leq \mu \leq w$,
we have $y' \in [x,w]$.
As $y' \in \mathcal{V}$,
it follows from the \hyperref[claim:binary_inclusion]{claim}
that $y'$ belongs to $\mathcal{U}$.
So $y' \in \mathcal{U} \intersect \left[0,\mu\right]$,
and thus $y' \leq \sup \left( \mathcal{U} \intersect \left[ 0, \mu \right] \right) = y$.
\qed

\subsection{Proof of \texorpdfstring{\Cref{proposition:priorshift} (\cpageref{proposition:priorshift})}{Proposition \ref{proposition:priorshift} (p. \pageref{proposition:priorshift})}}
\label{app:proofs:pf_priorshift}

Fix any atomless full-support distributions $F_0 \neq G_0$; we shall find a regular $u : [0,1] \to \R$ for which \eqref{eq:mcs_incr_prior} fails.
If $F_0$ is not less informative than $G_0$, then \eqref{eq:mcs_incr_prior} fails for any strictly convex $u : [0,1] \to \R$, since $F_0$ ($G_0$) is uniquely optimal for $u$ given $F_0$ ($G_0$).
Assume for the remainder that $F_0$ is less informative than $G_0$; we shall construct a regular $u : [0,1] \to \R$ and distributions $F$ and $G$ that are uniquely optimal for $u$ given $F_0$ and given $G_0$, respectively, such that $F$ is not less informative than $G$.

For an atomless full-support distribution $H$ and $x < y$ in $[0,1]$, write
\begin{equation*}
	M_H(x,y) \coloneqq \frac{1}{H(y)-H(x)}\int_x^y z H(\dd z).
\end{equation*}

\begin{claim}
	\label{claim:alphabet_soup}
	We may choose $\widehat{x} < x < y < \widehat{y} < \widehat{z} < z < \widehat{w} < w$ in $(0,1)$ such that $y = M_{F_0}(x,z)$, $\widehat{y} = M_{G_0}(\widehat{x},\widehat{z})$, $w = M_{F_0}(z,1)$ and $\widehat{w} = M_{G_0}(\widehat{z},1)$.
\end{claim}

\begin{proof}[Proof of the {\hyperref[claim:alphabet_soup]{claim}}]
	\renewcommand{\qedsymbol}{$\square$}
	For any atomless full-support distribution $H$, integrating by parts%
		\footnote{Licensed by e.g. Theorem~18.4 in \textcite{Billingsley1995}.}
	yields that for each $x' \in (0,1)$,
	\begin{equation*}
		M_H(x',1)
		= \frac{ 1 - x' H(x') - \int_{x'}^1 H }{ 1-H(x') }
		= 1 + \frac{ (1-x') H(x') - \int_{x'}^1 H }{ 1-H(x') } .
	\end{equation*}
	We have $\int_{x'}^1 F_0 \geq \int_{x'}^1 G_0$ for every $x' \in (0,1)$ since $F_0$ is less informative than $G_0$.
	Since in addition $F_0 \neq G_0$, it cannot be that $F_0$ is first-order stochastically dominated by $G_0$, so there is an $x \in (0,1)$ such that $F_0(x) < G_0(x)$.
	It follows that $M_{F_0}(x,1) < M_{G_0}(x,1)$. 
	Then, since $G_0$ is continuous, we have $M_{F_0}(x,1) < M_{G_0}(\widehat{x},\widehat{z})$ for some $\widehat{x} \in (0,x)$ and $\widehat{z} \in (M_{F_0}(x,1),1)$.
	Set $\widehat{y} \coloneqq M_{G_0}(\widehat{x},\widehat{z})$ and $\widehat{w} \coloneqq M_{G_0}(\widehat{z},1)$.
	Since $F_0$ is continuous and $\widehat{w} < M_{F_0}(\widehat{w},1)$, there exists a $z \in (\widehat{z},\widehat{w})$ such that $\widehat{w} < M_{F_0}(z,1)$.
	Set $y \coloneqq M_{F_0}(x,z)$ and $w \coloneqq M_{F_0}(z,1)$. Note that $x < y \le M_{F_0}(x,1) < \widehat{y} < \widehat{z}$ and $w < 1$. Hence $0 < \widehat{x} < x < y < \widehat{y} < \widehat{z} < z < \widehat{w} < w < 1$.
\end{proof}
\renewcommand{\qedsymbol}{$\blacksquare$}

Let $F$ be a distribution that is monotone partitional with respect to $F_0$ such that $[0,x]$ is an $F$-revealing interval and $[x,z]$ and $[z,1]$ are $F$-pooling. Let $G$ be a distribution that is monotone partitional with respect to $G_0$ and such that $[0,\widehat{x}]$ is a $G$-revealing interval and $[\widehat{x},\widehat{z}]$ and $[\widehat{z},1]$ are $G$-pooling. $F$ is not less informative than $G$, since $\max \supp(G) = \widehat{w} < w \in \supp(F)$.

\begin{figure}
	\centering








\begin{tikzpicture}[x=15cm, y=2.8cm, scale=0.8,
  line cap=round,
  axis/.style={thin,-{Stealth[length=4pt]}},
  curve/.style={line width=1.2pt,line cap=round,line join=round},
  pcurve/.style={curve,draw=pcolor},
  qcurve/.style={curve,draw=qcolor,dash pattern=on 3.5pt off 4.0pt},
  ucurve/.style={curve,draw=ucolor,line width=1.05pt},
  thinblackcurve/.style={curve,draw=black,line width=0.3pt}
]

\colorlet{pcolor}{black}
\colorlet{qcolor}{black}
\colorlet{ucolor}{gray}

\pgfmathdeclarefunction{f}{1}{%
  \pgfmathparse{20*#1*#1 - 371*#1/30 + 521/200}}
\pgfmathdeclarefunction{p}{1}{%
  \pgfmathparse{#1 <= .30 ? f(#1) :
    (#1 <= .70 ? .5*#1 + 109/200 : 1.5*#1 - 31/200)}}
\pgfmathdeclarefunction{q}{1}{%
  \pgfmathparse{#1 <= .15 ? f(#1) :
    (#1 <= .60 ? 57/40 - 1.5*#1 : 16*#1/7 - 237/280)}}

\pgfmathdeclarefunction{bump}{1}{%
  \pgfmathparse{#1*#1*(1-#1)*(1-#1)}}

\pgfmathdeclarefunction{ua}{1}{%
  \pgfmathsetmacro{\s}{(#1-.30)/.09}%
  \pgfmathparse{.5*#1 + 109/200
    - (13/15)*.09*\s*(1-\s)^2*(1-\s/13+2*\s^2)}}
\pgfmathdeclarefunction{ub}{1}{%
  \pgfmathsetmacro{\s}{(#1-.39)/.10}%
  \pgfmathparse{.74 + .5*(#1-.39) - 10*(#1-.39)^2
    - (49/270)*bump(\s)}}
\pgfmathdeclarefunction{uc}{1}{%
  \pgfmathsetmacro{\s}{(#1-.49)/.34}%
  \pgfmathparse{57/40 - 1.5*#1
    + (901/700)*(\s^2*(35/34-6*\s/17)
      - (35/34+(1520/27)*(17/50)/(53/7))*bump(\s))}}
\pgfmathdeclarefunction{ud}{1}{%
  \pgfmathparse{1471/1400 + (16/7)*(#1-.83)
    - (55/14)*(#1-.83)^2}}
\pgfmathdeclarefunction{ue}{1}{%
  \pgfmathparse{1.5*#1 - 31/200 - 20*(#1-.93)^2}}
\pgfmathdeclarefunction{u}{1}{%
  \pgfmathparse{#1 <= .30 ? f(#1) :
    (#1 <= .39 ? ua(#1) : (#1 <= .49 ? ub(#1) :
    (#1 <= .83 ? uc(#1) : (#1 <= .93 ? ud(#1) : ue(#1)))))}}

  \foreach \t/\lab in {
    .15/\widehat{x}, .30/x, .39/y, .49/\widehat{y},
    .60/\widehat{z}, .70/z, .83/\widehat{w}, .93/w}{
    \pgfmathsetmacro{\top}{max(p(\t),q(\t))}
    \draw[gray!35,densely dotted] (\t,0) -- (\t,\top);
    \draw (\t,.022) -- (\t,-.022)
      node[below=0pt] {$\strut \lab$};
  }
  \draw[<->] (0.07,0) -- (1.055,0);

  \draw[qcurve] (.15,{q(.15)}) -- (.60,{q(.60)}) -- (1,{q(1)});
  \draw[pcurve] plot[domain=.15:.30,samples=70] (\x,{f(\x)})
    -- (.70,{p(.70)}) -- (1,{p(1)});
  \draw[ucurve]
    plot[domain=0.07:.30,samples=70] (\x,{f(\x)});
  \draw[thinblackcurve]
    plot[domain=0.07:.30,samples=70] (\x,{f(\x)});
  \draw[qcurve]
    plot[domain=0.07:.15,samples=70] (\x,{f(\x)});

  \foreach \a/\b/\fun in {
    .30/.39/ua, .39/.49/ub, .49/.83/uc, .83/.93/ud, .93/1/ue}{
    \draw[ucurve] plot[domain=\a:\b,samples=100] (\x,{\fun(\x)});
  }

  \foreach \t in {.30,.70}
    \fill[pcolor] (\t,{p(\t)}) circle[radius=2.2pt];
  \foreach \t in {.15,.60}
    \fill[qcolor] (\t,{q(\t)}) circle[radius=2.2pt];

  \foreach \t in {.39,.49,.83,.93}
    \draw[draw=ucolor,fill=white,line width=.8pt]
      (\t,{u(\t)}) circle[radius=1.8pt];

  \node[right=4pt,text=pcolor] at (1,{p(1)}) {$\boldsymbol{p}$};
  \node[above right=2pt,text=qcolor] at (1,{q(1)}) {$\boldsymbol{q}$};
  \node[below right=2pt,text=ucolor] at (1,{u(1)}) {$\boldsymbol{u}$};
\end{tikzpicture}
	\caption{Proof of \Cref{proposition:priorshift}. Here $p=q=u$ on $[0,\widehat{x}]$ and $p=u$ on $[\widehat{x},x]$.}
	\label{fig:prior}
\end{figure}

Observe (referring to \Cref{fig:prior}) that we may choose continuous convex $p,q : [0,1] \to \R$ such that $p$ is strictly convex on $[0,x]$, affine on $[x,z]$ and on $[z,1]$, and kinked at $x$ and at $z$, $q$ is strictly convex on $[0,\widehat{x}]$, affine on $[\widehat{x},\widehat{z}]$ and on $[\widehat{z},1]$, and kinked at $\widehat{x}$ and at $\widehat{z}$, and $p = q$ on $[0,\widehat{x}]$, $p < q$ on $(\widehat{x},y] \cup [w,1]$, and $p > q$ on $[\widehat{y},\widehat{w}]$. Then (referring again to \Cref{fig:prior}) we may choose a regular $u : [0,1] \to \R$ such that $p \ge u \le q$, $\{u = p\} = [0,x] \cup \{y,w\}$ and $\{u = q\} = [0,\widehat{x}] \cup \{\widehat{y},\widehat{w}\}$.

Observe that $H \in \mpc{F_0}$ satisfies $\int u \dd H = \int p \dd F_0$ iff $H=F$, and that $H \in \mpc{G_0}$ satisfies $\int u \dd H = \int q \dd G_0$ iff $H=G$. Thus $F$ ($G$) is uniquely optimal for $u$ given $F_0$ (given $G_0$), by Theorem~1 and Proposition~1 of \textcite{DworczakMartini2019}.
\qed

\subsection{Proof of \texorpdfstring{\hyperref[theorem:incr_general]{\Cref*{theorem:incr}$\boldsymbol{'}$} (\cpageref{theorem:incr_general})}{Theorem \ref{theorem:incr}' (p. \pageref{theorem:incr_general})}}
\label{app:proofs:pf_theorem_incr_general}

For the first half (sufficiency), fix a prior distribution $F_0$, and let $u,v : E \to \R$ be strongly regular with $u$ coarsely less convex than $v$. If $u$ is concave, then it is strictly concave by strong regularity, so the point mass at $\mu_0 \coloneqq \int x F_0(\dd x)$ is uniquely optimal for $u$ given $F_0$, so \eqref{eq:mcs_incr_general} holds. If $u$ is convex, then it is strictly convex by strong regularity, and hence so is $v$, in which case $F_0$ is uniquely optimal for $v$ given $F_0$, so \eqref{eq:mcs_incr_general} holds.

For the second half (necessity), given any $S \subseteq E$, let $\interior(S)$ denote its relative interior. Recall that $E$ is a non-empty, compact and convex subset of $\R^n$, where $n \in \N$.

\begin{lemma}
	\label{lemma:indefinite}
	Assume that $E$ is not one-dimensional and that $u : E \to \R$ is strongly regular. If there is no $x \in \interior(E)$ at which the Hessian matrix $(\nabla^2 u)(x)$ is indefinite, then $u$ is either concave or convex.
\end{lemma}

\begin{proof}[Proof of \Cref{lemma:indefinite}]
	\renewcommand{\qedsymbol}{$\square$}%
	Suppose toward a contradiction that $u$ is neither concave nor convex. Since the Hessian $\nabla^2 u$ is nowhere indefinite, it follows that there exist distinct $\widehat{x},\widehat{y} \in \interior(E)$ such that $u$ is locally convex at $\widehat{x}$ and locally concave at $\widehat{y}$. 
	Let $w : [0,1] \to \R$ be given by $w(\alpha) \coloneqq u(\alpha \widehat{x} + (1-\alpha)\widehat{y})$ for each $\alpha \in [0,1]$.
	Since $u$ is locally convex at $\widehat{x}$ and locally concave at $\widehat{y}$, and strongly regular, there are $\underline \alpha < \overline{\alpha}$ in $(0,1)$ such that $w''(\underline\alpha) < 0 < w''(\overline{\alpha})$.
	Then, fixing a compact convex two-dimensional $E' \subseteq E$ such that $\co(\{\widehat{x},\widehat{y}\}) \subseteq E'$, the restriction $\widehat{u}$ of $u$ to $E'$ is (strongly regular and) neither concave nor convex. Moreover, viewing $\widehat{u}$ as a function of two arguments, the Hessian of $\widehat{u}$ is nowhere indefinite. Therefore, by considering $\widehat{u}$ instead of $u$, we may assume without loss that $n = 2$. Since $u$ is twice continuously differentiable with $w''(\underline \alpha) < 0 < w''(\overline{\alpha})$, we may further assume (applying an affine change of coordinates if required) that $[0,1]^2 \subseteq \interior(E)$ and $u_{22}(\cdot,0) < 0 < u_{22}(\cdot,1)$ on $[0,1]$.

	Let $\phi,\psi : [0,1] \to [0,1]$ be such that for each $x_1 \in [0,1]$, $\phi(x_1)$ is the largest $\delta \in (0,1)$ such that $u(x_1,\cdot)$ is concave on $[0,\delta]$ (so that $\phi(x_1) < 1$), and $\psi(x_1)$ is the largest $\delta \in (\phi(x_1),1]$ such that $u(x_1,\cdot)$ is convex on $[\phi(x_1),\delta]$.
	We shall exhibit a closed proper interval $I \subseteq [0,1]$ on which $\phi$ and $\psi$ are continuous.
	To this end, let $K : [0,1] \to \{0,1,\dots,\infty\}$ be such that for each $x_1 \in [0,1]$, $K(x_1)$ is the number of sign changes of $u_{22}(x_1,\cdot)$.%
		\footnote{Given a non-empty interval $J \subseteq \R$ and a map $\Phi : J \to \R$, the \emph{number of sign changes} of $\Phi$ is the supremum over $\kappa \in \{0,1,2,\dots\}$ such that there exists an increasing sequence $(x_k)_{k=0}^\kappa \subseteq J$ with $\Phi(x_{k-1})\Phi(x_k) < 0$ for each $k \in \{1,2,\dots,\kappa\}$.} 
	Note that $K \ge 1$ since $u_{22}(\cdot,0) < 0 < u_{22}(\cdot,1)$ on $(0,1)$, and that $K < \infty$ since $u$ is strongly regular.
	Note also that $K$ is lower semi-continuous, hence bounded on a closed proper interval $I \subseteq [0,1]$,%
		\footnote{$\{K > m\}$ is open for each $m \in \N$ since $K$ is lower semi-continuous, and $\bigcap_{m \in \N} \{K > m\} = \varnothing$ since $K < \infty$, so that $\{K > m\}$ fails to be dense for some $m \in \N$.}
	so that $\argmax_{I} K$ is non-empty and open relative to $I$. Then, after replacing $I$ with an appropriate sub-interval of $\argmax_I K$ if necessary, we may assume without loss that $K$ is constant on $I$.
	Observe that $K$ is discontinuous wherever either $\phi$ or $\psi$ is, because $u$ is strongly regular, so that $\phi$ and $\psi$ are continuous on $I$.

	We now exhibit a closed proper interval $J \subseteq I$ on which $\phi$ is either concave or convex. 
	Because it holds on $I$ that $\phi$ and $\psi$ are continuous and satisfy $\phi < \psi$, we may assume without loss (replacing $I$ with a sub-interval if necessary) that $\max_I \phi < \min_I \psi$.
	Then, fixing an $\bar{x}_2 \in (\max_I \phi, \min_I \psi)$, $u$ is locally strictly convex on $D_+ \coloneqq \{(x_1,x_2) \in I \times [0,\bar{x}_2] : x_2 > \phi(x_1)\}$ and locally strictly concave on $D_- \coloneqq \{(x_1,x_2) \in I \times [0,\bar{x}_2] : x_2 < \phi(x_1)\}$, since $u$ is strongly regular and the Hessian of $u$ is nowhere indefinite on $\interior(E)$. 
	Moreover, since $u$ is strongly regular, given any $x_1 \in \interior(I)$ and $\lambda \in \R$, there exists a $\delta > 0$ such that $(u \circ \gamma_\delta)$ is either strictly convex or strictly concave, where $\gamma_\delta : [0,1] \to I \times [0,\bar{x}_2]$ is affine with $\gamma_\delta(0) = (x_1,\phi(x_1))$ and $\gamma_\delta(1) = (x_1+\delta,\phi(x_1)+\delta\lambda)$. 
	Hence the map $t \mapsto \phi(x_1+t)-\phi(x_1)-\lambda t$ is either positive or negative on $[0,\delta]$. 
	By Theorem~12 in \textcite{AgronskyBrucknerLaczkovichPreiss1985},%
		\footnote{Their theorem features a measure $\omega$; we apply it with $\omega$ being the point mass at 1.}
	it follows that $\phi$ is either convex or concave on some closed proper interval $J \subseteq I$.

	We assume for the remainder that $\phi$ is convex on $J$, omitting the analogous argument in case $\phi$ is concave on $J$.
	Let $\mathcal{G} \coloneqq \{(x_1,\phi(x_1)) : x_1 \in J\} \subseteq E$ be the graph of the restriction of $\phi$ to $J$, fix distinct $x,y \in \mathcal{G}$, and let $w_\star : [0,1] \to \R$ be given by $w_\star(\alpha) \coloneqq u(\alpha x + (1-\alpha) y)$ for each $\alpha \in [0,1]$.
	Because $u$ is locally strictly convex (concave) on $D_+$ ($D_-$) and twice continuously differentiable, we have at each $x \in \mathcal{G}$ that the Hessian matrix $(\nabla^2 u)(x)$ is both positive semi-definite and negative semi-definite, which is to say that $(\nabla^2 u)(x)$ is the zero matrix.
	Hence $\nabla u$ is constant on $\mathcal{G}$, since $\phi$ is convex on $J$,%
		\footnote{To see why, observe that there exists a Lipschitz continuous and onto map $\gamma : [0,1] \to \mathcal{G}$ such that $\gamma_1$ is strictly increasing, and that $\frac{\dd}{\dd \alpha} [(\nabla u) \circ \gamma](\alpha) = [ (\nabla^2 u) (\gamma(\alpha)) ] [(\nabla\gamma)(\alpha)] = 0$ for a.e. $\alpha \in [0,1]$, so that $(\nabla u) \circ \gamma$ is constant.}
	so that $w_\star'(0) = (\nabla u)(y) \cdot (x-y) = (\nabla u)(x) \cdot (x-y) = w_\star'(1)$.
	Moreover, $\co(\mathcal{G}) \subseteq \cl(D_+)$ since $\phi$ is convex, so that $u$ is convex on $\co(\mathcal{G})$ and, hence, $w_\star'$ is increasing.
	Thus $w_\star$ is affine, a contradiction with the strong regularity of $u$.
\end{proof}
\renewcommand{\qedsymbol}{$\blacksquare$}

\begin{lemma}
	\label{lemma:regular}
	Suppose that $E$ has dimension $n$ and that $0 \in \interior(E)$.
	Let $w : E \to \R$ be twice continuously differentiable with bounded derivatives. 
	Then there exists a convex $\psi : E \to \R$ such that $w+\psi$ is strongly regular and $\abs*{\psi(x)}/\norm*{x}^2 \to 0$ as $x \to 0$.
\end{lemma}

\begin{proof}[Proof of \Cref{lemma:regular}]
	\renewcommand{\qedsymbol}{$\square$}%
	Let $w^0 : E \to \R$ be the second-order Taylor expansion of $w$ around $0$, i.e. $w^0(x) \coloneqq w(0) + (\nabla w)(0) \cdot x + \frac 12 x^\top [(\nabla^2 w)(0)] x$ for each $x \in E$.
	Since $w$ is twice continuously differentiable with bounded derivatives, we may define $\gamma : \R_+ \to \R_+$ by 
	\begin{equation*}
		\gamma(r) = \max_{x \in E : \norm*{x} \le r} \norm*{ \left(\nabla^2 w\right)(x) - \left(\nabla^2 w\right)(0) } 
		\quad \text{for each $r \in \R_+$,}
	\end{equation*}
	where $\norm*{\cdot}$ is the operator norm ($\norm*{A} \coloneqq \max_{y \in \R^n : \norm*{y}=1} \norm*{Ay}$ for any $n\times n$ matrix $A$). Let $\delta,\Delta: \R_+ \to \R_+$ be given by 
	\begin{equation*}
		\delta(r) = 2(r + \gamma(2r)) \quad \text{and}\quad \Delta(r) = \int_0^r (r-s)\delta(s)\dd s
		\quad \text{for each $r \in \R_+$,}
	\end{equation*}
	let $\phi : E  \to \R$ be the radial function $x \mapsto \Delta(\norm*{x})$ and set $\psi \coloneqq w^0-w+\phi$.

	Note that $\gamma$ is continuous with $\gamma(0) = 0$, so that $\Delta$ is twice continuously differentiable with $\Delta(r)/r^2 \to 0$ as $r \to 0$ by Taylor's theorem. Then $\psi$ is twice continuously differentiable with bounded derivatives and, invoking Taylor's theorem again, $\abs*{\psi(x)}/\norm*{x}^2 \to 0$ as $x \to 0$.

	We next show that $y^\top [(\nabla^2 \psi)(x)] y \ge 0$ for all $x \in \interior(E) \setminus \{0\}$ and $y \in \R^n$, so that $\psi$ is convex.
	Fix such $x$ and $y$, and observe that
	\begin{equation}
		\left(\nabla^2 \phi\right)(x) 
		= \Delta''(r)\bar x \bar x^\top+\frac{\Delta'(r)}r\left(I_n-\bar x \bar x^\top\right) ,
		\label{eq:radial_hessian}
	\end{equation}
	where $r \coloneqq \norm*{x}$, $\bar x \coloneqq x/r$, and $I_n$ is the $n \times n$ identity matrix.
	Then multiplication by $(\nabla^2 \phi)(x)$ scales $\bar x$ by $\delta(r)$ and scales any vector that is perpendicular to $\bar x$ by $\lambda(r)$,
	implying that $(\nabla^2 \phi)(x)$ has eigenvalues $\Delta''(r) = \delta(r)$ and $\Delta'(r)/r \eqqcolon \lambda(r)$.
	Moreover, $\delta(r) \ge \frac 1 r \int_0^r \delta = \lambda(r)$ since $\delta$ is increasing, and 
	\begin{equation*}
		\frac 1 r \int_0^r \delta = r + \frac 2 r \int_0^r \gamma(2s) \dd s \ge  r + \gamma(r),
	\end{equation*}
	where the inequality holds since $\gamma$ is non-negative and increasing.
	Thus $\delta(r) \ge r + \gamma(r) \le \lambda(r)$, whence
	\begin{equation*}
		y^\top \left[\left(\nabla^2 \phi\right)(x)\right] y \ge (r+\gamma(r))\norm*{y}^2 \ge r\norm*{y}^2 + y^\top \left[ \left(\nabla^2 w\right)(x) - \left(\nabla^2 w\right)(0) \right] y ,
	\end{equation*}
	so that $y^\top [(\nabla^2 \psi)(x)] y \ge r\norm*{y}^2 \ge 0$.

	It remains to prove that $w+\psi$ is strongly regular. 
	Since $w$ and $\psi$ are both twice continuously differentiable with bounded derivatives, it suffices to show that, for all $x,y \in \R^n$ such that $x \cdot y = 0$, $\norm*{y} = 1$ and $T \coloneqq \interior(\{t \in \R : x + ty \in E) \ne \varnothing$, the second derivative of $(w+\psi)(x+ty)$ with respect to $t$ is strictly quasiconvex on $T$.
	To that end, fix $x,y \in \R^n$, and note that for all $t \in T$ such that $x +ty \ne 0$, it holds that
	\begin{multline*}
		\frac{\dd^2 (w+\psi)}{\dd t^2}(x+ty)
		= y^{\top}\left[ \left(\nabla^2 w\right)(0) + \left(\nabla^2 \phi\right)(x+ty) \right] y
		\\
		= y^{\top} \left[ \left(\nabla^2 w\right)(0) \right] y + \frac{t^2 \delta(\norm*{x+ty}) + \norm*{x}^2 \lambda(\norm*{x+ty})}{\norm*{x}^2+t^2},
	\end{multline*}
	where the second equality follows from \eqref{eq:radial_hessian}.
	Viewed as a function of $t$, the right-hand side is strictly decreasing on $T \cap (-\infty,0)$ and strictly increasing on $T \cap (0,\infty)$, since $\delta$ and $\lambda$ are strictly increasing on $(0,\infty)$ with $\delta \ge \lambda$. 
\end{proof}
\renewcommand{\qedsymbol}{$\blacksquare$}

Fix a strongly regular $u : E \to \R$ that is neither concave nor convex; we shall find a strongly regular $v : E \to \R$ that is coarsely more convex than $u$ and an atomless convex-support distribution $F_0$ such that \eqref{eq:mcs_incr_general} fails.

Assume without loss that $E$ has dimension $n$, and note that $n \geq 2$ by hypothesis. By \Cref{lemma:indefinite}, there is an $x \in \interior(E)$ at which the Hessian matrix $(\nabla^2 u)(x)$ is indefinite. Assume without loss that $x = 0$. Since $(\nabla^2 u)(0)$ is indefinite, it admits eigenvalues $\lambda_1,\lambda_2$ such that $\lambda_2 < 0 < \lambda_1$.
As $(\nabla^2 u)(0)$ is symmetric, its eigenvectors (appropriately rescaled) form an orthonormal basis of $\R^n$. We henceforth express elements of $E$ in coordinates relative to this basis, with the eigenvectors associated with $\lambda_1$ and $\lambda_2$ as (respectively) the first and second basis vectors.
Then $u_{11}(0) = \lambda_1$, $u_{12}(0) = 0$ and $u_{22}(0) = \lambda_2$, where subscripts denote partial derivatives.
Assume without loss that $u(0) = u_1(0) = u_2(0) = 0$.
Let 
\begin{equation*}
	S \coloneqq \left\{x \in \R^n : \norm*{x} \leq 1 \text{ and $x_i = 0$ for $i > 2$}\right\} ,
\end{equation*}
and note that since $0 \in \interior(E)$, we may assume without loss that $S \subseteq E$.

Let $u^0 : S \to \R$ be given by 
\begin{equation*}
	u^0(x)
	\coloneqq \tfrac{1}{2} \left( \lambda_1 x_1^2 + \lambda_2 x_2^2 \right)
	\quad \text{for each $x \in S$.}
\end{equation*}
A second-order Taylor expansion of $u$ around $0$ yields that
\begin{equation}
	\abs*{u(x)-u^0(x)}
	\bigm/
	\norm*{x}^2
	\rightarrow 0
	\quad \text{as $x \rightarrow 0$ in $S$.}
	\label{eq:u_approx}
\end{equation}
Since $u$ is twice continuously differentiable, applying \Cref{lemma:regular} to $x \mapsto u(x) - \tfrac{1}{2}\lambda_2(x_1+x_2)^2$ delivers a convex and twice differentiable $\psi : E \to \R$ with $\abs*{\psi(x)} \bigm/ \norm*{x}^{2} \rightarrow 0$ as $x \rightarrow 0$ such that $v : E \to \R$ given by 
\begin{equation*}
	v(x)
	\coloneqq u(x) - \tfrac{1}{2}\lambda_2(x_1+x_2)^2 + \psi(x)
	\quad \text{for each $x \in E$}
\end{equation*}
is strongly regular. Since $v-u$ is convex, $v$ is coarsely more convex than $u$ by \Cref{corollary:CLC_suffsuff}. By a second-order Taylor expansion of $v$ around $0$, 
\begin{equation}
	\abs*{v(x)-v^0(x)}
	\bigm/
	\norm*{x}^2
	\rightarrow 0
	\quad \text{as $x \rightarrow 0$ in $S$,}
	\label{eq:v_approx}
\end{equation}
where $v^0 : S \to \R$ is given by 
\begin{equation*}
	v^0(x)
	\coloneqq u^0(x) - \tfrac{1}{2}\lambda_2(x_1+x_2)^2
	\quad \text{for each $x \in S$.}
\end{equation*}

Let $G_0$ be the uniform distribution on $S$.
Note that there are no distinct $x,y \in \interior(S)$ such that either $(\nabla u^0)(x) = (\nabla u^0)(y)$ or $(\nabla v^0)(x) = (\nabla v^0)(y)$.
Hence by Theorem~7 in \textcite{DworczakKolotilin2023}, 
\begin{equation*}
	\argmax_{F \in \mpc{G_0}} \int u^0 \dd F
	= \{F^\star\}
	\quad \text{and} \quad 
	\argmax_{F \in \mpc{G_0}} \int v^0 \dd F
	= \{G^\star\}
\end{equation*}
for some distributions $F^\star$ and $G^\star$.
We shall (a)~show that $F^\star$ is not less informative than $G^\star$, and then (b)~deduce that \eqref{eq:mcs_incr_general} fails for some atomless convex-support prior distribution $F_0$.

For part~(a), let $F$ be the posterior-mean distribution induced (given prior $G_0$) by a signal that reveals the first coordinate of the state and nothing else.
The map $p : S \to \R$ given by $p(x) \coloneqq \frac{1}{2} \lambda_1 x_1^2$ for each $x \in S$ is convex and Lipschitz with $p \geq u^0$, and it satisfies $\int (p-u^0)\dd F = 0$ since $F$ assigns probability 1 to $\{x \in S: x_2 = 0\}$.
Hence $F^\star = F$ by Theorem~5 in \textcite{DworczakKolotilin2023}.
Thus if $G^\star$ were more informative than $F^\star$, then \emph{any} distribution more informative than $F^\star$ would also be optimal for $v^0$, since $F^\star=F$ reveals the first coordinate of the state and $v^0(x_1,\cdot)$ is affine for each $x_1 \in [-1,1]$. As $G^\star$ is uniquely optimal for $v^0$ given $G_0$, it therefore cannot be more informative than $F^\star$.

For part~(b), define $u^\eps,v^\eps : S \to \R$ by $u^\eps(x) \coloneqq u(\eps x) / \eps^2$ and $v^\eps(x) \coloneqq v(\eps x) / \eps^2$ for each $x \in S$ and $\eps \in (0,1]$, and let $M_u,M_v : [0,1] \rightrightarrows \mpc{G_0}$ be given by, for each $\eps \in [0,1]$,  
\begin{equation*}
	M_u(\eps) = \argmax_{F \in \mpc{G_0}} \int u^\eps \dd F 
	\quad \text{and} \quad 
	M_v(\eps) = \argmax_{F \in \mpc{G_0}} \int v^\eps \dd F .
\end{equation*} 

Given a distribution $F$ supported within $S$ and an $\eps \in (0,1]$, let $F^\eps$ be the pushforward of $F$ by $x \mapsto \eps x$.
Since rescaling interim payoffs (by $1/\eps^2$) and the prior (by $1/\eps$) affects neither feasibility nor the sender's preferences,%
	\footnote{Writing $F \circ \eps^{-1}$ for the pushforward by $x \mapsto \eps x$ of a distribution $F$, (i)~a distribution $F$ belongs to $\mpc{G_0}$ iff $F \circ \eps^{-1}$ belongs to $\mpc{G^\eps_0}$, and (ii)~for $F,G$ concentrated on $S$, $\int u^\eps \dd F \geq \mathrel{(>)} \int u^\eps \dd G$ iff $\int u \dd (F \circ \eps^{-1}) \geq \mathrel{(>)} \int u \dd (G \circ \eps^{-1})$, and similarly for $v^\eps$ and $v$.}
\begin{align*}
	\argmax_{F \in \mpc{G^\eps_0}} \int u \dd F &= \{F^\eps : F \in M_u(\eps)\} \quad \text{and} \\
	\argmax_{F \in \mpc{G^\eps_0}} \int v \dd F &= \{F^\eps : F \in M_v(\eps)\}
\end{align*}
for each $\eps \in (0,1]$.
It therefore suffices to exhibit an $\eps \in (0,1]$ such that $M_u(\eps)$ is not lower than $M_v(\eps)$, since $F \in \mpc{G_0}$ is less informative than $G \in \mpc{G_0}$ if and only if $F^\eps$ is less informative than $G^\eps$.

To that end, for each $k \in \N$, fix arbitrary $F^k \in M_u(1/k)$ and $G^k \in M_v(1/k)$.
Endow $\mpc{G_0}$ with the topology of weak convergence and observe that it is compact, being a closed subset of the space of all distributions supported within $S$, which is compact by Prokhorov's theorem \parencite[e.g.][Theorem~5.1]{Billingsley1999}.
We may therefore assume, passing to a subsequence if necessary, that $F^k \to F$ and $G^k \to G$ for some $F,G \in \mpc{G_0}$. Note that for each $H \in \mpc{G_0}$, the maps $[0,1] \times \mpc{G_0} \to \R$ given by 
\begin{equation*}
	(\eps,F) \mapsto \int u^\eps \dd F
	\quad \text{and} \quad 
	(\eps,F) \mapsto \int v^\eps \dd F
\end{equation*} 
are continuous at $(0,H)$.%
	\footnote{To see why the first map is continuous at $(0,H)$, note that $u^\eps \rightarrow u^0$ uniformly since $\sup_{x \in S}\abs{u^\eps(x)-u^0(x)} = \sup_{x \in S}\abs{u(\eps x)-u^0(\eps x)}/\eps^2 \rightarrow 0$ as $\eps \downarrow 0$ by definition of $u^0$ and \eqref{eq:u_approx}. Then, for any sequence $(\eps_k,H^k)_{k \in \N}$ converging to $(0,H)$, $\abs*{\int u^{\eps_k} \dd H^k - \int u^0 \dd H} \le \sup_{x \in S}\abs{u^{\eps_k}(x)-u^0(x)} + \abs*{\int u^0 \dd (H^k-H)} \to 0$. The argument for the second map is analogous.} 
Hence $M_u$ and $M_v$ are upper semi-continuous at $0$ by the Berge maximum theorem \parencite[e.g.][Theorem~17.31]{AliprantisBorder2006}.
It follows that $F = F^\star$ and $G = G^\star$, so that $F^k$ is not less informative than $G^k$ for large enough $k$, since $F^\star$ is not less informative than $G^\star$.
Since our choice of $F^k$ and $G^k$ at each $k \in \N$ was arbitrary, this shows that there is a $k \in \N$ such that $M_u(1/k)$ is not lower than $M_v(1/k)$.
\qed

\section{More on mean-preserving spreads}
\label{app:mps}

In this appendix, we supplement §\ref{sec:mps} by proving \Cref{lemma:mps_mpc} and \Cref{proposition:mps_prior}, and establishing two propositions which parallel \Cref{proposition:mpsxy} except that only constraint sets of the form $\mps{F_0}$ (=$\mpsI{[0,1]}{F_0}$) are considered.

The following two propositions are the analogues of, respectively, \Cref{theorem:nondecr,theorem:incr}, when the constraint set is $\mps{F_0}$ (rather than $\mpc{F_0}$ or $\mpsI{I}{F_0}$). Recall from \cref{app:proofs:pf_thm_nondecr_binary} the notation `$\mathcal{I}^u$'.

\begin{proposition}
	\label{proposition:mps_nodecr}
	Let $u,v : [0,1] \to \R$ be upper semi-continuous.
	The following are equivalent:
	
	\begin{enumerate}[label=(\Roman*)]
		
		\item \label{item:mps_clc}
		For every $I \in \mathcal{I}^u$, there exists a $J \in \mathcal{I}^v$ such that $I \subseteq J$, and if $I=J$ then $\{v = \cav v\} \cap I \subseteq \{u = \cav u\} \cap I$.
		
		\item \label{item:mps_thm1_nondecr} 
		For every distribution $F_0$, \eqref{eq:mps_nohigher} holds.

	\end{enumerate}
\end{proposition}

\begin{proposition}
	\label{proposition:mps_incr}
	Let $u : [0,1] \to \R$ be upper semi-continuous.
	The following are equivalent:
	
	\begin{enumerate}[label=(\Roman*)]
		
		\setcounter{enumi}{2}
		
		\item \label{item:mps_crater}
		Either $u$ is concave, or $(1-\alpha) u(0)+\alpha u(1) \ge u(\alpha)$ for all $0 < \alpha < 1$.
		
		\item \label{item:mps_thm2_incr}
		For every upper semi-continuous $v : [0,1] \to \R$ satisfying \ref{item:mps_clc}, and for every distribution $F_0$, \eqref{eq:mps_lower} holds.

	\end{enumerate}
\end{proposition}

Property~\ref{item:mps_clc} is a notion of $u$ being `very coarsely' less convex than $v$, and by \Cref{proposition:mpsxy}, it is implied by (but does not imply) $u$ being coarsely less convex than $v$. Property~\ref{item:mps_crater} is very strong; however, by \Cref{proposition:mpsxy}, it can be dispensed with if we consider only $v : [0,1] \to \R$ which are coarsely (rather than merely `very coarsely') more convex than $u$.

In the rest of this appendix, we prove \Cref{lemma:mps_mpc} and \Cref{proposition:mps_prior,proposition:mps_nodecr,proposition:mps_incr}. Recall the concepts, results and notation reviewed in \cref{app:proofs:pf_thm_nondecr_binary}, and the notion of a distribution family introduced in \cref{app:proofs:pf_thm_nondecr:lessvex_implies_cs_intervals}.

\subsection{Proof of \texorpdfstring{\Cref{lemma:mps_mpc} (\cpageref{lemma:mps_mpc})}{Lemma \ref{lemma:mps_mpc} (p. \pageref{lemma:mps_mpc})}}
\label{app:mps:lemma_mps_mpc}

\begin{namedthm}[Claim A.]
	\label{claim:mps}
	Given an upper semi-continuous $u : [0,1] \to \R$, $(G^u_x)_{x \in [0,1]}$ and $(H^u_x)_{x \in [0,1]}$ are distribution families.
	Moreover, for any distribution $F_0$,
	\begin{equation*}
		G^u \coloneqq \int G^u_x F_0(\dd x) \quad \text{and} \quad H^u \coloneqq \int H^u_x F_0(\dd x)
	\end{equation*}  
	are, respectively, the least and the most informative elements of
	\begin{equation*}
		\argmax_{F \in \mps{F_0}} \int u \dd F .
	\end{equation*}
\end{namedthm}

\begin{proof}[Proof of {\hyperref[claim:mps]{claim A}}]
	For the first part, note that since $u$ is upper semi-continuous, there is a set $\mathcal{J}$ of disjoint open intervals such that
	\begin{equation*}
		\{ x \in [0,1] : u(x) \ne (\cav u)(x) \} = \bigcup_{J \in \mathcal{J}} J .
	\end{equation*}
	Hence, for any continuous $w : [0,1] \to \R$, the map $x \mapsto \int w \dd G^u_x$ is Borel measurable, since it is affine on each $I \in \mathcal{J}$ and equal to $w$ on $\{u = \cav u\}$. Similarly, the map $x \mapsto \int w \dd H^u_x$ is Borel measurable, as it is affine on each $I \in \mathcal{I}^u$ and equal to $w$ on $[0,1] \setminus (\bigcup_{I \in \mathcal{I}^u} I)$. Thus $(G^u_x)_{x \in [0,1]}$ and $(H^u_x)_{x \in [0,1]}$ are distribution families.

	For the second part, note first that for any convex $\psi : [0,1] \to \R$,
	\begin{equation*}
		\int \psi \dd G^u
		= \int \int \psi \dd G^u_x F_0(\dd x)
		\geq \int \psi \dd F_0
		\leq \int \int \psi \dd H^u_x F_0(\dd x)
		= \int \psi \dd H^u ,
	\end{equation*}
	where the equalities hold by \Cref{lemma:dist_family}\ref{item:fam_product} in \cref{app:proofs:pf_thm_nondecr:lessvex_implies_cs_intervals} and the inequalities hold since $\psi$ is convex and, for each $x \in [0,1]$, $G^u_x$ and $H^u_x$ have mean $x$. Thus $G^u \in \mps{F_0} \ni H^u$. Now, fix an arbitrary $F \in \mps{F_0}$; we will show that $\int u \dd G^u \ge \int u \dd F \le \int u \dd H^u$, and that at least one of the inequalities is strict unless $F$ is both more informative than $G^u$ and less informative than $H^u$.
	By Blackwell's theorem \parencite[e.g.][p.~94]{Phelps2001}, there exists a distribution family $(\lambda_x)_{x \in [0,1]}$ such that $F = \int \lambda_x F_0(\dd x)$ and for each $x \in [0,1]$, $\lambda_x$ has mean $x$.
	For each $x \in [0,1]$, letting $F^x$ be the distribution with mean $x$ and $\supp(F^x) \subseteq \{0,1\}$, it holds that $\lambda_x \in \mpc{F^x}$, so that $\int u \dd G^u_x = \int u \dd H^u_x \ge \int u \dd \lambda_x$, where the inequality is strict unless $\lambda_x$ is both more informative than $G^u_x$ and less informative than $H^u_x$. Hence
	\begin{equation*}
		\int u \dd G^u
		= \int \int u \dd G^u_x F_0(\dd x)
		\geq \int\int u \dd \lambda_x F_0(\dd x)
		= \int u \dd F ,
	\end{equation*}
	where the equalities hold by \Cref{lemma:dist_family}\ref{item:fam_product} in \cref{app:proofs:pf_thm_nondecr:lessvex_implies_cs_intervals}. Similarly, $\int u \dd H^u \geq \int u \dd F$. Furthermore, both inequalities are strict unless for $F_0$-a.e. $x \in [0,1]$, $\lambda_x$ is both more informative than $G^u_x$ and less informative than $H^u_x$.

	It remains only to show that if for $F_0$-a.e. $x \in [0,1]$, $\lambda_x$ is both more informative than $G^u_x$ and less informative than $H^u_x$, then $F$ is both more informative than $G^u$ and less informative than $H^u$. This follows since for any convex $\psi : [0,1] \to \R$, if $\int \psi \dd G^u_x \le \int \psi \dd \lambda_x$ for $F_0$-a.e. $x \in [0,1]$, then
	\begin{equation*}
		\int \psi \dd G^u
		= \int \int \psi \dd G^u_x F_0(\dd x)
		\le \int \int \psi \dd \lambda_x F_0(\dd x)
		= \int \psi \dd F, 
	\end{equation*}
	where the equalities follow from \Cref{lemma:dist_family}\ref{item:fam_product} in \cref{app:proofs:pf_thm_nondecr:lessvex_implies_cs_intervals}, and similarly if $\int \psi \dd \lambda_x \leq \int \psi \dd H^u_x$ for $F_0$-a.e. $x \in [0,1]$ then $\int \psi \dd F \leq \int \psi \dd H^u$.
\end{proof}

\begin{namedthm}[Claim B.]
	\label{claim:mps_clc}
	For upper semi-continuous $u,v : [0,1] \to \R$, `\eqref{eq:mpc_nohigher} holds for every distribution $F_0$ with $\supp(F_0) \subseteq \{0,1\}$' implies \ref{item:mps_clc}.
\end{namedthm}

\begin{proof}[Proof of {\hyperref[claim:mps_clc]{claim B}}]
	We prove the contra-positive. Suppose that \ref{item:mps_clc} fails, and consider two cases.

	\emph{Case~1: there exists $I \in \mathcal{I}^u$ such that $I \subseteq J$ for no $J \in \mathcal{I}^v$.}
	In this case, there exists an $x \in I$ such that 
	$H^v_x = \1_{[x,1]} = G^v_x$. Moreover, $H^u_x \ne H^v_x$ since $x \in I \in \mathcal{I}^u$. Hence \eqref{eq:mpc_nohigher} fails for the distribution $F_0$ with mean $x$ and $\supp(F_0) \subseteq \{0,1\}$.

	\emph{Case~2: for all $I \in \mathcal{I}^u$, there exists $J \in \mathcal{I}^v$ such that $I \subseteq J$.}
	In this case, since \ref{item:mps_clc} fails, there exists an $I \in \mathcal{I}^u \cap \mathcal{I}^v$ and an $x \in I$ such that $u(x) < (\cav u)(x)$ and $v(x) = (\cav v)(x)$. Then $G^v_x = \1_{[x,1]}$ and $G^u_x \ne G^v_x$. Since $I = J$, we have $H^u_v = H^v_x$. Thus \eqref{eq:mpc_nohigher} fails for the distribution $F_0$ with mean $x$ and $\supp(F_0)\subseteq \{0,1\}$.
\end{proof}

\begin{proof}[Proof of \Cref{lemma:mps_mpc}]
	The `only if' claims were proved in the text. To establish the second `if' claim, namely that `\eqref{eq:mpc_lower} for every distribution $F_0$ with $\supp(F_0) \subseteq \{0,1\}$' implies `\eqref{eq:mps_lower} for every distribution $F_0$', it suffices in light of \hyperref[claim:mps]{claim A} to show that both of the following hold:
	
	\begin{enumerate}[label=(\roman*)]
	
		\item \label{item:mps_G}
		if $G^u_x$ is less informative than $G^v_x$ for every $x \in [0,1]$, then for any distribution $F_0$, $\int G^u_x F_0(\dd x)$ is less informative than $\int G^v_x F_0(\dd x)$, and
	
		\item \label{item:mps_H} 
		if $H^u_x$ is less informative than $H^v_x$ for every $x \in [0,1]$, then for any distribution $F_0$, $\int H^u_x F_0(\dd x)$ is less informative than $\int H^v_x F_0(\dd x)$.

	\end{enumerate}
	We establish \ref{item:mps_G} and omit the analogous argument for \ref{item:mps_H}.
	Suppose that $G^u_x$ is less informative than $G^v_x$ for each $x \in [0,1]$, and fix a distribution $F_0$. For any convex function $\psi : [0,1] \to \R$, $\int \psi \dd G^u_x \le \int \psi \dd G^v_x$ for each $x \in [0,1]$, so
	\begin{multline*}
		\int \psi \dd \left(\int G^u_x F_0(\dd x)\right)
		= \int \int \psi \dd G^u_x F_0(\dd x)
		\\
		\le \int \int \psi \dd G^v_x F_0(\dd x)
		= \int \psi \dd \left(\int G^v_x F_0(\dd x)\right) ,
	\end{multline*}
	where the equalities follow from \Cref{lemma:dist_family}\ref{item:fam_product} in \cref{app:proofs:pf_thm_nondecr:lessvex_implies_cs_intervals}.
	So $\int G^u_x F_0(\dd x)$ is less informative than $\int G^v_x F_0(\dd x)$.

	It remains to prove the first `if' claim in \Cref{lemma:mps_mpc}, namely that `\eqref{eq:mpc_nohigher} for every distribution $F_0$ with $\supp(F_0) \subseteq \{0,1\}$' implies `\eqref{eq:mps_nohigher} for every distribution $F_0$'.
	Suppose that \eqref{eq:mpc_nohigher} holds for every distribution $F_0$ with $\supp(F_0) \subseteq \{0,1\}$, and fix an arbitrary distribution $F_0$; we must show that \eqref{eq:mps_nohigher} holds.
	By \hyperref[claim:mps_clc]{claim B}, \ref{item:mps_clc} holds, so that $H^u_x$ is less informative than $H^v_x$ for each $x \in [0,1]$.
	It follows by \ref{item:mps_H} above that $H^u \coloneqq \int H^u_x F_0(\dd x)$ is less informative than $H^v \coloneqq \int H^v_x F_0(\dd x)$.
	Let $\mathcal{I}$ be the set of all $F_0$-non-null elements of $\mathcal{I}^u$, and consider two cases.

	\smallskip 
	\emph{Case~1: there are $I \in \mathcal{I}$ and $J \in \mathcal{I}^v$ such that $I \subsetneq J$.}
	Fix a continuous and strictly convex $\psi : [0,1] \to \R$.
	By \hyperref[claim:mps]{claim A}, to establish \eqref{eq:mps_nohigher}, it suffices to show that $\int \psi \dd H^u < \int \psi \dd H^v$. By the case-1 hypothesis, $\int \psi \dd H^u_x < \int \psi \dd H^v_x$ holds for each $x \in I$. Since $\int \psi \dd H^u_x \leq \int \psi \dd H^v_x$ for every $x \in [0,1]$, it follows that
	\begin{equation*}
		\int \psi \dd H^u = \int \int \psi \dd H^u_x F_0(\dd x) < \int \int \psi \dd H^v_x F_0(\dd x) = \int \psi \dd H^v ,
	\end{equation*}	
	where the equalities follow from \Cref{lemma:dist_family}\ref{item:fam_product} in \cref{app:proofs:pf_thm_nondecr:lessvex_implies_cs_intervals}, and the inequality holds since $I$ is $F_0$-non-null.

	\smallskip 
	\emph{Case~2: there are no $I \in \mathcal{I}$ and $J \in \mathcal{I}^v$ such that $I \subsetneq J$.}
	By \hyperref[claim:mps]{claim A} and since $H^u$ is less informative than $H^v$, it suffices to show that $G^u \coloneqq \int G^u_x F_0(\dd x)$ is less informative than $G^v \coloneqq \int G^v_x F_0(\dd x)$.
	To that end, fix an arbitrary convex $\psi : [0,1] \to \R$; we will show that $\int \psi \dd G^u \leq \int \psi \dd G^v$. Given any $x \in \bigcup_{I \in \mathcal{I}} I$, $G^u_x$ is less informative than $G^v_x$ by \ref{item:mps_clc} and the case-2 hypothesis, so that $\int \psi \dd G^u_x \le \int \psi \dd G^v_x$.
	And for each $x \in [0,1] \setminus (\bigcup_{I \in \mathcal{I}^u} I)$, $\int \psi \dd G^u_x = \psi(x) \le \int \psi \dd G^v_x$.
	Hence
	\begin{align*}
		\int \psi \dd G^u = \int \int \psi \dd G^u_x F_0(\dd x) \le \int \int \psi \dd G^v_x F_0(\dd x) = \int \psi \dd G^v ,
	\end{align*}
	where the equalities follow from \Cref{lemma:dist_family}\ref{item:fam_product} in \cref{app:proofs:pf_thm_nondecr:lessvex_implies_cs_intervals}, and the inequality holds since $(\bigcup_{I \in \mathcal{I}^u} I) \setminus (\bigcup_{I \in \mathcal{I}} I)$ is $F_0$-null.
\end{proof}

\subsection{Proof of \texorpdfstring{\Cref{proposition:mps_prior} (\cpageref{proposition:mps_prior})}{Proposition \ref{proposition:mps_prior} (p. \pageref{proposition:mps_prior})}}
\label{app:mps:lemma_mps_prior}

For the `if' part, if $F_0$ is not less informative than $G_0$, then for any strictly concave $u : [0,1] \to \R$,
\begin{equation*}
	\argmax_{F \in \mps{F_0}} \int u \dd F = \{F_0\}
	\notlowerthan
	\{G_0\} = \argmax_{F \in \mps{G_0}} \int u \dd F .
\end{equation*}

For the `only if' part, suppose that $F_0$ is less informative than $G_0$, and fix an upper semi-continuous $u : [0,1] \to \R$.
In light of \hyperref[claim:mps]{claim~A} in §\ref{app:mps:lemma_mps_mpc} above, it suffices to show that, for all convex $\psi : [0,1] \to \R$, $\int \psi \dd G^u_\star \ge \int \psi \dd G^u$ and $\int \psi \dd H^u_\star \ge \int \psi \dd H^u$, where $G^u \coloneqq \int G^u_x \dd F_0$, $G^u_\star \coloneqq \int G^u_x \dd G_0$, $H^u \coloneqq \int H^u_x \dd F_0$ and $H^u_\star \coloneqq \int H^u_x \dd G_0$. Fix $\psi$ and define $\psi_G,\psi_H : [0,1] \to \R$ by $\psi_G(x) \coloneqq \int \psi \dd G^u_x$ and $\psi_H(x) \coloneqq \int \psi \dd H^u_x$ for all $x \in [0,1]$. Observe that $\psi_G$ and $\psi_H$ are convex. Then, 
\begin{equation*}
	\int \psi \dd G^u_\star = \int \psi_G \dd G_0 \ge \int \psi_G \dd F_0 = \int \psi \dd G^u
\end{equation*}
where the inequality holds since $F_0$ is less informative than $G_0$, and the equalities follow from \Cref{lemma:dist_family}\ref{item:fam_product} in \cref{app:proofs:pf_thm_nondecr:lessvex_implies_cs_intervals}. Similarly, $\int \psi \dd H^u_\star \ge \int \psi \dd H^u$.
\qed

\subsection{Proof of \texorpdfstring{\Cref{proposition:mps_nodecr}}{Proposition \ref{proposition:mps_nodecr}}}
\label{app:mps:prop_mps_nodecr}

By \Cref{lemma:mps_mpc}, it suffices to show that \ref{item:mps_clc} holds if and only if \eqref{eq:mpc_nohigher} holds for every distribution $F_0$ with $\supp(F_0) \subseteq \{0,1\}$.
The `if' direction follows from \hyperref[claim:mps_clc]{claim B} in §\ref{app:mps:lemma_mps_mpc} above.
For the `only if' direction, suppose that \ref{item:mps_clc} holds, fix a distribution $F_0$ with $\supp(F_0) \subseteq \{0,1\}$, and let $x$ be its mean; we must show that \eqref{eq:mpc_nohigher} holds. This is immediate if $x \notin I$ for every $I \in \mathcal{I}^u$, since then $G^u_x = \1_{[x,1]} = H^u_x$. Assume instead that $x \in I$ for some $I \in \mathcal{I}^u$. By \ref{item:mps_clc}, there is a $J \in \mathcal{I}^v$ such that $I \subseteq J$. If $I \neq J$, then $H^u_x$ is less informative than $H^v_x$ and $H^u_x \neq H^v_x$, so \eqref{eq:mpc_nohigher} holds. Suppose instead that $I=J$. Then by \ref{item:mps_clc} again, we have $\{v = \cav v\} \cap I \subseteq \{u = \cav u\} \cap I$, so $G^u_x$ is less informative than $G^v_x$ and $H^u_x = H^v_x$, and thus \eqref{eq:mpc_nohigher} holds.
\qed

\subsection{Proof of \texorpdfstring{\Cref{proposition:mps_incr}}{Proposition \ref{proposition:mps_incr}}}
\label{app:mps:prop_mps_incr}

To prove that \ref{item:mps_crater} implies \ref{item:mps_thm2_incr}, suppose that \ref{item:mps_crater} holds, and fix an upper semi-continuous $v : [0,1] \to \R$ such that \ref{item:mps_clc} holds. 
By \Cref{lemma:mps_mpc}, it suffices to show that \eqref{eq:mpc_lower} holds for all distributions $F_0$ with $\supp(F_0) \subseteq \{0,1\}$. Fix one such distribution $F_0$, and let $x$ be its mean. By \ref{item:mps_clc}, $H^u_x$ is less informative than $H^v_x$, so it suffices to show that $G^u_x$ is less informative than $G^v_x$. If $u$ is concave, this holds because $G^u_x = \1_{[x,1]}$ (that is, $G^u_x$ is a point mass). If $u$ is not concave, then $\mathcal{I}^u = \{(0,1)\}$ by \ref{item:mps_crater}, so $\{v = \cav v\} \subseteq \{u = \cav u\}$ by \ref{item:mps_clc}, whence $G^u_x$ is less informative than $G^v_x$.

For the converse, assume that \ref{item:mps_crater} fails; we will show that \ref{item:mps_thm2_incr} also fails. Let $v : [0,1] \to \R$ be affine. Since \ref{item:mps_crater} fails, we have $\mathcal{I}^u \ne \{(0,1)\} = \mathcal{I}^v$, so \ref{item:mps_clc} holds. Since $u$ is not concave (as \ref{item:mps_crater} fails), there exists an $x \in (0,1)$ such that $u(x) < \cav u(x)$, so that $G^u_x \neq \1_{[x,1]}$. We have $G^v_x = \1_{[x,1]}$ since $v$ is affine, so $G^u_x$ is not less informative than $G^v_x$. Then, letting $F^x$ be the distribution with support $\{0,1\}$ and mean $x$, we see that \eqref{eq:mpc_lower} fails for $F_0 = F^x$. Since $\mpc{F^x} = \mps{\1_{[x,1]}}$, it follows that \eqref{eq:mps_lower} fails for the distribution $F_0 = \1_{[x,1]}$. So \ref{item:mps_thm2_incr} fails.
\qed

\section{More on comparative-statics theory}
\label{app:mcs_lit}

We discussed in §\ref{sec:intro:lit_mcs} how our results relate to comparative-statics theory. In this appendix, we supplement that discussion by reviewing relevant definitions and results from \textcite{QuahStrulovici2007extensions} (§\ref{app:mcs_lit:qs07}), proving that I-quasi-supermodularity is very demanding in the persuasion model (§\ref{app:mcs_lit:iqsm}), and discussing our use of the weak rather than the strong set order (§\ref{app:mcs_lit:wso}).

\subsection{Review of \texorpdfstring{\textcite{QuahStrulovici2007extensions}}{Quah and Strulovici (2007)}}
\label{app:mcs_lit:qs07}

We rely in several places (notably in the proof in \cref{app:proofs:pf_thm_nondecr} of \Cref{theorem:nondecr}) on definitions and results from an unpublished note by \textcite{QuahStrulovici2007extensions}. We review these here for completeness.

Fix a non-empty set $\mathcal{A}$ equipped with a partial order $\lesssim$. Given $a,c \in \mathcal{A}$, write $[a,c] \coloneqq \left\{ b \in \mathcal{A} : a \lesssim b \lesssim c \right\}$. ($[a,c]$ is called an \emph{order interval.}) For $A,B \subseteq \mathcal{A}$, say that $A$ is \emph{lower} than $B$ iff for any $a \in A$ and $b \in B$ there exist $a' \in A$ and $b' \in B$ such that $a \lesssim b'$ and $a' \lesssim b$, that $A$ is \emph{strictly lower} than $B$ iff $A$ is lower than $B$ and $B$ is not lower than $A$, and that $B$ is \emph{(strictly) higher} than $A$ iff $A$ is (strictly) lower than $B$.

We shall consider functions $U : \mathcal{A} \to [-\infty,\infty)$ which are \emph{proper,} meaning that $U(a)>-\infty$ for at least one $a \in \mathcal{A}$. (In the persuasion model, $\mathcal{A} = \mpc{F_0}$ for some distribution $F_0$, $\mathord{\lesssim} = \text{`less informative than'}$, and $U(F) = \int u \dd F$ for each $F \in \mathcal{A}$, where $u : [0,1] \to \R$ is upper semi-continuous. This $U$ is necessarily proper, and is $\R$-valued in case $u$ is bounded.)

\begin{definition}
	\label{definition:ido_general}
	Let $U,V : \mathcal{A} \to [-\infty,\infty)$ be proper. Given $a,c \in \mathcal{A}$ such that $a \lesssim c$, we say that \emph{$U$ is dominated by $V$ on $[a,c]$} iff
	\begin{equation*}
		U(c) > -\infty , \quad
		V(a) > -\infty , \quad
		\text{and} \quad
		c \in \argmax_{b \in [a,c]} U(b)
	\end{equation*}
	implies that $V(c) \geq V(a)$, with the inequality strict if $U(c) > U(a)$.
	We say that \emph{$U$ is interval-dominated by $V$} iff for all $a \lesssim c$ in $\mathcal{A}$, $U$ is dominated by $V$ on $[a,c]$.
\end{definition}

Quah and Strulovici's original definition of interval dominance is for $\R$-valued functions; our above definition extends it to proper $[-\infty,\infty)$-valued functions, in such a way that the following result remains true.

\begin{namedthm}[Quah and Strulovici's (\citeyear{QuahStrulovici2007extensions}) Proposition~5.]
	\label{prop:qs07prop5}
	Let $U,V : \mathcal{A} \to [-\infty,\infty)$ be proper, and suppose that $U$ is interval-dominated by $V$. If $a \in \argmax_{\mathcal{A}} U$, $c \in \argmax_{\mathcal{A}} V$ and $c \lesssim a$, then $c \in \argmax_{\mathcal{A}} U$ and $a \in \argmax_{\mathcal{A}} V$. (Hence $\argmax_{\mathcal{A}} U$ is not strictly higher than $\argmax_{\mathcal{A}} V$.)
\end{namedthm}

\begin{proof}
	Since $a \in \argmax_{\mathcal{A}} U$ and $c \in \argmax_{\mathcal{A}} V$, and since $U$ and $V$ are proper, we have $U(a) > -\infty < V(c)$ and $a \in \argmax_{b \in [c,a]} U(b)$. Since $U$ is dominated by $V$ on $[c,a]$, it follows that $V(a) \geq V(c)$, with the inequality strict if $U(a) > U(c)$. Since $U(a) \geq U(c)$ and $V(a) \leq V(c)$, it follows that $U(a) = U(c)$ and $V(a) = V(c)$, so $c \in \argmax_{\mathcal{A}} U$ and $a \in \argmax_{\mathcal{A}} V$.
\end{proof}

Recall that $(\mathcal{A},\mathord{\lesssim})$ is a \emph{lattice} iff every two-element set $\{a,b\} \subseteq \mathcal{A}$ has a greatest lower bound, denoted $a \wedge b$, and a least upper bound, denoted $a \vee b$.
(In the persuasion model, $(\mathcal{A},\mathord{\lesssim})$ is a lattice, as shown in \cref{app:proofs:product}.)

\begin{definition}
	\label{definition:iqsm}
	Assume that $(\mathcal{A},\mathord{\lesssim})$ is a lattice, and fix $U : \mathcal{A} \to \R$. We say that $U$ is \emph{I-quasi-supermodular} iff for any $a,b \in \mathcal{A}$, $a \in \argmax_{c \in [a \wedge b,a]} U(c)$ implies that $U(a \vee b) \geq U(b)$, with the inequality strict if $U(a) > U(a \wedge b)$.
\end{definition}

\begin{namedthm}[Quah and Strulovici's (\citeyear{QuahStrulovici2007extensions}) Theorem~1.]
	\label{prop:thm1}
	Assume that $(\mathcal{A},\mathord{\lesssim})$ is a lattice, fix $U,V : \mathcal{A} \to \R$, and suppose that $U$ is interval-dominated by $V$ and that either $U$ or $V$ is I-quasi-supermodular. If $a \in \argmax_{\mathcal{A}} U$ and $b \in \argmax_{\mathcal{A}} V$, then $a \wedge b \in \argmax_{\mathcal{A}} U$ and $a \vee b \in \argmax_{\mathcal{A}} V$. (Hence $\argmax_{\mathcal{A}} U$ is lower than $\argmax_{\mathcal{A}} V$.)
\end{namedthm}

\begin{proof}
	Suppose that $V$ is I-quasi-supermodular.

	Suppose toward a contradiction that $a \vee b \notin \argmax_{\mathcal{A}} V$. Then $V(a \vee b) < V(b)$. Since $V$ is I-quasi-supermodular, it follows that there is a $c \in [a \wedge b,a]$ such that $V(a)<V(c)$. Since $U$ is dominated by $V$ on $[c,a]$, it follows that $a \notin \argmax_{d \in [c,a]} U(d)$, which is impossible because $a \in \argmax_{\mathcal{A}} U$.

	Suppose toward a contradiction that $a \wedge b \notin \argmax_{\mathcal{A}} U$. Then $U(a) > U(a \wedge b)$. For each $c \in [a \wedge b,a]$, we have $a \in \argmax_{d \in [c,a]} U(d)$ since $a \in \argmax_{\mathcal{A}} U$, which since $U$ is dominated by $V$ on $[c,a]$ implies that $V(a) \geq V(c)$; so $a \in \argmax_{c \in [a \wedge b,a]} V(c)$. For $c = a \wedge b$, we obtain the additional conclusion that $V(a) > V(a \wedge b)$. Since $V$ is I-quasi-supermodular, it follows that $V(a \vee b) > V(b)$, which is impossible since $b \in \argmax_{\mathcal{A}} V$.

	Suppose instead that $U$ is I-quasi-supermodular. Then
	\begin{equation}
		a \vee b \in \argmax_{c \in [b,a \vee b]} U(c) ,
		\label{eq:qs07U}
	\end{equation}
	since if there were a $c \in [b,a \vee b]$ such that $U(a \vee c) < U(c)$ then by I-quasi-supermodularity there would be a $d \in [a \wedge c,a]$ such that $U(a) < U(d)$, which is impossible since $a \in \argmax_{\mathcal{A}} U$.

	Suppose toward a contradiction that $a \vee b \notin \argmax_{\mathcal{A}} V$. Then $V(a \vee b) < V(b)$. Since $U$ is dominated by $V$ on $[b,a \vee b]$, it follows that there is a $c \in [b,a \vee b]$ such that $U(a \vee b) < U(c)$, a contradiction with \eqref{eq:qs07U}.

	Suppose toward a contradiction that $a \wedge b \notin \argmax_{\mathcal{A}} U$. Then $U(a) > U(a \wedge b)$. We have $a \in \argmax_{c \in [a \wedge b,a]} U(c)$ since $a \in \argmax_{\mathcal{A}} U$. Because $U$ is I-quasi-supermodular, it follows that $U(a \vee b) > U(b)$. This together with \eqref{eq:qs07U} and the fact that $U$ is dominated by $V$ on $[b,a \vee b]$ implies that $V(a \vee b) > V(b)$, which is impossible since $b \in \argmax_{\mathcal{A}} V$.
\end{proof}

\subsection{I-quasi-supermodularity in the persuasion model}
\label{app:mcs_lit:iqsm}

The weakest supermodularity-type domain restriction in the comparative-statics literature is `I-quasi-supermodularity', defined in §\ref{app:mcs_lit:qs07} above. In the persuasion model, even I-quasi-supermodularity is highly demanding:

\begin{proposition}
	\label{proposition:iqsm}
	Let $v : [0,1] \to \R$ regular. The sender's objective $F \mapsto \int v \dd F$ is I-quasi-supermodular on $\mpc{F_0}$ for any distribution $F_0$ if and only if $v$ is either concave or convex.
\end{proposition}

Regularity is not dispensable: we show in \cref{app:robust:kink} that there exist (knife-edge) examples of non-differentiable functions $v$ which are neither concave nor convex, yet have $F \mapsto \int v \dd F$ I-quasi-supermodular.

\begin{proof}
	Note that if $w : [0,1] \to \R$ is regular, then it is bounded, so $F \mapsto \int w \dd F$ is $\R$-valued. Now, let $v : [0,1] \to \R$ be regular, and suppose that $F \mapsto \int v \dd F$ is I-quasi-supermodular on $\mpc{F_0}$ for any distribution $F_0$; we will show that $v$ must be either concave or convex. Recall from the proof of \Cref{theorem:nondecr} (\cref{app:proofs:pf_thm_nondecr}) that if a regular $u : [0,1] \to \R$ is coarsely less convex than $v$, then $F \mapsto \int u \dd F$ is interval-dominated by $F \mapsto \int v \dd F$. Hence by Quah and Strulovici's (\citeyear{QuahStrulovici2007extensions}) Theorem~1 (see §\ref{app:mcs_lit:qs07} above), \eqref{eq:mcs_incr} on \cpageref{eq:mcs_incr} holds for every regular $u : [0,1] \to \R$ that is coarsely less convex than $v$ and every distribution $F_0$. Then by \Cref{proposition:W_down}, $v$ must be either concave or convex.
\end{proof}

\subsection{The weak and strong set orders}
\label{app:mcs_lit:wso}

In the comparative-statics literature, the desired `increase' of $\argmax_{\mathcal{A}} U$ is often formalised using the \emph{strong set order.} The strong set order is defined only in case $(\mathcal{A},\mathord{\lesssim})$ is a lattice. Under this assumption, given $A,B \subseteq \mathcal{A}$, we say that \emph{$A$ is lower than $B$ in the strong set order} if and only if for any $a \in A$ and $b \in B$, their greatest lower bound $a \wedge b$ belongs to $A$, and their least upper bound $a \vee b$ belongs to $B$. Evidently a set $A \subseteq \mathcal{A}$ is lower than itself in the strong set order if and only if it is a sublattice. For this reason, the strong set order is usually used only to compare sublattices.

In the persuasion model, the action set $(\mathcal{A},\mathord{\lesssim})$ is in fact a lattice (see \cref{app:proofs:product}), but the set $\argmax_{\mathcal{A}} U$ need not be a sublattice. (This is true even in special cases, e.g. if only binary priors $F_0$ are considered.)

Our results are phrased in terms of the \emph{weak} set order: given $A,B \subseteq \mathcal{A}$, we say that \emph{$A$ is lower than $B$ (in the weak set order)} iff for any $a \in A$ and $b \in B$, there is an $a' \in A$ such that $a' \leq b$ and there is a $b' \in B$ such that $a \leq b'$. Evidently strong set ordering implies weak set ordering, but not vice-versa. The two are equivalent for singletons $A=\{a\}$ and $B=\{b\}$.

We chose the weak set order for two reasons. Firstly, we consider the weak set order more natural. (Our reading of the literature is that the strong set order is widespread not because its extra strength is interpretable, but rather because it yields clean necessity results.) Secondly, as noted above, the strong set order is not very natural for (and therefore usually not even defined for) comparing non-sublattice sets, such as argmaxes in the persuasion model.

Some of our results do remain true if `lower than in the weak set order' is replaced by `lower than in the strong set order'. Others do not, for example \Cref{proposition:ido_binary} (choose $u=v$, where $u$ and $F_0$ are such that $\argmax_{F \in \mpc{F_0}} \int u \dd F$ is not a sublattice).

\section{More on \texorpdfstring{\Cref{lemma:CLC_suff} (\cpageref{lemma:CLC_suff})}{Lemma \ref{lemma:CLC_suff} (p.~\pageref{lemma:CLC_suff})}}
\label{app:more_clc}

In this appendix, we say more about payoff transformations which produce coarse-convexity shifts, thereby extending \Cref{lemma:CLC_suff}. In §\ref{app:more_clc:tight}, we establish a partial converse to \Cref{lemma:CLC_suff}, as well as an exact (but somewhat more complicated) characterisation of coarse-convexity-increasing transformations.

We then revisit the sender--receiver interpretation of the persuasion model from §\ref{sec:model:interpretation}, in which a \emph{receiver} chooses an action at the interim stage, informed by the realisation of the signal chosen by the sender. In §\ref{app:more_clc:align}, we ask when increased alignment of interests between the sender and receiver leads the sender's reduced-form interim payoff $u$ to become coarsely more convex.

Finally, in §\ref{app:more_clc:gitmez_molavi}, we prove that if $u$ is regular and S-shaped, then if its slope $u'$ becomes less convex, $u$ itself becomes coarsely more convex. This was stated without proof in §\ref{sec:mcs_incr:s}, and extends an insight from \textcite{GitmezMolavi2023}.

\subsection{Tightness of \texorpdfstring{\Cref{lemma:CLC_suff}}{Lemma \ref{lemma:CLC_suff}}}
\label{app:more_clc:tight}

\Cref{lemma:CLC_suff} is nearly tight, in the following sense:

\begin{namedthm}[Partial converse of \Cref*{lemma:CLC_suff}.]
	\label{lemma:CLC_suff_converse}
	If $\Phi : \R \times [0,1] \to \R$ is such that for every upper semi-continuous $u : [0,1] \to \R$, $u$ is coarsely less convex than $x \mapsto \Phi(u(x),x)$,
	then $\Phi$ must be convex on $\R \times (0,1)$
	with $\Phi(\cdot,x)$ increasing for every $x \in (0,1)$.
\end{namedthm}

This partial converse is implied by the following result, which closes the small gap between \Cref{lemma:CLC_suff} and its converse by giving an exact characterisation of coarse-convexity-increasing transformations $\Phi$. This result has other useful consequences, such as the fact that $u$ is coarsely less convex than $\max\{u,\psi\}$ whenever $\psi : [0,1] \to \R$ is strictly convex.

\begin{namedthm}[\Cref*{lemma:CLC_suff}$\boldsymbol{^*}$.]
	\label{proposition:CLC_charac}
	For a map $\Phi : \R \times [0,1] \to \R$, the following are equivalent:
	
	\begin{enumerate}[label=(\roman*)]
		
		\item \label{bullet:proposition:CLC_charac:shift}
		For every $u : [0,1] \to \R$, $u$ is coarsely less convex than $x \mapsto \Phi(u(x),x)$.
		
		\item \label{bullet:proposition:CLC_charac:shift_usc}
		For every upper semi-continuous $u : [0,1] \to \R$, $u$ is coarsely less convex than $x \mapsto \Phi(u(x),x)$.
	
		\item \label{bullet:proposition:CLC_charac:charac}
		For any $x < y$ in $[0,1]$, $\alpha \in (0,1)$ and $a,b,c \in \R$ such that $c \leq \mathrel{(<)} \alpha a + (1-\alpha)b$, we have $\Phi(c,\alpha x + (1-\alpha)y) \leq \mathrel{(<)} \alpha\Phi(a,x) + (1-\alpha)\Phi(b,y)$.
	
	\end{enumerate}
	
\end{namedthm}

For the proof, we write $a_\alpha b \coloneqq \alpha a + (1-\alpha) b$ for $a,b \in \R$ and $\alpha \in [0,1]$.

\begin{proof}[Proof of the {\hyperref[lemma:CLC_suff_converse]{partial converse of \Cref*{lemma:CLC_suff}}}]
	By \hyperref[proposition:CLC_charac]{\Cref*{lemma:CLC_suff}$^*$}, it suffices to show that property~\ref{bullet:proposition:CLC_charac:charac} implies that $\Phi$ is convex on $\R \times (0,1)$ and that $\Phi(\cdot,x)$ is increasing for each $x \in (0,1)$.
	So let $\Phi$ satisfy \ref{bullet:proposition:CLC_charac:charac}, and note that it follows that for each $c \in \R$, $\Phi(c,\cdot)$ is convex, hence continuous on $(0,1)$.

	For convexity, property~\ref{bullet:proposition:CLC_charac:charac} immediately implies that $\Phi(\alpha(a,x)+(1-\alpha)(b,y)) \leq \Phi(a,x)_\alpha \Phi(b,y)$ for any $\alpha \in (0,1)$ and any $(a,x),(b,y) \in \R \times [0,1]$ such that $x \neq y$.
	To show that the same holds when $x=y=z \in (0,1)$,
	(in other words, that $\Phi(\cdot,z)$ is convex for each $z \in (0,1)$)
	observe that for any $x \in (0,z)$ and $y \in (z,1)$ such that $x_\alpha y = z$, we have $\Phi(a_\alpha b, z ) \leq \Phi(a,x)_\alpha \Phi(b,y)$, so letting $x,y \to z$ yields
	$\Phi(a_\alpha b, z) \leq \Phi(a,z)_\alpha \Phi(b,z)$
	by continuity.

	For monotonicity, take any $z \in (0,1)$ and $c<a$ in $\R$; we must show that $\Phi(c,z) \leq \Phi(a,z)$. 
	For any $x \in (0,z)$ and $y \in (z,1)$ such that $\frac{1}{2} x + \frac{1}{2} y = z$, property~\ref{bullet:proposition:CLC_charac:charac} implies $\Phi(c,z) < \frac{1}{2}\Phi(a,x) + \frac{1}{2}\Phi(a,y)$, which as $x,y \to z$ yields $\Phi(c,z) \leq \Phi(a,z)$ by continuity.
\end{proof}

\begin{proof}[Proof of {\hyperref[proposition:CLC_charac]{\Cref*{lemma:CLC_suff}$^*$}}]
	\ref{bullet:proposition:CLC_charac:charac} implies \ref{bullet:proposition:CLC_charac:shift} since for any $u : [0,1] \to \R$ and any $x<y$ in $[0,1]$ such that $u(x_\beta y) \leq u(x)_\beta u(y)$ for every $\beta \in (0,1)$, property~\ref{bullet:proposition:CLC_charac:charac} (with $a \coloneqq u(x)$, $b \coloneqq u(y)$ and $c \coloneqq u(x_\alpha y)$) implies for each $\alpha \in (0,1)$ that $\Phi(u(x_\alpha y),x_\alpha y) \leq \Phi(u(x),x)_\alpha \Phi(u(y),y)$, with strict inequality if $u(x_\alpha y) < u(x)_\alpha u(y)$. \ref{bullet:proposition:CLC_charac:shift} immediately implies \ref{bullet:proposition:CLC_charac:shift_usc}.
	To show that \ref{bullet:proposition:CLC_charac:shift_usc} implies \ref{bullet:proposition:CLC_charac:charac}, we prove the contra-positive: let $\Phi$ violate \ref{bullet:proposition:CLC_charac:charac}, meaning that there are $x < y$ in $[0,1]$, $\alpha \in (0,1)$ and $a,b,c \in \R$ such that either
	
	\begin{enumerate}
	
		\item \label{bullet:proposition:CLC_charac:weak}
		$c \leq a_\alpha b$ and $\Phi(c,x_\alpha y) > \Phi(a,x)_\alpha \Phi(b,y)$, or

		\item \label{bullet:proposition:CLC_charac:strict}
		$c < a_\alpha b$ and $\Phi(c,x_\alpha y) \geq \Phi(a,x)_\alpha \Phi(b,y)$.
	
	\end{enumerate}
	
	\noindent
	To show that \ref{bullet:proposition:CLC_charac:shift_usc} fails, define $u : [0,1] \to \R$ by
	$u \coloneqq a$ on $[0,x]$,
	$u(x_\alpha y) \coloneqq c$,
	$u \coloneqq b$ on $[y,1]$
	and $u \coloneqq \min\{a,b,c\}-1$ on $(x,x_\alpha y) \union (x_\alpha y,y)$.
	Clearly $u$ is upper semi-continuous. We have $u(x_\beta y) \leq u(x)_\beta u(y)$ for every $\beta \in (0,1)$, with strict inequality at $\beta=\alpha$ in case~\ref{bullet:proposition:CLC_charac:strict}, and furthermore $\Phi(u(x_\alpha y),x_\alpha y) \geq \Phi(u(x),x)_\alpha \Phi(u(y),y)$, with strict inequality in case~\ref{bullet:proposition:CLC_charac:weak}. Thus $u$ is not coarsely less convex than $x \mapsto \Phi(u(x),x)$.
\end{proof}

\subsection{Alignment and coarse convexity}
\label{app:more_clc:align}

Recall the sender--receiver interpretation from §\ref{sec:model:interpretation}.
There is a non-empty set $\mathcal{A}$ of actions, and the sender's and receiver's interim payoffs $U_S(a,m)$ and $U_R(a,m)$ depend on the chosen action $a \in \mathcal{A}$ and on the mean $m \in [0,1]$ of their (posterior) belief about the state.
When the posterior mean is $m \in [0,1]$, the receiver chooses action $A(m) \in \mathcal{A}$, so the sender's reduced-form interim payoff is $u(m) \coloneqq U_S( A(m), m )$.
We assume that $A : [0,1] \to \mathcal{A}$ is \emph{$U_R$-optimal,} i.e. a selection from the correspondence $m \mapsto \argmax_{a \in \mathcal{A}} U_R(a,m)$.

We consider shifts of the sender's interim payoff from $(a,m) \mapsto U_S(a,m)$ to $(a,m) \mapsto \Psi\bigl( U_S(a,m), U_R(a,m), m \bigr)$, where $\Psi : \R^2 \times [0,1] \to \R$ is strictly increasing in its first argument---that is, $\Psi$ is a \emph{utility transformation.}
We are interested in \emph{alignment-increasing} utility transformations $\Psi$, meaning those that are increasing in their second argument (the receiver's payoff).

\begin{proposition}
	\label{proposition:alignment}
	Let $\Psi : \R^2 \times [0,1] \rightarrow \R$ be convex with $\Psi(\cdot,\ell,x)$ strictly increasing and $\Psi(k,\cdot,x)$ increasing for all $k,\ell \in \R$ and $x \in [0,1]$.
	Then for any action set $\mathcal{A}$, any sender's and receiver's payoffs $U_S,U_R : \mathcal{A} \times [0,1] \rightarrow \R$, and any $U_R$-optimal $A : [0,1] \to \mathcal{A}$,
	the map $x \mapsto U_S(A(x),x)$
	is coarsely less convex than
	the map $x \mapsto \Psi(U_S(A(x),x),U_R(A(x),x),x)$.
\end{proposition}

In words, applying a \emph{convex} alignment-increasing utility transformation $\Psi$ to the sender's payoff $U_S$ always makes her reduced-form interim payoff $u$ coarsely more convex.
Convexity is satisfied by many natural alignment-increasing utility transformations, such as $(k,\ell,x) \mapsto k + \rho \ell$ for $\rho \geq 0$.

\begin{proof}
	For $x,y \in \R$ and $\alpha \in [0,1]$, write $x_\alpha y \coloneqq \alpha x + (1-\alpha) y$.
	Define $u,v : [0,1] \to \R$
	by $u(x) \coloneqq U_S(A(x),x)$ and $v(x) \coloneqq U_R(A(x),x)$ for each $x \in [0,1]$.
	Choose any $x<y$ in $[0,1]$ such that $u(x_\beta y) \leq u(x)_\beta u(y)$ for every $\beta \in (0,1)$, and fix an $\alpha \in (0,1)$.
	Note that $v(x_\alpha y) \leq v(x)_\alpha v(y)$ since $v$ is convex (as $A$ is $U_R$-optimal).
	Thus
	\begin{align*}
		\Psi\bigl(u(x_\alpha y),v(x_\alpha y),x_\alpha y\bigr)
		&\leq \Psi\bigl(u(x_\alpha y),v(x)_\alpha v(y),x_\alpha y\bigr)
		\\
		&\leq \Psi\bigl(u(x)_\alpha u(y), v(x)_\alpha v(y),x_\alpha y\bigr)
		\\
		&\leq \Psi\bigl(u(x),v(x),x)_\alpha \Psi(u(y),v(y),y\bigr) ,
	\end{align*}
	where the first inequality holds since $\Psi\bigl(u(x_\alpha y),\cdot,x_\alpha y\bigr)$ is increasing,
	the second holds since $\Psi\bigl(\cdot,v(x)_\alpha v(y),x_\alpha y\bigr)$ is (strictly) increasing,
	and the final inequality holds since $\Psi$ is convex.
	Moreover, the second inequality is strict if $u(x_\alpha y) < u(x)_\alpha u(y)$, as $\Psi\bigl(\cdot,v(x)_\alpha v(y),x_\alpha y\bigr)$ is strictly increasing.
\end{proof}

The convexity-of-$\Psi$ hypothesis is essential, indeed nearly necessary: \Cref{proposition:alignment} has a partial converse similar to that of \Cref{lemma:CLC_suff} (see §\ref{app:more_clc:tight} above).
It is therefore not \emph{generally} true that increased alignment of interests leads to coarse-convexity shifts.

\begin{example}
	\label{example:align_concave}
	Consider the alignment-increasing utility transformation $\Psi$ defined by $\Psi(k,\ell,x) \coloneqq k + \phi(\ell)$ for all $k,\ell \in \R$ and $x \in [0,1]$, where $\phi : \R \to \R$ is strictly increasing.
	It is natural for $\phi$ to be concave, as this captures inequality-aversion in the sender's evaluation of (distributions of interim) receiver welfare.
	But when $\phi$ is concave and not convex,
	$x \mapsto U_S(A(x),x)$
	fails to be coarsely less convex than
	$x \mapsto \Psi(U_S(A(x),x),U_R(A(x),x),x)$
	for some $U_S,U_R : \mathcal{A} \times [0,1] \to \R$ and some $U_R$-optimal $A : [0,1] \to \mathcal{A}$.
\end{example}

\subsection{Gitmez and Molavi's payoff shift}
\label{app:more_clc:gitmez_molavi}

We shall prove the following claim, which was made in §\ref{sec:mcs_incr:s}.

\begin{claim}
	\label{claim:gitmez_molavi}
	Suppose that $u$ and $v$ are regular, that $u$ is S-shaped, and that $u'$ is more convex than $v'$ in the conventional sense: $u' = \phi \circ v'$ for some convex and strictly increasing $\phi : \co(v'([0,1])) \to \R$. Then $u$ is coarsely less convex than $v$.
\end{claim}

The proof is a slight extension of Gitmez and Molavi's (\citeyear{GitmezMolavi2023}, appendix~A) core argument.

\begin{proof}
	Clearly $v$ is also S-shaped, with the same inflection point $\bar{x}$.
	For any $x \in [0,1)$, define $R_u^x : [x,1] \to \R$ by $R_u^x(y) \coloneqq [u(y)-u(x)]/(y-x)$ for each $y \in (x,1]$ and $R_u^x(x) \coloneqq \lim_{y \downarrow x} R_u^x(y)$.
	For any $y \in (x,1]$, since $u$ is S-shaped,
	$R_u^x$ is increasing on $[x,y]$ iff $R_u^x$ is strictly increasing on $[x,y]$ iff $u(\alpha x + (1-\alpha) y) \leq \alpha u(x) + (1-\alpha) u(y)$ for every $\alpha \in [0,1]$ iff $u(\alpha x + (1-\alpha) y) < \alpha u(x) + (1-\alpha) u(y)$ for every $\alpha \in (0,1)$.
	The same applies to $R_v^x$, analogously defined.
	What must be shown is therefore that for any $x<y$ in $[0,1]$, if $R_u^x$ is increasing on $[x,y]$, then so is $R_v^x$.
	So fix any $x<y$ in $[0,1]$.
	Since $R_u^x$ and $R_v^x$ are strictly quasi-concave, it suffices to show that their respective maximisers $z$ and $w$ satisfy $z \leq w$.
	This is immediate if $w=1$, so assume that $w<1$.
	If $x \geq \bar{x}$, then $u$ and $v$ are strictly concave on $[x,1]$, so $z=x=w$; so assume that $x < \bar{x}$, and note that this implies $w>x$ since $v$ is strictly convex on $[x,\bar{x}]$.
	The first-order conditions are $R_u^x(z) \leq u'(z)$, with equality if $z < 1$, and $R_v^x(w) = v'(w)$.
	Thus since $z,w \in [\bar{x},1]$, $z \leq w$ holds iff $R_u^x(w) \geq u'(w)$. And indeed
	$R_u^x(w)
	= (w-x)^{-1} \int_x^w \phi \circ v'
	\geq \phi( (w-x)^{-1} \int_x^w v' )
	= \phi( R_v^x(w) )
	= \phi( v'(w) )
	= u'(w)$
	by Jensen's inequality, since $\phi$ is convex.
\end{proof}

\section{More on the limits of comparative statics}
\label{app:robust}

This appendix supplements §\ref{sec:limits} by examining the robustness, in several senses, of our comparative-statics results. In §\ref{app:robust:tight}, we show that our binary-priors result \Cref{proposition:ido_binary} is tight (so \Cref{theorem:incr} is robust): on any restricted domain of priors that contains at least one non-binary prior, `increasing' comparative statics conclusions cannot generally be drawn without a crater-property-type shape restriction on the interim payoff $u$. In §\ref{app:robust:specific}, we show that the crater property remains necessary even if instead of considering all coarse-convexity shifts of payoffs, we contemplate only more specific shifts (e.g. increased convexity in the conventional sense).

In §\ref{app:robust:blackwell}, we show that if our notion of `less informative than' is replaced by the \textcite{Blackwell1951,Blackwell1953}, then \Cref{theorem:nondecr} remains true as stated, but \Cref{theorem:incr} does not. Hence `coarsely less convex than' remains the correct notion of payoff shift for comparative statics, but the crater property is no longer sufficient for a coarse-convexity shift to yield `increasing' comparative statics.

Finally, we examine the role of regularity, in particular differentiability. We show in §\ref{app:robust:kink} that it is possible to construct knife-edge non-differentiable examples in which comparative-statics conclusions can be drawn without the crater property. In §\ref{app:robust:disappointment}, we apply this finding to recover a comparative-statics result of \textcite{CerreiavioglioDillenbergerOrtoleva2020}.

\subsection{Robustness and tightness}
\label{app:robust:tight}

In this section, we show that \Cref{theorem:incr} is robust (so \Cref{proposition:ido_binary} is tight): binary priors are the \emph{only} priors under which `increasing' comparative-statics conclusions can be drawn generally without a crater-like restriction.

Recall that an \emph{M-shaped} function $u : [0,1] \rightarrow \R$ is one which is continuous and concave--convex--concave. Unlike S and W shapes, M-shaped functions can violate the crater property.

\begin{proposition}
	\label{proposition:tightness}
	For any distribution $F_0$ that is not binary,
	there are regular $u,v : [0,1] \to \R$
	such that $u$ is coarsely less convex than $v$, and yet \eqref{eq:mcs_incr} fails.
	These $u$ and $v$ may be chosen to be M- and S-shaped, respectively.
\end{proposition}

In other words, binary distributions $F_0$ are the \emph{only} ones for which \eqref{eq:mcs_incr} holds between any $u$ and any coarsely more convex $v$, even if attention is restricted to very well-behaved $u,v : [0,1] \to \R$ (in particular regular and, respectively, M- and S-shaped).
`Increasing' comparative statics can thus be guaranteed \emph{only} by either restricting the domain of interim payoffs $u$ (as in \Cref{theorem:incr}) or by focussing on binary priors $F_0$ (as in \Cref{proposition:ido_binary}).

\begin{proof}
	The argument is close to the proof in §\ref{sec:mcs_incr:pf_thm2} of the converse (necessity) half of \Cref{theorem:incr}.
	Fix a distribution $F_0$ that is not binary.
	Choose an $X \in (0,1)$ such that $0 < \lim_{z \uparrow X}F_0(z) \le F_0(X) < 1$ and $F_0 < F_0(X)$ on $[0,X)$.
	Define 
	\begin{equation*}
		x \coloneqq \frac{1}{F_0(X)} \int_{[0,X]} \xi F_0(\dd \xi)
		\quad \text{and} \quad
		w \coloneqq \frac{1}{1-F_0(X)} \int_{(X,1]} \xi F_0(\dd \xi),
	\end{equation*}
	and note that $0 < x < X < w\le 1$.
	Fix a convex $p : [0,1] \rightarrow \R$ that is affine on $[0,X]$ and on $[X,1]$, but not affine on $[0,1]$.
	Clearly we may choose a regular and M-shaped $u : [0,1] \to \R$ such that $p=u$ on $\{x,w\}$ and $p>u$ on $[0,1] \setminus \{x,w\}$, and such that $u$ is convex on $[X,y]$ and concave on $[y,1]$ for some $y \in (X,1)$, with $\lim_{z \uparrow 1}u'(z)=(p(1)-p(X))/(1-X)$ if $w = 1$.
	Let $G$ be the distribution supported on $\{x,w\}$ whose mean is the same as that of $F_0$. Then $G$ is uniquely optimal for $u$ given $F_0$, since any other $F \in \mpc{F_0}$ has $\int u \dd F < \int p \dd F \leq \int p \dd F_0 = \int p \dd G = \int u \dd G$, where the weak inequality holds since $p$ is convex and $F$ belongs to $\mpc{F_0}$, the first equality holds since $p$ is affine on $[0,X]$ and on $[X,1]$, and the final equality holds since $p=u$ $G$-a.e.

	Since $u'$ is bounded, we may choose a regular $v : [0,1] \to \R$
	that coincides with $u$ on $[X,1]$
	and that weakly exceeds $u$ and is strictly convex on $[0,X]$. Then $v$ is S-shaped and coarsely more convex than $u$.
	Observe that there are $a \in [0,X]$ and $\pi \in [0,1]$ such that
	\begin{equation*}
		\frac{ v(b) - v(a) }{ b-a } = v'(b) ,
		\quad \text{where} \quad
		b \coloneqq \frac{ \pi\delta a + \int_{(a,1]} \xi F_0( \dd \xi ) }{\pi\delta +1-F_0(a)} \in (0,w] 
	\end{equation*}
	and $\delta \coloneqq F_0(a) - \lim_{z \uparrow a} F_0(z)$.
	Define $F$ by $F \coloneqq F_0$ on $[0,a)$, $F \coloneqq F_0(a)-\pi\delta$ on $[a,b)$, and $F \coloneqq 1$ on $[b,1]$.
	(That is, $F$ reveals $[0,a)$, pools $(a,1]$, reveals $a$ with probability $1-\pi$, and otherwise pools it with $(a,1]$.)
	Let $q : [0,1] \to \R$ be affine on $[a,1]$ and satisfy $q \geq v$, with equality on $[0,a] \union \{b\}$.
	The distribution $F$ is optimal for $v$ given $F_0$ since for any (other) $H \in \mpc{F_0}$, we have $\int v \dd H \leq \int q \dd H \leq \int q \dd F_0 = \int q \dd F = \int v \dd F$, where the second inequality holds since $q$ is convex and $H$ belongs to $\mpc{F_0}$, the first equality holds since $q$ is affine on $[a,1]$, and the final equality holds since $q=v$ $F$-a.e.

	Since $p(X) > u(X)$, it must be either that $a<X$ or that $a = X$ and $\pi\delta > 0$.
	Then $b < w$, since $F_0 < F_0(X)$ on $[0,X)$. Thus $F$ is not more informative than $G$, so \eqref{eq:mcs_incr} fails.
\end{proof}

\subsection{Specific shifts}
\label{app:robust:specific}

In this appendix, we show that the crater property remains necessary for `increasing' comparative statics when attention is confined to shifts of the sender's interim payoff $u$ that are more specific than coarse-convexity shifts: in particular, conventional increased convexity, adding a convex function, and arbitrarily small coarse-convexity shifts.

\begin{proposition}
	\label{proposition:specific}
	Let $u : [0,1] \to \R$ be regular.
	The following are equivalent:

	\begin{enumerate}[label=(\roman*)]
		
		\item \label{item:specific:crater} 
		$u$ satisfies the crater property.
		
		\item \label{item:specific:more_cvex} 
		For any regular $v : [0,1] \to \R$ such that $v = \phi \circ u$ for some convex and strictly increasing $\phi : \co(u([0,1])) \to \R$,
		\eqref{eq:mcs_incr} holds for every atomless convex-support distribution $F_0$.

		\item \label{item:specific:cvex_diff} 
		For any regular $v : [0,1] \to \R$ such that $v = u + \psi$ for some convex $\psi : [0,1] \to \R$,
		\eqref{eq:mcs_incr} holds for every atomless convex-support distribution $F_0$.
		
	\end{enumerate}
	Furthermore, for any $\eps>0$, \ref{item:specific:crater}--\ref{item:specific:cvex_diff} are equivalent to:

	\begin{enumerate}[label=(\roman*$_\eps$)]

		\setcounter{enumi}{3}
	
		\item \label{item:specific:marg} For any regular $v : [0,1] \to \R$ that is coarsely more convex than $u$ and satisfies $\sup_{x \in [0,1]} \abs*{v(x)-u(x)} < \eps$,
		\eqref{eq:mcs_incr} holds for every atomless convex-support distribution $F_0$.
	
	\end{enumerate}
\end{proposition}

\begin{proof}
	\ref{item:specific:crater} implies \ref{item:specific:more_cvex}, \ref{item:specific:cvex_diff} and \ref{item:specific:marg} by \Cref{corollary:CLC_suffsuff} and \Cref{theorem:incr}.

	To show that \ref{item:specific:cvex_diff} implies \ref{item:specific:crater}, we shall prove the contra-positive, by arguing that in the proof of the necessity half of \Cref{theorem:incr} (§\ref{sec:mcs_incr:pf_thm2}), $v$ can be chosen so that $v-u$ is convex. Since $u$ is regular, we may choose a regular $w : [0,1] \to \R$ such that $w = u$ on $[X,1]$ and, on each sub-interval of $[0,X]$ on which $u$ is convex (concave), $w-u$ is affine ($w$ is affine).
	Note that $w-u$ is convex, and that $w$ is convex on $[0,X]$.
	Fix any $\chi : [0,1] \to \R$ that is continuously differentiable with bounded derivative, is strictly convex on $[0,X]$, and vanishes on $[X,1]$.
	Then $v \coloneqq w + \chi$ weakly exceeds $u$, is strictly convex on $[0,X]$, and coincides with $u$ on $[X,1]$; and evidently $v-u$ is convex.

	To show that \ref{item:specific:more_cvex} implies \ref{item:specific:crater}, we shall modify the proof in §\ref{sec:mcs_incr:pf_thm2} of the necessity half of \Cref{theorem:incr}.
	By replacing $x$ and $x'$ ($w$ and $w'$) with larger (smaller) values if necessary, we can ensure that $u(x) \neq u(w)$, without loss $u(x) < u(w)$, that $X \in (y,z)$, and that for some $z' \in (z,w)$, $u$ is strictly increasing and strictly concave on $[z',w']$ with $u'(w) > 0$ and $\max_{[x',z']}u = u(z')$.
	Observe that, for all $\eps > 0$, there exists a map $\phi : \co(u([0,1])) \to \R$ that is continuously differentiable, strictly increasing and convex, equal to the identity on $(-\infty,u(w)-\eps]$ and affine on $[u(w),\infty)$, and such that $v \coloneqq \phi \circ u$ is strictly concave on $[z',w']$ and regular, and the tangents to $v$ and $u$ at $w$ cross within $(z',w)$.
	Clearly, the tangent to $v$ at $w$ approaches the tangent to $u$ at $w$ as $\eps$ vanishes.
	Thus (recalling the properties of $p$ and $F_0$) for sufficiently small $\eps$, there exists a function $q : [0,1] \to \R$, 
	an $x^\star \in (x,y)$, an $X^\star \in (X,z)$ and a $w^\star \in (w,w')$
	such that
	$q$ is affine on $[x',X^\star]$ and on $[X^\star,w']$,
	strictly exceeds $v$ on $[x',w'] \setminus \{x^\star,w^\star\}$,
	is tangent to $v$ at $x^\star$ and at $w^\star$, and satisfies
	\begin{equation*}
		\frac{1}{F_0(X^\star)} \int_0^{X^\star} \xi F_0(\dd \xi) = x^\star
		\quad \text{and} \quad
		\frac{1}{1-F_0(X^\star)} \int_{X^\star}^1 \xi F_0(\dd \xi) = w^\star.
	\end{equation*}
	Then the monotone partitional distribution $G$ with $G$-pooling intervals $[0,X^\star]$ and $[X^\star,1]$ is uniquely optimal for $v$.%
	\footnote{Observe that for any other distribution $H \in \mpc{F_0}$, $\int v \dd H < \int q \dd H \le \int q \dd F_0 = \int v \dd G$, where the first inequality holds since $[x',w'] \supseteq \supp(H) \ne \{x^\star,w^\star\}$ and $q > v$ on $[x',w'] \setminus \{x^\star,w^\star\}$, the second since $q$ is convex and $H \in \mpc{F_0}$, and the equality holds since $q$ is affine on $[x',X^\star]$ and on $[X^\star,w']$ with $q = v$ on $\{x^\star,w^\star\}$.}
	Since $X^\star \neq X$, $G$ pools states on either side of $X$, so \eqref{eq:mcs_incr} fails.

	To show that \ref{item:specific:marg} implies \ref{item:specific:crater}, simply note that $\sup_{x \in [0,1]} \abs*{v(x)-u(x)}$ vanishes with $\eps$ in the argument of the previous paragraph.
\end{proof}

\subsection{Blackwell informativeness}
\label{app:robust:blackwell}

As mentioned in §\ref{sec:model:informativeness}, our notion of `less informative than' is in the spirit of, but distinct from, Blackwell's (\citeyear{Blackwell1951,Blackwell1953}) notion of `less informative than'. In this appendix, we show that \Cref{theorem:nondecr} remains true as stated if `less informative than' is replaced by Blackwell's notion, and we explain why the same is not true of \Cref{theorem:incr}.

Let $\mathcal{F}$ be the set of all distributions (CDFs $[0,1] \to [0,1]$). Equip $\mathcal{F}$ with the topology of weak convergence, and write $\Delta(\mathcal{F})$ for the set of all Borel probability measures on $\mathcal{F}$. Write $\mu_F$ for the mean of a distribution $F \in \mathcal{F}$. Given $F \in \mathcal{F}$, let $\Sigma(F)$ be the set of all $\sigma \in \Delta(\mathcal{F})$ such that $\int G \sigma(\dd G) = F$.

Return to §\ref{sec:model}. Fix the prior $F_0$. The sender chooses a signal, each realisation of which induces a posterior belief $\in \mathcal{F}$ via Bayes's rule. Each signal thus induces a random posterior belief, with some distribution $\in \Delta(\mathcal{F})$. Say that $\sigma \in \Delta(\mathcal{F})$ is \emph{feasible} (given $F_0$) iff it is the posterior-belief distribution induced by some signal. By the splitting lemma \parencite{HardyLittlewoodPolya1929,Blackwell1951}, a probability measure $\sigma \in \Delta(\mathcal{F})$ is feasible if and only if it belongs to $\Sigma(F_0)$. The sender's problem is to choose among the feasible posterior-belief distributions $\sigma$ to maximise $\int u(\mu_F) \sigma(\dd F)$.

When the posterior-belief distribution is $\sigma \in \Delta(\mathcal{F})$, the distribution of the posterior \emph{mean} is $F_\sigma \in \mathcal{F}$ given by $F_\sigma(m) \coloneqq \sigma\left( \left\{ F \in \mathcal{F} : \mu_F \leq m \right\} \right)$ for each $m \in [0,1]$. Recall from §\ref{sec:model} that a distribution $\in \mathcal{F}$ is equal to $F_\sigma$ for \emph{some} $\sigma \in \Sigma(F_0)$ if and only if it belongs to $\mpc{F_0}$, and (thus) that the sender's problem reduces to choosing $F \in \mpc{F_0}$ to maximise $\int u \dd F$.

\begin{definition}
	\label{definition:blackwell}
	For $\sigma,\tau \in \Delta(\mathcal{F})$, we call $\sigma$ \emph{Blackwell less informative} than $\tau$ if and only if $\int \psi \dd \sigma \leq \int \psi \dd \tau$ for every convex and continuous $\psi : \mathcal{F} \to \R$.
\end{definition}

A Blackwell more informative posterior-belief distribution is precisely one that is preferred ex-ante by every expected-utility decision-maker.%
	\footnote{Explicitly,
	$\sigma$ is Blackwell less informative than $\tau$ iff
	for any non-empty and compact (action) set $\mathcal{A}$ and any continuous (payoff) $V : \mathcal{A} \times \mathcal{F} \to \R$ such that $V(a,\cdot)$ is expected-utility (that is, affine) for each $a \in \mathcal{A}$,
	we have $\int \sup_{a \in \mathcal{A}} V(a,F) \sigma( \dd F )
	\leq \int \sup_{a \in \mathcal{A}} V(a,F) \tau( \dd F )$.}
By comparison, for $\sigma,\tau \in \Delta(\mathcal{F})$, the posterior-mean distribution $F_\sigma$ is less informative than $F_\tau$ (per the definition in §\ref{sec:model:informativeness}) if and only if $\tau$ is preferred ex-ante to $\sigma$ by every expected-utility decision-maker who cares about the state only through its mean, i.e. has a risk-neutral attitude to uncertainty about the state.%
	\footnote{For each $a \in A$, if $V(a,\cdot) : \mathcal{F} \to \R$ is expected-utility, then there is a (von Neumann--Morgenstern payoff) function $U(a,\cdot) : [0,1] \to \R$ such that $V(a,F) = \int U(a,x) F(\dd x)$ for each $F \in \mathcal{F}$, and risk-neutrality means affineness of $U(a,\cdot)$.}
It obviously follows that if $\sigma$ is Blackwell less informative than $\tau$, then $F_\sigma$ is less informative than $F_\tau$.

Focussing on `less informative than' seems appropriate, since the sender herself is assumed to care about the state only through its mean. But one can ask how our results change if Blackwell's notion is used instead. To that end, define the weak set order as usual: for two sets $S,S' \subseteq \Sigma(F_0)$, we call $S$ \emph{B-lower} than $S'$
if and only if for any $\sigma \in S$ and $\tau \in S'$, there is an element of $S'$ that is Blackwell more informative than $\sigma$,
and there is an element of $S$ that is Blackwell less informative than $\tau$. Say that $S$ is \emph{strictly B-lower} than $S'$
if and only if $S$ is B-lower than $S'$
and $S'$ is not B-lower than $S$, and call $S'$ \emph{(strictly) B-higher} than $S$
if and only if $S$ is (strictly) B-lower than $S'$.

\begin{proposition}
	\label{proposition:nondecr_blackwell}
	\Cref{theorem:nondecr} remains true when \eqref{eq:mcs} is replaced by
	\begin{equation}
		\argmax_{\sigma \in \Sigma(F_0)} \int u(\mu_F) \sigma(\dd F)
		\Bnotstrictlyhigherthan
		\argmax_{\sigma \in \Sigma(F_0)} \int v(\mu_F) \sigma(\dd F) .
		\label{eq:mcs_blackwell}
		\tag{$\star_\text{B}$}
	\end{equation}
\end{proposition}

\begin{proposition}
	\label{proposition:incr_blackwell}
	\Cref{theorem:incr} becomes \emph{false} when \eqref{eq:mcs_incr} is replaced by
	\begin{equation}
		\argmax_{\sigma \in \Sigma(F_0)} \int u(\mu_F) \sigma(\dd F)
		\Blowerthan
		\argmax_{\sigma \in \Sigma(F_0)} \int v(\mu_F) \sigma(\dd F) .
		\label{eq:mcs_incr_blackwell}
		\tag{$\star\star_\text{B}$}
	\end{equation}
	In particular, there exist regular $u,v : [0,1] \to \R$ and a distribution $F_0$ such that $u$ is strictly concave, $v$ is S-shaped, $F_0$ is atomless, and \eqref{eq:mcs_incr_blackwell} fails.
\end{proposition}

\begin{proof}[Proof of \Cref{proposition:incr_blackwell}]
	Let $u,v : [0,1] \to \R$ be given by $u(x) \coloneqq -x^2$ and $v(x) \coloneqq x^2-2x^3/3$ for each $x \in [0,1]$. Both are regular, $u$ is strictly concave, and $v$ is S-shaped. Hence $u$ and $v$ both satisfy the crater property, and $u$ is coarsely less convex than $v$. Let $F_0$ be the uniform distribution.

	Then (as noted at the start of §\ref{sec:mcs_incr:pf_thm2}) there are monotone partitional distributions $F$ and $G$ such that $F$ ($G$) is uniquely optimal for $u$ (for $v$) given $F_0$. Since $u$ is strictly concave, $F$ is the point mass at $1/2$: $F = \1_{[1/2,1]}$. Since $v$ is S-shaped, $G$ has the `upper censorship' form: it has $G$-revealing interval $[0,a]$ and $G$-pooling interval $[a,1]$, where $a \in (0,1)$ satisfies
	\begin{equation*}
		\frac{ v(b) - v(a) }{ b-a } = v'(b) ,
		\quad \text{where} \quad
		b = \frac{1}{1-F_0(a)} \int_a^1 \xi F_0( \dd \xi ) .\footnote{For a proof that this $G$ is optimal, see \cref{footnote:upper_censor_opt} on \cpageref{footnote:upper_censor_opt}.}
	\end{equation*}
	It is easily verified that $a=1/4$ and $b=5/8$, so $G(x)=x$ for $x \in [0,1/4)$, $G(x)=1/4$ for $x \in [1/4,5/8)$, and $G(x)=1$ for $x \in [5/8,1]$.

	We must show that \eqref{eq:mcs_incr_blackwell} fails. In light of the above, it suffices to exhibit a $\sigma \in \Sigma(F_0)$ such that $F_\sigma=F$ and, for every $\tau \in \Sigma(F_0)$ such that $F_\tau=G$, $\sigma$ fails to be Blackwell less informative than $\tau$. To that end, consider the posterior-belief distribution $\sigma$ induced by a signal which reveals (only) whether or not the state belongs to $(1/4,3/4]$;%
		\footnote{That is, $\sigma(\{H\})=1/2=\sigma(\{H'\})$ where $H$ and $H'$ are the distributions given by $H(x)=0$ and $H'(x)=2x$ for $x \in [0,1/4)$, $H(x)=2x-1/2$ and $H'(x)=1/2$ for $x \in [1/4,3/4)$, and $H(x)=1$ and $H'(x)=2x-1$ for $x \in [3/4,1]$.}
	evidently $F_\sigma=F$. 

	Suppose toward a contradiction that $\tau \in \Sigma(F_0)$ satisfies $F_\tau=G$ and that $\sigma$ is Blackwell less informative than $\tau$. For any measurable set $S \subseteq [0,1]$, write $\mathcal{F}_S \coloneqq \left\{ H : \supp(H) \subseteq \cl(S) \right\} \subseteq \mathcal{F}$. Let $\lambda$ be the Lebesgue measure on $[0,1]$, and note that its CDF is $F_0$. Since $[0,1/4]$ is a $G$-revealing interval, $\tau\bigl( \bigcup_{x \in [0,1/4]} \mathcal{F}_{\{x\}} \bigr) = 1/4 = \lambda([0,1/4])$, which since $\tau \in \Sigma(F_0)$ implies that $\tau\bigl( \mathcal{F}_{[0,1/4]} \bigr) = 1/4$. Since $\sigma$ is Blackwell less informative than $\tau$, we have $\tau\bigl( \mathcal{F}_{(1/4,3/4]} \bigr) = 1/2$ (and $\tau\bigl( \mathcal{F}_{[0,1/4] \cup (3/4,1]} \bigr) = 1/2$). Then $\tau\bigl( \mathcal{F}_{[0,1/4]} \bigr) = \lambda([0,1/4])$ and $\tau\bigl( \mathcal{F}_{(1/4,3/4]} \bigr) = \lambda((1/4,3/4])$, which since $\tau \in \Sigma(F_0)$ implies that $\tau\bigl( \mathcal{F}_{(3/4,1]} \bigr) = \lambda((3/4,1]) = 1/4 > 0$. But each distribution $H \in \mathcal{F}_{(3/4,1]}$ has mean $\mu_H \geq 3/4 > 5/8$, contradicting the supposition that $F_\tau=G$.
\end{proof}

\begin{proof}[Proof of \Cref{proposition:nondecr_blackwell}]
	For the first (sufficiency) part of \Cref{theorem:nondecr}, fix upper semi-continuous $u,v : [0,1] \to \R$ with $u$ coarsely less convex than $v$, and further fix a distribution $F_0$. It suffices to show that if
	\begin{equation}
		\argmax_{\sigma \in \Sigma(F_0)} \int u(\mu_F) \sigma(\dd F)
		\Bhigherthan
		\argmax_{\sigma \in \Sigma(F_0)} \int v(\mu_F) \sigma(\dd F) 
		\label{eq:mcs_decr_blackwell_pf1}
	\end{equation}
	then
	\begin{equation}
		\argmax_{\sigma \in \Sigma(F_0)} \int u(\mu_F) \sigma(\dd F)
		=
		\argmax_{\sigma \in \Sigma(F_0)} \int v(\mu_F) \sigma(\dd F) .
		\label{eq:mcs_decr_blackwell_pf2}
	\end{equation}
	So, suppose that \eqref{eq:mcs_decr_blackwell_pf1} holds. Then, since `$\sigma$ is Blackwell less informative than $\tau$' implies `$F_\sigma$ is less informative than $F_\tau$', we have
	\begin{equation*}
		\argmax_{F \in \mpc{F_0}} \int u \dd F
		\higherthan
		\argmax_{F \in \mpc{F_0}} \int v \dd F .
	\end{equation*}
	As shown in the proof of \Cref{theorem:nondecr} (\cref{app:proofs:pf_thm_nondecr}), the fact that $u$ is coarsely less convex than $v$ implies that $F \mapsto \int u \dd F$ is interval-dominated by $F \mapsto \int v \dd F$. Hence by Proposition~5 in \textcite{QuahStrulovici2007extensions},%
		\footnote{For completeness, this proposition is stated and proved in \cref{app:mcs_lit}.}
	we have
	\begin{equation*}
		\argmax_{F \in \mpc{F_0}} \int u \dd F
		=
		\argmax_{F \in \mpc{F_0}} \int v \dd F ,
	\end{equation*}
	which is equivalent to \eqref{eq:mcs_decr_blackwell_pf2}.

	For the converse (necessity) part of \Cref{theorem:nondecr}, fix upper semi-continuous $u,v : [0,1] \to \R$ such that $u$ is not coarsely less convex than $v$. The proof of \Cref{theorem:nondecr} (\cref{app:proofs:pf_thm_nondecr}) constructs a binary distribution $F_0$ such that
	\begin{equation*}
		\argmax_{F \in \mpc{F_0}} \int u \dd F
		\strictlyhigherthan
		\argmax_{F \in \mpc{F_0}} \int v \dd F .
	\end{equation*}
	Because $F_0$ is binary, this implies that
	\begin{equation*}
		\argmax_{\sigma \in \Sigma(F_0)} \int u(\mu_F) \sigma(\dd F)
		\Bstrictlyhigherthan
		\argmax_{\sigma \in \Sigma(F_0)} \int v(\mu_F) \sigma(\dd F) . \qedhere
	\end{equation*}
\end{proof}

\subsection{Relaxing differentiability}
\label{app:robust:kink}

Unlike \Cref{theorem:nondecr}, \Cref{theorem:incr} restricts attention to \emph{regular} interim payoffs $u,v : [0,1] \to \R$, meaning those which are differentiable, are nowhere affine, and satisfy some mild well-behavedness conditions (see the definition on \cpageref{definition:regular}). Dropping the nowhere-affineness restriction makes no difference, as we show in \cref{app:incr_affine}. Here, we examine dropping differentiability.

The following shows that there exist (knife-edge) examples of non-differentiable interim payoffs for which comparative statics holds even though the conditions of \Cref{theorem:incr} and \Cref{proposition:W_down,proposition:ido_binary} are not satisfied.

\begin{proposition}
	\label{proposition:kink_knife}
	Let $v \coloneqq \min\{w_c,w_a\}$, where $w_a : [0,1] \to \R$ is affine, $w_c : [0,1] \to \R$ is Lipschitz continuous and twice continuously differentiable with $w_c''>0$, and $w_c(0)<w_a(0)$ and $w_c(1)>w_a(1)$. Then for any upper semi-continuous and bounded-below $u : [0,1] \to \R$ that is coarsely less convex than $v$ and any distribution $F_0$, \eqref{eq:mcs_incr} holds.
\end{proposition}

Note that $v$ is S-shaped and non-regular (with a single concave kink). \Cref{proposition:kink_knife} delivers comparative statics between $u$ and $v$ (whatever the prior $F_0$) provided $u$ is coarsely less convex than $v$, even though $v$ is neither concave nor convex (cf. \Cref{proposition:W_down}), and even if $u$ violates the crater property (cf. \Cref{theorem:incr}) and $F_0$ is not binary (cf. \Cref{proposition:ido_binary}). This result is knife-edge because if we replace $v$ by an (arbitrarily close) \emph{regular} S-shaped function, then \Cref{proposition:W_down} delivers the opposite conclusion: there exists a (regular and) coarsely less convex $u : [0,1] \to \R$ and a distribution $F_0$ such that \eqref{eq:mcs_incr} \emph{fails.}

The shape of $v$ contemplated in \Cref{proposition:kink_knife} appears in (and was inspired by) Proposition~4 in \textcite[]{CerreiavioglioDillenbergerOrtoleva2020}, which we discuss in §\ref{app:robust:disappointment} below.

\begin{proof}
	Let $v : [0,1] \to \R$ satisfy the hypotheses of \Cref{proposition:kink_knife}, and write $V(F) \coloneqq \int v \dd F$ for any distribution $F$. By Theorem~1 in \textcite{QuahStrulovici2007extensions}, it suffices to show that for any distribution $F_0$, $V$ is I-quasi-supermodular on $\mpc{F_0}$. (This theorem and the definition of I-quasi-supermodularity are both reviewed in \cref{app:mcs_lit}.)

	So, fix a distribution $F_0$. Let $y$ be the unique $z \in (0,1)$ such that $w_c(z)=w_a(z)$. Given any $F \in \mpc{F_0}$, let $C_F : [0,1] \rightarrow \R$ be given by $C_F(x) \coloneqq \int_0^x F$ for each $x \in [0,1]$. Recall from \cref{app:proofs:product} that any two distributions $F,G \in \mpc{F_0}$ have a greatest lower bound $F \wedge G$ and a least upper bound $F \vee G$ with respect to `less informative than', that $C_{F\wedge G}$ is the convex envelope of $\min\{C_F,C_G\}$, and that $C_{F\vee G}=\max\{C_F,C_G\}$. 

	\begin{namedthm}[Claim.]
		\label{claim:iqsm_kink}
		For $F,G \in \mpc{F_0}$, if $F$ is less informative than $G$ and $C_F(y)=C_G(y)$, then $V(F) \le V(G)$, with equality if and only if $C_F=C_G$ on $[0,y]$.
	\end{namedthm}

	\begin{proof}[Proof of the {\hyperref[claim:iqsm_kink]{claim}}]
		\renewcommand{\qedsymbol}{$\square$}
		Fix an $F \in \mpc{F_0}$. Integrating by parts%
			\footnote{Licensed by e.g. Theorem~18.4 in \textcite{Billingsley1995}.}
		twice yields
		\begin{equation*}
			\int_{[0,y]} v \dd F
			= v(y) F(y) - \int_0^y v' F
			= v(y) F(y) - w_c'(y) C_F(y) + \int_0^y w_c'' C_F .
		\end{equation*}
		Similarly, letting $k \in \R$ denote the slope of $w_a$, we have
		\begin{equation*}
			\int_{(y,1]} v \dd F
			= w_a(1) - v(y) F(y) - k C_F(1) + k C_F(y)
		\end{equation*}
		since $v''=w_a''=0$ on $(y,1)$. Hence
		\begin{equation*}
			V(F)
			= w_a(1) - k C_F(1)
			- [w_c'(y)-k] C_F(y)
			+ \int_0^y w_c'' C_F .
		\end{equation*}
		The claim follows since a less informative $F$ is exactly one whose $C_F$ is pointwise lower, with equality at $1$, and since $w_c'(y)>k$ and $w_c''>0$.
	\end{proof}%
	\renewcommand{\qedsymbol}{$\blacksquare$}

	Now, fix any two distributions $F,G \in \mpc{F_0}$, write ${F_\star} \coloneqq F\wedge G$ and ${F^\star} \coloneqq F\vee G$, and suppose that $F$ maximises $V$ on $[{F_\star},F]$; we must show that $V(F^\star) \geq V(G)$, with the inequality strict if $V(F) > V(F_\star)$.

	It must be that $C_F(y)=C_{F_\star}(y)$. To see why, suppose toward a contradiction that $C_F(y)>C_{F_\star}(y)$. Then $F$ has mass on both sides of $y$. Pooling some of this mass from both sides into its mean $y$ yields a distribution $H \in \mpc{F_0}$ that is less informative than $F$ and satisfies $C_{F_\star}(y)<C_H(y)<C_F(y)$ and $V(H)>V(F)$. But ${F_\star}\vee H$ belongs to $[{F_\star},F]$, and $C_{{F_\star}\vee H}(y)=C_H(y)$. So $V({F_\star}\vee H)\ge V(H)>V(F)$ by the \hyperref[claim:iqsm_kink]{claim}, a contradiction.

	We therefore have $C_F(y)=C_{F_\star}(y)\le C_G(y)$, whence $C_{F^\star}(y)=C_G(y)$. Thus $V({F^\star})\ge V(G)$ by the \hyperref[claim:iqsm_kink]{claim}.

	It remains to show that $V({F^\star}) > V(G)$ if $V(F) > V(F_\star)$. So suppose that $V(F) > V(F_\star)$, and suppose toward a contradiction that $V({F^\star})=V(G)$. Then the \hyperref[claim:iqsm_kink]{claim} implies that $C_{F^\star}=C_G$ on $[0,y]$,
	hence $C_F\le C_G$ on $[0,y]$. It follows that $C_{F_\star}=C_F$ on $[0,y]$, since pasting $C_F$ on $[0,y]$ to $C_{F_\star}$ on $[y,1]$ yields a
	convex minorant of $\min\{C_F,C_G\}$, which the convex envelope $C_{F_\star}$ must exceed. On the other hand, since $C_F(y)=C_{F_\star}(y)$, the \hyperref[claim:iqsm_kink]{claim} implies that there is a $z \in (0,y)$ at which $C_{F_\star}(z)<C_F(z)$. Contradiction.
\end{proof}

\subsection{Application to disappointment-averse persuasion}
\label{app:robust:disappointment}

\textcite{CerreiavioglioDillenbergerOrtoleva2020} give an explicit representation for Gül's (\citeyear{Gul1991}) \emph{disappointment-aversion} (an early non-expected-utility model of preferences over lotteries), and use it, among other things, to show that a disappointment-averse sender optimally provides less information. The proof is arduous. Here, we exhibit a short proof using \Cref{proposition:kink_knife}.%
	\footnote{We thank Simone Cerreia-Vioglio for suggesting this application.}

A disappointment-averse decision-maker's behaviour is described by a continuous and strictly increasing interim payoff $v : [0,1] \to \R$ and a disappointment-aversion parameter $\beta \in [0,\infty)$. For each $k \in \R$, define $\phi_{k,\beta} : \R \to \R$ by
\begin{equation*}
	\phi_{k,\beta}(\ell) \coloneqq \min\left\{ \ell, \frac{\ell+\beta k}{1+\beta} \right\}
	\quad \text{for each $\ell \in \R$,}
\end{equation*}
and let $\mathcal{U}(v,\beta) \coloneqq \bigl\{ \phi_{k,\beta} \circ v : k \in v([0,1]) \bigr\}$. The $(v,\beta)$ decision-maker's preferences over distributions are represented by $F \mapsto \min_{u \in \mathcal{U}(v,\beta)} u^{-1}\left( \int u \dd F \right)$. This explicit representation is due to \textcite[Theorem~3]{CerreiavioglioDillenbergerOrtoleva2020}. Note that $\beta=0$ is expected utility.

One application in \textcite{CerreiavioglioDillenbergerOrtoleva2020} is to the standard `single-moment' persuasion model, with the twist that the sender is disappointment-averse. The authors' Proposition~4 asserts that if $v \coloneqq \min\{w_c,w_a\}$, where $w_a : [0,1] \to \R$ is affine, $w_c : [0,1] \to \R$ is Lipschitz continuous and twice continuously differentiable with $w_c''>0$, and $w_c(0)<w_a(0)$ and $w_c(1)>w_a(1)$, then for any $\beta>0$ and any prior $F_0$, a $(v,\beta)$ sender optimally provides less information than a $(v,0)$ sender. The proof is arduous.

\Cref{proposition:kink_knife} in §\ref{app:robust:kink} above furnishes an easy proof. Fix a distribution $F$ that is optimal given $F_0$ for the $(v,\beta)$ sender; we will show that there is a distribution which is optimal for the $(v,0)$ sender and more informative than $F$. By a minmax principle (see the authors' Proposition~3), there is an element $u$ of the sup-norm closure of the convex hull of $\mathcal{U}(v,\beta)$ such that $F$ is optimal given $F_0$ for the $(u,0)$ sender (that is, $F \in \argmax_{H \in \mpc{F_0}} \int u \dd H$). Clearly $u$ is less convex than $v$ in the conventional sense,%
	\footnote{By Choquet's theorem (since every extreme point of the sup-norm closure of the convex hull of $\mathcal{U}(v,\beta)$ belongs to $\mathcal{U}(v,\beta)$), there is a Borel probability measure $\mu$ concentrated on $v([0,1])$ such that $u = \int \left(\phi_{k,\beta} \circ v\right) \mu(\dd k)$, or equivalently $u = \phi \circ v$ where $\phi \coloneqq \int \phi_{k,\beta} \mu(\dd k)$. The function $\phi$ is concave and strictly increasing since $\phi_{k,\beta}$ is for each $k \in \R$.}
hence coarsely less convex by \Cref{corollary:CLC_suffsuff}. So by \Cref{proposition:kink_knife}, there is a distribution $G$ which is optimal for the $(v,0)$ sender (i.e. $G \in \argmax_{H \in \mpc{F_0}} \int v \dd H$) and is more informative than $F$.

This argument also highlights the role of the shape assumptions on $v$: if we perturb $v$ to be regular and S-shaped, then \Cref{proposition:kink_knife} becomes inapplicable and \Cref{proposition:W_down} (\cpageref{proposition:W_down}) becomes applicable, yielding a negative conclusion (no comparative-statics conclusion can be drawn in general).

\section{More on delegation}
\label{app:delegation}

In this appendix, we supplement §\ref{sec:appl:delegation} by giving details on the delegation--persuasion equivalence and proving \Cref{proposition:delegation_clc_b,proposition:delegation_clc_G,proposition:delegation_crater}.

\subsection{Details on the delegation--persuasion equivalence}
\label{app:delegation:equivalence}

\textcite{KolotilinZapechelnyuk2025,KleinerMoldovanuStrack2021} establish (slightly different) equivalences between delegation and persuasion. For our purposes, it is important that the persuasion problem feature a prior $F_0$ which does not depend on the parameters $b$ or $G$, since otherwise (by \Cref{proposition:priorshift}) no positive comparative-statics conclusions can be drawn regarding shifts of these parameters. To construct such a prior $F_0$, we rely on (and slightly extend) the arguments in \textcite{KolotilinZapechelnyuk2025}.

\begin{lemma}
	\label{lemma:delegation_oo}
	Given $B > 0$, there are $\underline a \in (-\infty,(c')^{-1}(\min X))$ and $\bar a \in ((c')^{-1}(\max X),\infty)$ such that, for any Borel-measurable $b : X \to \R$ satisfying $\abs*{b} \le B$, an IC mechanism $A$ is undominated only if $U_A(x) \ge \max\{U(\underline a,x),U(\bar a,x)\}$ for all $x \in X$.
\end{lemma}

\begin{proof}
	Let $\mathcal{F}$ be the set of CDFs $F$ on $\R$ such that $c$ is $F$-integrable. Note that for any $F \in \mathcal{F}$, the identity function is $F$-integrable since $\abs*{c'(a)} \to \infty$ as $\abs*{a} \to \infty$, so that $F$ has finite mean, denoted $\mu_F$.
	Call $F \in \mathcal{F}$ \emph{dominated by $0$} iff for any Borel-measurable $b_\star : X \to \R$ with $\abs*{b_\star} \le B$ and any $x \in X$,
	\begin{equation*}
		\max\left\{\int U(a,x) F(\dd a), \int [ab_\star(x)-c(a)] F(\dd a) \right\} < -c(0) .
	\end{equation*}
	Note that there exists a $\gamma > 0$ such that any $F$ with $\mu_F \notin [-\gamma,\gamma]$ is dominated by $0$.%
		\footnote{Set $B' \coloneqq \max\{B,\min\abs*{X},\max\abs*{X}\}$. Since $\abs*{c'(a)} \to \infty$ as $\abs*{a} \to \infty$, there exists a $\gamma > 0$ such that $B'\abs*{\mu} - c(\mu) < -c(0)$ for all $\mu \notin [-\gamma,\gamma]$. Then, for any $F$ with $\mu_F \notin [-\gamma,\gamma]$, any $b_\star : X \to \R$ with $\abs*{b_\star} \le B$ and any $x \in X$, $\max\{\int U(a,x) F(\dd a),\int [ab_\star(x)-c(a)] F(\dd a)\} \le B'\abs*{\mu_F} - c(\mu_F) < -c(0)$, by Jensen's inequality.}
	Moreover, there exists a $\delta \in \R$ such that any $F$ satisfying $\mu_F \in [-\gamma,\gamma]$ and $\int U(a,x)F(\dd a) < \delta$ for some $x \in X$ is dominated by $0$.%
		\footnote{Set $B' \coloneqq \max\{B,\abs*{\min X},\abs*{\max X}\}$ and $\delta \coloneqq -c(0)-3B'\gamma$. Then, for any such $F$ and $x$, and for any $y \in X$ and $b_\star : X \to \R$ with $\abs*{b_\star} \le B$, $\int U(a,y)F(\dd a) = \int U(a,x)F(\dd a) + (y-x)\int a F(\dd a) < \delta + 2B'\gamma < -c(0)$ and, similarly, $\int [ab_\star(y) - c(a)]F(\dd a) < \delta + 3B'\gamma = -c(0)$.}
	Choosing $\underline a < (c')^{-1}(\min X)$ sufficiently small and $\bar a > (c')^{-1}(\max X)$ sufficiently large ensures that $\max\{U(a,x) : a \in \{\underline a,\bar a\}, x \in X\} \le \delta$, since $\abs*{c'(a)} \to \infty$ as $\abs*{a} \to \infty$.

	Now, fix an IC mechanism $A$ such that $U_A(x) < \max\{U(\underline a,x),U(\bar a,x)\}$ for some $x \in X$; we must show that it is dominated.
	For each $y \in X$, we have $A_y \in \mathcal{F}$ since $\abs*{c'(a)} \to \infty$ as $\abs*{a} \to \infty$, and $\int U(a,x)A_y(\dd a) \le U_A(x) < \delta$ since $A$ is IC.
	Hence $A_y$ is dominated by $0$ for each $y \in X$, so $A$ is dominated.
\end{proof}

\begin{lemma}
	\label{lemma:delegation_wlog}
	The assumptions that $\underline a = 0 = c'(0)$ and that $\bar a = 1 = c'(1)$ involve no loss of generality.
\end{lemma}
\begin{proof}
	Define $\alpha: \R \to \R$ by $\alpha(a) = (a-\underline a)/(\bar a - \underline a)$ for all $a \in \R$, $\beta: \R \to \R$ by $\beta(x) \coloneqq (x-c'(\underline a))/(c'(\bar a)-c'(\underline a))$ for all $x \in \R$, and $\ell \coloneqq (\bar a -\underline a)(c'(\bar a)-c'(\underline a))$.
	Relabelling each action $a$ as $\alpha(a)$ and each agent's type $x$ as $\beta(x)$ implies relabelling $\underline a$ as $0$ and $\bar a$ as $1$, respectively, and results in type space $X_\star \coloneqq \beta(X)$, type distribution $G_\star \coloneqq G \circ \beta^{-1}$, and ex-post payoffs $U(\alpha^{-1}(a),\beta^{-1}(x))$ and $V(\alpha^{-1}(a),\beta^{-1}(x))$ for the agent and principal, respectively.
	Moreover, the payoff function $(a,x) \mapsto U(\alpha^{-1}(a),\beta^{-1}(x))$ induces the same preferences over actions, for each type of the agent, as $(a,x) \mapsto (U(\alpha^{-1}(a),\beta^{-1}(x))-\underline a \beta^{-1}(x))/\ell = ax-c_\star(a)$, where $c_\star : \R \to \R$ is given by 
	\begin{equation*}
		c_\star(a) \coloneqq \frac{c((\bar a - \underline a) a + \underline a)-(\bar a - \underline a)c'(\underline a) a}{\ell}
	\end{equation*} 
	 for all $a \in \R$.
	Similarly, $(a,x) \mapsto V(\alpha^{-1}(a),\beta^{-1}(x))$ induces the same preferences over IC mechanisms for the principal as $(a,x) \mapsto [V(\alpha^{-1}(a),\beta^{-1}(x))- \underline a (b \circ \beta^{-1})(x)]/\ell = ab_\star(x)-c_\star(a)$, where $b_\star : X_\star \to \R$ is given by $b_\star \coloneqq \beta \circ b \circ \beta^{-1}$.
	Finally, $c_\star'(a) = a$ for each $a \in \{0,1\}$, and $X_\star$, $G_\star$, $b_\star$ and $c_\star$ satisfy the hypotheses imposed on $X$, $G$, $b$ and $c$, respectively.
\end{proof}
\

\subsection{Proof of \texorpdfstring{\Cref{proposition:delegation_clc_b} (\cpageref{proposition:delegation_clc_b})}{Proposition \ref{proposition:delegation_clc_b} (p. \pageref{proposition:delegation_clc_b})}}
\label{app:delegation:proof_clc_b}

Fix $x < y$ in $[0,1]$ and $\alpha \in (0,1)$, and set $z \coloneqq \alpha x + (1-\alpha) y$; we must show that $\alpha u(x) + (1-\alpha) u(y) \ge \mathrel{(>)} u(z)$ implies $\alpha v(x) + (1-\alpha) v(y) \ge \mathrel{(>)} v(z)$.

We may extend $\lambda$ to a function $[0,1] \to [0,1]$ that is decreasing if \ref{item:delegation_b_pos} holds and increasing if \ref{item:delegation_b_neg} holds.
Then, whether \ref{item:delegation_b_pos} or \ref{item:delegation_b_neg} holds, 
\begin{align*}
	\int_x^z \lambda(w)[b(w)-w] G(\dd w) &\ge \lambda(z)\int_x^z [b(w)-w] G(\dd w) \quad \text{and} 
	\\
	\int_z^y \lambda(w)[w - b(w)] G(\dd w) &\ge \lambda(z)\int_z^y [w-b(w)] G(\dd w).
\end{align*}
Define $\psi : [0,1] \to \R$ by $\psi(w) \coloneqq \int_0^w [w-t] G(\dd t)$ for each $w \in [0,1]$. We have
\begin{equation*}
	u(w) = \psi(w) + \int_0^w [t-b(t)]G(\dd t) \quad \text{and} \quad v(w) = \psi(w) + \int_0^w [t-d(t)] G(\dd t)
\end{equation*}
for all $w \in [0,1]$. 
Then, letting $\delta \coloneqq \alpha \psi(x) + (1-\alpha)\psi(y)-\psi(z)$,
\begin{multline*}
	\alpha v(x) + (1-\alpha)v(y)-v(z)
	\\
	\begin{aligned}
		&= \delta+\alpha\int_x^z [d(w)-w] G(\dd w) + (1-\alpha) \int_z^y [w-d(w)] G(\dd w) 
		\\
		&= \delta+\alpha \int_x^z \lambda(w) [b(w)-w] G(\dd w) + (1-\alpha) \int_z^y \lambda(w) [w - b(w)] G(\dd w) 
		\\
		&\ge \delta+\lambda(z) \left(\alpha \int_x^z [b(w)-w] G(\dd w) + (1-\alpha) \int_z^y [w-b(w)] G(\dd w)\right)
		\\
		&= \delta+\lambda(z) \bigl( \alpha u(x)+(1-\alpha) u(y) - u(z)-\delta \bigr).
	\end{aligned}
\end{multline*}
Hence, $\alpha u(x) + (1-\alpha)u(y) \ge u(z)$ implies $\alpha v(x) + (1-\alpha)v(y) \ge v(z)$, because $\lambda$ has image in $[0,1]$ and $\psi$ is convex. 
Moreover, $u$ is affine on $[0,\min X]$ and on $[\max X,1]$, and $\psi$ is strictly convex on $X$, so that $\alpha u(x) + (1-\alpha)u(y) > u(z)$ implies $\alpha v(x) + (1-\alpha)v(y) > v(z)$.
\qed

\subsection{Proof of \texorpdfstring{\Cref{proposition:delegation_clc_G} (\cpageref{proposition:delegation_clc_G})}{Proposition \ref{proposition:delegation_clc_G} (p. \pageref{proposition:delegation_clc_G})}}
\label{app:delegation:proof_clc_G}

Fix $x < y$ in $[0,1]$ and $\alpha \in (0,1)$, and set $z \coloneqq \alpha x + (1-\alpha) y$; we must show that $\alpha u(x) + (1-\alpha) u(y) \ge \mathrel{(>)} u(z)$ implies $\alpha v(x) + (1-\alpha) v(y) \ge \mathrel{(>)} v(z)$.
Define $\psi : [0,1] \to \R$ by 
\begin{equation*}
	\psi(w) = 
	\begin{cases}
		\alpha (b(w)-x) & \text{if $w \in [x,z) \cap X$}\\
		(1-\alpha)(y-b(w)) & \text{if $w \in [z,y] \cap X$}\\
		0 & \text{otherwise.}
	\end{cases}
\end{equation*}
and note that there exists a $\widehat{w} \in X$ such that $\psi(w) \ge \mathrel{(\le)} 0$ for all $w < \mathrel{(>)} \widehat{w}$ if \ref{item:delegation_G_down} holds, and $\psi(w) \le \mathrel{(\ge)} 0$ for all $w < \mathrel{(>)} \widehat{w}$ if \ref{item:delegation_G_up} holds.
Choose $g, h : \interior(X) \to (0,\infty)$ such that $h/g$ is decreasing if \ref{item:delegation_G_down} holds and increasing if \ref{item:delegation_G_up} holds, and set $g(\min X) \coloneqq 1 \eqqcolon g(\max X)$.
Fix $\eps > 0$ and set $h(w') \coloneqq \max\{\eps,\min\{1/\eps,\lim_{w \rightarrow w'} h/g\}\}$ for each $w' \in \{\min X,\max X\}$.
Then
\begin{multline*}
	\alpha v(x) + (1-\alpha)v(y)-v(z) 
	= \int \psi \dd H
	= \int \frac{h}{g}\psi \dd G \ge \frac{h(\widehat{w})}{g(\widehat{w})}\int \psi \dd G 
	\\
	= \frac{h(\widehat{w})}{g(\widehat{w})}
	\bigl( \alpha u(x) + (1-\alpha)u(y)-u(z) \bigr) ,
\end{multline*}
where the inequality holds provided $\eps$ is small enough. Since $h(\widehat{w})/g(\widehat{w}) > 0$, the result follows.
\qed

\subsection{Proof of \texorpdfstring{\Cref{proposition:delegation_crater} (\cpageref{proposition:delegation_crater})}{Proposition \ref{proposition:delegation_crater} (p. \pageref{proposition:delegation_crater})}}
\label{app:delegation:proof_crater}

Define $\gamma : \interior(X) \to \R$ by $\gamma(x) = (2-b'(x))/(x-b(x))$ for all $x \in \interior(X)$.
Note that $u$ is continuously differentiable with bounded derivative $u'(x) = G(x) + G'(x)(x-b(x))$, and that $u'$ is differentiable on $\interior(X)$ with derivative $u''(x) = G'(x)(2-b'(x)) + G''(x)(x-b(x))$.
Since $G'$ is log-concave and $X = \supp(G)$, we have $G'(x) > 0$ for all $x \in \interior(X)$.
Then since $b(x) \ne x$,
\begin{equation*}
	\frac{u''(x)}{G'(x)(x-b(x))} = \gamma(x)+\frac{G''(x)}{G'(x)}
	\quad \text{for each $x \in \interior(X)$.}
\end{equation*}
By hypothesis, $\gamma$ and $G''/G'$ are decreasing. Moreover, since $X$ is convex and $b$ is continuous, either $b(x) > x$ for all $x \in \interior(X)$, or $b(x) < x$ for all $x \in \interior(X)$. 
If $b(x) < \mathrel{(>)} x$ for all $x \in \interior(X)$, there are $x \leq y$ in $X$ such that $u$ is strictly convex (strictly concave) on $[\min X,x]$, affine on $[x,y]$, and strictly concave (strictly convex) on $[y,\max X]$. 
Moreover, $u$ is affine on $[0,\min X]$ and differentiable at $\min X$ unless $\min X = 0$, and affine on $[\max X,1]$ and differentiable at $\max X$ unless $\max X = 1$.
Hence $u$ is weakly regular and either S-shaped or reverse-S-shaped.

It remains to show that $u$ satisfies the crater property for \emph{weakly} regular interim payoffs $u$ (which may have affine segments), as defined in \cref{app:incr_affine}. Since $u$ is either S-shaped or reverse-S-shaped, and is affine on $[0,\min X]$ and $[\max X,1]$, what must be shown is simply that the tangents to $u$ at $0$ and at $1$ cross at coordinates $(X,Y) \in \R^2$ such that $Y \leq u(X)$ \emph{and} $u$ is strictly convex on an interval containing $X$. Since $u=0$ on $[0,\min X]$ and $u(x) = x - \int b \dd G$ for each $x \in [\max X,1]$, what the crater property requires in case $u$ is S-shaped (is reverse-S-shaped) is precisely that $0 \geq \min X - \int b \dd G$ (that $0 \leq \max X - \int b \dd G$), which holds since $\int b \dd G \in X$.
\qed

\section{\texorpdfstring{\Cref{theorem:incr}}{Theorem~\ref{theorem:incr}} with affine segments}
\label{app:incr_affine}

In this appendix, we show that regularity can be weakened in \Cref{theorem:incr} (and \Cref{proposition:W_down}) to allow affine segments, at the cost of a longer proof.

Say that $u : [0,1] \to \R$ \emph{weakly regular}
iff
(i)~$u$ is continuous and possesses a continuous and bounded derivative $u' : (0,1) \to \R$, and
(ii)~$[0,1]$ may be partitioned into finitely many intervals,
on each of which $u$ is either strictly convex, strictly concave, or affine.
This is exactly regularity as defined on \cpageref{definition:regular}, except with (ii) modified to permit affine segments.

We defined the crater property (\cpageref{definition:crater}) only for regular $u : [0,1] \to \R$. For the more general weakly regular case, the definition is as follows:
a weakly regular $u : [0,1] \to \R$ satisfies the
crater property
if and only if
for any $x < y < z < w$ in $[0,1]$
such that $u$ is concave on $[x,y]$ and $[z,w]$ and strictly convex on $[y,z]$,
the tangents to $u$ at $x$ and at $w$
cross at coordinates $(X,Y) \in \R^2$
satisfying $y \leq X \leq z$
and $Y \leq u(X)$.

\begin{namedthm}[\Cref*{theorem:incr}$\boldsymbol{^\dag}$.]
	\label{theorem:incr_affine}
	Let $u : [0,1] \to \R$ be weakly regular.
	If $u$ satisfies the crater property,
	then for every weakly regular $v : [0,1] \to \R$ that is coarsely more convex than $u$
	and every atomless convex-support distribution $F_0$,
	\begin{equation}
		\argmax_{F \in \mpc{F_0}} \int u \dd F
		\lowerthan
		\argmax_{F \in \mpc{F_0}} \int v \dd F .
		\label{eq:mcs_incr_affine}
		\tag{$\star\star$}
	\end{equation}
	Conversely, if \eqref{eq:mcs_incr} holds
	for every weakly regular $v$ that is coarsely more convex than $u$
	and every atomless convex-support distribution $F_0$,
	then $u$ satisfies the crater property.
\end{namedthm}

We view the restriction to atomless and convex-support priors $F_0$ as a mild form of well-behavedness. A simple way of dropping this restriction is to replace it with the (generic) requirement that there be a unique distribution optimal given $F_0$ for $u$ and for $v$; with this substitution, \hyperref[theorem:incr_affine]{\Cref*{theorem:incr}$^\dag$} remains true as stated.%
	\footnote{The first (sufficiency) half follows from \hyperref[theorem:incr_affine]{\Cref*{theorem:incr}$^\dag$} and the facts that when the space of distributions has the topology of weak convergence, it is sequentially compact (by Prokhorov's theorem, e.g. Theorem~5.1 in \textcite{Billingsley1999}), the atomless convex-support distributions form a dense subset, $F_0 \mapsto \argmax_{F \in \mpc{F_0}} \int u \dd F$ is upper hemi-continuous, and the binary relation `is less informative than' is continuous.}

Like \Cref{theorem:incr}, \Cref{proposition:W_down} remains true when regularity is replaced by weak regularity and only atomless convex-support prior distributions $F_0$ are considered. The only change to the statement is that `concave or convex' is replaced by `concave or strictly convex', and the only change to the proof is that when establishing the `only if' (necessity) part, the possibility that $v$ is affine on $[z,w']$ must be dealt with as a separate case; we omit the details.

The proof of the converse (necessity) half of \Cref*{theorem:incr}$^\dag$ follows from the proof in §\ref{sec:mcs_incr:pf_thm2} of the necessity half of \Cref{theorem:incr}, except (again) that the possibility that $u$ is affine on $[z,w']$ must be dealt with as a separate case; we omit the details. The proof of the sufficiency half is long; below, we first (§\ref{app:proofs:pf_thm_incr:sufficiency}) prove it using a lemma, then (§\ref{app:proofs:pf_thm_incr:lemma_ordered_prices}) prove the lemma.

\subsection{Proof of the sufficiency part of \texorpdfstring{\hyperref[theorem:incr_affine]{\Cref*{theorem:incr}$\boldsymbol{^\dag}$}}{Theorem \ref*{theorem:incr}†}}
\label{app:proofs:pf_thm_incr:sufficiency}

Given any distribution $F$, let $C_F : [0,1] \rightarrow \R$ be given by $C_F(x) \coloneqq \int_0^x F$ for each $x \in [0,1]$.
We shall make free use of the order isomorphism described in \cref{app:proofs:product}
between distributions $F$ ordered by informativeness and convex functions $C_F$ ordered by pointwise inequality.

The sufficiency proof relies on three lemmata.
The first is a version of Dworczak and Martini's (\citeyear{DworczakMartini2019}) duality theorem.
Given any weakly regular $u : [0,1] \rightarrow \R$, let $\mathcal{M}(u)$ denote the space of all convex and Lipschitz continuous functions $p : [0,1] \rightarrow \R$ satisfying $p \geq u$.

\begin{lemma}
	\label{lemma:duality_weak}
	Let $u : [0,1] \rightarrow \R$ be weakly regular, and let $F_0$ be an atomless distribution. Then
	\begin{equation*}
		\min_{p \in \mathcal{M}(u)} \int p \dd F_0
		= \max_{F \in \mpc{F_0}} \int u \dd F ,
	\end{equation*}
	where both sides are well-defined.
	Moreover, for $p \in \mathcal{M}(u)$ and $F \in \mpc{F_0}$ to solve (respectively) the minimisation and maximisation problems,
	it is necessary and sufficient that both
	
	\begin{enumerate}[label=(\alph*)]
		\item \label{item:duality_weak:convex}
		$p$ is affine on any interval on which $C_F < C_{F_0}$, and

		\item \label{item:duality_weak:support}
		$p = u$ on $\supp(F)$.
	\end{enumerate}
	
\end{lemma}

\begin{proof}[Proof of \Cref{lemma:duality_weak}]
	Fix a distribution $F_0$.
	The result is trivial if $F_0$ is degenerate, so suppose not.
	Since $u$ is weakly regular, for any convex and continuous $q : [0,1] \rightarrow \R$ such that $q \geq u$, there is a $p \in \mathcal{M}(u)$ such that $p \leq q$.
	Thus the first part follows from Theorem~1(ii) in \textcite{DizdarKovac2020} applied to the restriction of $u$ to $\co(\supp(F_0))$, since $u$ is weakly regular.

	For the second part, fix any $p \in \mathcal{M}(u)$ and any $F \in \mpc{F_0}$.
	Since $F_0$ is atomless,
	we have $F_0(0)=0$
	and thus $F(0)=0$.%
		\footnote{We have $C_F \leq C_{F_0}$ and $C_{F_0}(0) = 0 \leq C_F(0)$, whence $[C_F(x)-C_F(0)] / x \leq [C_{F_0}(x)-C_{F_0}(0)] / x$ for every $x \in (0,1]$,
		so that letting $x \downarrow 0$
		yields $F(0) \leq F_0(0) = 0$.}
	Because $p$ is convex and Lipschitz, it admits a bounded right-derivative $p' : [0,1) \to \R$, which we may extend continuously to $[0,1]$.
	Then for any distribution $G$ with $G(0)=0$,
	integrating by parts twice,%
		\footnote{This is licensed by e.g. Theorem~18.4 in \textcite{Billingsley1995}.}
	\begin{equation*}
		\int p \dd G
		= p(1) - \int p' G
		= p(1) - p'(1)C_G(1) + \int C_G \dd p' ,
	\end{equation*}
	where the last term is to be understood in the Lebesgue--Stieltjes sense.
	Thus
	\begin{equation*}
		\int p \dd F_0 \geq \int p \dd F \geq \int u \dd F ,
	\end{equation*}
	where the first inequality is strict unless \ref{item:duality_weak:convex} holds,
	while the second is strict unless \ref{item:duality_weak:support} holds since $p$ and $u$ are continuous.
\end{proof}

\begin{lemma}
	\label{lemma:propertyW}
	Let $u : [0,1] \to \R$ be weakly regular and satisfy the crater property,
	and suppose there are $x<z$ in $[0,1]$
	such that the tangent to $u$ at $x$ (at $z$)
	weakly exceeds $u$ on $[x,z]$.
	Then there is a $y \in (x,z]$ (a $y \in [x,z)$)
	such that $u$ is concave on $[x,y]$ (on $[y,z]$)
	and strictly convex on $[y,z]$ (on $[x,y]$).
\end{lemma}

\begin{proof}[Proof of \Cref{lemma:propertyW}]
	Suppose that the tangent to $u$ at $x$ weakly exceeds $u$ on $[x,z]$;
	the other case is analogous.
	Let $y$ be the largest $y' \in [x,z]$ such that $u$ is concave on $[x,y']$.
	We have $y > x$ since $u$ is weakly regular.
	It remains to show that $u$ is strictly convex on $[y,z]$.
	This is immediate if $y=z$, so suppose for the remainder that $y<z$.

	Let $\widehat{z}$ be the largest $w \in [y,1]$ such that $u$ is strictly convex on $[y,w]$;
	clearly $\widehat{z} > y$ by the weak regularity of $u$.
	We must show that $\widehat{z} \geq z$,
	so suppose toward a contradiction that $\widehat{z} < z$.
	Then by weak regularity,
	$u$ is concave on $[\widehat{z},w]$
	for some $w \in (\widehat{z},z]$.
	But then $u$ violates the crater property,
	since the tangent to $u$ at $x$
	strictly exceeds $u$ on $[y,\widehat{z}]$ (as $u$ is strictly convex on $[y,\widehat{z}]$).
\end{proof}

\begin{lemma}
	\label{lemma:ordered_prices}
	Let $u,v : [0,1] \rightarrow \R$ be weakly regular,
	and suppose that $u$ satisfies the crater property and is coarsely less convex than $v$.
	Let $F_0$ be an atomless convex-support distribution.
	Then for any
	\begin{equation*}
		p \in \argmin_{r \in \mathcal{M}(u)} \int r \dd F_0
		\quad \text{and} \quad
		q \in \argmin_{r \in \mathcal{M}(v)} \int r \dd F_0 ,
	\end{equation*}
	if $q$ is affine on an interval $[x,y] \subseteq \supp(F_0)$,
	then so is $p$.
\end{lemma}

\Cref{lemma:ordered_prices} is proved in the next section.

\begin{proof}[Proof of the first (sufficiency) part of \Cref{theorem:incr}]
	Fix weakly regular $u,v : [0,1] \to \R$
	such that $u$ satisfies the crater property
	and is coarsely less convex than $v$,
	let $F_0$ be an atomless convex-support distribution,
	and fix
	\begin{equation*}
		G' \in \argmax_{F \in \mpc{F_0}} \int u \dd F
		\quad \text{and} \quad
		H' \in \argmax_{F \in \mpc{F_0}} \int v \dd F .
	\end{equation*}
	We shall construct
	\begin{equation*}
		G'' \in \argmax_{F \in \mpc{F_0}} \int u \dd F
		\quad \text{and} \quad
		H'' \in \argmax_{F \in \mpc{F_0}} \int v \dd F 
	\end{equation*}
	such that $G''$ is less informative than $H'$ and $G'$ is less informative than $H''$.

	We derive $G''$ from $G'$ by fully pooling signal realisations over each concavity interval of $u$,
	in the following precise sense.
	Assume without loss of generality that $u$ is not strictly convex,
	and enumerate the maximal proper intervals on which $u$ is concave as
	$([x_k,z_k])_{k=1}^K$ (where $K \in \N$).
	For each $k$, let $y_k$ denote the mean of $G'$ conditional on the event $[x_k,z_k]$.
	(In case $[x_k,z_k]$ is $G'$-null,
	let $y_k$ be an arbitrary element of $[x_k,z_k]$.)
	Define a distribution $G''$ by
	\begin{equation*}
		G''(w) \coloneqq
		\begin{cases}
			G'(x_k-)
			& \text{if $w \in [x_k,y_k)$ for some $k \in \{1,\dots,K\}$} \\
			G'(z_k)
			& \text{if $w \in [y_k,z_k]$ for some $k \in \{1,\dots,K\}$} \\
			G'(w)
			& \text{otherwise,}
		\end{cases}
	\end{equation*}
	where `$G(x-)$' is shorthand for $\lim_{y \uparrow x} G(y)$.
	For any $G'$-non-null $[x_k,z_k]$,
	the distribution `$G''$ conditional on $[x_k,z_k]$'
	is less informative than
	the distribution `$G'$ conditional on $[x_k,z_k]$',%
		\footnote{Explicitly:
		the distribution $\1_{(z_k,1]} + \1_{[x_k,z_k]} \times [ G'' - G''(x_k-) ] / [ G''(z_k) - G''(x_k-) ]$
		is less informative than the distribution
		$\1_{(z_k,1]} + \1_{[x_k,z_k]} \times [ G' - G'(x_k-) ]  / [ G'(z_k) - G'(x_k-) ]$.}
	so $\int_{[x_k,z_k]} u \dd G'' \geq \int_{[x_k,z_k]} u \dd G'$.
	And we have $G''=G'$ on $\mathcal{X} \coloneqq [0,1] \setminus \Union_{k=1}^K [x_k,z_k]$,
	so that $\int_{\mathcal{X}} u \dd G'' = \int_{\mathcal{X}} u \dd G'$ since $\mathcal{X}$ is open.
	Thus $\int u \dd G'' \geq \int u \dd G'$,
	which since $G'$ is optimal for $u$ given prior $F_0$
	implies that $G''$ is, too.

	We similarly derive $H''$ from $H'$
	by spreading signal realisations over each convexity interval of $v$
	as much as possible subject to keeping $H''$ less informative than the prior $F_0$.
	Formally, assume without loss of generality that $v$ is not strictly concave,
	enumerate the maximal proper intervals on which $v$ is convex as $(I_\ell)_{\ell = 1}^L$ (where $L \in \N$),
	and define $I \coloneqq \Union_{\ell=1}^L I_\ell$.
	Let $C$ be the convex envelope of $\1_I C_{F_0} + \1_{[0,1] \setminus I} C_{H'}$,
	and let the distribution $H''$ be defined by $C_{H''} = C$.
	We have $H''=H'$ off $I$,
	and clearly `$H''$ conditional on $I_\ell$' is more informative than `$H'$ conditional on $I_\ell$' for each $H'$-non-null $I_\ell$,
	so
	$\int v \dd (H''-H')
	= \sum_{\ell=1}^L \int_{I_\ell} v \dd (H''-H')
	\geq 0$,
	which since $H'$ is optimal for $v$ given prior $F_0$ implies that $H''$ is, too.

	It remains to prove that $G''$ is less informative than $H'$ and that $G'$ is less informative than $H''$.
	We shall rely on the following claim, whose proof (relegated to the end) hinges on \Cref{lemma:propertyW,lemma:ordered_prices}.

	\begin{namedthm}[Claim.]
		\label{claim:W_sufficiency}
		Let $G$ and $H$ be optimal (given prior $F_0$) for $u$ and $v$, respectively. 
		Then for any $a < b$ in $[0,1]$ such that
		$C_H < C_{F_0}$ on $(a,b)$ and
		$C_H = C_{F_0}$ on $\{a,b\}$,
		there are $c \leq d$ in $\supp(G)$ such that $C_G \leq C_H$ on $[a,b] \setminus (c,d)$, $u$ is affine on $[c,d]$, and $\supp(G) \cap [c,d]$ is infinite unless $u$ is locally concave at $c$ and $d$.
	\end{namedthm}

	To prove that $G''$ is less informative than $H'$,
	it suffices to show that 
	for any $a < b$ in $[0,1]$
	such that $C_{H'} < C_{F_0}$ on $(a,b)$ and $C_{H'} = C_{F_0}$ on $\{a,b\}$,
	we have $C_{G''} \leq C_{H'}$ on $(a,b)$.
	So fix such a pair $a<b$.
	By the \hyperref[claim:W_sufficiency]{claim},
	there are $c \leq d$ in $\supp(G'')$ such that $C_{G''} \leq C_{H'}$ on $[a,b] \setminus (c,d)$, $u$ is affine on $[c,d]$, and $\supp(G'') \cap [c,d]$ is infinite unless $u$ is locally concave at $c$ and $d$.
	Since $u$ is concave on $[c,d]$, $1 \le \abs*{\supp(G'') \cap [c,d]} \le 3$ with the first inequality strict only if $u$ fails to be locally concave at either $c$ or $d$, by definition of $G''$.
	Then, $c = d$, so that $C_{G''} \leq C_{H'}$ on $[a,b]$.

	Similarly, to prove that $G'$ is less informative than $H''$, 
	it suffices to show that 
	for any $a < b$ in $[0,1]$
	such that $C_{H''} < C_{F_0}$ on $(a,b)$ and $C_{H''} = C_{F_0}$ on $\{a,b\}$,
	we have $C_{G'} \leq C_{H''}$ on $(a,b)$.
	So fix such a pair $a<b$.
	By the \hyperref[claim:W_sufficiency]{claim},
	there are $c \leq d$ in $\supp(G')$ such that $C_{G'} \leq C_{H''}$ on $[a,b] \setminus (c,d)$ and $u$ is affine on $[c,d]$.
	If $[a,b]$ and $[c,d]$ are disjoint, then we are done.
	Suppose for the remainder that $[a,b] \intersect [c,d]$ is non-empty.
	We must show that $C_{G'} \leq C_{H''}$ on $[a',b'] \coloneqq [a,b] \intersect [c,d]$.

	$v$ is convex on $[a',b']$ since $[a',b'] \subseteq [c,d]$,
	so by definition of $H''$,
	the restriction of $C_{H''}$ to $[a',b']$ equals the convex envelope of
	$\1_{(a',b')} C_{F_0} + \1_{\{a',b'\}} C_{H''}$.
	We have $C_{G'} \leq \1_{(a',b')} C_{F_0} + \1_{\{a',b'\}} C_{H''}$ on $[a',b']$
	by hypothesis and the fact that $G'$ is less informative than the prior $F_0$.%
		\footnote{At $a'$,
		we have if $a'=c$ that $C_{G'}(a') \leq C_{H''}(a')$,
		and if not then $a'=a$,
		in which case $C_{G'}(a') \leq C_{F_0}(a') = C_{H''}(a')$ since $G'$ is less informative than $F_0$.
		Similarly at $b'$.}
	Thus since $C_{G'}$ is convex,
	it must satisfy $C_{G'} \leq C_{H''}$ on $[a',b']$.

	\begin{proof}[Proof of the claim]%
		\renewcommand{\qedsymbol}{$\square$}
		Fix $a < b$ in $[0,1]$ such that
		$C_H < C_{F_0}$ on $(a,b)$ and
		$C_H = C_{F_0}$ on $\{a,b\}$.
		Note that $[a,b] \subseteq \supp(F_0)$ since the latter is convex.
		Since $u$ and $v$ are weakly regular,
		\Cref{lemma:duality_weak} provides
		that there exist
		\begin{equation*}
			p \in \argmin_{r \in \mathcal{M}(u)} \int r \dd F_0
			\quad \text{and} \quad
			q \in \argmin_{r \in \mathcal{M}(v)} \int r \dd F_0 ,
		\end{equation*}
		and that $q$ is affine on $[a,b]$.
		By \Cref{lemma:ordered_prices}, it follows that $p$ is also affine on $[a,b]$.
		Write $[a',b']$ for the maximal interval $I$
		such that $p$ is affine on $I$
		and $[a,b] \subseteq I \subseteq \supp(F_0)$.
		We have $C_G = C_{F_0}$ on $\{a',b'\}$ by \Cref{lemma:duality_weak},
		which since $\supp(F_0)$ is convex and contains $[a',b']$
	 	implies that $(a',b') \intersect \supp(G)$ is non-empty.
		Define
		\begin{equation*}
			c \coloneqq \inf \left[ (a',b') \intersect \supp(G) \right]
			\quad \text{and} \quad
			d \coloneqq \sup \left[ (a',b') \intersect \supp(G) \right].
		\end{equation*}

		Observe that $\supp(G) \cap [c,d]$ is infinite if $u$ fails to be locally concave at $c$ ($d$), as $c = a'$ ($d = b'$) in this case, since $u$ is regular, $p$ is affine on $[a',b']$ with $p \ge u$, and $p = u$ on $\{c,d\}$, by \Cref{lemma:duality_weak}.

		We now show that $C_G \leq C_H$ on $[a,b] \setminus (c,d)$.
		This is trivial if $c \leq a$ and $b \leq d$, so suppose not.
		Assume that $a < c$; we will show that $C_G \leq C_H$ on $[a,\min\{b,c\}]$.
		(We omit the analogous argument that $C_G \leq C_H$ on $[\max\{a,d\},b]$ when $d<b$.)
		By definition of $c$, $C_G$ is affine on $[a',c]$.
		Since $C_G \leq C_{F_0}$ with equality at $a'$,
		where $C_{F_0}$ is convex and differentiable at $a'$ ($F_0$ being atomless),
		$C_G$ coincides on $[a',c]$ with the tangent to $C_{F_0}$ at $a'$.
		Similarly, since $C_H \leq C_{F_0}$ with equality at $a$
		and $C_H$ is convex,
		we have on $[a,1]$ that $C_H$ exceeds the tangent to $C_{F_0}$ at $a$.
		Since the latter tangent exceeds the former on $[a,1]$,
		it follows that $C_G \leq C_H$
		on $[a',c] \intersect [a,1] = [a,c] \supseteq [a,\min\{b,c\}]$.

		It remains to show that $u$ is affine on $[c,d]$.
		Since $u$ is weakly regular, it suffices to show that $u$ is affine on $[x,w]$ for any $x < w$ in $(a',b') \intersect \supp(G)$.
		Fix such a pair $x < w$, and note that by \Cref{lemma:duality_weak}, $p$ is tangent to $u$ at $x$ and at $w$.
		Then since $p \geq u$ and $u$ satisfies the crater property, 
		\Cref{lemma:propertyW} provides that
		there are $y \in (x,w]$ and $z \in [x,w)$
		such that $u$ is concave on $[x,y]$ and on $[z,w]$
		and strictly convex on $[x,z]$ and on $[y,w]$.
		Clearly it must be that $y=w$ and $z=x$,
		so that $u$ is concave on $[x,w]$.
		Since $p$ is convex and $p \geq u$ on $[x,w]$ with equality on $\{x,w\}$,
		it follows that $u$ is affine on $[x,w]$.
	\end{proof}%
	\renewcommand{\qedsymbol}{$\blacksquare$}

	With the claim established, the proof is complete.
\end{proof}

\subsection{Proof of \texorpdfstring{\Cref{lemma:ordered_prices}}{Lemma \ref{lemma:ordered_prices}}}
\label{app:proofs:pf_thm_incr:lemma_ordered_prices}

We rely on the following result, which follows from \Cref{lemma:duality_weak,lemma:propertyW}.

\begin{corollary}
	\label{corollary:duality}
	Let $u : [0,1] \rightarrow \R$ be weakly regular, let $F_0$ be an atomless convex-support distribution, and let $p$ minimise $\int p \dd F_0$ over $\mathcal{M}(u)$.
	Then

	\begin{enumerate}[label=(\roman*)]
		
		\item \label{item:duality_weak:linear}
		for any $x < z$ such that $[x,z]$ is maximal among the intervals of affineness of $p$ within $\supp(F_0)$, 
		there are 
		\begin{equation*}
			x < y \leq \frac{\int_x^z \xi F_0(\dd \xi)}{F_0(z)-F_0(x)} \leq y' < z
		\end{equation*}
		such that $p(y) = u(y)$ and $p(y') = u(y')$, and 
		
		\item \label{item:duality_weak:non_linear}
		if $p(y) > u(y)$ for some $y \in \supp(F_0)$ such that $F_0(y) > 0$ ($F_0(y) < 1$), then $y > 0$ and there is $x \in [0,y)$ ($y < 1$ and there is $z \in (y,1]$) such that $p$ is affine on $[x,y]$ (on $[y,z]$).
		
	\end{enumerate}

	\noindent
	Moreover, if $u$ satisfies the crater property, then 

	\begin{enumerate}[label=(\roman*),resume]
		
		\item \label{item:duality_weak:propertyW}
		given $x < y$ such that $[x,y]$ is maximal among the intervals of affineness of $p$ within $\supp(F_0)$, and $F_0(x) > 0$ ($F_0(y) < 1$),
		it holds that $p(x) = u(x)$ ($p(y) = u(y)$), that $u$ is convex and not affine on some open interval $I$ containing $x$ ($y$), and that
		\begin{equation*}
			u' < \mathrel{(>)} \frac{p(y)-p(x)}{y-x}
			\quad \text{on $(0,x) \intersect I$ (on $(y,1) \intersect I$).}
		\end{equation*}

	\end{enumerate}
	
\end{corollary}

\begin{proof}[Proof of \Cref{corollary:duality}]
	Fix $F$ maximising $\int u \dd F$ on $\mpc{F_0}$.
	For \ref{item:duality_weak:linear}, fix $x < z$ such that $[x,z]$ is maximal among intervals of affineness of $p$ within $\supp(F_0)$.
	Then $C_F = C_{F_0}$ on $\{x,z\}$ by \Cref{lemma:duality_weak}.%
		\footnote{If e.g. $C_F(x) < C_{F_0}(x)$, then $x$ lies in the interior of $\supp(F_0)$, $C_F < C_{F_0}$ on a neighbourhood of $x$, and $p$ is affine on this neighbourhood by \ref{item:duality_weak:convex}, contradicting the definition of $[x,z]$.}
	Then $(x,z)$ is $F$-non-null since $F_0$ has convex support,%
		\footnote{Since $F_0$ has convex support, $C_{F_0}$ is not affine on $[x,z]$. 
		Then, neither is $C_F$, and thus $\supp(F) \intersect (x,z)$ is not empty.}
	and thus there are $y,y' \in \supp(F)$ such that 
	\begin{equation*}
		x < y \leq \frac{\int_{(x,z)} \xi F(\dd \xi)}{\int_{(x,z)} \dd F} \leq y' < z.
	\end{equation*}
	By \ref{item:duality_weak:support}, $p(y) = u(y)$ and $p(y') = u(y')$.
	Finally, since $C_F = C_{F_0}$ on $\{x,z\}$ and $F_0$ is atomless, 
	$F = F_0$ on $\{x,z\}$ and $F$ is continuous at $x$ and $z$, so that  
	\begin{equation*}
		\frac{\int_x^z \xi F(\dd \xi)}{\int_{(x,z)} \dd F} 
		= \frac{zF(z) - x F(x) - [C_F(z)-C_F(x)]}{F(z)-F(x)} 
		= \frac{\int_x^z \xi F_0(\dd \xi)}{F_0(z) - F_0(x)}.
	\end{equation*}
	This proves \ref{item:duality_weak:linear}.

	For \ref{item:duality_weak:non_linear}, suppose that $p(y) > u(y)$ for some $y \in \supp(F_0)$ such that $F_0(y) > 0$ (the case $F_0(y) < 1$ is analogous).
	Then $y \notin \supp(F)$ by \ref{item:duality_weak:support}, so that $C_F$ is affine on a neighbourhood of $y$.
	Moreover, $y > \min \supp(F_0)$ since $F_0$ is atomless.
	Then, $y > 0$ and, since $\supp(F_0)$ is convex and $C_{F_0}$ is strictly convex on $\supp(F_0)$, there is $x \in [0,y)$ such that $C_F < C_{F_0}$ on $[x,y)$.
	Hence, $p$ is affine on $[x,y]$ by \ref{item:duality_weak:convex}, as $p$ is continuous.

	For \ref{item:duality_weak:propertyW}, fix $x < y$ such that $[x,y]$ is maximal among intervals of affineness of $p$ within $\supp(F_0)$, and $F_0(x) > 0$ (the case $F_0(y) < 1$ is analogous).
	By \ref{item:duality_weak:linear}, there is $w \in (x,y)$ such that $p(w) = u(w)$, so that $p$ is tangent to $u$ at $w$.
	Then, there is $z \in [x,w)$ such that $u$ is strictly convex on $[x,z]$ and concave on $[z,w]$, by \Cref{lemma:propertyW}.
	Let $b \coloneqq \min \supp(F_0)$ and $a$ be the smallest $a' \in [b,x]$ such that $p$ is affine on $[a',x]$.
	We consider two cases.

	\medskip

	\emph{Case~1: $a = x$.}
	Note that $x > b$ since $F_0$ is atomless and $F_0(x) > 0$.
	Then, by the hypothesis of this case, there exists an increasing sequence $(x_k)_{k \in \N} \subseteq (b,x)$ such that $\lim_k x_k = x$ and on which $C_F = C_{F_0}$, by \ref{item:duality_weak:convex}.
	Then, there exists an increasing sequence $(y_k)_{k \in \N} \subseteq (b,x) \intersect \supp(F)$ such that $\lim_k y_k = x$, since $C_{F_0}$ is strictly convex on $\supp(F_0)$.
	By \ref{item:duality_weak:support}, $p(y_k) = u(y_k)$ for each $k \in \N$.
	Then, since $p$ is convex and $u$ is weakly regular, by the hypothesis of this case, $u$ is convex and not affine on $[y_{k'},x]$ for some $k' \in \N$, and 
	\begin{equation*}
		u' < \frac{p(y)-p(x)}{y-x} \quad \text{on $(y_{k'},x)$.}%
			\footnote{To see why this last property must hold, suppose it were to fail. Then there is a sequence $(z_k)_{k=1}^\infty \subseteq (b,x)$ with $\lim_{k \to \infty} z_k = x$ such that $u'(z_k) \geq [ p(y)-p(x) ]/(y-x)$. Since $u$ is weakly regular, it follows that $u' \geq [ p(y)-p(x) ]/(y-x)$ on $[y_{k'},x]$ for some $k' \in \N$. 
			But then $p$ is affine on $[y_{k'},x]$ since $u(y_{k'}) = p(y_{k'})$, contradicting the hypothesis of this case.}
	\end{equation*}
	Moreover, $p(x) = u(x)$ and thus $u$ is affine on $[x,w]$ if $z = x$, since $p \geq u$ with equality on $\{x,w\}$ and $u$ is concave on $[z,w]$.
	The result follows by choosing $I = (y_{k'},z)$ if $z > x$, and $I = (y_{k'},w)$ otherwise.

	\medskip

	\emph{Case~2: $a < x$.}
	In this case, there is $\widehat{x} \in (a,x)$ such that $p(\widehat{x}) = u(\widehat{x})$, by \ref{item:duality_weak:linear}.
	Then, $p$ is tangent to $u$ at $\widehat{x}$, and thus there is $\widehat{y} \in (\widehat{x},x]$ such that $u$ is concave on $[\widehat{x},\widehat{y}]$, and strictly convex on $[\widehat{y},x]$, by \Cref{lemma:propertyW}.
	Define
	\begin{equation*}
		I \coloneqq 
		\begin{cases}
			(\widehat{y},z) & \quad \text{if $\widehat{y} < x < z$}\\
			(\widehat{x},z) & \quad \text{if $\widehat{y} = x$}\\
			(\widehat{y},w) & \quad \text{if $x = z$}.
		\end{cases}
	\end{equation*}
	Note that $\widehat{y} < z$, for otherwise $u$ would be concave on $[\widehat{x},w]$ and thus $p$ would be affine on $[\widehat{x},w]$ (since $p = u$ on $\{\widehat{x},w\}$), contradicting $\widehat{x} < x$.
	Then $I$ contains $x$, since $\widehat{x} < \widehat{y} \leq x \leq z < w$.
	
	To show that $u$ is convex and not affine on $I$, note that $u$ is strictly convex on $[\widehat{y},z]$, as it is weakly regular and strictly convex on $[\widehat{y},x]$ and $[x,z]$.
	Then $p(x) = u(x)$, since $u$ satisfies the crater property and, clearly, the tangents to $u$ at $\widehat{x}$ and $w$ intersect at $(x,p(x))$.
	Hence $u$ is affine on $[\widehat{x},x]$ (on $[x,w]$) if $\widehat{y} = x$ ($x = z$), since $u$ is concave on $[\widehat{x},\widehat{y}]$ with $u(\widehat{x}) = p(\widehat{x})$ (on $[z,w]$ with $u(w) = p(w)$).
	Since $\widehat{y} < z$ and $u$ is strictly convex on $[\widehat{y},z]$, $u$ is convex and not affine on $I$.

	It remains to show that
	\begin{equation*}
		u' < \frac{p(y)-p(x)}{y-x}
		\quad \text{on $(0,x) \intersect I$.}
	\end{equation*}
	To this end, since $u$ is convex on $I$, we may assume without loss of generality that
	\begin{equation*}
		u'(x) \geq \frac{p(y)-p(x)}{y-x}.
	\end{equation*}
	Then $x = z$ and equality holds, since $p \geq u$ with equality at $x$ and $u$ is strictly convex on $[x,z]$.
	The result follows since $\widehat{y} < z$ and $u$ is strictly convex on $[\widehat{y},x]$.
\end{proof}

\begin{proof}[Proof of \Cref{lemma:ordered_prices}]
	Fix $F_0$, $p$ and $q$. 
	Suppose toward a contradiction that there exist $\widetilde{x} < \widetilde{z}$ in $\supp(F_0)$ such that $q$ is affine on $[\widetilde{x},\widetilde{z}]$, but $p$ is not.
	Assume without loss of generality that $[\widetilde{x},\widetilde{z}]$ is maximal among the intervals of affineness of $q$ within $\supp(F_0)$.
	We consider two cases.

	\medskip

	\emph{Case~1: $u$ is convex on $[\widetilde{x},\widetilde{z}]$.} 
	We shall construct $a \in [0,\widetilde{x}]$ such that $u$ is concave on $[a,\widetilde{z}]$ and $p(a) = u(a)$.
	A similar argument yields $b \in [\widetilde{z},1]$ such that $u$ is concave on $[\widetilde{x},b]$ and $p(b) = u(b)$.
	Then $u$ is concave on $[a,b]$ and thus $p$ is affine on $[a,b]$, contradicting the fact that $p$ is not affine on $[\widetilde{x},\widetilde{z}] \subseteq [a,b]$.

	To construct $a$, note that $v$ is convex on $[\widetilde{x},\widetilde{z}]$ by the hypothesis of this case, since $u$ is coarsely less convex than $v$.
	Then $v$ is affine on $[\widetilde{x},\widetilde{z}]$ by \ref{item:duality_weak:linear}
	(since \ref{item:duality_weak:linear} implies that $q(y) = v(y)$ for some $y \in (\widetilde{x},\widetilde{z})$).
	Then so is $u$, as it is coarsely less convex than $v$.
	Then, if $p(\widetilde{x}) = u(\widetilde{x})$, we may take $a = \widetilde{x}$.
	Hence, assume without loss of generality that $p(\widetilde{x}) > u(\widetilde{x})$.

	Let $\bar z$ be the largest $z \in [\widetilde{x},1]$ such that $p$ is affine on $[\widetilde{x},z]$. 
	Then $\bar z < \widetilde{z}$ by hypothesis, and $\bar z > \widetilde{x}$ by \ref{item:duality_weak:non_linear} (which is applicable since $F_0(\widetilde x) < 1$).
	Let $\bar x$ be the smallest $x \in [0,\widetilde{x}] \intersect \supp(F_0)$ such that $p$ is affine on $[x,\bar z]$.
	By \ref{item:duality_weak:linear}, there is $a \in (\bar x,\bar z)$ such that $p(a) = u(a)$.
	And $a$ belongs to $[0,\widetilde{x}]$ since $u$ and $p$ are affine on $[\widetilde x,\bar z]$ and since $p \geq u$, with strict inequality at $\widetilde x$.
	
	It remains to prove that $u$ is concave on $[a,\widetilde{z}]$.
	As $u$ is affine on $[\widetilde{x},\widetilde{z}]$ and weakly regular, and $\widetilde{x} < \bar z < \widetilde{z}$, it suffices to show that $u$ is concave on $[a,\bar z]$.
	Note that $p$ is tangent to $u$ at $a$ as $\bar x < a < \bar z$ and $p(a) = u(a)$.
	Then $u$ is concave on $[a,\bar z]$ by \Cref{lemma:propertyW}, as $p \geq u$ on $[a,\bar z]$, and $u$ and $p$ are affine on $[\widetilde{x},\bar z]$.%
		\footnote{Indeed, \Cref{lemma:propertyW} yields $y \in (a,\bar z]$ such that $u$ is concave on $[a,y]$ and strictly convex on $[y,\bar z]$. And $y = \bar z$ since $u$ is affine on $[\widetilde{x},\bar z]$.}

	\medskip

	\emph{Case~2: $u$ is not convex on $[\widetilde{x},\widetilde{z}]$.}
	In this case, since $u$ is weakly regular, there are $\widetilde{x} \leq c < d \leq \widetilde{z}$ such that $[c,d]$ is maximal among the intervals in $[\widetilde{x},\widetilde{z}]$ on which $u$ is strictly concave.
	Then $p$ and $u$ differ somewhere in $(c,d)$ and thus, by \ref{item:duality_weak:non_linear}, $p$ is not strictly convex on $(c,d)$.
	Hence there are $\bar x < \bar z$ such that $[\bar x,\bar z]$ is maximal among the intervals of affineness of $p$ within $\supp(F_0)$, and $[\bar x,\bar z] \intersect (c,d)$ is not empty.
	Since $p$ is not affine on $[\widetilde{x},\widetilde{z}]$, either $\widetilde{x} < \bar x$ or $\bar z < \widetilde{z}$.
	We consider the case $\widetilde{x} < \bar x$; the other is analogous.

	Note that $\bar x < d \leq \widetilde z$, where the strict inequality holds as $[\bar x,\bar z] \intersect (c,d)$ is not empty. 
	We shall exhibit a $w \in (\bar x,\widetilde z]$ such that
	\begin{equation}
		u(\bar x)_\alpha u(w) \geq u(\bar x_\alpha w) \quad \text{for all $\alpha \in (0,1)$,}
		\label{eq:u_chord_xw}
	\end{equation}
	a $\widetilde{y} \in (\bar x,w)$ such that $q(\widetilde{y}) = v(\widetilde{y})$, and show that $v(\bar x) < q(\bar x)$.
	To see why this suffices, note that it implies that given $\alpha \in (0,1)$ such that $\bar x_\alpha w = \widetilde{y}$,
	\begin{equation*}
		v(\bar x)_\alpha v(w) < q(\bar x)_\alpha q(w) = q(\widetilde{y}) = v(\widetilde{y}),
	\end{equation*}
	where the strict inequality holds since $\alpha \in (0,1)$, $v(\bar x) < q(\bar x)$ and $v(w) \leq q(w)$, and the first equality holds as $q$ is affine on $[\widetilde{x},\widetilde{z}] \supseteq [\bar x,w]$.
	Together with \eqref{eq:u_chord_xw}, this contradicts the fact that $u$ is coarsely less convex than $v$.

	To construct $w$ note that, by \ref{item:duality_weak:linear}, there is 
	\begin{equation*} 
		\bar x < \frac{\int_{\bar x}^{\bar z} \xi F_0(\dd \xi)}{F_0(\bar z) - F_0(\bar x)} \leq \bar y < \bar z 
	\end{equation*}
	such that $p(\bar y) = u(\bar y)$.
	Define $w \coloneqq \min\left\{\bar y, \widetilde{z}\right\}$ and note that $w \in (\bar x,\widetilde{z}]$.
	To establish \eqref{eq:u_chord_xw}, note $p$ is tangent to $u$ at $\bar y$, so that there is $\gamma \in [\bar x,\bar y)$ such that $u$ is strictly convex on $[\bar x,\gamma]$ and concave on $[\gamma,\bar y]$, by \Cref{lemma:propertyW}.
	Then \eqref{eq:u_chord_xw} holds since $p(\bar x) = u(\bar x)$ by \ref{item:duality_weak:propertyW} (which is applicable since $F_0(\bar x) > 0$ and $\widetilde x < \bar x$).%
		\footnote{This is easily seen graphically. It follows from the facts that $p$ is affine on $[\bar x,\bar y]$, that $p \geq u$ on $[\bar x,\bar y]$ with equality on $\{\bar x,\bar y\}$, that $u$ is convex on $[\bar x,\gamma]$ and concave on $[\gamma,\bar y]$, and that $\bar x < w \leq \bar y$.}

	To construct $\widetilde{y} \in (\bar x,w)$ such that $q(\widetilde{y}) = v(\widetilde{y})$, let $[a,b]$ be the maximal interval of convexity of $u$ containing $\bar x$.
	(This is well-defined since $u$ is weakly regular).
	Note that if $\bar x \in [c,d)$ then $\gamma = \bar x$, as $u$ is concave on $(c,d)$ and strictly convex on $[\bar x,\gamma]$. But then $u$ would be affine on $[\bar x,\bar y]$ since $u \le p$ is concave on $[\gamma,\bar y]$ and $p = u$ on $\{\bar x,\bar y\}$, contradicting the fact that $u$ is strictly concave on $[c,d]$. Hence $\bar x < c$ as $\bar x < d$.
	Then $b \leq c$, and by \ref{item:duality_weak:propertyW} (applicable since $F_0(\bar x) > 0$ and $\widetilde x < \bar x$) we have that $a < \bar x < b$, that $u$ is not affine on $[a,b]$, and that
	\begin{equation}
		u' < \frac{p(\bar z)-p(\bar x)}{\bar z - \bar x} \quad \text{on $(a,\bar x)$.}
		\label{eq:up_strict}
	\end{equation}
	We rely on the following claim, proved at the end. 
	
	\begin{namedthm}[Claim.]
		\label{claim:vthrehsolds_bounds}
		$a \leq \widetilde{x}$ and $\widetilde{z} \leq \bar z$.
	\end{namedthm}
	
	By \ref{item:duality_weak:linear}, we may choose  
	\begin{equation*}
		\widetilde{x} < y \leq \frac{\int_{\widetilde x}^{\widetilde z} \xi F_0(\dd \xi)}{F_0(\widetilde{z})-F_0(\widetilde{x})} < \widetilde{z}
	\end{equation*}
	such that $q(y) = v(y)$.
	Note that $y < \min\{\bar y,\widetilde z\} = w$ since $y < \widetilde{z}$ and 
	\begin{equation*}
		y \leq \frac{\int_{\widetilde x}^{\widetilde z} w \dd F_0(w)}{F_0(\widetilde{z})-F_0(\widetilde{x})} < \frac{\int_{\bar x}^{\bar z} w \dd F_0(w)}{F_0(\bar{z})-F_0(\bar{x})} \leq \bar y,
	\end{equation*}
	where the strict inequality holds as $F_0$ has convex support,
	$\widetilde{x} < \bar x$ and, by the \hyperref[claim:vthrehsolds_bounds]{claim}, $\widetilde{z} \leq \bar z$.
	Thus we may take $\widetilde{y} \coloneqq y$ if $y > \bar x$.
	If instead $y \leq \bar x$, note that $v$ is convex on $[a,b]$, as $u$ is coarsely less convex than $v$ and convex on $[a,b]$.
	Moreover, $q$ is affine on $[\widetilde x,\widetilde z]$ and $q \geq v$ with equality at $y$.
	Since $a \leq \widetilde x < y \leq \bar x < b \leq c \leq \widetilde z$, it follows that $v = q$ on $[\widetilde x,b] = [a,b] \intersect [\widetilde x,\widetilde z]$.
	As $\bar x < w$, we may then choose any $\widetilde y \in (\bar x,\min\{b,w\})$.

	It remains to prove that $v(\bar x) < q(\bar x)$.
	Note that, by \eqref{eq:u_chord_xw} and \eqref{eq:up_strict}, 
	\begin{equation*}
		u(\widetilde{x})_\alpha u(w) > u(\widetilde{x}_\alpha w) \quad \text{for all $\alpha \in (0,1)$,}
	\end{equation*}
	since $a \leq \widetilde{x} < \bar x$, and $u$ is convex on $[a,\bar x]$, strictly convex on $[\bar x,\gamma]$ and concave on $[\gamma,\bar y]$, and $p \ge u$ is affine on $[\bar x,\bar z] \ni \bar y$ with $p = u$ on $\{\bar x,\bar y\}$.%
		\footnote{In detail, $u$ cannot be affine on an open interval containing $\bar x$, for otherwise $u=p$ on $[\bar x,\bar y]$, which contradicts \eqref{eq:up_strict} since $p$ is affine on $[\bar x,\bar z] \ni \bar y$.
		Then, the continuous map that matches $u$ on $[\widetilde{x},\bar x] \union \{w\}$ and is affine on $[\bar x,w]$, is convex and not affine on $[\widetilde{x},w]$. The result then follows from \eqref{eq:u_chord_xw}.}
	Hence, choosing $\alpha \in (0,1)$ such that $\bar x = \widetilde{x}_\alpha w$, 
	\begin{equation*}
		q(\bar x) = q(\widetilde{x})_\alpha q(w) \geq v(\widetilde{x})_\alpha v(w) > v(\bar x),
	\end{equation*}
	where the equality holds since $q$ is affine on $[\widetilde{x},\widetilde{z}] \supseteq [\widetilde x,w]$, the weak inequality as $q \geq v$, and the strict inequality holds since $u$ is less convex than $v$.
	
	\begin{proof}[Proof of the claim]
		\renewcommand{\qedsymbol}{$\square$}
		We begin by exhibiting $\widetilde{x} \leq c' < d' \leq \widetilde{z}$ such that $u$ is strictly convex on $[\widetilde{x},c']$ and $[d',\widetilde{z}]$, and concave on $[c',d']$.
		By \ref{item:duality_weak:linear}, 
		\begin{equation*}
			v(\widetilde{x}_\alpha \widetilde{z}) = q(\widetilde{x}_\alpha \widetilde{z}) = q(\widetilde{x})_\alpha q(\widetilde{z}) \geq v(\widetilde{x})_\alpha v(\widetilde{z})
			\quad \text{for some $\alpha \in (0,1)$,}
		\end{equation*}
		where the second equality holds since $q$ is affine on $[\widetilde{x},\widetilde{z}]$, and the inequality holds since $q \geq v$.
		Then 
		\begin{equation}
			u(\widetilde{x}_{\alpha} \widetilde{z}) \geq u(\widetilde{x})_\alpha u(\widetilde{z}) \quad \text{for some $\alpha \in (0,1)$,}
			\label{eq:u_no_chord}
		\end{equation}
		since $u$ is coarsely less convex than $v$.
		Hence the tangent to $u$ at some $a_\star \in (\widetilde{x},\widetilde{z})$ weakly exceeds $u$ on $[\widetilde{x},\widetilde{z}]$, as $u$ is weakly regular.
		Therefore, by \Cref{lemma:propertyW}, there are $c' \in [\widetilde{x},a_\star)$ and $d' \in (a_\star,\widetilde{z}]$ such that $u$ is strictly convex on $[\widetilde{x},c']$ and $[d',\widetilde{z}]$, and concave on $[c',a_\star]$ and $[a_\star,d']$.
		As $u$ is weakly regular, it is concave on $[c',d']$, as desired.

		Note that $b \leq c < d \leq \widetilde{z}$.
		Then $a \leq \widetilde{x}$ since $u$ is weakly regular.
		Indeed, if $\widetilde{x} < a$ then, by definition of $a$ and $b$, there would exist $\widetilde{x} \leq a' < a$ and $b < b' \leq \widetilde{z}$ such that $u$ is strictly concave on $[a',a]$ and $[b,b']$.
		But then $c' \leq a'$ and $b' \leq d'$, contradicting the fact that $u$ is convex and not affine on $[a,b]$.

		It remains to show that $\widetilde{z} \leq \bar z$.
		Suppose this fails and seek a contradiction.
		Then $p(\bar z) = u(\bar z)$ by \ref{item:duality_weak:propertyW}, and thus 
		\begin{equation}
			u(\bar x)_\alpha u(\bar z) = p(\bar x) _\alpha p(\bar z) = p(\bar x_\alpha \bar z) \geq u(\bar x_\alpha \bar z) \quad \text{for all $\alpha \in (0,1)$,}
			\label{eq:olc_if}
		\end{equation}
		where the first equality holds since $p(\bar x) = u(\bar x)$, and the second since $p$ is affine on $[\bar x,\bar z]$.
		Moreover, $u$ is convex and not affine on some open interval $I$ containing $\bar z$, by \ref{item:duality_weak:propertyW}.
		Then
		\begin{equation*}
			c' \leq c < d \leq \bar z,
		\end{equation*}
		where the first inequality holds since $\widetilde{x} \leq c < d$ and $u$ is strictly convex on $[\widetilde{x},c']$ and strictly concave on $[c,d]$, and the last inequality holds since $[\bar x,\bar z] \intersect (c,d) \neq \varnothing$ and $u$ is strictly concave on $[c,d]$ and convex on $I \ni \bar z$.
		Then $u$ is convex on $[\bar z,\widetilde{z}]$, as it is convex and not affine on $I \ni \bar z$, concave on $[c',d']$, and strictly convex on $[d',\widetilde{z}]$.
		Then \eqref{eq:up_strict} and \eqref{eq:olc_if} contradict \eqref{eq:u_no_chord}, since $a \leq \widetilde{x} < \bar x$ and $u$ is convex on $[a,\bar x]$.%
			\footnote{To see why, note that the map $\1_{[\tilde{x},\bar x] \union [\bar z,\tilde{z}]} u + \1_{(\bar x,\bar z)} p$ is convex and not affine on $[\widetilde{x},\widetilde{z}]$.}
	\end{proof}
	\renewcommand{\qedsymbol}{$\blacksquare$}

	With the \hyperref[claim:vthrehsolds_bounds]{claim} established,
	the proof is complete.
\end{proof}

\end{appendices}



\printbibliography[heading=bibintoc]

@string{ecta	= "Econometrica"									}

@string{aer 	= "American Economic Review"						}

@string{aerins 	= "American Economic Review: Insights"				}

@string{jpe		= "Journal of Political Economy"					}

@string{restud	= "Review of Economic Studies"						}

@string{jet		= "Journal of Economic Theory"						}

@string{geb		= "Games and Economic Behavior"						}

@string{te		= "Theoretical Economics"							}

@string{rand	= "RAND Journal of Economics"						}

@string{mit		= "MIT Press"						}

@article{AgronskyBrucknerLaczkovichPreiss1985,
	author		= "Steven J. Agronsky and Andrew M. Bruckner and Miklós Laczkovich and David Preiss",
	year		= "1985",
	title		= "Convexity conditions and intersections with smooth functions",
	journaltitle= "Transactions of the American Mathematical Society",
	volume		= "289",
	number		= "2",
	pages		= "659--677",
	langid		= "american",
	doi			= "10.1090/S0002-9947-1985-0784008-9",
}

@article{AkbarpourDworczakKominers2024,
	author		= "Mohammad Akbarpour and Piotr Dworczak and Scott Duke Kominers",
	year		= "2024",
	title		= "Redistributive allocation mechanisms",
	journaltitle= jpe,
	volume		= "132",
	number		= "6",
	pages		= "1793--2178",
	langid		= "american",
	doi			= "10.1086/728111",
}

@article{AlonsoMatouschek2008,
	author		= "Ricardo Alonso and Niko Matouschek",
	year		= "2008",
	title		= "Optimal delegation",
	journaltitle= restud,
	volume		= "75",
	number		= "1",
	pages		= "259--293",
	langid		= "american",
	doi			= "10.1111/j.1467-937X.2007.00471.x",
}

@article{AmadorWerningAngeletos2006,
	author		= "Manuel Amador and Iván Werning and George-Marios Angeletos",
	year		= "2006",
	title		= "Commitment vs. flexibility",
	journaltitle= ecta,
	volume		= "74",
	number		= "2",
	pages		= "365--396",
	langid		= "american",
	doi			= "10.1111/j.1468-0262.2006.00666.x",
}

@article{AmirRietzke2025,
	author		= "Rabah Amir and David Rietzke",
	year		= "2025",
	title		= "A comment on",
	subtitle	= "``{M}onotone comparative statics''",
	journaltitle= ecta,
	volume		= "93",
	number		= "4",
	pages		= "1481--1490",
	langid		= "american",
	doi			= "10.3982/ECTA23292",
}

@article{AndersonSmith2024,
	author		= "Axel Anderson and Lones Smith",
	year		= "2024",
	title		= "The comparative statics of sorting",
	journaltitle= aer,
	volume		= "114",
	number		= "3",
	pages		= "709--51",
	langid		= "american",
	doi			= "10.1257/aer.20210890",
}

@article{ArieliBabichenkoSmorodinskyYamashita2023,
	author		= "Itai Arieli and Yakov Babichenko and Rann Smorodinsky and Takuro Yamashita",
	year		= "2023",
	title		= "Optimal persuasion via bi-pooling",
	journaltitle= te,
	volume		= "18",
	number		= "1",
	pages		= "15--36",
	langid		= "american",
	doi			= "10.3982/TE4663",
}

@article{ArmstrongZhou2022,
	author		= "Mark Armstrong and Jidong Zhou",
	year		= "2022",
	title		= "Consumer information and the limits to competition",
	journaltitle= aer,
	volume		= "112",
	number		= "2",
	pages		= "534--77",
	langid		= "american",
	doi			= "10.1257/aer.20210083",
}

@article{Becker1973,
	author		= "Gary S. Becker",
	year		= "1973",
	title		= "A theory of marriage",
	subtitle 	= "Part {I}",
	journaltitle= jpe,
	volume		= "81",
	number		= "4",
	pages		= "813--846",
	langid		= "american",
	doi			= "10.1086/260084",
}

@article{BergemannHeumannMorris2022,
	author		= "Dirk Bergemann and Tibor Heumann and Stephen Morris",
	year		= "2026",
	title		= "Screening with persuasion",
	journaltitle= jpe,
	volume		= "134",
	number		= "2",
	pages		= "570--625",
	langid		= "american",
	doi			= "10.1086/738342",
}

@article{Bergemannetal2022,
	author		= "Dirk Bergemann and Tibor Heumann and Stephen Morris and Constantine Sorokin and Eyal Winter",
	year		= "2022",
	title		= "Optimal Information Disclosure in Classic Auctions",
	journaltitle= aerins,
	volume		= "4",
	number		= "3",
	pages		= "371--388",
	langid		= "american",
	doi			= "10.1257/aeri.20210504",
}

@article{Blackwell1953,
	author = "David Blackwell",
	year = "1953",
	title = "Equivalent comparisons of experiments",
	journaltitle= "Annals of Mathematical Statistics",
	volume = "24",
	number = "2",
	pages = "265--272",
	doi = "10.1214/aoms/1177729032",
}

@article{BrocasCarillo2007,
	author		= "Isabelle Brocas and Juan D. Carrillo",
	year		= "2007",
	title		= "Influence through ignorance",
	journaltitle= rand,
	volume		= "38",
	number		= "4",
	pages		= "931--947",
	langid		= "american",
	doi			= "10.1111/j.0741-6261.2007.00119.x",
}

@article{CerreiavioglioDillenbergerOrtoleva2020,
	author		= "Simone Cerreia-Vioglio and David Dillenberger and Pietro Ortoleva",
	year		= "2020",
	title		= "An explicit representation for disappointment aversion and other betweenness preferences",
	journaltitle= te,
	volume		= "15",
	number		= "4",
	pages		= "1509--1546",
	langid		= "american",
	doi			= "10.3982/TE3440",
}

@article{ChambersLiuRehbeck2020,
	author		= "Christopher P. Chambers and Ce Liu and John Rehbeck",
	year		= "2020",
	title		= "Costly information acquisition",
	journaltitle= jet,
	volume		= "186",
	langid		= "american",
	doi			= "10.1016/j.jet.2019.104979",
}

@article{Denti2022,
	author		= "Tommaso Denti",
	year		= "2022",
	title		= "Posterior separable cost of information",
	journaltitle= aer,
	volume		= "112",
	number		= "10",
	pages		= "3215--3259",
	langid		= "american",
	doi			= "10.1257/aer.20211252",
}

@article{DiamondStiglitz1974,
	author		= "Peter A. Diamond and Joseph E. Stiglitz",
	year		= "1974",
	title		= "Increases in risk and in risk aversion",
	journaltitle= jet,
	volume		= "8",
	number		= "3",
	pages		= "337--360",
	langid		= "american",
	doi			= "10.1016/0022-0531(74)90090-8",
}

@article{DizdarKovac2020,
	author		= "Deniz Dizdar and Eugen Kováč",
	year		= "2020",
	title		= "A simple proof of strong duality in the linear persuasion problem",
	journaltitle= geb,
	volume		= "122",
	pages		= "407--412",
	langid		= "american",
	doi			= "10.1016/j.geb.2020.04.017",
}

@article{DoganHu2022,
	author		= "Mustafa Do\u{g}an and Ju Hu",
	year		= "2022",
	title		= "Consumer search and optimal information",
	journaltitle= rand,
	volume		= "53",
	number		= "2",
	pages		= "261--449",
	langid		= "american",
	doi			= "10.1111/1756-2171.12411",
}

@article{DworczakKolotilin2023,
	author			= "Piotr Dworczak and Anton Kolotilin",
	year			= "2024",
	title			= "The persuasion duality",
	journaltitle= te,
	volume		= "19",
	number		= "4",
	pages		= "1701--1755",
	langid		= "american",
	doi			= "10.3982/TE5900",
}

@article{DworczakMartini2019,
	author		= "Piotr Dworczak and Giorgio Martini",
	year		= "2019",
	title		= "The simple economics of optimal persuasion",
	journaltitle= jpe,
	volume		= "127",
	number		= "5",
	pages		= "1993--2048",
	langid		= "american",
	doi			= "10.1086/701813",
}

@article{GanuzaPenalva2010,
	author		= "Juan-José Ganuza and José S. Penalva",
	year		= "2010",
	title		= "Signal orderings based on dispersion and the supply of private information in auctions",
	journaltitle= ecta,
	volume		= "78",
	number		= "3",
	langid		= "american",
	doi			= "10.3982/ECTA6640",
}

@article{GentzkowKamenica2016,
	author		= "Matthew Gentzkow and Emir Kamenica",
	year		= "2016",
	title		= "A {R}othschild--{S}tiglitz approach to {B}ayesian persuasion",
	journaltitle= "American Economic Review: Papers \& Proceedings",
	volume		= "106",
	number		= "5",
	pages		= "597--601",
	langid		= "american",
	doi			= "10.1257/aer.p20161049",
}

@article{Gul1991,
	author		= "Faruk Gül",
	year		= "1991",
	title		= "A theory of disappointment aversion",
	journaltitle= ecta,
	volume		= "59",
	number		= "3",
	pages		= "667--686",
	langid		= "american",
	doi			= "10.2307/2938223",
}

@article{HardyLittlewoodPolya1929,
	author		= "G. H. Hardy and J. E. Littlewood and G. Pólya",
	year		= "1929",
	title		= "Some simple inequalities satisfied by convex functions",
	journaltitle= "Messenger of Mathematics",
	volume		= "58",
	pages		= "145--152",
	langid		= "american",
}

@article{Jewitt1986,
	author		= "Ian Jewitt",
	year		= "1986",
	title		= "A note on comparative statics and stochastic dominance",
	journaltitle= "Journal of Mathematical Economics",
	volume		= "15",
	number		= "3",
	pages		= "249--254",
	langid		= "american",
	doi			= "10.1016/0304-4068(86)90013-3",
}

@article{Jewitt1987,
	author		= "Ian Jewitt",
	year		= "1987",
	title		= "Risk aversion and the choice between risky prospects",
	subtitle	= "The preservation of comparative statics results",
	journaltitle= restud,
	volume		= "54",
	number		= "1",
	pages		= "73--85",
	langid		= "american",
	doi			= "10.2307/2297447",
}

@article{KamenicaGentzkow2011,
	author		= "Emir Kamenica and Matthew Gentzkow",
	year		= "2011",
	title		= "{B}ayesian persuasion",
	journaltitle= aer,
	volume		= "101",
	number		= "6",
	pages		= "2590--2615",
	langid		= "american",
	doi			= "10.1257/aer.101.6.2590",
}

@article{KarlinNovikoff1963,
	author		= "Samuel Karlin and Albert Novikoff",
	year		= "1963",
	title		= "Generalized convex inequalities",
	journaltitle= "Pacific Journal of Mathematics",
	volume		= "13",
	number		= "4",
	pages		= "1251--1279",
	langid		= "american",
	doi			= "10.2140/pjm.1963.13.1251",
}

@article{Karr1983,
	author		= "Alan F. Karr",
	year		= "1983",
	title		= "Extreme points of certain sets of probability measures, with applications",
	journaltitle= "Mathematics of Operations Research",
	volume		= "8",
	number		= "1",
	pages		= "74--85",
	langid		= "american",
	doi			= "10.1287/moor.8.1.74",
}

@article{KartikKleinerVanweelden2021,
	author		= "Navin Kartik and Andreas Kleiner and Richard {Van Weelden}",
	year		= "2021",
	title		= "Delegation in veto bargaining",
	journaltitle= aer,
	volume		= "111",
	number		= "12",
	pages		= "4046--87",
	langid		= "american",
	doi			= "10.1257/aer.20201817",
}

@article{KleinerMoldovanuStrack2021,
	author		= "Andreas Kleiner and Benny Moldovanu and Philipp Strack",
	year		= "2021",
	title		= "Extreme points and majorization",
	subtitle	= "Economic applications",
	journaltitle= ecta,
	volume		= "89",
	number		= "4",
	pages		= "1557--1593",
	langid		= "american",
	doi			= "10.3982/ECTA18312",
}

@article{Kolotilin2018,
	author		= "Anton Kolotilin",
	year		= "2018",
	title		= "Optimal information disclosure",
	subtitle	= "A linear programming approach",
	journaltitle= te,
	volume		= "13",
	number		= "2",
	pages		= "607--635",
	langid		= "american",
	doi			= "10.3982/TE1805",
}

@article{KolotilinCorraoWolitzky2023,
	author			= "Anton Kolotilin and Roberto Corrao and Alexander Wolitzky",
	year			= "2025",
	title			= "Persuasion and matching",
	subtitle	= "Optimal productive transport",
	journaltitle= jpe,
	volume		= "133",
	number		= "4",
	pages		= "1083--1400",
	langid		= "american",
	doi			= "10.1086/734095",
}

@article{KolotilinMylovanovZapechelnyuk2022,
	author		= "Anton Kolotilin and Tymofiy Mylovanov and Andriy Zapechelnyuk",
	year		= "2022",
	title		= "Censorship as optimal persuasion",
	journaltitle= te,
	volume		= "17",
	number		= "2",
	pages		= "561--585",
	langid		= "american",
	doi			= "10.3982/TE4071",
}

@article{KolotilinZapechelnyuk2025,
	author		= "Anton Kolotilin and Andriy Zapechelnyuk",
	year		= "2025",
	title		= "Persuasion meets delegation",
	journaltitle= ecta,
	volume		= "93",
	number		= "1",
	pages		= "195--228",
	langid		= "american",
	doi			= "10.3982/ECTA17051",
}

@article{LoertscherMuir2022,
	author		= "Simon Loertscher and Ellen V. Muir",
	year		= "2022",
	title		= "Monopoly pricing, optimal randomization, and resale",
	journaltitle= jpe,
	volume		= "130",
	number	= "3",
	pages		= "566--635",
	langid		= "american",
	doi			= "10.1086/717894",
}

@article{MilgromShannon1994,
	author		= "Paul Milgrom and Chris Shannon",
	year		= "1994",
	title		= "Monotone comparative statics",
	journaltitle= ecta,
	volume		= "62",
	number		= "1",
	pages		= "157--180",
	langid		= "american",
	doi			= "10.2307/2951479",
}

@article{MullerScarsini2006,
	author		= "Alfred Müller and Marco Scarsini",
	year		= "2006",
	title		= "Stochastic order relations and lattices of probability measures",
	journaltitle= "{SIAM} Journal on Optimization",
	volume		= "16",
	number		= "4",
	pages		= "1024--1043",
	langid		= "american",
	doi			= "10.1137/040611021",
}

@article{OnuchicRay2023,
	author		= "Paula Onuchic and Debraj Ray",
	year		= "2023",
	title		= "Conveying value via categories",
	journaltitle= te,
	volume		= "18",
	number		= "4",
	pages		= "1407--1439",
	langid		= "american",
	doi			= "10.3982/TE5026",
}

@article{Pratt1964,
	author		= "John W. Pratt",
	year		= "1964",
	title		= "Risk aversion in the small and in the large",
	journaltitle= ecta,
	volume		= "32",
	number		= "1--2",
	pages		= "122--136",
	langid		= "american",
	doi			= "10.2307/1913738",
}

@article{QuahStrulovici2009,
	author		= "John K.-H. Quah and Bruno Strulovici",
	year		= "2009",
	title		= "Comparative statics, informativeness, and the interval dominance order",
	journaltitle= ecta,
	volume		= "77",
	number		= "6",
	pages		= "1949--1992",
	langid		= "american",
	doi			= "10.3982/ECTA7583",
}

@article{RavidRoeslerSzentes2022,
	author		= "Doron Ravid and Anne-Katrin Roesler and Balázs Szentes",
	year		= "2022",
	title		= "Learning before trading",
	subtitle	= "On the inefficiency of ignoring free information",
	journaltitle= jpe,
	volume		= "130",
	number		= "2",
	pages		= "346--387",
	langid		= "american",
	doi			= "10.1086/717350",
}

@article{RayoSegal2010,
	author		= "Luis Rayo and Ilya Segal",
	year		= "2010",
	title		= "Optimal information disclosure",
	journaltitle= jpe,
	volume		= "118",
	number		= "5",
	pages		= "949--987",
	langid		= "american",
	doi			= "10.1086/657922",
}

@article{RoeslerSzentes2017,
	author		= "Anne-Katrin Roesler and Balázs Szentes",
	year		= "2017",
	title		= "Buyer-optimal learning and monopoly pricing",
	journaltitle= aer,
	volume		= "107",
	number		= "7",
	pages		= "2072--2080",
	langid		= "american",
	doi			= "10.1257/aer.20160145",
}

@article{Sinander2022,
	author		= "Ludvig Sinander",
	year		= "2022",
	title		= "The converse envelope theorem",
	journaltitle= ecta,
	volume		= "90",
	number		= "6",
	pages		= "2795--2819",
	langid		= "american",
	doi			= "10.3982/ECTA18119",
}

@article{Topkis1978,
	author		= "Donald M. Topkis",
	year		= "1978",
	title		= "Minimizing a submodular function on a lattice",
	journaltitle= "Operations Research",
	volume		= "26",
	number		= "2",
	pages		= "305--321",
	langid		= "american",
	doi			= "10.1287/opre.26.2.305",
}

@article{Vatter2022,
	author		= "Benjamin Vatter",
	year		= "2025",
	title		= "Quality disclosure and regulation",
	subtitle	= "Scoring design in {M}edicare {A}dvantage",
	journaltitle= ecta,
	volume		= "93",
	number		= "3",
	pages		= "959--1001",
	langid		= "american",
	doi			= "10.3982/ECTA21182",
}

@article{Whitmeyer2024,
	author		= "Mark Whitmeyer",
	year		= "2026",
	title		= "Making information more valuable",
	journaltitle= jpe,
	volume		= "134",
	number		= "3",
	pages		= "978--1016",
	langid		= "american",
	doi			= "10.1086/739326",
}

@article{Yoder2022,
	author		= "Nathan Yoder",
	year		= "2022",
	title		= "Designing incentives for heterogeneous researchers",
	journaltitle= jpe,
	volume		= "130",
	number		= "8",
	pages		= "2018--2054",
	langid		= "american",
	doi			= "10.1086/720072",
}

@article{Zapechelnyuk2020,
	author		= "Andriy Zapechelnyuk",
	year		= "2020",
	title		= "Optimal quality certification",
	journaltitle= aerins,
	volume		= "2",
	number		= "2",
	pages		= "161--176",
	langid		= "american",
	doi			= "10.1257/aeri.20190387",
}

@unpublished{Armstrong1995,
	author		= "Mark Armstrong",
	year		= "1995",
	title		= "Delegation and discretion",
	note		= "working paper, Nov 1995",
	langid		= "american",
	url			= "https://mpra.ub.uni-muenchen.de/id/eprint/17069",
}

@unpublished{CurelloSinanderWhitmeyer2025,
	author		= "Gregorio Curello and Ludvig Sinander and Mark Whitmeyer",
	year		= "2025",
	title		= "Outside options and risk attitude",
	note		= "working paper, 18 Sep 2025",
	langid		= "american",
	doi			= "10.48550/arXiv.2509.14732",
}

@unpublished{AugiasUhe2025,
	author		= "Victor Augias and Lina Uhe",
	year		= "2025",
	title		= "The economics of convex function intervals",
	note		= "working paper, 27 Oct 2025",
	langid		= "american",
	doi			= "10.48550/arXiv.2510.20907",
}

@unpublished{CheKimKojima2021,
	author		= "Yeon-Koo Che and Jinwoo Kim and Fuhito Kojima",
	year		= "2021",
	title		= "Weak monotone comparative statics",
	note		= "working paper, 25 Nov 2021",
	langid		= "american",
	doi			= "10.48550/arXiv.1911.06442",
}

@unpublished{Delara2025,
	author			= "Michel de Lara",
	year			= "2025",
	title			= "Increasing value of information implies separable utility",
	note			= "working paper, 14 Oct 2025",
	langid			= "american",
	doi			= "10.48550/arXiv.2510.11102",
}

@unpublished{GitmezMolavi2023,
	author			= "A. Arda Gitmez and Pooya Molavi",
	year			= "2023",
	title			= "Informational autocrats, diverse societies",
	note			= "working paper, 29 Aug 2023; nominally superseded by \textcite{GitmezMolaviSonin2025}",
	langid			= "american",
	doi				= "10.48550/arXiv.2203.12698",
}

@unpublished{GitmezMolaviSonin2025,
	author			= "A. Arda Gitmez and Pooya Molavi and Konstantin Sonin",
	year			= "2025",
	title			= "Propaganda and repression in diverse societies",
	note			= "working paper, 4 Sep 2025",
	langid			= "american",
}

@thesis{Holmstrom1977,
	author		= "Bengt Holmström",
	year		= "1977",
	title		= "On incentives and control in organizations",
	type		= "doctoral thesis",
	institution	= "Stanford University",
	langid		= "american",
}

@unpublished{HwangKimBoleslavsky2023,
	author			= "Ilwoo Hwang and Kyungmin Kim and Raphael Boleslavsky",
	year			= "2023",
	title			= "Competitive advertising and pricing",
	note			= "working paper, 21 Jan 2023",
	langid			= "american",
}

@unpublished{JewittQuigley2023,
	author			= "Ian Jewitt and Daniel Quigley",
	year			= "2023",
	title			= "Conjugate persuasion",
	note			= "working paper, Aug 2023",
	langid			= "american",
	url				= "https://ora.ox.ac.uk/objects/uuid:16eb7865-d71a-483f-8cb9-dc9789a255fa",
}

@unpublished{Kleiner2026,
	author			= "Andreas Kleiner",
	year			= "2026",
	title			= "Optimal delegation in a multidimensional world",
	note			= "working paper, Sep 2026",
	langid			= "american",
	doi				= "10.48550/arXiv.2208.11835",
}

@unpublished{Kolotilin2014,
	author			= "Anton Kolotilin",
	year			= "2014",
	title			= "Optimal information disclosure",
	subtitle		= "Quantity vs. quality",
	note			= "working paper, Mar 2014; nominally superseded by \textcite{Kolotilin2018}",
	langid			= "american",
	doi				= "10.2139/ssrn.2302448",
}

@unpublished{Kreutzkamp2023,
	author			= "Sophie Kreutzkamp",
	year			= "2023",
	title			= "Endogenous information acquisition in cheap-talk games",
	note			= "working paper, 18 Jan 2023",
	langid			= "american",
}

@unpublished{MenschMalik2023,
	author			= "Jeffrey Mensch and Komal Malik",
	year			= "2024",
	title			= "Posterior-mean separable costs of information acquisition",
	note			= "working paper, 6 Nov 2024",
	langid			= "american",
	doi				= "10.48550/arXiv.2311.09496",
}

@unpublished{MenschRavidThereze2026,
	author			= "Jeffrey Mensch and Doron Ravid and Joao Thereze",
	year			= "2026",
	title			= "Monopoly, quality, and costly information",
	note			= "working paper, 1 Apr 2026",
	langid			= "american",
}

@unpublished{QuahStrulovici2007extensions,
	author		= "John K.-H. Quah and Bruno Strulovici",
	year		= "2007",
	title		= "Comparative statics with the interval dominance order",
	subtitle	= "Some extensions",
	note		= "working paper, 9 Dec 2007",
	langid		= "american",
}

@unpublished{SinanderStrassen,
	author			= "Ludvig Sinander",
	year			= "2026",
	title			= "Not {S}trassen's theorem",
	note			= "working paper, 11 Nov 2026",
	langid			= "american",
	doi				= "",
}

@unpublished{Thereze2023adv,
	author			= "João Thereze",
	year			= "2026",
	title			= "Adverse selection and endogenous information",
	note			= "working paper, 2 Mar 2026",
	langid			= "american",
}

@unpublished{Zeng2026,
	author			= "Yishu Zeng",
	year			= "2026",
	title			= "Optimal delegation with lotteries",
	note			= "working paper, Oct 2026",
	langid			= "american",
}

@book{AliprantisBorder2006,
	author		= "Charalambos D. Aliprantis and Kim C. Border",
	year		= "2006",
	title		= "Infinite dimensional analysis",
	subtitle	= "A hitchhiker's guide",
	edition		= "3rd",
	publisher	= "Springer",
	address		= "Berlin",
	langid		= "american",
}

@book{AumannMaschler1995,
	author		= "Robert J. Aumann and Michael B. Maschler",
	year		= "1968/1995",
	title		= "Repeated games with incomplete information",
	edition		= "",
	publisher	= mit,
	address		= "Cambridge, MA",
	langid		= "american",
	note		= "circulated 1966--68, published 1995",
}

@book{Billingsley1995,
	author		= "Patrick Billingsley",
	year		= "1995",
	title		= "Probability and measure",
	edition		= "3rd",
	publisher	= "Wiley",
	address		= "New York, NY",
	langid		= "american",
}

@book{Billingsley1999,
	author		= "Patrick Billingsley",
	year		= "1999",
	title		= "Convergence of probability measures",
	edition		= "2nd",
	publisher	= "Wiley",
	address		= "New York, NY",
	langid		= "american",
}

@incollection{Blackwell1951,
	author			= "David Blackwell",
	year			= "1951",
	title			= "Comparison of experiments",
	editor			= "Jerzy Neyman",
	booktitle		= "Second {B}erkeley Symposium on Mathematical Statistics and Probability",
	pages			= "93--102",
	publisher		= "University of California Press",
	address			= "Berkeley, CA",
	langid			= "american",
}

@incollection{Holmstrom1984,
	author			= "Bengt Holmström",
	year			= "1984",
	title			= "On the theory of delegation",
	editor			= "Marcel Boyer and Richard E. Kihlstrom",
	booktitle		= "Bayesian Models in Economic Theory",
	pages			= "115--141",
	publisher		= "North-Holland",
	address			= "New York, NY",
	langid			= "american",
}

@book{Phelps2001,
	author		= "Robert R. Phelps",
	year		= "2001",
	title		= "Lectures on {C}hoquet's theorem",
	edition		= "2nd",
	publisher	= "Springer",
	series		= "Lecture Notes in Mathematics",
	address		= "Berlin",
	langid		= "american",
}

@book{ShakedShanthikumar2007,
	author		= "Moshe Shaked and J. George Shanthikumar",
	year		= "2007",
	title		= "Stochastic orders",
	series		= "Springer Series in Statistics",
	edition		= "",
	publisher	= "Springer",
	address		= "New York, NY",
	langid		= "american",
}

@book{Srivastava1998,
  author    = "Sashi Mohan Srivastava",
  year      = "1998",
  title     = "A course on {B}orel sets",
  series    = "Graduate Texts in Mathematics",
  edition   = "",
  publisher = "Springer",
	address		= "New York, NY",
	langid		= "american",
}

@book{Warga1972,
	author		= "Jack Warga",
	year		= "1972",
	title		= "Optimal control of differential and functional equations",
	edition		= "",
	publisher	= "Academic Press",
	address		= "Princeton, NJ",
	langid		= "american",
}


\end{document}